\documentclass[reqno]{amsart}%
\usepackage{enumerate}
\usepackage{tikz}
\providecommand{\U}[1]{\protect\rule{.1in}{.1in}}

\newtheorem{theorem}{Theorem}[section]

\newtheorem{corollary}[theorem]{Corollary}

\newtheorem{definition}[theorem]{Definition}

\newtheorem{lemma}[theorem]{Lemma}

\newtheorem{proposition}[theorem]{Proposition}
\newtheorem{remark}[theorem]{Remark}

\numberwithin{equation}{section}

\begin{document}
\title{On the Theory of Weak Turbulence for the Nonlinear Schr\"odinger Equation.}

\author[Miguel Escobedo]{Miguel Escobedo$^{1,2}$}
\address{$^{1}$Departamento de Matem\'aticas,
Universidad del Pa\'{\i}s Vasco,
Apartado 644, 48080 Bilbao, Spain. miguel.escobedo@ehu.es}
\address{$^{2}$Basque Center for Applied Mathematics
(BCAM), Alameda de Mazarredo 14, E--48009 Bilbao, Spain.}
\author[Juan J. L. Vel\'azquez]{Juan J. L. Vel\'azquez$^{3}$}
\address{$^{2}$Institute of Applied Mathematics, University of Bonn, Endenicher Allee 60,
53115 Bonn, Germany. E-mail: velazquez@iam.uni-bonn.de}



\thanks{This work has been supported by DGES Grant 2011-29306-C02-00, Basque Government Grant IT641-13, the
Hausdorff Center for Mathematics of the University of Bonn and the Collaborative Research Center {\it The Mathematics of Emergent Effects} (DFG SFB 1060, University of Bonn).
}

\begin{abstract}
We study the Cauchy problem for  a kinetic equation arising in the weak turbulence theory  for the cubic nonlinear Schr\"odinger equation. We define suitable concepts of weak and mild solutions and prove local and global well posedness results. Several qualitative properties of the solutions, including long time asymptotics, blow up results and  condensation in finite time are obtained. We also prove the existence of a family of solutions that exhibit pulsating behavior. 
\end{abstract}

\maketitle

\section{Introduction}

The name weak turbulence \index{weak turbulence}is often used in the physical literature to describe
the transfer of energy between different frequencies which takes place in
several nonlinear wave equations with weak nonlinearities.

The theory of weak turbulence has been extensively developed in the last
decades and many applications are available today. From the mathematical point
of view, the starting point of all the problems which can be studied using the
weak turbulence \index{weak turbulence}approach is a set of nonlinear wave equations with weak
nonlinearities. We will denote as $\varepsilon$ a small number which measures
the strength of the nonlinear interactions. If $\varepsilon$ is set to zero,
the problem becomes a linear system of wave equations which will be termed in
the following as the linearized problem. In order to simplify the presentation
we will restrict this introductory description of weak turbulence \index{weak turbulence}theory to
the cases in which the set of nonlinear equations is solved in the whole space
$x\in\mathbb{R}^{N},$ for $t\in\mathbb{R}$ and where the equations are
invariant under space and time translations. This allows to solve the
linearized problem using standard Fourier transforms, but in principle the
same ideas could be applied to nonhomogeneous systems. Suppose that the set of
magnitudes by the wave equations is denoted as $u=u\left(  t,x\right)  ,$
where $u\in\mathbb{C}^{L}$ or $u\in\mathbb{R}^{L}.$ Then, the linearized
problem admits solutions proportional to $e^{i\left(  kx+\omega t\right)  }$
with $\omega=\Omega\left(  k\right)  ,$ where $\Omega$ is a function, perhaps
multivalued, which is often referred as dispersion relation. In conservative
(vs. dissipative) problems, the function $\Omega\left(  \cdot\right)  $ is
real. A large class of initial data for the linearized problem can be
decomposed in Fourier modes $e^{ikx}$ and then, the solution of the linear
equation is given by the form:%
\[
u\left(  t,x\right)  =\int a\left(  t,k\right)  e^{i\left(  kx-\omega
t\right)  }d^{N}k\ \ \text{with\ \ }u\left(  x,0\right)  =u_{0}\left(
x\right)  =\int a\left(  0,k\right)  e^{ikx}d^{N}k
\]

A crucial quantity in weak turbulence \index{weak turbulence}theories is the density in the
wavenumber space $f\left(  t,k\right)  =\left\vert a\left(  t,k\right)
\right\vert ^{2}.$ Since $\Omega\left(  \cdot\right)  $ is real, the function
$f\left(  t,k\right)  $ is constant in time for the solutions of the
linearized problem. However, the dynamics of $f\left(  t,k\right)  $ becomes
nontrivial if the nonlinear terms in the original system of wave equations are
taken into account. Typically, due to the effect of resonances between some
specific wavenumbers $k,$ the function $f$ changes in time with a rate that
usually is a power law of the strength of the nonlinearities.

In principle, it is not possible to write a closed evolution equation for the
function $f\left(  t,k\right)  $ because the dynamics of the function
$a\left(  t,\cdot\right)  $ does not depend only on $\left\vert a\left(
t,k\right)  \right\vert $ but also in the phase of $a\left(  t,\cdot\right)
.$ However, one of the key hypothesis in weak turbulence \index{weak turbulence}theory is that for a
suitably chosen class of initial data $u_{0},$ it is possible to approximate
the evolution of $f$ by means of a kinetic equation. Moreover, in the limit of
weak nonlinear interactions, it is possible to give an interpretation of the
evolution of $f\left(  t,k\right)  $ by means of a particle model. The
evolution of $f$ is driven to the leading order in $\varepsilon$ by resonances
between linear modes with different values of $k.$ This resonance condition
can be given the interpretation of a collision between a number of particles,
which results in another group of particles. The numbers of particles involved
in these fictitious collisions depend on the dispersion relation for the
linearized problem as well as in the form of the nonlinear terms. The
resonance condition between modes can be understood as a condition for the
conservation of the moment and energy of the particles in the collision
process, assuming that $k$ and $\omega$ are given the interpretation of moment
and energy of the particles respectively. This makes this particle
interpretation for the effect of the nonlinearities particularly appealing.

The precise conditions that  allow to approximate the dynamics of wave
equations by the kinetic models of weak turbulence \index{weak turbulence}have not been obtained in a
fully rigorous manner. However, the physical derivations of the kinetic models
of weak turbulence assume the statistical independence of the phases of the
modes $a\left(  0,k\right)  .$ From this point of view the derivation of the
kinetic models of weak turbulence starting from wave equations have several
analogies with the formal derivations of the Boltzmann equation starting from
the dynamics of a particle system which can be found in the physical
literature. It is also worth mentioning that the theory of weak turbulence
assumes that the solutions of the underlying wave equation can be approximated
to the leading order by means of solutions of the linearized problem. However,
it is well known that effects induced by the nonlinear terms in the equation,
which can become relevant for some ranges of $k,$ can have a strong influence
in the distribution function $f$ (cf. \cite{DNPZ}, \cite{MMT}, \cite{NR},
\cite{ZDP}). 

The
collision kernels arising in the kinetic equations of weak turbulence \index{weak turbulence}theory
depend strongly on the details of the problem under consideration, as well as
the number of particles involved in the collisions. Nevertheless this approach
has been shown to be very fruitful in several physical problems, including
water surface and capillary water waves (cf. \cite{J}, \cite{ZFww},
\cite{ZF}), internal waves on density stratifications (cf. \cite{CZ}, \cite{LT}), nonlinear
optics (cf. \cite{DNPZ}) and\ waves in Bose-Einstein condensates, planetary
Rossby waves (cf. \cite{BN}, \cite{LG}), and vibrating elastic plates (cf.
\cite{DJR}) among others. Many more applications as well as an extensive
references list can be found in \cite{N}.

The first derivation of a kinetic model of weak turbulence \index{weak turbulence}was obtained, to
our knowledge, in \cite{Peierls} in the context of the study of phonon
interactions in anharmonic crystals. Derivations which take as starting point
wave equations arising in a large variety of physical contexts and yielding
analogous kinetic models were obtained in the 1960's in \cite{BS}, \cite{GK1},
\cite{Hass1}, \cite{Hass2}, \cite{New1}, \cite{New2}, \cite{SG}. There has
been a large increase in the number of applications of weak turbulence theory
in the last fifteen years. References about these more recent developments can
also be found in \cite{N}.

One of the most relevant mathematical results for the kinetic models of weak
turbulence was the discovery by V. E. Zakharov of a class of stationary power
law solutions for many models of weak turbulence.\index{weak turbulence} The earliest solutions of
this class can be found in \cite{Z1}, \cite{Z2}. Some of the solutions found
in \cite{Z1}, \cite{Z2} are just thermodynamic equilibria. These equilibria, 
take the form of a power law and they are usually termed as Rayleigh Jeans \index{Rayleigh Jeans}equilibria. 
However, some of the solutions found by V. E. Zakharov are characterized by the presence of
fluxes of some physical magnitude (typically number of particles of energy)
between different regions of the space $k.$ From this point of view they have
very strong analogies with the Kolmogorov solutions for the theory of
turbulence in fluids, although in this last case the nonlinearities of the underlying
problem (namely Euler's equations) are much stronger than in the case of weak
turbulence. Due to this analogy, power law stationary solutions of kinetic
models of weak turbulence \index{weak turbulence}which describe fluxes between different regions of
the phase space are usually termed as Kolmogorov-Zakharov \index{Kolmogorov-Zakharov}solutions. Some of
the earliest examples of such type of solutions can be found in
\cite{KatsKant}, \cite{ZF}. Several other examples can be found in \cite{ZF}.
Methods to study linear stability for the Kolmogorov-Zakharov \index{Kolmogorov-Zakharov}solutions in
several models of weak turbulence can be found in \cite{BZ}.

One of the simplest, and most widely studied models of weak turbulence, \index{weak turbulence}is the
one in which the underlying nonlinear wave equation is the nonlinear
Schr\"{o}dinger equation \index{Schr\"{o}dinger equation}\index{Schr\"{o}dinger equation}(cf.   \cite{DNPZ}, \cite{N}, \cite{Zbook} and references therein).
More precisely, the function $u=u\left(  x,t\right)
\in\mathbb{C}$ satisfies:%
\begin{equation}
i\partial_{t}u=-\Delta u+\varepsilon\left\vert u\right\vert ^{2}%
u\ \ ,\ \ u\left(  0,\cdot\right)  =u_{0} \label{M1E1}%
\end{equation}
We will assume by definiteness that this problem is considered in $\left(
t,x\right)  \in\mathbb{R}^{3}\times\mathbb{R}.$ If $\varepsilon=0$, equation
(\ref{M1E1}) becomes the linear Schr\"{o}dinger equation \index{Schr\"{o}dinger equation}whose solutions are
given by integrals of the form $\int_{\mathbb{R}^{3}}e^{i\left(  kx-\omega
t\right)  }d\mu\left(  k\right)  $ for a large class of measures $\mu
\in\mathcal{M}_{+}\left(  \mathbb{R}^{3}\right)  $ with $\omega=k^{2}.$ Weak
turbulence theory suggests that, for a choice of initial data $u_{0}$
according to a suitable class of probability measures homogeneous in space,
the dynamics of the solutions of (\ref{M1E1}) for small $\varepsilon$ can be
obtained by means of the kinetic equation (cf. \cite{DNPZ}):
\begin{eqnarray}
\partial_{t}F_{1}  &  = & \frac{\varepsilon^{2}}{\pi}\iiint_{\left(
\mathbb{R}^{3}\right)  ^{3}}\delta\left(  k_{1}+k_{2}-k_{3}-k_{4}\right)
\delta\left(  \omega_{1}+\omega_{2}-\omega_{3}-\omega_{4}\right)
\cdot\label{M1E2}\\
&&\hskip 2cm   \cdot\left[  F_{3}F_{4}\left(  F_{1}+F_{2}\right)  -F_{1}F_{2}\left(
F_{3}+F_{4}\right)  \right]  dk_{2}dk_{3}dk_{4}\ \nonumber
\end{eqnarray}
with $\omega=\Omega\left(  k\right)  =k^{2},$ and where from now on we will
use the notation $F_{\ell}=F\left(  t,k_{\ell}\right)  ,\ \ell=1,2,3,4.$
Equation (\ref{M1E2}) is one of the most important examples of kinetic
model arising in weak turbulence theory. \index{weak turbulence}
It allows to understand some of the solutions
of the nonlinear equation (\ref{M1E1}) in terms of particle collisions.  Equation (\ref{M1E2}) has been extensively used to study problems in optical turbulence and Bose Einstein \index{Bose Einstein}condensation \index{condensation}(cf. \cite{DNPZ}, \cite{JPR}, \cite{LLPR}, \cite{LY1, LY2, LY3},  \cite{N},  \cite{P}, \cite{SK1, SK2}, \cite{Spohn}, \cite{Zbook}). 

In this paper we only consider the isotropic case of  equation
(\ref{M1E2}). The main reason for such a restriction is that it is not possible to 
give a meaning to the
operator in the right-hand side of (\ref{M1E2}) in a unique way if $F$
contains Dirac masses, as it was noticed in in \cite{Lu1} for an equation closely related, namely  the Nordheim  \index{Nordheim} equation which will be discussed in Section \ref{Nordheim} .

Suppose therefore that we look for solutions of (\ref{M1E2}) with the form:
$F\left(  t,k\right)  =f\left(  t,\omega\right)  ,$ $\omega=k^{2}.$ Then,
after rescaling the time variable $t$ in order to eliminate from the equation
some constants we obtain that $f$ solves:%
\begin{equation}
\partial_{t}f_{1}=\iint  W\left[  \left(  f_{1}+f_{2}\right)  f_{3}%
f_{4}-\left(  f_{3}+f_{4}\right)  f_{1}f_{2}\right]  d\omega_{3}d\omega
_{4},\ \ t>0  \label{E1}
\end{equation}
where $f_{k}=f\left(  t,\omega_{k}\right)  ,\ k=1,2,3,4$ and
\begin{eqnarray}
W=\frac{\min\left\{  \sqrt{\omega_{1}},\sqrt{\omega_{2}},\sqrt{\omega_{3}%
},\sqrt{\omega_{4}}\right\}  }{\sqrt{\omega_{1}}}\ \ ,\ \ \ \omega_{2}%
=\omega_{3}+\omega_{4}-\omega_{1} \label{E2}
\end{eqnarray}
We are interested in
the initial value problem associated to (\ref{E1}), (\ref{E2}). We will then
assume that (\ref{E1}), (\ref{E2}) is solved with initial value $f_{0}\left(
\omega\right)  ,$ i.e.:%
\begin{equation}
f\left(  0,\omega\right)  =f_{0}\left(  \omega\right)  \geq0\ \ ,\ \ \omega
\geq0 \label{E2a}%
\end{equation}

The function $f$ is not a particle density in the space of frequencies
$\omega,$ due to the presence of some nontrivial jacobians. A magnitude that
is proportional to the density of particles in the space $\left\{  \omega
\geq0\right\}  $ is the function $g$ defined by means of:%
\begin{equation}
g= \sqrt \omega \, f \label{E2b}%
\end{equation}

Then $g$ solves:
\begin{equation}
\partial_{t}g_{1}=\iint\Phi\left[\left(\frac{g_{1}}{\sqrt{\omega_{1}}%
}+\frac{g_{2}}{\sqrt{\omega_{2}}}\right)  \frac{g_{3}g_{4}}{\sqrt{\omega
_{3}\omega_{4}}}-\left(  \frac{g_{3}}{\sqrt{\omega_{3}}}+\frac{g_{4}}%
{\sqrt{\omega_{4}}}\right)  \frac{g_{1}g_{2}}{\sqrt{\omega_{1}\omega_{2}}%
}\right]  d\omega_{3}d\omega_{4} \label{S2E1}%
\end{equation}
where:%
\begin{equation}
\Phi=\min\left\{  \sqrt{\omega_{1}},\sqrt{\omega_{2}},\sqrt{\omega_{3}}%
,\sqrt{\omega_{4}}\right\}  \label{S2E1a}%
\end{equation}%
\begin{equation}
g\left(  0,\omega\right)  =g_{0}\left(  \omega\right)  =\sqrt \omega \, f_0(\omega )\geq0\ \ ,\ \ \omega\geq0 \label{S2E1b}%
\end{equation}

The integrations in (\ref{E1}), (\ref{S2E1}) are always restricted to the
region 
$$D(\omega  _1)=\left\{\omega_{3}\geq0,\ \omega_{4}\geq0; \,\,\omega _3+\omega _4\ge \omega _1\right\},\,\,\, \omega _1\ge 0.$$ 
In order to simplify the notation, we will assume in all the remainder of the paper
that $\Phi=0$ in $\mathbb{R}^{4}\setminus\left[  0,\infty\right)  ^{4}.$

As we already mentioned, the equations (\ref{M1E1}) and (\ref{M1E2}) have been used to study different questions related with the Bose Einstein \index{Bose Einstein}condensation.\index{condensation} In particular,  on the basis of physical arguments, as well as formal asymptotics and numerical simulations, it has been  accepted that, under some conditions, the solutions of equation (\ref{M1E2}) would contain, at least after some time, a Dirac mass at the origin (cf.  \cite{DNPZ}, \cite{LY1, LY2, LY3}, \cite{P}, \cite{SK1, SK2}), a question that we will consider in some detail later. Let us only say here that, in the mentioned literature,  this property is considered  as  reminiscent of the Bose Einstein condensation phenomena. Then, with  some abuse of language, we will refer  to the solutions of (\ref{M1E2}) that have a Dirac measure at the origin as solutions with a condensate.  \index{condensate}

\subsection{Main results}
\label{results}
The main goals of this paper are to study the Cauchy problem associated to  (\ref{S2E1})-(\ref{S2E1b}) (or equivalently (\ref{E1})-(\ref{E2a})),  to obtain some of the qualitative behavior of the solutions and describe their long time asymptotic behavior.  Although we  prove
several well-posedness results for  initial data with infinite number of particles, i.e.  $\int g_{0}=\infty$, we  restrict most of the analysis in this paper  to the case where$\int g_{0}<\infty.$ 

The stationary solutions of the equation (\ref{S2E1})-(\ref{S2E1b}) have been studied in the physics literature. On the other hand, the Cauchy problem for the Nordheim equation,  (cf. Subsection \ref{Nordheim}),  has been studied in \cite{EMV1, EMV2}, \cite{EV1, EV2}, \cite{Lu1, Lu2, Lu3}.

 The main results that are proved in this paper are the following:
 \subsubsection{Existence results \& stationary solutions.}\hfill\break
 
 \textbf{1.}- Existence of bounded mild solutions \index{mild solution} of (\ref{E1})-(\ref{E2a}), i.e. solutions 
of  the equation in the sense of the integral formulation which results from
the Duhamel's formula. These solutions are  locally defined in time for a large class
of bounded initial data (cf. Theorem \ref{localExBounded}).

\textbf{2.}- Existence of global weak solutions \index{weak solution} of (\ref{E1})-(\ref{E2a}), i.e.
solutions in the sense of distributions, globally defined in time for a large
class of initial measures with total finite mass (cf. Theorem
\ref{globalWeakSol}). 

\textbf{3.} Characterisation of the weak stationary solutions of (\ref{S2E1})-(\ref{S2E1a}) with finite mass as the
Dirac masses $g_{stat}=\delta_{R},$ with $R\geq0$ (cf.  Theorem \ref{StatIsot}).

 \subsubsection{Qualitative behavior of the solutions.}\hfill\break
 
 \textbf{4.} Characterization of the long time asymptotics of the weak solutions \index{weak solution} of
(\ref{S2E1})-(\ref{S2E1b}) with finite mass $\int g\left(  t,d\omega\right)
<\infty$  in terms of the properties of the initial data $g_{0}$ (cf. 
Theorem \ref{Asympt}). This result states that $g\left(  t,\cdot\right)
\rightharpoonup\delta_{R_{\ast}}$ with $R_{\ast}=\inf\left[
\operatorname*{supp}g_{0}\right]  .$ 

\textbf{5.} Transport of the energy of the system  towards $\omega\rightarrow\infty.$ 
as $t\rightarrow \infty,$ for a large class of weak solutions \index{weak solution} of (\ref{S2E1})-(\ref{S2E1b}) with
finite energy the  (cf. Corollary \ref{AsEnergy}).

\textbf{6.} Optimal upper estimates for the rate of transport of the
energy towards large values of $\omega$ for the solutions described in the
point 8 (cf. Proposition \ref{coarsening}).

\textbf{7.}
If $g$ is a weak solution \index{weak solution} of (\ref{S2E1})-(\ref{S2E1b})
globally defined in time and if  we define $R_{\ast}=\inf\left[  \operatorname*{supp}%
g_{0}\right]  =0,$ it is possible to show the following alternative: Either
$\int_{\left\{  0\right\}  }g\left(  t,d\omega\right)>0$ for $t>t_{\ast},$
or the mass of $g$ approaches towards $\omega=0$ in a "pulsating manner" \index{pulsating}
(cf. Theorem \ref{AsympOsc}).

\textbf{8.} (Blow-up in finite time). 
Existence of solutions of (\ref{E1})-(\ref{E2a}) with initial data such that $\left\Vert
f_{0}\right\Vert _{L^{\infty}\left(  \left[  0,\infty\right)  \right)
}<\infty$ for which $\lim\sup_{t\rightarrow T}\left\Vert f\left(
t,\cdot\right)  \right\Vert _{L^{\infty}\left(  \left[  0,\infty\right)
\right)  }=\infty$ for some $T<\infty$ (cf. Theorem
\ref{BU}).

\textbf{9.} There exist initial data $g_{0}$ such that the first alternative stated
in the point 6 takes place (cf. Theorem \ref{weakCond}). Moreover, there
exist also initial data $g_{0}$ such that the second alternative stated in the
point (7) takes place (cf. Theorem \ref{globOsc}).

It may be useful to precise the meaning of the pulsating solutions \index{pulsating}
mentioned in the point (7). Suppose that $R_{\ast}=0.$ Then, either there
exists $t^{\ast}\in\left(  0,\infty\right)  $ such that $\int_{\left\{  0\right\}  }g\left(  t,d\omega\right)  >0$ for $t>t_{\ast}$, or,
alternatively, during most of the times, there exists $\rho=\rho\left(  t\right)  >0$ such that $\frac{1}%
{\rho\left(  t\right)  }g\left(  t,\frac{\cdot}{\rho\left(  t\right)
}\right)  $ is close to the Dirac mass $M\delta_{1}\left(  \cdot\right)  $ in
the weak topology. It turns out that the function $\rho\left(  t\right)  $
does not change its position continuously in general.\ On the
contrary, we study in detail a class of initial data $g_{0}$ for which
the function $\rho\left(  t\right)  $ can be shown to change by means of some
jumps which take place at specific times (cf. Chapter 4). At those times $g$ ceases being
close to a Dirac mass and its mass spreads among a large set of values
$\omega.$ After a transient time $g\left(  t,\cdot\right)  $ concentrates its
mass again close to a Dirac mass whose position is closer to the origin than
the previous one. This process is iterated infinitely many times as $g\left(
t,\cdot\right)  $ approaches to $M\delta_{0}\left(  \cdot\right)  $ as
$t\rightarrow\infty.$ During all this evolution the solution satisfies
$\int_{\left\{  0\right\}  }g\left(  t,d\omega\right)  =0.$

The dynamics of the solutions of  (\ref{S2E1}) has been extensively studied in the physics literature by means of physical simulations and formal asymptotic arguments. Two of the  main questions that have been discussed are the finite time condensation \index{condensation}and the asymptotic behaviour as $t\to +\infty$.

When considering the long time behavior for the solutions of kinetic equations of type (\ref{S2E1}) (or equivalently (\ref{E1})) for gravity waves, it was seen  in \cite{Hass2} that the Dirac masses where stationary solutions of the corresponding weak turbulence \index{weak turbulence}equation, but it was suggested that generic solutions should converge to the Rayleigh Jeans \index{Rayleigh Jeans}equilibria. In \cite{DNPZ}, using dimensional and scaling arguments, the authors  indicate that as $t$ tends to $\infty$,  the solutions of  (\ref{S2E1}), in presence of a condensate, converge towards a Dirac mass located at the origin containing the total number of particles. The same result is also described in \cite{P} as well as in \cite{Zbook},  where finite time condensation \index{condensation}is also briefly described. Different scenarios  of  condensate formation in finite or infinite time have been discussed in
\cite{KSS},\cite{LY1, LY2, LY3}, \cite{S}.
It is now widely believed that a generic mechanism for the formation of a condensate  is the one described  in \cite{JPR}, \cite{LLPR}, \cite{SK1}, \cite{SK2}. In these papers, using numerical simulations and asymptotic arguments, it is shown  how the condensate arises by means of a finite time blow up \index{blow up}of the solutions of the equation (\ref{S2E1}). Near the blow up \index{blow up} point the particle distribution $f$ is given by a self similar solution of the second kind. Additional information about these issues may be found in See \cite{N} Chapter 15.

Our results  of points (3) and (4) above prove that all the weak solutions \index{weak solution} of (\ref{S2E1})-(\ref{S2E1a}), without flux at the origin and with finite mass, converge, in the weak sense of measures, to a Dirac delta containing all the mass of the solution and located at a suitable value of $\omega $. This asymptotic behavior can take place either with the formation of condensate in finite time or without it.  The results in points (7) and (9) show that both possibilities can take place.

It has been shown in  \cite{Z3} that the equation (\ref{E1}), (\ref{E2}) has two Kolmogorov-Zakharov \index{Kolmogorov-Zakharov}solutions namely  $f_1=\omega^{-7/6}$ and $f_2=\omega ^{-3/2}$. The first, $f_1$,  has a constant flux of particles towards the origin and a zero flux of energy. The second or $f_2$, has a constant flux of energy towards large values of $\omega$ and zero flux of particles, (cf  \cite{DNPZ} and \cite{Z3}).  However, since the integrals that define the fluxes for $f_2$ are divergent, we will only consider in this paper the solution $f_1$.
It is nevertheless interesting to notice that  the finite mass, zero flux weak 
solutions obtained in points (4) and (8)
present both fluxes, in the directions predicted by these two  Kolmogorov-Zakharov \index{Kolmogorov-Zakharov}solutions. 
This behavior  indicates a
tendency of these solutions  to transport particles
towards small values of $k.$ Since the energy is conserved in the particle
collisions and the energy is reduced if $\left\vert k\right\vert $ is reduced,
the inward particle flux must be compensated with an outward particle flux.
The tendency to have these particle and energy fluxes will be made precise in
this paper for isotropic solutions. We will derive in Section
\ref{Transfer} estimates for the rate of transfer of energy towards infinity
for particle distributions satisfying $\int g_{0}\left(  d\omega\right)
<\infty.$ Some heuristic estimates about the characteristic time scales for
the transfer of energy for arbitrary distributions $g_{0}$ are also discussed in Section \ref{Transfer2}.

It is interesting to compare the results concerning energy fluxes towards
infinity with those obtained for the nonlinear Schr\"{o}dinger equation \index{Schr\"{o}dinger equation}
obtained in the articles \cite{CKSTT}, \cite{K1} and \cite{K2}. The results in those
 papers show the existence of
solutions of the NLS equation for which the energy can be transferred to 
large values of the frequency. The results in our paper concern just the
kinetic approximation of the NLS equation, but prove rigorously the escape of
the energy towards large values of $\left\vert k\right\vert $ as
$t\rightarrow\infty.$ In the absence of a precise rigorous results relating
the solutions of the NLS equation and the corresponding kinetic theory of weak
turbulence it is hard to precise the connections between both types of
results. \

\subsection{Relation with the Nordheim  \index{Nordheim} equation}
\label{Nordheim}
Several  of the methods and results  in this paper bear some analogies
with those in \cite{EV1} for the Nordheim  \index{Nordheim}
equation. This equation, arises in the study of rarefied gases of
quantum particles and takes the following form for homogeneous isotropic
distributions:%
\begin{align}
\partial_{t}f_{1}  &  =\iint W\left[  \left(  1+f_{1}+f_{2}\right)
f_{3}f_{4}-\left(  1+f_{3}+f_{4}\right)  f_{1}f_{2}\right]  d\omega_{3}%
d\omega_{4}\label{E5}\\
W  &  =\frac{\min\left\{  \sqrt{\omega_{1}},\sqrt{\omega_{2}},\sqrt{\omega
_{3}},\sqrt{\omega_{4}}\right\}  }{\sqrt{\omega_{1}}} \label{E6}%
\end{align}

Equation (\ref{E5}), (\ref{E6}) differs from (\ref{E1}), (\ref{E2}) in the onset of the quadratic terms $f_3f_4-f_1f_2$. These additional terms come from the use of the Bose Einstein statistics, instead of  the classical, in the counting of the particles in the collisions.

The connection between the two equations (\ref{E1}),  (\ref{E5}) has been  already noticed by several authors,  (cf. for example \cite{BZ}, \cite{LLPR},  \cite{N} and the references therein). It is suggested in particular that the cubic terms in (\ref{E5}) should be dominant  in the limit of large occupation numbers.  
As a matter of fact, the results for the solutions of (\ref{E1}) (\ref{E5}) in points (1), (2) and (8) of Section (\ref{results}) above,  have also been obtained for the solutions of (\ref{E5}), (\ref{E6}) in \cite{EV1} and their proofs are very similar. 

\subsection{Plan of the paper.}

The plan of this paper is the following. Section 2 contains the
definition of the different concepts of solutions of (\ref{S2E1}%
)-(\ref{S2E1b}) which will be used in this paper, the relation
between then and several well posedness results. Two main concepts of solutions will be used in this paper,
namely mild solutions \index{mild solution}  (i.e. solutions in the sense of the variation of
constants formula), and weak solutions,\index{weak solution} which satisfy the equation in the
sense of distributions. This Section also contains a complete classification of the stationary solutions with finite mass. 
We end Section  2 explaining how the Kolmogorov-Zakharov \index{Kolmogorov-Zakharov}solutions \index{Kolmogorov-Zakharov}and some related ones, fit into the framework of this paper.
Section 3 describes several qualitative properties of the  solutions of (\ref{S2E1})-(\ref{S2E1b}). We obtain a classification of all the possible long time asymptotics of  one of the types of weak solutions that we have defined, namely those with interacting condensate and finite mass. We also derive some estimates for the rate of transfer of the energy towards large values of $\omega $. We also prove a refined theorem concerning the long time asymptotic of the solutions which shows that if the mass concentrates at the origin, either a condensate appears in finite time or the solution exhibits a behavior that we will denote as pulsating. \index{pulsating} Finally we also prove in this chapter that solutions blow up \index{blow up}in finite time.
Section 4 contains a construction of a large family of initial data for which the solutions do not condensate in finite time but exhibit pulsating behavior as $t$ goes to infinity. This is one of the most technical parts of the paper.
Section 5 gives a description by means of heuristic arguments  of how is the precise transfer of mass and energy for the different types of weak solutions considered in this paper. This chapter contains a list of open problems suggested by the results of this paper.
Section 6 contains several results that are basically adaptation of previous results obtained in \cite{EV1}.

\section{Well-Posedness Results}

We consider in this article different types of solutions for the equations  (\ref{E1}), (\ref{E2}) and
(\ref{S2E1}), (\ref{S2E1a}). Some of them are measured valued solutions that do not solve the equations in classical form. Therefore, we need suitable concepts of generalized solutions for these equations. 

An analysis of the  physical literature shows that two different types of solutions of (\ref{E1}), (\ref{E2}) have been implicitly considered,  depending on the interaction that is assumed between the condensate and the remaining particles of the system. For example, in \cite{LLPR}, \cite{JPR} it is assumed that there is no difference in the interactions between particles, whether they are or not in the condensate.  On the contrary, for the Kolmogorov-Zakharov \index{Kolmogorov-Zakharov}solutions, and related ones, it is implicitly assumed that the particles in the condensate do not interact with the remaining particles of the system. The difference between the two situations may be seen as reminiscent of the case of diffusive  particles reaching a boundary, where either reflecting or absorbing boundary conditions can be imposed. Motivated by these two different situations we define two different types of weak solutions.

We will also use in this paper mild solutions of (\ref{E1}), (\ref{E2}) and (\ref{S2E1}), (\ref{S2E1a}). They will be shown to be a subclass of  weak solutions, and several of their properties will be studied later in this paper. Mild solutions of a regularized version of (\ref{S2E1}), (\ref{S2E1a}) will be used as technical tool in one of the existence result for weak solutions.

We define:
\begin{eqnarray}
&& \Phi_{\sigma}=\min\left\{  \sqrt{\left(  \omega_{k}-\sigma\right)  _{+}}, k=1, 2, 3, 4.\right\},\,\, \hbox{for}\,\sigma>0;
\,\,\, \Phi_{0}=\Phi. \label{Z2E2}
 \end{eqnarray}
and  introduce for further reference the analogous of equation (\ref{S2E1})
with the collision kernel $\Phi$ replaced by $\Phi_{\sigma}:$
\begin{eqnarray}
&&\partial_{t}g_{1}=\iint \Phi_{\sigma}\left[  \left(  \frac{g_{1}}%
{\sqrt{\omega_{1}}}+\frac{g_{2}}{\sqrt{\omega_{2}}}\right)  \frac{g_{3}g_{4}%
}{\sqrt{\omega_{3}\omega_{4}}}-\left(  \frac{g_{3}}{\sqrt{\omega_{3}}}%
+\frac{g_{4}}{\sqrt{\omega_{4}}}\right)  \frac{g_{1}g_{2}}{\sqrt{\omega
_{1}\omega_{2}}}\right]  d\omega_{3}d\omega_{4}. \label{Z2E2a}%
\end{eqnarray}

\subsection{Weak solutions with interacting condensate.} \index{interacting condensate}
\label{weakinteracting}

The definition of weak solution that we introduce in this Section is similar to the one given for the Nordheim equation \index{Nordheim} in \cite{Lu1}.

We denote  as $\mathcal{M}_{+}\left(  \left[  0,\infty
\right)  \right)  $ the set of\ nonnegative Radon measures in $\left[
0,\infty\right)  .$ Given $\rho\in \mathbb{R}$, we will denote as $\mathcal{M}_{+}\left(
\left[  0,\infty\right)  :\left(  1+\omega\right)  ^{\rho}\right)  $ the set
of measures $\mu\in\mathcal{M}_{+}\left(  \left[  0,\infty\right)  \right)  $
such that:%
\[
||\mu || _{ \rho  }=\sup _{ R>0 } \frac {1} {(1+R)^\rho }\frac{1}{R}\int_{\frac{R}{2}}^{R}\mu\left(  d\omega\right)< +\infty.
\]
We will use also the functional space $L_{+}^{\infty}\left(
\mathbb{R}_{+}:\sqrt{\omega}\left( 1+\omega\right)^{\rho-\frac{1}{2}
}\right)  $ which is the space of locally, nonnegative, bounded functions $h$
such that:%
\[
||h|| _{L_{+}^{\infty}\left(
\mathbb{R}_{+}:\sqrt{\omega}\left( 1+\omega\right)^{\rho-\frac{1}{2}
}\right)   }=\sup_{\omega>0}\frac{h\left(  \omega\right)  }{\sqrt{\omega}\left(
1+\omega\right)  ^{\rho-\frac{1}{2}}}<\infty\ 
\]

\begin{remark}
\label{distWeak}We will use also at several points in the arguments 
that $\mathcal{M}_{+}\left(  \left[  0,\infty\right)  :\left(
1+\omega\right)  ^{\rho}\right)  $ endowed with the weak topology is
metrizable. We will denote the corresponding metric as $dist_{\ast}.$ The dependence of this distance in $\rho$ will not be written explicitly,
since it will not play any role in the arguments.
\end{remark}

We will assume in several of the results below that $\rho<-\frac{1}{2}.$ This
exponents corresponds to the slowest rate of decay which allows to define the
integrals appearing in the definitions of the solutions in classical form. The
typical behaviour of a function in $L_{+}^{\infty}\left(  \mathbb{R}_{+}%
:\sqrt{\omega}\left(  1+\omega\right)  ^{\rho-\frac{1}{2}}\right)  $ with
$\rho=-\frac{1}{2}$ is $g\left(  \omega\right)  \sim\frac{1}{\sqrt{\omega}}$
as $\omega\rightarrow\infty$ or equivalently $f\left(  \omega\right)
\sim\frac{1}{\omega}$ as $\omega\rightarrow\infty.$ This corresponds to
thermal equilibrium.

Notice that the range of powers $\rho<-\frac{1}{2}$ includes some functions
$g_{in}$ such that $\int_{0}^{\infty}g_{in}\left(  d\omega\right)  =\infty.$
We will impose additional constraints on $\rho$ if we need to consider
solutions either with finite number of particles or finite energy.

\begin{definition}
\label{weakSolution} 
\index{interacting condensate} 
Given $\sigma\geq0,$ and $\rho<-\frac{1}{2}$ we will say
that the measure valued function $g\in C\left(  \left[  0,T\right)  :\mathcal{M}_{+}\left(  \left[
0,\infty\right)  :\left(  1+\omega\right)  ^{\rho}\right)  \right)  $ is a
weak solution \index{weak solution}  of (\ref{Z2E2a}) with interacting condensate  and with initial datum $g_{0}\in\mathcal{M}
_{+}\left(  \left[  0,\infty\right)  :\left(  1+\omega\right)  ^{\rho}\right)
$ if the following identity holds for any test function $\varphi\in C_{0}%
^{2}\left(  \left[  0,T\right)  \times\left[  0,\infty\right)  \right)  :$%
\begin{align}
 & \int_{\left[  0,\infty\right)  }g\left(  t_{\ast},\omega\right)
\varphi\left(  t_{\ast},\omega\right)  d\omega  -\int_{\left[  0,\infty\right)
}g_{0}\varphi\left(  0,\omega\right)  d\omega
  =\int_{0}^{t_{\ast}}\int_{\left[  0,\infty\right)  }g\partial_{t}\varphi
d\omega dt+ \label{Z2E1}\\
&  +\int_{0}^{t_{\ast}}\iiint_{\left[  0,\infty\right)^3 }\frac{g_{1}g_{2}g_{3}%
\Phi_{\sigma}}{\sqrt{\omega_{1}\omega_{2}\omega_{3}}}\times \nonumber \\
&\hskip 2cm  \times\left[  \varphi\left(
\omega_{1}+\omega_{2}-\omega_{3}\right)  +\varphi\left(  \omega_{3}\right)
-\varphi\left(  \omega_{1}\right)  -\varphi\left(  \omega_{2}\right)  \right]
d\omega_{1}d\omega_{2}d\omega_{3}dt \nonumber
\end{align}
for any $t_{\ast}\in\left[  0,T\right)  .$
\end{definition}

\begin{remark}
The reason to assume the condition $\rho<-\frac{1}{2},$ 
is  to guarantee that  the integrals on the right-hand side of (\ref{Z2E1}) converge  for large values of $\omega$.
\end{remark}

It is important to prove that the nonlinear operator in the last term of
(\ref{Z2E1}) is well defined for $g\in C\left(  \left[  0,T\right)
:\mathcal{M}_{+}\left(  \left[  0,\infty\right)  :\left(  1+\omega\right)
^{\rho}\right)  \right)  $, $\rho <-1/2$. This is a consequence of the following Lemma.

\begin{lemma}
\label{Cont}Suppose that $\varphi\in C_{0}^{2}\left(  \left[  0,\infty\right)
\right)  .$ Then, for all $\sigma \in [0, 1]$, the functions defined by means of:%
\begin{equation}
\Delta_{\varphi,\sigma}\left(  \omega_{1},\omega_{2},\omega_{3}\right)
=\frac{\Phi_{\sigma}}{\sqrt{\omega_{1}\omega_{2}\omega_{3}}}\left[
\varphi\left(  \omega_{1}+\omega_{2}-\omega_{3}\right)  +\varphi\left(
\omega_{3}\right)  -\varphi\left(  \omega_{1}\right)  -\varphi\left(
\omega_{2}\right)  \right]  \ \label{Z2E6}%
\end{equation}
with $\Phi_{\sigma}$ as in (\ref{Z2E2}), are uniformly continuous 
on compact subsets of  $\left(
\omega_{1},\omega_{2},\omega_{3}\right)  \in\left[  0,\infty\right)
^{3},\ \omega_{4}=\omega_{1}+\omega_{2}-\omega_{3}.$
\end{lemma}

\begin{proof}
We just need to prove uniform continuity of $\Delta_{\varphi,\sigma}$ near the
boundary of $\left[  0,\infty\right)  ^{3}.$ We first derive an uniform
estimate for $\Delta_{\varphi,\sigma}$ near the lines $\Gamma_{k,j}=\left\{
\omega_{k}=\omega_{j}=0\ \ ,\ \ k\neq j\right\}  $ for $k,j=1,2,3.$ Using that
$\Phi_{\sigma}\leq\Phi_{0}$ we obtain:%
\begin{equation}
\Delta_{\varphi,\sigma}\left(  \omega_{1},\omega_{2},\omega_{3}\right)
\leq\Delta_{\varphi,0}\left(  \omega_{1},\omega_{2},\omega_{3}\right)
\label{Z2E3}%
\end{equation}

Notice that the line $\Gamma_{1,2}$ is contained in the set $\left\{
\omega_{3}\geq\left(  \omega_{1}+\omega_{2}\right)  \right\}  $ where
$\Phi_{\sigma}\leq\Phi_{0}=0.$ Then
$\Delta_{\varphi,\sigma}$ vanish in the set
$\left\{  \omega_{3}\geq\left(  \omega_{1}+\omega_{2}\right)  \right\}  $ and
then, they are uniformly continuous there. We now examine the lines
$\Gamma_{1,3},\ \Gamma_{2,3}.$ Due to the symmetry of the functions
$\Delta_{\varphi,\sigma}$ we can restrict to study to the line $\Gamma_{1,3}.$
Suppose that $\omega_{1}\leq\omega_{2},$ $\omega_{3}\leq\left(  \omega
_{1}+\omega_{2}\right)  .$ We expand the term between brackets in (\ref{Z2E6})
using Taylor at the point $\left(  \omega_{1},\omega_{2},\omega_{3}\right)
=\left(  0,\omega_{2},0\right)  $. Then:%
\begin{align}
&\Delta_{\varphi,0}\left(  \omega_{1},\omega_{2},\omega_{3}\right)  \leq
\frac{\Phi_{0}}{\sqrt{\omega_{1}\omega_{2}\omega_{3}}}\times \\
&\hskip 2cm \left(  \left\vert
\varphi^{\prime}\left(  \omega_{2}\right)  \left(  \omega_{1}-\omega
_{3}\right)  -\varphi^{\prime}\left(  0\right)  \left(  \omega_{1}-\omega
_{3}\right)  \right\vert +C\left[  \left(  \omega_{1}\right)  ^{2}+\left(
\omega_{3}\right)  ^{2}\right]  \right)\nonumber 
\end{align}
for some constant $C>0$ depending only on $\varphi$ and its derivatives,
whence:
\[
\Delta_{\varphi,0}\left(  \omega_{1},\omega_{2},\omega_{3}\right)  \leq
\frac{C\Phi_{0}}{\sqrt{\omega_{1}\omega_{2}\omega_{3}}}\left(  \omega
_{2}\left\vert \omega_{1}-\omega_{3}\right\vert +\left(  \omega_{1}\right)
^{2}+\left(  \omega_{3}\right)  ^{2}\right)
\]

We now estimate $\Phi_{0}$ by $\min\left\{  \sqrt{\omega_{1}},\sqrt{\omega
_{3}}\right\}  .$ Then:%
\begin{equation}
\Delta_{\varphi,0}\left(  \omega_{1},\omega_{2},\omega_{3}\right)  \leq
\frac{C}{\max\left\{  \sqrt{\omega_{1}},\sqrt{\omega_{3}}\right\}
\sqrt{\omega_{2}}}\left[  \omega_{2}\left\vert \omega_{1}-\omega
_{3}\right\vert +\left(  \max\left\{  \omega_{1},\omega_{3}\right\}  \right)
^{2}\right]  \label{Z2E3a}%
\end{equation}

Combining (\ref{Z2E3}), (\ref{Z2E3a}) we obtain the desired uniform
convergence of the functions $\Delta_{\varphi,\sigma}$ in a neighbourhood of
the line $\Gamma_{1,3}$, including the origin $\left(  \omega_{1},\omega
_{2},\omega_{3}\right)  =\left(  0,0,0\right)  .$ It only remains to obtain
uniform continuity of the functions $\Delta_{\varphi,\sigma}$ in a
neighbourhood of the planes $\Pi_{k}=\left\{  \omega_{k}=0\right\}
,\ k=1,2,3.$ This follows from the fact that after removing the neighbourhoods
of the lines $\Gamma_{k,j},$ $k,j=1,2,3$ indicated above, we can restrict the
analysis to points where at least two of the coordinates $\omega
_{j},\ j=1,2,3,$ are bounded from below. Suppose that the remaining coordinate
is $\omega_{\ell}.$ The function $\frac{\Phi_{\sigma}}{\sqrt{\omega_{\ell}}}$
is then uniformly continuous in a neighbourhood of $\Pi_{k}$ and the result follows.
\end{proof}

\subsection{Weak solutions with non interacting condensate.}\index{non interacting condensate} \index{interacting condensate}
\label{weaknoninteracting}
Although, most of the results that we will obtain in this paper are for weak solutions with interacting condensate, we wish to have a precise  functional framework which allows to treat solutions that behave like the  Kolmogorov-Zakharov \index{Kolmogorov-Zakharov}solutions for small values of $\omega $.  
We recall that  the  Kolmogorov-Zakharov \index{Kolmogorov-Zakharov}solution 
$f_{s}\left(
\omega\right)  = K\omega ^{-7/6}$, $g_{s}\left(
\omega\right)  =K\omega ^{-2/3}$.

Some general properties of the solutions defined in this Section will be discussed in Section \ref{fluxes}.

\begin{definition}
\label{weakSolutionNI}Given $\sigma\geq0,$ and $\rho<-\frac{1}{2}$ we will say
that the measure valued function $g\in C\left(  \left[  0,T\right)  :\mathcal{M}_{+}\left(  \left[
0,\infty\right)  :\left(  1+\omega\right)  ^{\rho}\right)  \right)  $ is a
weak solution \index{weak solution} of (\ref{Z2E2a}) with non interacting condensate and with initial datum $g_{0}\in\mathcal{M}%
_{+}\left(  \left[  0,\infty\right)  :\left(  1+\omega\right)  ^{\rho}\right)
$ if the following identity holds for any test function $\varphi\in C_{0}%
^{2}\left(  \left[  0,T\right)  \times\left[  0,\infty\right)  \right)  :$%
\begin{align}
 & \int_{\left[  0,\infty\right)  }g\left(  t_{\ast},\omega\right)
\varphi\left(  t_{\ast},\omega\right)  d\omega  -\int_{\left[  0,\infty\right)
}g_{0}\varphi\left(  0,\omega\right)  d\omega
  =\int_{0}^{t_{\ast}}\int_{\left[  0,\infty\right)  }g\partial_{t}\varphi
d\omega dt+ \label{Z2E1NI}\\
&  +\int_{0}^{t_{\ast}}\iiint_{\left(  0,\infty\right)^3 }\frac{g_{1}g_{2}g_{3}%
\Phi_{\sigma}}{\sqrt{\omega_{1}\omega_{2}\omega_{3}}}\times \nonumber \\
&\hskip 2cm  \times\left[  \varphi\left(
\omega_{1}+\omega_{2}-\omega_{3}\right)  +\varphi\left(  \omega_{3}\right)
-\varphi\left(  \omega_{1}\right)  -\varphi\left(  \omega_{2}\right)  \right]
d\omega_{1}d\omega_{2}d\omega_{3}dt \nonumber
\end{align}
for any $t_{\ast}\in\left[  0,T\right)  .$
\end{definition}

\begin{remark}
\label{parentesis}
The difference between Definitions \ref{weakSolution} and \ref{weakSolutionNI} is extremely subtle. The domain of integration in the triple integral in the right hand side of (\ref{Z2E1}) is $[0, +\infty)^3$, while the corresponding domain of integration in  (\ref{Z2E1NI}) is $(0, +\infty)^3$. The reason for this difference is to avoid, in the second case,  any interaction between the particles in any possible condensate and the remaining particles. Notice also that as long as $\int  _{ \{0\} }g(t, d\omega ) = 0$, both Definitions are equivalent
\end{remark}

\subsection{Mild solutions.}

We will use two different concepts of mild solutions,  \index{mild solution} namely, measured valued
mild solutions and bounded mild solutions. The idea behind these definitions
is that they satisfy the equations in the sense of Duhamel's formula. 

The reason to introduce the bounded mild solutions is due to our interest to prove  finite time blow up \index{blow up}in $L^\infty$ norm  for some of these solutions. Our interest in the measured valued mild solutions is twofold. First, we use them as technical tools in order to obtain global existence of weak solutions with interacting condensate in the sense of Definition (\ref{weakSolution}).  
\index{weak solution} \index{interacting condensate}On the other hand, we construct a family of measured valued mild solutions with some peculiar asymptotic behavior (pulsating solutions) in Chapter 4. \index{pulsating}

We need to define some auxiliary functions.
\begin{lemma}
\label{LAg}Let $\rho<-\frac{1}{2}.$ Suppose that either $g\in\mathcal{M}%
_{+}\left(  \left[  0,\infty\right)  :\left(  1+\omega\right)  ^{\rho}\right)
$ and $\sigma>0$ or $g\in L_{+}^{\infty}\left(  \mathbb{R}_{+}:\sqrt{\omega
}\left(  1+\omega\right)  ^{\rho-\frac{1}{2}}\right)  $ and $\sigma\geq0.$
Then $A_{\sigma}\left(  \omega_{1}\right)  $ defined by means of%
\begin{equation}
A_{\sigma}\left(  \omega_{1}\right)  =-\iint \Phi_{\sigma}\left[
\frac{2g_{2}g_{3}}{\sqrt{\omega_{1}\omega_{2}\omega_{3}}}-\frac{g_{3}g_{4}%
}{\sqrt{\omega_{1}\omega_{3}\omega_{4}}}\right]  d\omega_{3}d\omega_{4}
\label{S7E3n}%
\end{equation}
where $\omega_{2}=\omega_{3}+\omega_{4}-\omega_{1}$ defines a continuous
function in $\left[  0,\infty\right)  .$\ Moreover, we have:%
\begin{equation}
A_{\sigma}\left(  \omega_{1}\right)  \geq0\ \ ,\ \ \omega_{1}\in\left[
0,\infty\right)  \label{S7E5n}%
\end{equation}

\end{lemma}

\begin{remark}
Given $g\in\mathcal{M}_{+}\left(  \left(  1+\omega\right)  ^{\rho}\right)  ,$
we define the measure:
$$\iint \Phi_{\sigma}\left[  \frac{2g_{2}g_{3}}
{\sqrt{\omega_{1}\omega_{2}\omega_{3}}}\right]  d\omega_{3}d\omega_{4}$$
 by
means of its action over a test function $\varphi\in C_{0}^{2}\left(  \left[
0,\infty\right)  \right)  ,$ namely:%
\begin{align}
&\int_0^\infty \varphi\left(  \omega_{1}\right)  \iint \Phi_{\sigma}\left[
\frac{2g_{2}g_{3}}{\sqrt{\omega_{1}\omega_{2}\omega_{3}}}\right]  d\omega
_{3}d\omega_{4}d\omega_{1}=\\
&\hskip 2cm =\int_0^\infty \iint\Phi_{\sigma}\left[  \frac{2g_{2}%
g_{3}\varphi\left(  \omega_{3}+\omega_{4}-\omega_{2}\right)  }{\sqrt
{(\omega _3+\omega _4-\omega _2)\omega_{2}\omega_{3}}}\right]  d\omega_{3}d\omega_{4}d\omega_{2}\nonumber
\end{align}

\end{remark}

\begin{proof}
Suppose first that $g\in\mathcal{M}_{+}\left(  \left[  0,\infty\right)
:\left(  1+\omega\right)  ^{\rho}\right)  $ and $\sigma>0.$ The function
$\frac{\Phi}{\sqrt{\omega_{1}}}$ is continuous for $\omega_{1}>0$ and $\left(
\omega_{1},\omega_{2}\right)  \in\mathbb{R}_{+}^{2}.$ Moreover, $\Phi_{\sigma
}=0$ if $\min\left\{  \omega_{1},\omega_{2},\omega_{3}\right\}  \leq\sigma.$
Therefore, each of the terms $\frac{2\Phi_{\sigma}}{\sqrt{\omega_{1}\omega
_{2}\omega_{3}}},$ $\frac{\Phi_{\sigma}}{\sqrt{\omega_{1}\omega_{3}\omega_{4}%
}}$ are bounded. Then, $\frac{2\Phi_{\sigma}g_{2}g_{3}}{\sqrt{\omega_{1}%
\omega_{2}\omega_{3}}}$ and $\frac{\Phi_{\sigma}g_{3}g_{4}}{\sqrt{\omega
_{1}\omega_{3}\omega_{4}}}$ Radon measures in $\mathbb{R}_{+}^{2}%
$.\ Therefore, (\ref{S7E3n}) defines a continuous function in $\left\{
\omega_{1}>0\right\}  .$ Notice that the continuity at $\omega_{1}=0$ follows
from the fact that $\Phi_{\sigma}=0$ if $\omega_{1}\leq\sigma.$ Convergence of
the integrals for large values of $\omega_{3},\omega_{4}$ are a consequence of
the fact that $g\in\mathcal{M}_{+}\left(  \left[  0,\infty\right)  :\left(
1+\omega\right)  ^{\rho}\right)  $ with $\rho<-\frac{1}{2}.$

In order to prove (\ref{S7E5n}) we rewrite $A\left(  \omega_{1}\right)$ using  that:%
\begin{eqnarray}
&&\iint\Phi\frac{2g_{2}g_{3}}{\sqrt{\omega_{1}\omega_{2}\omega_{3}}}%
d\omega_{3}d\omega_{4}=\iint\Phi\frac{g_{2}g_{3}}{\sqrt{\omega_{1}%
\omega_{2}\omega_{3}}}d\omega_{3}d\omega_{4}+\label{S7E4n}\\
&&\hskip 6 cm +\iint\Phi\frac{g_{2}g_{4}%
}{\sqrt{\omega_{1}\omega_{2}\omega_{4}}}d\omega_{3}d\omega_{4} \nonumber
\end{eqnarray}

We now use the change of variables $\omega_{2}=\omega_{3}+\omega_{4}%
-\omega_{1},$ $d\omega_{2}=d\omega_{4}$ in the first integral and $\omega
_{2}=\omega_{3}+\omega_{4}-\omega_{1},$ $d\omega_{2}=d\omega_{3}$ in the
second one. Then, replacing the variable $\omega_{2}$ by $\omega_{4}$ in the
first resulting integral and $\omega_{2}$ by $\omega_{3}$ in the second, we
obtain that the integral in (\ref{S7E4n}) becomes
\[
\iint\frac{g_{3}g_{4}}{\sqrt{\omega_{1}\omega_{3}\omega_{4}}}\Psi
d\omega_{3}d\omega_{4},\,\,\,\Psi=\Psi_{1}+\Psi_{2}%
\]
where:
\begin{align*}
\Psi_{1}  &  =\chi_{\left\{  \omega_{3}\geq\omega_{4}\right\}  }\chi_{\left\{
\omega_{3}\geq\omega_{1}\right\}  }\sqrt{\left(  \omega_{1}+\omega_{4}%
-\omega_{3}\right)  _{+}}+\chi_{\left\{  \omega_{3}\leq\omega_{4}\right\}
}\chi_{\left\{  \omega_{3}\geq\omega_{1}\right\}  }\sqrt{\omega_{1}}+\\
&  +\chi_{\left\{  \omega_{3}\geq\omega_{4}\right\}  }\chi_{\left\{
\omega_{3}\leq\omega_{1}\right\}  }\sqrt{\omega_{4}}+\chi_{\left\{  \omega
_{3}\leq\omega_{4}\right\}  }\chi_{\left\{  \omega_{3}\leq\omega_{1}\right\}
}\sqrt{\omega_{3}}%
\end{align*}%
\begin{align*}
\Psi_{2}  &  =\chi_{\left\{  \omega_{3}\leq\omega_{4}\right\}  }\chi_{\left\{
\omega_{4}\geq\omega_{1}\right\}  }\sqrt{\left(  \omega_{1}+\omega_{3}%
-\omega_{4}\right)  _{+}}+\chi_{\left\{  \omega_{3}\geq\omega_{4}\right\}
}\chi_{\left\{  \omega_{4}\geq\omega_{1}\right\}  }\sqrt{\omega_{1}}+\\
&  +\chi_{\left\{  \omega_{3}\leq\omega_{4}\right\}  }\chi_{\left\{
\omega_{4}\leq\omega_{1}\right\}  }\sqrt{\omega_{3}}+\chi_{\left\{  \omega
_{3}\geq\omega_{4}\right\}  }\chi_{\left\{  \omega_{4}\leq\omega_{1}\right\}
}\sqrt{\omega_{4}}%
\end{align*}
Notice that $\Psi\geq\Phi$, whence (\ref{S7E5n}) follows.

If $g\in L_{+}^{\infty}\left(  \mathbb{R}_{+}:\sqrt{\omega}\left(
1+\omega\right)  ^{\rho-\frac{1}{2}}\right)  $ we argue in a similar way.
Notice that we can then assume that $\sigma=0$ because the boundedness of $g$
by $C\sqrt{\omega}$ for small $\omega$ implies the convergence of the
integrals in (\ref{S7E3n}).
\end{proof}

\begin{lemma}
\label{LQg}Let $\rho<-\frac{1}{2}.$ Suppose that either $g\in\mathcal{M}%
_{+}\left(  \left[  0,\infty\right)  :\left(  1+\omega\right)  ^{\rho}\right)
$ and $\sigma>0$ or $g\in L_{+}^{\infty}\left(  \mathbb{R}_{+}:\sqrt{\omega
}\left(  1+\omega\right)  ^{\rho-\frac{1}{2}}\right)  $ and $\sigma\geq0.$\\
The following formula defines mappings $\mathcal{O}_{\sigma}%
:\mathcal{M}_{+}\left(  \left[  0,\infty\right)  :\left(  1+\omega\right)
^{\rho}\right)  \rightarrow\mathcal{M}_{+}\left(  \left[  0,\infty\right)
:\left(  1+\omega\right)  ^{\rho}\right)  $ and $\mathcal{O}_{\sigma}%
:L_{+}^{\infty}\left(  \mathbb{R}_{+}:\sqrt{\omega}\left(  1+\omega\right)
^{\rho-\frac{1}{2}}\right)  \rightarrow L_{+}^{\infty}\left(  \mathbb{R}%
_{+}:\sqrt{\omega}\left(  1+\omega\right)  ^{\rho-\frac{1}{2}}\right)  $
respectively:%
\begin{equation}
\mathcal{O}_{\sigma}\left[  g\right]  =\iint\Phi_{\sigma}\frac{g_{2}%
g_{3}g_{4}}{\sqrt{\omega_{2}\omega_{3}\omega_{4}}}d\omega_{3}d\omega
_{4}\ ,\ \omega_{2}=\omega_{3}+\omega_{4}-\omega_{1} \label{Z2E7}%
\end{equation}
where the action of the measure $\mathcal{O}_{\sigma}\left[  g\right]  $
acting over a test function $\varphi\in C_{0}\left(  \left[  0,\infty\right)
\right)  $ is given by:%
\begin{equation}
\left\langle \mathcal{O}_{\sigma}\left[  g\right]  ,\varphi\right\rangle
=\iiint\Phi_{\sigma}\frac{g_{2}g_{3}g_{4}}{\sqrt{\omega_{2}\omega
_{3}\omega_{4}}}\varphi\left(  \omega_{1}\right)  d\omega_{3}d\omega
_{4}d\omega_{1} \label{Z2E8}%
\end{equation}

If $g\in L_{+}^{\infty}\left(  \mathbb{R}_{+}:\sqrt{\omega}\left(
1+\omega\right)  ^{\rho-\frac{1}{2}}\right)  $ we can define directly
$\mathcal{O}_{\sigma}\left[  g\right]  $ by means of the integration in
(\ref{Z2E7}) for any $\sigma\geq0.$
\end{lemma}

\begin{proof}
If $g\in\mathcal{M}_{+}\left(  \left[  0,\infty\right)  :\left(
1+\omega\right)  ^{\rho}\right)  $ and $\sigma>0$ the function $\frac
{\Phi_{\sigma}}{\sqrt{\omega_{2}\omega_{3}\omega_{4}}}$ is bounded and
continuous in $\left[  0,\infty\right)  ^{3}.$ We can then compute the
integral (\ref{Z2E8}) which must be understood as:%
\begin{equation}
\left\langle \mathcal{O}_{\sigma}\left[  g\right]  ,\varphi\right\rangle
=\iiint\Phi_{\sigma}\frac{g_{2}g_{3}g_{4}}{\sqrt{\omega_{2}\omega
_{3}\omega_{4}}}\varphi\left(  \omega_{3}+\omega_{4}-\omega_{2}\right)
d\omega_{2}d\omega_{3}d\omega_{4} \label{Z3E1}%
\end{equation}

Since the function $\varphi\left(  \omega_{3}+\omega_{4}-\omega_{1}\right)  $
is not compactly supported, we must examine carefully the convergence of the
integral in (\ref{Z3E1}). Not convergence problems arise for small $\omega
_{2},\ \omega_{3},\ \omega_{4}$ due to the cutoff in $\Phi_{\sigma}.$ Since
$g\in\mathcal{M}_{+}\left(  \left[  0,\infty\right)  :\left(  1+\omega\right)
^{\rho}\right)  $ with $\rho<-1$ and $\varphi$ is bounded, we then obtain
convergence of the integral in (\ref{Z3E1}). It remains to prove that
$\mathcal{O}_{\sigma}\left[  g\right]  \in\mathcal{M}_{+}\left(  \left[
0,\infty\right)  :\left(  1+\omega\right)  ^{\rho}\right)  .$ To this end we
consider an increasing sequence of test functions $\left\{  \varphi_{n}\left(
\cdot\right)  \right\}  \subset C_{0}\left(  \left[  0,\infty\right)  \right)
$ and such that $\lim_{n\rightarrow\infty}\varphi_{n}\left(  \omega\right)
=\left(  1+\omega\right)  ^{\rho}$ uniformly for $\omega$ in compact sets of
$\left[  0,\infty\right)  .$ Notice that, since the support of the functions
$\varphi_{n}\left(  \cdot\right)  $ is contained in $\left\{  \omega
\geq0\right\}  ,$ we can restrict the integral in (\ref{Z3E1}) to the set
where $\omega_{2}\leq\left(  \omega_{3}+\omega_{4}\right)  .$ Due to the
symmetry under the permutation $\omega_{3}\longleftrightarrow\omega_{4}$ we
can assume that $\omega_{3}\leq\omega_{4}.$ Using the test functions
$\varphi=\varphi_{n}$ as well as the fact that $\Phi_{\sigma}\leq\sqrt
{\omega_{2}}$ we can estimate the integral on the right-hand side of
(\ref{Z3E1}) as:%
\[
C\int g_{2}d\omega_{2}\iint\frac{g_{3}g_{4}}{\sqrt{\omega_{3}\omega_{4}}%
}\left(  1+\left(  \omega_{4}\right)  ^{\rho-1}\right)  d\omega_{3}d\omega
_{4}<\infty
\]
whence $\mathcal{O}_{\sigma}\left[  g\right]  \in\mathcal{M}_{+}\left(
\left[  0,\infty\right)  :\left(  1+\omega\right)  ^{\rho}\right)  $ using
Monotone Convergence.

If $g\in L_{+}^{\infty}\left(  \mathbb{R}_{+}:\sqrt{\omega}\left(
1+\omega\right)  ^{\rho-\frac{1}{2}}\right)  $ and $\sigma\geq0$ we obtain
convergence for the integral in (\ref{Z2E7}), even if $\sigma=0,$ in the
region where $\omega_{2},\ \omega_{3},\ \omega_{4}$ are smaller than one using
the fact that $g\left(  \omega\right)  \leq C\sqrt{\omega}$ for $\omega\geq0.$
The convergence of the integrals for large values of $\omega$ follows from the
fact that $g\left(  \omega\right)  \leq C\omega^{-\rho},\ \rho<-1,$ if
$\omega\geq1.$ In order to prove that $\mathcal{O}_{\sigma}\left[  g\right]
\in L_{+}^{\infty}\left(  \mathbb{R}_{+}:\sqrt{\omega}\left(  1+\omega\right)
^{\rho-\frac{1}{2}}\right)  $ we argue as follows. Since $\Phi_{\sigma}%
\leq\sqrt{\omega_{1}}$ and $g\left(  \omega\right)  \leq C\sqrt{\omega}$ for
$\omega\geq0$ we obtain:%
\[
\mathcal{O}_{\sigma}\left[  g\right]  \left(  \omega_{1}\right)
=C\sqrt{\omega_{1}}\iint\frac{g_{3}g_{4}}{\sqrt{\omega_{3}\omega_{4}}%
}d\omega_{3}d\omega_{4}=C\sqrt{\omega_{1}}\left(  \int\frac{g\left(
\omega\right)  }{\sqrt{\omega}}d\omega\right)  ^{2}%
\]
whence the estimate $\mathcal{O}_{\sigma}\left[  g\right]  \left(  \omega
_{1}\right)  \leq C\sqrt{\omega_{1}}$ for $\omega_{1}\leq1$ follows. In order
to obtain estimates for $\omega_{1}\geq1,$ notice that $\min\left\{
\omega_{3},\omega_{4}\right\}  \geq\frac{\omega_{1}}{2}.$ We can assume,
without loss of generality that $\min\left\{  \omega_{3},\omega_{4}\right\}
=\omega_{3}$ by means of a symmetrization argument$.$ Then, using
$\Phi_{\sigma}\leq\sqrt{\omega_{2}}$ in the region where $\omega_{2}\leq
\omega_{1}$ and $\Phi_{\sigma}\leq\sqrt{\omega_{1}}$ if $\omega_{2}>\omega
_{1}$ we would obtain:%
\begin{equation}
\mathcal{O}_{\sigma}\left[  g\right]  \left(  \omega_{1}\right)  \leq
\iint_{\left\{  \omega_{2}\leq\omega_{1}\right\}  }\frac{g_{2}g_{3}g_{4}}%
{\sqrt{\omega_{3}\omega_{4}}}d\omega_{3}d\omega_{4}+\sqrt{\omega_{1}}
\iint_{\left\{  \omega_{2}>\omega_{1}\right\}  }\frac{g_{2}g_{3}g_{4}}%
{\sqrt{\omega_{2}\omega_{3}\omega_{4}}}d\omega_{3}d\omega_{4} \label{Z3E2}%
\end{equation}

In order to estimate the first integral on the right we symmetrize the
integral to have $\omega_{4}\geq\omega_{3}.$ In this integral we have then
$\omega_{4}$ of order $\omega_{1}\geq1.$ Therefore:%
\[
\iint_{\left\{  \omega_{2}\leq\omega_{1}\right\}  }\frac{g_{2}g_{3}g_{4}%
}{\sqrt{\omega_{3}\omega_{4}}}d\omega_{3}d\omega_{4}\leq\frac{C}{\left(
\omega_{1}\right)  ^{\frac{1}{2}+\rho}}\iint_{\left\{  \omega_{2}\leq
\omega_{1}\right\}  }\frac{g_{2}g_{3}}{\sqrt{\omega_{3}}}d\omega_{3}%
d\omega_{4}%
\]

We can now change the variable $\omega_{4}$ to $\omega_{2}$ and integrate also
$\omega_{3}$ in the whole space. Both integrals are finite and we obtain the
estimate:%
\[
\iint_{\left\{  \omega_{2}\leq\omega_{1}\right\}  }\frac{g_{2}g_{3}g_{4}%
}{\sqrt{\omega_{3}\omega_{4}}}d\omega_{3}d\omega_{4}\leq\frac{C}{\left(
\omega_{1}\right)  ^{\frac{1}{2}+\rho}}\ \ ,\ \ \omega_{1}\geq1
\]

We now estimate the second integral in (\ref{Z3E2}). We can assume also by
means of a symmetrization argument that $\omega_{4}\geq\omega_{3}.$ Then
$\omega_{4}\geq C\omega_{1}.$ We also replace the integration in $\omega_{4}$
by the integration in $\omega_{2}.$ Therefore
\begin{align*}
&  \sqrt{\omega_{1}}\iint_{\left\{  \omega_{2}>\omega_{1}\right\}  }%
\frac{g_{2}g_{3}g_{4}}{\sqrt{\omega_{2}\omega_{3}\omega_{4}}}d\omega
_{3}d\omega_{4}\\
&  \leq\frac{C\sqrt{\omega_{1}}}{\left(  \omega_{1}\right)  ^{\frac{1}{2}%
+\rho}}\iint_{\left\{  \omega_{2}>\omega_{1}\right\}  }\frac{g_{2}g_{3}%
}{\sqrt{\omega_{2}\omega_{3}}}d\omega_{3}d\omega_{2}\\
&  \leq\frac{C}{\left(  \omega_{1}\right)  ^{\rho}}\frac{1}{\left(  \omega
_{1}\right)  ^{\rho-\frac{1}{2}}}\int\frac{g_{3}}{\sqrt{\omega_{3}}}%
d\omega_{3}\leq\frac{C}{\left(  \omega_{1}\right)  ^{2\rho-\frac{1}{2}}}%
\end{align*}

Therefore, we obtain, combining these estimates:%
\[
\mathcal{O}_{\sigma}\left[  g\right]  \left(  \omega_{1}\right)  \leq
C\min\left\{  \sqrt{\omega_{1}},\frac{1}{\left(  \omega_{1}\right)  ^{\frac
{1}{2}+\rho}}\right\}
\]
whence $\mathcal{O}_{\sigma}\left[  g\right]  \in L_{+}^{\infty}\left(
\mathbb{R}_{+}:\sqrt{\omega}\left(  1+\omega\right)  ^{\rho-\frac{1}{2}%
}\right)  .$
\end{proof}

\subsubsection{Measured valued mild solutions. \index{mild solution} }

We now define measured valued mild solutions for $\sigma >0$.

\begin{definition}
\label{MeasMildSol}Let$\ \rho<-\frac{1}{2},\ \sigma>0$ and $T\in\left(
0,\infty\right]  $.  Given  $g_{in}\in\mathcal{M}_{+}\left(  \left[  0,\infty
\right)  :\left(  1+\omega\right)  ^{\rho}\right)  $ we will say that $g\in
C\left(  \left[  0,T\right]  :\mathcal{M}_{+}\left(  \left[  0,\infty\right)
:\left(  1+\omega\right)  ^{\rho}\right)  \right)  $ is a measured valued mild
solution of (\ref{Z2E2a}) with initial value $g\left(  \cdot,0\right)
=g_{in}$ if the following identity holds in the sense of measures:%
\begin{equation}
g\left(  t,\cdot\right)  =g_{in}\left(  \cdot\right)  \exp\left(  -\int
_{0}^{t}A_{\sigma}\left(  s,\cdot\right)  ds\right)  +\int_{0}^{t}\exp\left(
-\int_{s}^{t}A_{\sigma}\left(  \xi,\cdot\right)  d\xi\right)  \mathcal{O}%
_{\sigma}\left[  g\right]  \left(  s,\cdot\right)  ds \label{Z3E3}%
\end{equation}
for $0\leq t<T,$ where $A\left(  \cdot,s\right)  $ is defined as in Lemma
\ref{LAg} for each $g\left(  \cdot,s\right)  $ and $\mathcal{O}\left[
g\right]  \left(  \cdot,s\right)  $ is defined as in Lemma \ref{LQg} for each
$g\left(  \cdot,s\right)  .$
\end{definition}

\begin{remark}
The main reason to restrict the definition to $\sigma>0$ is because for
$\sigma=0$ the operator $\mathcal{O}_{\sigma}\left[  g\right]  $ cannot be
defined as a finite measure for arbitrary measures $g\in\mathcal{M}_{+}\left(
\left[  0,\infty\right)  :\left(  1+\omega\right)  ^{\rho}\right)  .$ The
solutions defined here will be used as an auxiliary tool in
order to construct global weak solutions \index{weak solution} of (\ref{Z2E2a}) for $\sigma=0.$ 
\end{remark}

\begin{remark}
We understand by solutions in the sense of measures, solutions defined by
means of their action over a test function i.e. (\ref{Z3E3}) means:%
\begin{eqnarray}
&&\int\varphi\left(  \omega\right)  g\left(  t,d\omega\right)  = \int
\varphi\left(  \omega\right)  \exp\left(  -\int_{0}^{t}A_{\sigma}\left(
s,\omega\right)  ds\right)  g_{in}\left(  d\omega\right) + \label{Z3E4}\\
&&\hskip 3cm +\int_{0}^{t}%
\int\varphi\left(  \omega\right)  \exp\left(  -\int_{s}^{t}A_{\sigma}\left(
\xi,\omega\right)  d\xi\right)  \mathcal{O}_{\sigma}\left[  g\right]  \left(
s,d\omega\right)  ds\nonumber %
\end{eqnarray}
for any $\varphi\in C_{0}\left(  \left[  0,\infty\right)  \right)  .$ Notice
that, since $A_{\sigma}\left(  \cdot,s\right)  $ is a continuous function all
the terms in (\ref{Z3E4}) are well defined.
\end{remark}

\subsubsection{Bounded mild solutions. \index{mild solution} }\hfill\break

We now define solutions in the space of functions $C\left(  \left[  0,T\right]  :L_{+}^{\infty
}\left(  \mathbb{R}_{+}:\sqrt{\omega}\left(  1+\omega\right)  ^{\rho-\frac
{1}{2}}\right)  \right)  .$

\begin{definition}
\label{BoundMildSol}Let$\ \rho<-\frac{1}{2},\ \sigma\geq0$ and $T\in\left(
0,\infty\right]  $. Given any  $g_{in}\in L_{+}^{\infty}\left(  \mathbb{R}_{+}%
:\sqrt{\omega}\left(  1+\omega\right)  ^{\rho-\frac{1}{2}}\right)  $ we
say that $g\in C\left(  \left[  0,T\right]  :L_{+}^{\infty}\left(
\mathbb{R}_{+}:\sqrt{\omega}\left(  1+\omega\right)  ^{\rho-\frac{1}{2}%
}\right)  \right)  $ is a bounded mild solution \index{mild solution} of (\ref{Z2E2a}) with initial
value $g\left(  \cdot,0\right)  =g_{in}$ if the following identity holds in
the sense of measures:%
\begin{equation}
g\left(  t,\cdot\right)  =g_{in}\left(  \cdot\right)  \exp\left(  -\int
_{0}^{t}A_{\sigma}\left(  s,\cdot\right)  ds\right)  +\int_{0}^{t}\exp\left(
-\int_{s}^{t}A_{\sigma}\left(  \xi,\cdot\right)  d\xi\right)  \mathcal{O}%
_{\sigma}\left[  g\right]  \left(  s,\cdot\right)  ds\ \label{Z3E5}%
\end{equation}
for $0\leq t<T,$ where $A\left(  s,\cdot\right)  $ is defined as in Lemma
\ref{LAg} for each $g\left(  \cdot,s\right)  $ and $\mathcal{O}\left[
g\right]  \left(  s,\cdot\right)  $ is defined as in Lemma \ref{LQg} for each
$g\left(  s,\cdot\right)  .$
\end{definition}

\begin{remark}
Notice that, differently from Definition \ref{MeasMildSol}, in Definition
\ref{BoundMildSol} we allow $\sigma$ to take the value $0.$
\end{remark}

\subsubsection{Relation between the different concepts of solution.}

The relation between the different concepts of solution mentioned above is
described in the following result.

\begin{proposition}
\label{relSolutions}(i) Suppose that $\sigma\geq0,\ \rho<-\frac{1}{2}$ and $g\in C\left(  \left[0,T\right]\!:\!L_{+}^{\infty}\left(  \mathbb{R}%
_{+}\!:\sqrt{\omega}\left(  1+\omega\right)  ^{\rho-\frac{1}{2}}\right)
\right)  $ is a bounded mild solution \index{mild solution} of (\ref{Z2E2a}) in the sense of
Definition \ref{BoundMildSol}. \\
Then $g\in C\left(  \left[  0,T\right]
:\mathcal{M}_{+}\left(  \left[  0,\infty\right)  :\left(  1+\omega\right)
^{\rho}\right)  \right)  $ and it is also a measured valued mild solution \index{mild solution} of
(\ref{Z2E2a}) in the sense of Definition \ref{MeasMildSol}.

(ii) Suppose that $\sigma>0,\ \rho<-\frac{1}{2}$ and $g\in C\left(  \left[
0,T\right]  :\mathcal{M}_{+}\left(  \left[  0,\infty\right)  :\left(
1+\omega\right)  ^{\rho}\right)  \right)  $ is a measured valued mild solution \index{mild solution} 
of (\ref{Z2E2a}) in the sense of Definition \ref{MeasMildSol}. Then $g\in
C\left(  \left[  0,T\right)  :\mathcal{M}_{+}\left(  \left[  0,\infty\right)
\right)  \right)  $ and it is also a weak solution with interacting condensate of (\ref{Z2E2a}) in the
sense of Definition \ref{weakSolution}.  \index{weak solution} \index{interacting condensate}
\end{proposition}

\begin{proof}
The proof of (i) is immediate, because from  $g\in C\left(  \left[
0,T\right]  :L_{+}^{\infty}\left(  \mathbb{R}_{+}:\sqrt{\omega}\left(
1+\omega\right)  ^{\rho-\frac{1}{2}}\right)  \right)  $ it follows at once that one also has $g\in
C\left(  \left[  0,T\right]  :\mathcal{M}_{+}\left(  \left[  0,\infty\right)
:\left(  1+\omega\right)  ^{\rho}\right)  \right)  $ and (\ref{Z3E5}) implies
(\ref{Z3E3}) in the sense of measures (equivalently (\ref{Z3E4})).

In order to prove (ii), suppose now that $g\in C\left(  \left[  0,T\right]
:\mathcal{M}_{+}\left(  \left[  0,\infty\right)  :\left(  1+\omega\right)
^{\rho}\right)  \right)  $ is a measured valued mild   solution \index{mild solution} of
(\ref{Z2E2a}). This implies the identity (\ref{Z3E4}) for any\ $\varphi\in
C_{0}\left(  \left[  0,\infty\right)  \right)  .$ Using the regularity
properties of $g$ we can differentiate (\ref{Z3E4}) for $a.e.$ $t\in\left[
0,T\right]  $ and check that the following identity holds:%
\begin{align*}
&  \partial_{t}\left(  \int\varphi\left(  t,\omega_{1}\right)  g_{1}\left(
t,d\omega_{1}\right)  \right)  =\int\partial_{t}\varphi\left(  t,\omega_{1}\right)  g_{1}\left(
t,d\omega_{1}\right)+ \\
&  +\int\!\!\int\!\!\int\!\!\Phi_{\sigma}\left[  \left(  \frac{g_{1}}{\sqrt{\omega_{1}}%
}+\frac{g_{2}}{\sqrt{\omega_{2}}}\right)  \frac{g_{3}g_{4}}{\sqrt{\omega
_{3}\omega_{4}}}-\left(  \frac{g_{3}}{\sqrt{\omega_{3}}}+\frac{g_{4}}%
{\sqrt{\omega_{4}}}\right)  \frac{g_{1}g_{2}}{\sqrt{\omega_{1}\omega_{2}}%
}\right]  \varphi\left(  \omega_{1}\right)  d\omega_{3}d\omega_{4}d\omega_{1}%
\end{align*}

Symmetrizing the variables in the integrals of the right-hand side and
integrating in the interval $\left[  0,T\right]  $ we obtain (\ref{Z2E1}).
\end{proof}

\subsection{Existence of bounded mild solutions. \index{mild solution}}

We now prove the following result.

\begin{theorem}
\label{localExBounded}Let $\rho<-\frac{1}{2},\ \sigma\geq0$ be two given constants and 
$g_{in}\in L_{+}^{\infty}\left(  \mathbb{R}^{+};\sqrt{\omega}\left(
1+\omega\right)  ^{\rho-\frac{1}{2}}\right)  .$ There exists $T>0$,
$T\leq\infty,$ depending only on $\left\Vert g_{0}\left(  \cdot\right)
\right\Vert _{L_{+}^{\infty}\left(  \mathbb{R}^{+};\sqrt{\omega}\left(
1+\omega\right)  ^{\rho-\frac{1}{2}}\right)  }$ and $\sigma$, and 
a unique mild solution \index{mild solution} of (\ref{Z2E2a}) $g\in C\left(  \left[  0,T\right)
:L_{+}^{\infty}\left(  \mathbb{R}_{+}:\sqrt{\omega}\left(  1+\omega\right)
^{\rho-\frac{1}{2}}\right)  \right)  $ with initial value $g\left(
\cdot,0\right)  =g_{in}$ in the sense of Definition \ref{BoundMildSol}.

If $\rho>2,$ the obtained solution $g$ satisfies:%
\begin{equation}
\int_{0}^{\infty}g_{in}\left(  \omega\right)  \omega d\omega=\int_{0}^{\infty
}g\left(  t,\omega\right)  \omega d\omega\,,\ \ t\in\left(  0,T\right)  \,.
\label{F3E6c}%
\end{equation}

If $\sigma>0$ we have $T=\infty.$

If $\sigma=0,$ the function $f$ is in the space $W^{1,\infty}\left(  \left(
0,T\right)  ;L^{\infty}\left(  \mathbb{R}^{+}\right)  \right)  $ and it
satisfies (\ref{E1}) $a.e.\ \omega\in\mathbb{R}^{+}$ for any $t\in\left(
0,T_{\max}\right)  .$ Moreover, $f$ can be extended as a mild solution \index{mild solution} of
(\ref{E1})-(\ref{E2b}) to a maximal time interval $\left(  0,T_{\max
}\right)  $ with $0<T_{\max}\leq\infty.$ If $T_{\max}<\infty$ we have:%
\[
\lim\sup_{t\rightarrow T_{\max}^{-}}\left\Vert f\left(  t,\cdot\right)
\right\Vert _{L^{\infty}\left(  \mathbb{R}^{+}\right)  }=\infty.
\]

\end{theorem}

\begin{proof}
The proof of this result can be obtained in the same manner a the Proof of
Theorem 3.4 in \cite{EV1}. The key idea is to interpret mild solutions \index{mild solution} as a
fixed point for an operator $\mathcal{T}\left[  g\right]  $ which we define as
the right-hand side of (\ref{Z3E5}). Some of the main technical difficulties
in the Proof of Theorem 3.4 in \cite{EV1} are due to the need to control
quadratic terms in $f$ which appear in the Nordheim  \index{Nordheim} equation, but
which are not present in $\mathcal{O}_{\sigma}\left[  g\right]  $. We only
need to estimate then cubic terms and this allows to obtain well-posedness
results for a larger range of exponents than the one obtained in \cite{EV1}
(namely $\rho<-\frac{1}{2}$).

The key estimates needed to implement the fixed point argument are the
following ones:%
\begin{align}
\omega^{\rho}\mathcal{O}_{\sigma}\left[  g\right]  \left(  \omega\right)   &
\leq\frac{C}{\omega^{\min\left\{  \rho-\frac{1}{2},\frac{1}{2}\right\}  }%
}\left\Vert g\right\Vert _{L_{+}^{\infty}\left(  \mathbb{R}^{+};\sqrt{\omega
}\left(  1+\omega\right)  ^{\rho-\frac{1}{2}}\right)  }^{3}\ \ ,\ \ \omega
\geq1\label{K8E1}\\
0  &  \leq A_{\sigma}\left(  \omega\right)  \leq C\left\Vert g\right\Vert
_{L_{+}^{\infty}\left(  \mathbb{R}^{+};\sqrt{\omega}\left(  1+\omega\right)
^{\rho-\frac{1}{2}}\right)  }^{2}\ \ ,\ \ \omega\geq0 \label{K8E3}%
\end{align}
where $\mathcal{O}_{\sigma}\left[  g\right]  ,\ A_{\sigma}$ are as in
(\ref{S7E3n}), (\ref{Z2E7}) and $C$ depends in $\rho$. In order to derive
(\ref{K8E3}) we just notice that:%
\begin{eqnarray*}
A_{\sigma}\left(  \omega\right)  & \leq &\left\Vert g\right\Vert _{L_{+}^{\infty
}\left(  \mathbb{R}^{+};\sqrt{\omega}\left(  1+\omega\right)  ^{\rho-\frac
{1}{2}}\right)  }^{2}\int_{0}^{\infty}\int_{0}^{\infty}\frac{d\omega
_{3}d\omega_{4}}{\left(  1+\omega_{3}\right)  ^{\rho+\frac{1}{2}}\left(
1+\omega_{4}\right)  ^{\rho+\frac{1}{2}}}\\
&\leq& C\left\Vert g\right\Vert
_{L_{+}^{\infty}\left(  \mathbb{R}^{+};\sqrt{\omega}\left(  1+\omega\right)
^{\rho-\frac{1}{2}}\right)  }^{2}%
\end{eqnarray*}

To prove (\ref{K8E1}) we use  that we only need to integrate in the
domain $\left\{  \omega_{3}\geq0,\ \omega_{4}\geq0\ ,\ \omega_{2}%
\geq0\right\}  ,$ with $\omega_{2}$ as in (\ref{Z2E7}). We split the region in
the sets 
\begin{eqnarray*}
D_{I}=\left\{  0\leq\max\left\{  \omega_{3},\omega_{4}\right\}
\leq\omega_{1}\right\},  D_{II}=\left\{  \omega_{1}\leq\omega_{3}%
<\infty,\ 0\leq\omega_{4}<\omega_{1}\right\}\\
 D_{III}=\left\{  \omega
_{1}\leq\omega_{4}<\infty,\ 0\leq\omega_{3}<\omega_{1}\right\} ,
D_{IV}=\left\{  \omega_{1}\leq\omega_{3},\ \omega_{1}\leq\omega_{4}\right\}.
\end{eqnarray*} 
We then split the integrals which define the operator $\mathcal{O}_{\sigma
}\left[  g\right]  $ in the four domains. Notice that, since $g_{2}%
\leq\left\Vert g\right\Vert _{L_{+}^{\infty}\left(  \mathbb{R}^{+}%
;\sqrt{\omega}\left(  1+\omega\right)  ^{\rho-\frac{1}{2}}\right)  }\frac
{1}{\omega_{1}^{\rho}}$ in $D_{IV}:$
\[
\iint_{D_{IV}}\left[  \cdot\cdot\cdot\right]  \leq\frac{\left\Vert
g\right\Vert _{L_{+}^{\infty}\left(  \mathbb{R}^{+};\sqrt{\omega}\left(
1+\omega\right)  ^{\rho-\frac{1}{2}}\right)  }^{3}}{\omega_{1}^{\rho}}%
\int_{\omega_{1}}^{\infty}\int_{\omega_{1}}^{\infty}\frac{d\omega_{3}%
d\omega_{4}}{\left(  1+\omega_{3}\right)  ^{\rho+\frac{1}{2}}\left(
1+\omega_{4}\right)  ^{\rho+\frac{1}{2}}}%
\]
and the term on the right-hand side can be estimated by the right-hand side of
(\ref{K8E1}). On the other hand, using that $\omega_{3}\geq\omega_{1}$ in
$D_{II}$ we obtain:
\[
\int\!\!\int_{D_{II}}\!\!\left[  \cdot\cdot\cdot\right]  \leq\frac{\left\Vert
g\right\Vert _{L_{+}^{\infty}\left(  \mathbb{R}^{+};\sqrt{\omega}\left(
1+\omega\right)  ^{\rho-\frac{1}{2}}\right)  }^{3}}{\omega_{1}^{\rho+\frac
{1}{2}}}\int_{\omega_{1}}^{\infty}\!\!\int_{0}^{\omega_{1}}\!\!\!\!\!\frac{d\omega
_{3}d\omega_{4}}{\left(  1+\left(  \omega_{3}+\omega_{4}-\omega_{1}\right)
\right)  ^{\rho+\frac{1}{2}}\left(  1+\omega_{4}\right)  ^{\rho+\frac{1}{2}}}%
\]

Using that $\left(  \omega_{3}+\omega_{4}-\omega_{1}\right)  \geq\left(
\omega_{3}-\omega_{1}\right)  $ we can estimate the integral on the right-hand
side by the product of two integrals which can be bounded by the right-hand
side of (\ref{K8E1}). The estimate of the integral in the domain $D_{III}$ is similar.

Finally we estimate the integral in the domain $D_{I}.$ Due to the symmetry of
the integral it is enough to estimate the integrals in the region $\left\{
\omega_{3}\geq\omega_{4}\right\}  .$ We then have $\omega_{3}\geq\frac
{\omega_{1}}{2}.$ We also change variables in order to replace the integration
in $\omega_{3}$ by integration in $\omega_{2}.$ Then:%
\[
\iint_{D_{I}}\left[  \cdot\cdot\cdot\right]  \leq\frac{\left\Vert
g\right\Vert _{L_{+}^{\infty}\left(  \mathbb{R}^{+};\sqrt{\omega}\left(
1+\omega\right)  ^{\rho-\frac{1}{2}}\right)  }^{3}}{\omega_{1}^{\rho+\frac
{1}{2}}}\int_{0}^{\omega_{1}}\int_{0}^{\omega_{1}}\frac{d\omega_{2}d\omega
_{4}}{\left(  1+\left(  \omega_{2}\right)  \right)  ^{\rho+\frac{1}{2}}\left(
1+\omega_{4}\right)  ^{\rho+\frac{1}{2}}}%
\]
and this gives (\ref{K8E1}).

The rest of the fixed point argument can be made along the lines of the Proof
of Theorem 3.4 in \cite{EV1}. The fact that for $\sigma>0$ the solutions can
be extended to arbitrarily large values of $t$ follows just from the
boundedness of the function $\frac{\Phi_{\sigma}}{\sqrt{\omega_{1}\omega
_{2}\omega_{3}}}.$
\end{proof}

\subsection{Existence of global weak solutions with interacting condensate. \index{weak solution}} \index{interacting condensate}
\label{mildConst}

In order to prove well-posedness of measured valued weak solutions \index{weak solution}, we will
restrict our analysis to integrable distributions $g$ given that we will
consider the long time asymptotics of the solutions only in this case. More
precisely, we will assume that the initial data $g_{in}\in\mathcal{M}%
_{+}\left(  \left[  0,\infty\right)  :\left(  1+\omega\right)  ^{\rho}\right)
$ with $\rho<-1,$ and the resulting solutions $g\left(  t,\cdot\right)  $ will
be shown to be in the same space for $t\geq0.$ The reason to
assume that $\rho<-1$ is that we use in the argument yielding global
existence the finiteness of $\int g\left(  t,d\omega\right)  .$ It is likely
that global weak solutions \index{weak solution} could be obtained just with the assumption
$\rho<-\frac{1}{2}$ but the proof  would require a careful study of the
transfer of mass taking place at the region $\omega\rightarrow\infty.$
Therefore, this case will not be considered in this paper. Notice that due to
the cubic nonlinearity of the problem a simple Gronwall argument does not
allow to obtain global existence, in spite of the fact that the operator
$\mathcal{O}_{\sigma}\left[  g\right]  $ is defined for any $g\in
\mathcal{M}_{+}\left(  \left[  0,\infty\right)  :\left(  1+\omega\right)
^{\rho}\right)  ,\ \rho<-1.$

\begin{remark} 
\label{infinitemass}
On the other hand, we remark that the exponent $\rho=-\frac{1}{2}$ is in a
suitable sense optimal. Indeed, the operator $\mathcal{O}_{\sigma}\left[
g\right]  $ cannot be defined in general for $g\in\mathcal{M}_{+}\left(
\left[  0,\infty\right)  :\left(  1+\omega\right)  ^{\rho}\right)  $ with
$\rho\geq-\frac{1}{2}.$
\end{remark}

As a next step we prove a global existence Theorem of weak solutions \index{weak solution} for
(\ref{S2E1}). We use an idea similar to the one in \cite{Lu1} for the
Nordheim  \index{Nordheim} equation. We first regularize the problem using the kernels
$\Phi_{\sigma}$ with $\sigma>0.$ It is possible to obtain global weak
solutions in that case, just using the fact that mild solutions \index{mild solution} are weak
solutions. Finally we take the limit $\sigma\rightarrow0.$ The available
estimates for the solutions $g_{\sigma}$ will allow to prove that the limit
exists and it yields a global weak solution \index{weak solution} of (\ref{S2E1}). The main result
that we prove in this Section is the following.

\begin{theorem}
\label{globalWeakSol}Let $-2< \rho<-1$ and $g_{in}\in\mathcal{M}_{+}\left(
\left[  0,\infty\right)  :\left(  1+\omega\right)  ^{\rho}\right)  .$ There
exists a weak solution \index{weak solution} of (\ref{S2E1}) in the sense of Definition
\ref{weakSolution} with initial datum $g_{in}.$
\end{theorem}

\begin{remark}
Notice that if $\rho<-1$ taking a sequence of test functions $\varphi_{n}$
which converge uniformly to one as $n\rightarrow\infty,$ we can prove the
following identity for any weak solution \index{weak solution} of (\ref{S2E1}) in the sense of
Definition \ref{weakSolution}:
\begin{equation}
\int g_{in}\left(  d\omega\right)  =\int g\left(  t,d\omega\right)
\ \ ,\ \ a.e.\ t\geq0\ \label{K1}%
\end{equation}
Moreover, if $g_{in}\in\mathcal{M}_{+}\left(  \left[  0,\infty\right)
:\left(  1+\omega\right)  ^{\rho}\right)  $ with $\rho<-2$ we would obtain,
using a similar argument that:%
\begin{equation}
\int\omega g_{in}\left(  d\omega\right)  =\int\omega g\left(  t,d\omega
\right)  \label{K2}%
\end{equation}
\end{remark}
We split the proof of Theorem (\ref{globalWeakSol}) in the different Subsections

\subsubsection{Global measured valued weak solutions for the problem with
regularized kernel $\Phi_{\sigma},\ \sigma>0$}

As a first step we prove the existence of global weak solutions \index{weak solution} for the
regularized problem (\ref{Z2E2a}) with $\sigma>0.$

\begin{lemma}
\label{globRegul}Let $\rho<-1,\ \sigma>0$ and $g_{in}\in\mathcal{M}_{+}\left(
\left[  0,\infty\right)  :\left(  1+\omega\right)  ^{\rho}\right)  .$ Then,
there exist $g_{\sigma}\in C\left(  \left[  0,\infty\right)  :\mathcal{M}%
_{+}\left(  \left[  0,\infty\right)  :\left(  1+\omega\right)  ^{\rho}\right)
\right)  $ which is a global weak solution \index{weak solution} of (\ref{Z2E2a}) in the sense of
Definition \ref{weakSolution} and satisfies $g\left(  0,\cdot\right)
=g_{in}.$
\end{lemma}

\begin{proof}
We first construct a global mild solution \index{mild solution} in the sense of Definition
\ref{MeasMildSol}. The main idea for this construction is to reformulate
(\ref{Z3E5}) as a fixed point Theorem. Given $T>0,$ we define the following
operator
\[
\mathcal{T}_{\sigma}:C\left(  \left[  0,T\right)  :\mathcal{M}_{+}\left(
\left[  0,\infty\right)  :\left(  1+\omega\right)  ^{\rho}\right)  \right)
\rightarrow C\left(  \left[  0,T\right)  :\mathcal{M}_{+}\left(  \left[
0,\infty\right)  :\left(  1+\omega\right)  ^{\rho}\right)  \right)
\]
by means of:%
\begin{eqnarray*}
&&\mathcal{T}_{\sigma}\left[  g\right]  \left(  t,\cdot\right)  =g_{in}\left(
\cdot\right)  \exp\left(  -\int_{0}^{t}A_{\sigma}\left(  s,\cdot\right)
ds\right) +\\
&&\hskip 5cm  +\int_{0}^{t}\exp\left(  -\int_{s}^{t}A_{\sigma}\left(  \xi
,\cdot\right)  d\xi\right)  \mathcal{O}_{\sigma}\left[  g\right]  \left(
s,\cdot\right)  ds
\end{eqnarray*}

We now claim that the operators $g\rightarrow A_{\sigma},\ g\rightarrow
\mathcal{O}_{\sigma}\left[  g\right]  $ are continuous if we endow the space $C\left(
\left[  0,T\right)  :\mathcal{M}_{+}\left(  \left[  0,\infty\right)  :\left(
1+\omega\right)  ^{\rho}\right)  \right)  $ with the topology induced by the
metric:%
\[
dist\left(  g_{1},g_{2}\right)  =\sup_{0\leq t\leq
T}dist_{\ast}\left(  g_{1}\left(  t,\cdot\right)
,g_{2}\left(  t,\cdot\right)  \right)
\]
where $dist_{\ast}$ is as in Notation \ref{distWeak}.

We then need to prove that, given any test function $\varphi\in
C_{0}\left(  \left[  0,T\right)  \times\left[  0,\infty\right)  \right)$, the following functions depend continuously on $g$ in the weak
topology:%
\begin{align*}
I_{1}\left[  g\right]   &  =\int g_{in}\left(  \omega\right)  \exp\left(
-\int_{0}^{t}A_{\sigma}\left(  s,\omega\right)  ds\right)  \varphi\left(
\omega\right)  d\omega\\
I_{2}\left[  g\right]   &  =\iint_{0}^{t}\exp\left(  -\int_{s}^{t}%
A_{\sigma}\left(  \xi,\cdot\right)  d\xi\right)  \mathcal{O}_{\sigma}\left[
g\right]  \left(  s,\cdot\right)  \varphi\left(  \omega\right)  dsd\omega
\end{align*}

We now notice that since $\sigma>0$ the mapping $g\left(  \cdot,t\right)
\rightarrow A_{\sigma}\left(  t,\cdot\right)  $ defines a continuous map
between $\mathcal{M}_{+}\left(  \left[  0,\infty\right)  :\left(
1+\omega\right)  ^{\rho}\right)  ,$ endowed with the weak topology, and the
set of continuous bounded functions $C_{b}\left(  \left[  0,\infty\right)
\right)  ,$ endowed with the uniform convergence. This is due to the fact
that, since $\sigma>0$ the functions $\frac{\Phi_{\sigma}}{\sqrt{\omega
_{1}\omega_{2}\omega_{3}}},\ \frac{\Phi_{\sigma}}{\sqrt{\omega_{1}\omega
_{3}\omega_{4}}}$ are smooth, the values of $A_{\sigma}\left(  t,\omega
_{1}\right)  $ depend only on $g$ through integral quantities, and the decay
of the measure $g$ for large values implies that the contribution of the large
values of $\omega$ can be made small.  

Then, the operator $g\rightarrow A_{\sigma}$ from $C\left(
\left[  0,T\right)  :\mathcal{M}_{+}\left(  \left[  0,\infty\right)  :\left(
1+\omega\right)  ^{\rho}\right)  \right)  $ to $C\left(  \left[  0,T\right)
:C_{b}\left(  \left[  0,\infty\right)  \right)  \right)  $ is continuous.
Moreover, since $A_{\sigma}\geq0$ for $g\geq0$ it then follows that the
operator $I_{1}\left[  g\right]  ,$ which maps $C\left(  \left[  0,T\right)
:\mathcal{M}_{+}\left(  \left[  0,\infty\right)  :\left(  1+\omega\right)
^{\rho}\right)  \right)  $ to itself, is continuous.

On the other hand, the operator $\mathcal{O}_{\sigma}\left[  g\right]  $
defined in Lemma \ref{LQg} is a continuous operator from $\mathcal{M}%
_{+}\left(  \left[  0,\infty\right)  :\left(  1+\omega\right)  ^{\rho}\right)
$ to itself if this space is endowed with the weak topology if $\sigma>0$. The
proof of this uses the fact that the integrals in (\ref{Z3E1}) are well
defined as it can be seen from the arguments in the Proof of Lemma \ref{LQg}.
On the other hand the boundedness of $\frac{\Phi_{\sigma}}{\sqrt{\omega
_{2}\omega_{3}\omega_{4}}}$ implies that the functional $\mathcal{O}_{\sigma
}\left[  g\right]  $ depends continuously on convergent integrals of $g.$ The
continuity of the functional $g\rightarrow I_{2}\left[  g\right]  $ from
$C\left(  \left[  0,T\right)  :\mathcal{M}_{+}\left(  \left[  0,\infty\right)
:\left(  1+\omega\right)  ^{\rho}\right)  \right)  ,$ follows then similarly.
Therefore the transformation $\mathcal{T}_{\sigma}\left[  g\right]  $ defines
a continuous mapping from $C\left(  \left[  0,T\right)  :\mathcal{M}%
_{+}\left(  \left[  0,\infty\right)  :\left(  1+\omega\right)  ^{\rho}\right)
\right)  $ to itself if this space is endowed with the weak topology.
Moreover, this operator transforms the set
\begin{eqnarray}
&&\mathcal{Y}_{T}=\left\{  g\in C\left(  \left[  0,T\right)  :\mathcal{M}%
_{+}\left(  \left[  0,\infty\right)  :\left(  1+\omega\right)  ^{\rho}\right)
\right);   |||g||| _{ \rho , T }\leq2\left\Vert g_{in}\right\Vert  _{ \rho  }\right\}\\
&&|||g||| _{ \rho , T }=\sup_{0\leq t\leq T}||g(t)|| _{ \rho  }
\end{eqnarray}
into itself if $T$ is sufficiently small. 

Actually the operator $\mathcal{T}_{\sigma}\left[  g\right]  $ is compact in
the set $\mathcal{Y}_{T}$.  This a consequence of  the Arzela-Ascoli Theorem in metric
spaces (cf. \cite{DS}),  as well as the fact that the set
$$\left\{  g\in\mathcal{M}_{+}\left(  \left[  0,\infty\right)  :\left(
1+\omega\right)  ^{\rho}\right)  :\sup_{0\leq t\leq T}\frac{1}{\left[
1+R^{\rho}\right]  }\frac{1}{R}\int_{\frac{R}{2}}^{R}g\left(  t,d\omega
\right)  \leq2\left\Vert g_{in}\right\Vert \right\}  $$ is compact in
$\mathcal{M}_{+}\left(  \left[  0,\infty\right)  :\left(  1+\omega\right)
^{\rho}\right)  $ endowed with the weak topology. The uniform continuity of
$\mathcal{T}_{\sigma}\left[  g\right]  $ with respect to the  time variable follows from the
fact that the functions $t\rightarrow\Psi_{\varphi}\left(  t\right)
=\int\varphi\left(  \omega\right)  \mathcal{T}_{\sigma}\left[  g\right]
\left(  d\omega,t\right)  $ is Lifschitz continuous for any test function
$\varphi,$ as it can be seen from the definition of $\mathcal{T}_{\sigma
}\left[  g\right]$.

Local existence of solutions then follows using Schauder's Theorem. Notice
that, since $\sigma>0$ we can obtain that the corresponding fixed point, $g(t, \cdot)$ satisfies:
$$
||g(t)|| _{ \rho  }\le C_1+C_2\int _0^t ||g(s)|| _{ \rho  }ds
$$
where $C_1$ only depends on $\sigma $ and $||g_{in }|| _{ \rho  }$ and $C_2$ only depends on $\sigma $ and the total mass of $g (t)$ which is a constant and therefore depends on  total mass of $g _{ in }$. This is proved as follows. Integrating equation (\ref{Z3E5}) in the interval $(R, 2R)$ for any $R>0$,  the first term is immediately estimated using $||g _{ in }|| _{ \rho  }$. The integral of the second is estimated by splitting the domain in the subdomains $\{\omega _3+\omega _4 \ge 4R\}$ and   $\{\omega _3+\omega _4 \le 4R \}$.  Since $\sigma >0$ the term $(\omega _2\omega _3\omega _4)^{-1/2}$ is bounded by a constant depending on $\sigma $.  In the two resulting triple integrals, one of the integrations takes place in the interval $(R/4, 4R)$, and the two others are estimated by the total mass of $g(t)$, that is constant.
Using  Gronwall's lemma we deduce, 
$|||g||| _{ \rho , T }  \leq C\left(  T\right)$ for any finite $T.$ Iterating the construction it is then possible to prove
that the solution is global in time.

In order to conclude the proof of the Lemma, we just notice that mild
solutions \index{mild solution} of (\ref{Z2E2a}) in the sense of measures are weak solutions \index{weak solution} of
(\ref{Z2E2a}) in the sense of Definition \ref{weakSolution} due to Proposition
\ref{relSolutions}.
\end{proof}

\subsubsection{Monotonicity formula.\label{Mon}}

The following result is analogous to one that has been proved in \cite{EV1}, \cite{Lu3}.

\begin{proposition}
\label{atractiveness} Let $\sigma\geq0.$ Given $g\in\mathcal{M}_{+}\left(
\left[  0,\infty\right)  :\left(  1+\omega\right)  ^{\rho}\right)  $ we
define:%
\[
q\left[  g\right]  \left(  \omega_{1}\right)  =\left(  \frac{g_{1}}%
{\sqrt{\omega_{1}}}+\frac{g_{2}}{\sqrt{\omega_{2}}}\right)  \frac{g_{3}g_{4}%
}{\sqrt{\omega_{3}\omega_{4}}}-\left(  \frac{g_{3}}{\sqrt{\omega_{3}}}%
+\frac{g_{4}}{\sqrt{\omega_{4}}}\right)  \frac{g_{1}g_{2}}{\sqrt{\omega
_{1}\omega_{2}}}%
\]
with $\omega_{2}=\omega_{3}+\omega_{4}-\omega_{1}.$ Let us denote as
$\mathcal{S}^{3}$ the group of permutations of the three elements $\left\{
1,2,3\right\}  .$ Suppose that $\varphi\in C_{0}^{2}\left(  \left[
0,\infty\right)  \right)  $ is a test function. The following identity holds:%
\begin{eqnarray}
&&\int_{\left[  0,\infty\right)  ^{3}}\!\!\!\!\!\!d\omega_{1}d\omega_{3}d\omega_{4}%
\ \Phi_{\sigma}q\left[  g\right]  \left(  \omega_{1}\right)  \sqrt{\omega_{1}%
}\varphi\left(  \omega_{1}\right)  =\int_{\left[  0,\infty\right)  ^{3}%
}\!\!\!\!\!\!d\omega_{1}d\omega_{2}d\omega_{3}\,\frac{g_{1}\,g_{2}\,g_{3}}{\sqrt
{\omega_{1}\omega_{2}\omega_{3}}}\mathcal{G}_{\sigma,\varphi}\ \label{S1E12a}%
\end{eqnarray}
where:%
\begin{align*}
\mathcal{G}_{\sigma,\varphi}\equiv \mathcal{G}_{\sigma,\varphi}\left(  \omega_{1},\omega_{2},\omega_{3}\right)
&  =\frac{1}{6}\sum_{\sigma\in\mathcal{S}^{3}}H_{\varphi}\left(
\omega_{\sigma(1)},\omega_{\sigma(2)},\omega_{\sigma(3)}\right)  \Phi_{\sigma
}\left(  \omega_{\sigma(1)},\omega_{\sigma(2)};\omega_{\sigma(3)}\right)
\\
&  H_{\varphi}(x,y,z)=\varphi\left(  z\right)  +\varphi\left(  x+y-z\right)
-\varphi\left(  x\right)  -\varphi\left(  y\right)  
\end{align*}
with $\Phi_{\sigma}\left(  \omega_{1},\omega_{2};\omega_{3}\right)  $ given
as:%
\begin{eqnarray*}
&&\Phi_{\sigma}\left(  \omega_{1},\omega_{2};\omega_{3}\right)  =\min\left\{
\sqrt{\left(  \omega_{1}-\sigma\right)  _{+}},\sqrt{\left(  \omega_{2}%
-\sigma\right)  _{+}},\sqrt{\left(  \omega_{3}-\sigma\right)  _{+}}%
, \right. \\
&&\hskip 7.5cm \left. \sqrt{\left(  \omega_{1}+\omega_{2}-\omega_{3}-\sigma\right)  _{+}}\right\}
\end{eqnarray*}
and:%
\begin{equation}
\mathcal{G}_{\sigma,\varphi}\left(  \omega_{1},\omega_{2},\omega_{3}\right)
=\mathcal{G}_{\sigma,\varphi}\left(  \omega_{\sigma(1)},\omega_{\sigma
(2)},\omega_{\sigma(3)}\right)  \ \ \ \text{for any\ }\sigma\in\mathcal{S}^{3}
\label{S1E12four}%
\end{equation}
Moreover, if the function $\varphi$ is convex we have $\mathcal{G}_{\varphi
}\left(  \omega_{1},\omega_{2},\omega_{3}\right)  \geq0$ and if $\varphi$ is
concave we have $\mathcal{G}_{\varphi}\left(  \omega_{1},\omega_{2},\omega
_{3}\right)  \leq0.$ For any test function $\varphi$ the function
$\mathcal{G}_{\varphi}\left(  \omega_{1},\omega_{2},\omega_{3}\right)  $
vanishes along the diagonal $\left\{  \left(  \omega_{1},\omega_{2},\omega
_{3}\right)  \in\left[  0,\infty\right)  ^{3}:\omega_{1}=\omega_{2}=\omega
_{3}\right\}  $.
\end{proposition}

\begin{proof}
It is essentially identical to the Proof of Proposition 4.1 of \cite{EV1}. The
only difference is that we use $\Phi_{\sigma}$ instead of $\Phi.$ However,
using the fact that $\Phi_{\sigma}$ is invariant under permutations in their
variables, we can argue exactly as in the Proof of Proposition 4.1 of
\cite{EV1} by means of a symmetrization argument. The only relevant difference
with the result in \cite{EV1} is that due to the fact that $g$ are measures,
we must check the continuity of the functions which are integrated against
them. This follows from Lemma \ref{Cont}.
\end{proof}

We will need later a more detailed representation formula for the functions
$\mathcal{G}_{\sigma,\varphi}$ in the case $\sigma=0.$ To this end we define
the following functions which have been used also in \cite{EV1}.

\begin{definition}
\label{aux}We define auxiliary functions $\omega_{+},\ \omega_{0},\ \omega
_{-}$ from $\left[  0,\infty\right)  \times\left[  0,\infty\right)
\times\left[  0,\infty\right)  $ to $\left[  0,\infty\right)  $ as follows:%
\begin{align*}
\omega_{+}\left(  \omega_{1},\omega_{2},\omega_{3}\right)   &  =\max\left\{
\omega_{1},\omega_{2},\omega_{3}\right\}  ,\\
\omega_{-}\left(  \omega_{1},\omega_{2},\omega_{3}\right)   &  =\min\left\{
\omega_{1},\omega_{2},\omega_{3}\right\}  ,\\
\omega_{0}\left(  \omega_{1},\omega_{2},\omega_{3}\right)   &  =\omega_{k}%
\in\left\{  \omega_{1},\omega_{2},\omega_{3}\right\}  \setminus\left\{
\omega_{+},\omega_{-}\right\}  \text{ }%
\end{align*}
with $k\in\left\{  1,2,3\right\}  ,$ where we will assume that the set
$\left\{  \omega_{1},\omega_{2},\omega_{3}\right\}  $ has three different
elements even if some of the values of the elements $\omega_{j}$ are
identical. 
\end{definition}

\begin{lemma}
\label{strictConvex}The function $\mathcal{G}_{0,\varphi}$ defined in
Proposition \ref{atractiveness} can be written as:%
\[
\mathcal{G}_{0,\varphi}\left(  \omega_{1},\omega_{2},\omega_{3}\right)
=\frac{1}{3}\left[  \sqrt{\omega_{-}}H_{\varphi}^{1}\left(  \omega_{1}%
,\omega_{2},\omega_{3}\right)  +\sqrt{\left(  \omega_{0}+\omega_{-}-\omega
_{+}\right)  _{+}}H_{\varphi}^{2}\left(  \omega_{1},\omega_{2},\omega
_{3}\right)  \right]
\]%
\begin{align*}
H_{\varphi}^{1}\left(  \omega_{1},\omega_{2},\omega_{3}\right)   &
=\varphi\left(  \omega_{+}+\omega_{-}-\omega_{0}\right)  +\varphi\left(
\omega_{+}+\omega_{0}-\omega_{-}\right)  -2\varphi\left(  \omega_{+}\right) \\
H_{\varphi}^{2}\left(  \omega_{1},\omega_{2},\omega_{3}\right)   &
=\varphi\left(  \omega_{+}\right)  +\varphi\left(  \omega_{0}+\omega
_{-}-\omega_{+}\right)  -\varphi\left(  \omega_{0}\right)  -\varphi\left(
\omega_{-}\right)
\end{align*}
If $\varphi$ is concave both functions $H_{\varphi}^{1}$,\ $H_{\varphi}^{2}$
are nonpositive.
\end{lemma}

\begin{proof}
This result has been proved in \cite{EV1}.
\end{proof}

Using Proposition \ref{atractiveness} we can prove the following result.

\begin{lemma}
\label{Convex}For all $\sigma\geq0$ let  $g_{\sigma}\in C\left(  \left[
0,\infty\right)  :L_{+}^{\infty}\left(  \mathbb{R}_{+}:\sqrt{\omega}\left(
1+\omega\right)  ^{\rho-\frac{1}{2}}\right)  \right)  $ be a weak solution \index{weak solution} of
(\ref{Z2E2a}) in the sense of Definition \ref{weakSolution}. Let $\varphi\in
C\left(  \left[  0,\infty\right)  \right)  $ any convex function. Then:%
\[
\frac{d}{dt}\left(  \int_{0}^{\infty}g_{\sigma}\left(  t,\omega\right)
\varphi\left(  \omega\right)  d\omega\right)  \geq0,\ \ a.e.\ t\in\left[
0,\infty\right).
\]

\end{lemma}

\begin{proof}
It is just a consequence of Proposition \ref{atractiveness} as well as the identity:
\begin{equation}
\frac{d}{dt}\left(  \int_{0}^{\infty}g\left(  t,\omega\right)  \varphi\left(
\omega\right)  d\omega\right)  =\int_{0}^{\infty}\int_{0}^{\infty}\int
_{0}^{\infty}\frac{g_{1}g_{2}g_{3}}{\sqrt{\omega_{1}\omega_{2}\omega_{3}}%
}\mathcal{G}_{0,\varphi}\left(  \omega_{1},\omega_{2},\omega_{3}\right)
d\omega_{1}d\omega_{2}d\omega_{3}
\label{S2E5}
\end{equation}
a.e. $t\in\left[  0,\infty\right)$.
\end{proof}

\subsubsection{Tightness of the measures $\left\{  g_{\sigma}\right\}$}

The following result will be used several times in the following in order to
prove that the mass of the measures $g_{\sigma}$ cannot escape too far away.
In particular, since the Lemma provides uniform estimates in $\sigma$ of the
mass far away from the origin, it will play a crucial role taking the limit
$\sigma\rightarrow0,$ in order to prove the existence of weak solutions \index{weak solution} of of
(\ref{Z2E2a}) with $\sigma=0.$

\begin{lemma}
\label{cotInf}Suppose that $g_{\sigma}\in C\left(  \left[  0,\infty\right)
:\mathcal{M}_{+}\left(  \left[  0,\infty\right)  :\left(  1+\omega\right)
^{\rho}\right)  \right)  ,\ \rho<-\frac{1}{2}$ is a weak solution \index{weak solution} of
(\ref{Z2E2a}) in the sense of Definition \ref{weakSolution} for some
$\sigma\geq0$. Let $\eta>0,\ R>0.$ Suppose that $g_{\sigma}\left(
0,\cdot\right)  =g_{in}\left(  \cdot\right)  .$ Then:
\begin{equation}
\int_{\left[  0,L\right]  }g_{\sigma}\left(  t,d\omega\right)  \geq\left(
1-\eta\right)  \int_{\left[  0,R\right]  }g_{in}\left(  d\omega\right)
\ \ ,\ \ t\in\left[  0,T\right]  \label{S1E5}%
\end{equation}
where $L=\frac{R}{\eta}.$

Moreover, suppose that $g_{in}$ satisfies $\int g_{in}=1$ and $\int
_{R}^{\infty}g_{in}\left(  d\omega\right)  \leq AR^{\rho+1},\ $for some
$-2<\rho<-1$,\ $A>0$ and any $R\geq1.$ Then:%
\begin{equation}
\int_{0}^{R}g_{\sigma}\left(  t,d\omega\right)  \geq1-\frac{1}{R}%
-\frac{AR^{\rho+1}}{\left(  \rho+2\right)  }\ \ ,\ \ R\geq1 \label{S1E6}%
\end{equation}

\end{lemma}

\begin{proof}
We use the following test function:%
\[
\varphi\left(  \omega\right)  =\left(  1-K\omega\right)  _{+}%
\]
where $K>0$ is a constant to be precised. The function $\varphi$ is convex.
Applying Proposition \ref{atractiveness} and Lemma \ref{Convex} it follows
that:%
\begin{equation}
\int_{0}^{\infty}g_{\sigma}\left(  t,\omega\right)  \left(  1-K\omega\right)
_{+}d\omega\geq\int_{0}^{\infty}g_{in}\left(  \omega\right)  \left(
1-K\omega\right)  _{+}d\omega\ \ ,\ \ t\geq0\ \label{S1E6a}%
\end{equation}
whence, assuming that $KR\leq1$ and using that $\left(  1-K\omega\right)
_{+}\leq\chi_{\left(  0,\frac{1}{K}\right)  }:$%
\[
\int_{0}^{\frac{1}{K}}g_{\sigma}\left(  t,d\omega\right)  \geq \left(  1-KR\right)\int_{0}%
^{R}g_{in}\left(  \omega\right)    d\omega.
\]
Choosing then $K$ by means of $KR=\eta$ and writing $L=\frac{1}{K}$ we obtain
(\ref{S1E5}).

On the other hand, suppose that $\int g_{in}=1.$ We define $G_{in}\left(
\omega\right)  =\int_{\left[  \omega,\infty\right)  }g_{in}.$ Using
(\ref{S1E6a}) with $K=\frac{1}{R}$ we obtain:%
\[
\int_{0}^{R}g_{\sigma}\left(  t,d\omega\right)  \geq-\int_{0}^{\infty}%
\frac{dG_{in}}{d\omega}\left(  1-\frac{\omega}{R}\right)  _{+}d\omega
=1-\frac{1}{R}\int_{0}^{R}G_{in}\left(  \omega\right)  d\omega
\]

Using that $G_{in}\left(  \omega\right)  \leq A\omega^{\rho+1}$ if $\omega
\geq1$ and $G_{in}\left(  \omega\right)  \leq1$ for $\omega\geq0,$ it then
follows that:%
\[
\int_{0}^{R}g_{\sigma}\left(  t,d\omega\right)  \geq1-\frac{1}{R}%
-\frac{AR^{\rho+1}}{\left(  \rho+2\right)  }\ \ ,\ \ R\geq1
\]
if $\rho>-2$ whence (\ref{S1E6}) follows.
\end{proof}

\begin{remark}
It is important to notice that Lemma (\ref{cotInf}) also holds for $\sigma =0$. The proof of the existence of weak solutions for such a value of $\sigma $ is concluded in the next Subsection.
\end{remark}

\subsubsection{Limit $\sigma\rightarrow0.$ Global existence of weak
solutions.}

We can now prove Theorem \ref{globalWeakSol}:

\begin{proof}
[Proof of Theorem \ref{globalWeakSol}]We consider the solutions $\left\{
g_{\sigma}:\sigma>0\right\}  $ of the problems (\ref{Z2E2a}) which have been
found in Lemma \ref{globRegul}. Our goal is to prove suitable compactness
properties for these functions. The estimate (\ref{S1E6}) in Lemma
\ref{cotInf} imply uniform tightness on $\sigma$ for the measures $\left\{
g_{\sigma}:\sigma>0\right\}  .$ Moreover, this estimate yields also an uniform
estimate of the measures $g_{\sigma}$ in the space $\mathcal{M}_{+}\left(
\left[  0,\infty\right)  :\left(  1+\omega\right)  ^{\rho}\right)  $ with
$-2<\rho<-1.$ Therefore, the limit of these functions will be in the same
space. In order to prove the compactness of this family of measures in the space $C\left(
\left[  0,\infty\right)  :\mathcal{M}_{+}\left(  \left[  0,\infty\right)
:\left(  1+\omega\right)  ^{\rho}\right)  \right)  $ we need to obtain
estimates for the increments of time. It is enough to estimate the differences:
$$\int_{\left[  0,\infty\right)  }g_{\sigma}\left(  t_{2},\omega\right)
\varphi\left(  \omega\right)  d\omega-\int_{\left[  0,\infty\right)
}g_{\sigma}\left(  t_{1},\omega\right)  \varphi\left(  \omega\right)  d\omega$$
for any $\varphi\in C^{2}\left(  \left[  0,\infty\right)  \right)$, 
$t_{1},t_{2}\in\left[  0,\infty\right)  .$ Using (\ref{Z2E1}) and Lemma
\ref{Cont} we obtain:%
\[
\left\vert \int_{\left[  0,\infty\right)  }g_{\sigma}\left(  t_{2}%
,\omega\right)  \varphi\left(  \omega\right)  d\omega-\int_{\left[
0,\infty\right)  }g_{\sigma}\left(  t_{1},\omega\right)  \varphi\left(
\omega\right)  d\omega\right\vert \leq C\left\vert t_{2}-t_{1}\right\vert
\]
where $C>0$ is independent on $\sigma.$ The compactness of the
family $\left\{  g_{\sigma}:\sigma>0\right\}  $ follows then from Arzela-Ascoli
(cf. \cite{DS}). Taking a subsequence $\left\{  \sigma_{k}\right\}  $ we
obtain that $$g_{\sigma_{k}}\rightharpoonup g\in\mathcal{M}_{+}\left(  \left[
0,\infty\right)  :\left(  1+\omega\right)  ^{\rho}\right)  ,\ -2<\rho<-1.$$
Taking the limit in (\ref{Z2E1}) and using also Lemma \ref{Cont} we obtain
that $g$ is a weak solution \index{weak solution} of (\ref{S2E1}) in the sense of Definition
\ref{weakSolution} with initial datum $g_{in}$ and the result follows.
\end{proof}

\subsection{Stationary solutions.}

In this Section we discuss the stationary solutions of (\ref{M1E2}). It turns
out that in the isotropic case it is possible to obtain a complete
classification of the equlibria.

\subsubsection{Equilibria in the isotropic case.}

We first discuss the weak solutions \index{weak solution} in the sense of Definition
\ref{weakSolution} which do not depend on $t.$ Such solutions will be termed
as equilibria. In the isotropic case we can obtain a classification of all the equilibria.

\begin{theorem}
\label{StatIsot}Suppose that $g\in\mathcal{M}_{+}\left(  \left[
0,\infty\right)  :\left(  1+\omega\right)  ^{\rho}\right)  ,$ with
$\rho<-\frac{1}{2}$ has the property that the measure $\bar{g}\in C\left(
\left[  0,\infty\right)  :\mathcal{M}_{+}\left(  \left[  0,\infty\right)
:\left(  1+\omega\right)  ^{\rho}\right)  \right)  $ defined as $\bar
{g}\left(  t,\cdot\right)  =g\left(  \cdot\right)  $ for any $t\geq0$ is a
weak solution \index{weak solution} of (\ref{Z2E2a}) in the sense of Definition \ref{weakSolution}
with $\sigma=0.$ We will assume also that $\int_{\left[  0,\infty\right)
}g\left(  d\omega\right)  =m<\infty.$ Then there exists $\omega_{0}\geq0$ such
that:%
\[
g=m\delta_{\omega_{0}}%
\]

\end{theorem}

\begin{proof}
[Proof of Theorem \ref{StatIsot}] We can assume without loss of generality that $\int
g\left(  d\omega\right)  =1.$ Let us assume first that $\int_{\left\{
0\right\}  }g=0.$ Using the fact that $g$ is a weak solution of (\ref{Z2E2a})
in the sense of Definition \ref{weakSolution} it follows from Proposition
\ref{atractiveness} that for any concave test function $\varphi$ we have
\begin{equation}
\mathcal{G}_{0,\varphi}\leq0 \label{J0}%
\end{equation}

Since $g$ is an equilibrium it then follows that:
\begin{equation}
\int_{\left[  0,\infty\right)  ^{3}}d\omega_{1}d\omega_{2}d\omega_{3}%
\,\frac{g_{1}\,g_{2}\,g_{3}}{\sqrt{\omega_{1}\omega_{2}\omega_{3}}}%
\mathcal{G}_{0,\varphi}(\omega_{1}\,\omega_{2}\,\omega_{3})=0\ \label{J1}%
\end{equation}

We then apply Lemma \ref{estProd} to the function $\bar{g}\left(
t,\cdot\right)  =g\left(  \cdot\right)  $ which by assumption is a weak
solution of (\ref{Z2E2a}). Then:%
\[
T\int_{\mathcal{S}_{R,\rho}}\left[  \prod_{m=1}^{3}\,g_{m}\left(  d\omega
_{m}\right)  \right]  \leq\frac{2Bb^{\frac{7}{2}}R}{\rho^{2}\left(  \sqrt
{b}-1\right)  ^{2}},\
\]
where $R>0$, $b>1$ can be chosen arbitrarily close to one. and $\rho$
arbitrarily close to zero. The constant $B$ is independent of $b,\ \rho
,\ R,.T.$ The set $\mathcal{S}_{R,\rho}$ is contained in $\left(
0,\infty\right)  ^{3}.$ Taking the limit $T\rightarrow\infty$ it then follows
that:%
\[
\int_{\mathcal{S}_{R,\rho}}\left[  \prod_{m=1}^{3}\,g_{m}\left(  d\omega
_{m}\right)  \right]  =0
\]
and taking the limit $\rho\rightarrow0$ we obtain:%
\[
\int_{\left\{  \omega_{1}=\omega_{2}=\omega_{3}:\omega_{1}>0\right\}  }\left[
\prod_{m=1}^{3}\,g_{m}\left(  d\omega_{m}\right)  \right]  =0
\]

Therefore $g=\delta_{\omega_{0}}$ with $\omega_{0}>0$ and Theorem
\ref{StatIsot} would follow. Suppose then that $\int_{\left\{  0\right\}
}g>0.$ If $\int_{\left\{  0\right\}  }g=1$ the conclusion of the Theorem
follows with $\omega_{0}=0.$ If $m=\int_{\left\{  0\right\}  }g\in\left(
0,1\right)  $ there exists a bounded set $A$ such that $dist
\left(  A,\left\{  0\right\}  \right)  >0$ and $\int_{A}g>0.$ We then apply
(\ref{J1}) with the concave test function $\varphi\left(  \omega\right)
=\frac{\omega}{1+\omega}.$ Using (\ref{J1}), Lemma \ref{strictConvex}, as well as the fact that
$\varphi^{\prime\prime}\left(  \omega\right)  \leq-c_{1}<0$ in bounded sets we
obtain:
\[
0\leq-c_{0}m\left(  \int_{A}g\right)  ^{2}%
\]
with $c_{0}>0$. This gives a contradiction, whence $m\in\left\{
0,1\right\}$.
\end{proof}

\subsubsection{Equilibria in the nonisotropic case.}

The mathematical theory for the nonisotropic weak turbulence \index{weak turbulence}equation
(\ref{M1E2}) is far less developed than in the isotropic case. The main reason
for that is that the integral on the right-hand side of (\ref{M1E2}) does not
define a measure for an arbitrary measure $F.$ However, it is possible to
obtain a huge class of measures $F$ for which the right-hand side of
(\ref{M1E2}) is well defined in the sense of measures which actually vanishes.
The idea is construct measures $F$ with the form $\sum_{\ell}\delta_{k_{\ell}%
}$ where the values $k_{\ell}$ do not interact with each other.

 We first precise in which sense a measure $F\in\mathcal{M}_{+}\left(
\mathbb{R}^{3}\right)  $ is a stationary solution of (\ref{M1E2}).

\begin{definition}
\label{noniso}We will say that $F\in\mathcal{M}_{+}\left(  \mathbb{R}%
^{3}\right)  $ is a stationary solution of (\ref{M1E2}) if for any $\varphi\in
C_{0}\left(  \mathbb{R}^{3}\right)  $the following integrals are defined:
\begin{align*}
J_{k,\ell,m}  &  =\iiint_{\left(  \mathbb{R}^{3}\right)  ^{3}}%
\varphi\left(  k_{3}+k_{4}-k_{2}\right)  \delta\left(  \omega_{1}+\omega
_{2}-\omega_{3}-\omega_{4}\right)  F_{k}F_{\ell}F_{m}\\
\omega_{1}  &  =\left(  k_{3}+k_{4}-k_{2}\right)  ^{2}\ \ ,\ \ \omega
_{j}=k_{j}^{2}\ \ ,\ \ j=2,3,4
\end{align*}
with $\left(  k,\ell,m\right)  \in\left\{  \left(  3,4,1\right)  ,\left(
3,4,2\right)  ,\left(  1,2,3\right)  ,\left(  1,2,4\right)  \right\}  $ and,
moreover, the following identity holds:%
\[
J_{3,4,1}+J_{3,4,2}=J_{1,2,3}+J_{1,2,4}%
\]

\end{definition}

We can then construct infinitely many stationary solutions of (\ref{M1E2}) in
the sense of Definition \ref{noniso}. The possibility of obtaining stationary
solutions of weak turbulence \index{weak turbulence}equations by means of noninteracting particles
was already pointed out in \cite{Hass2}.

\begin{theorem}
Given $L=1,2,3,..,\infty,$ it is possible to choose vectors $\left\{
K_{j}\right\}  _{j=1}^{L}$, $K_{j}\in\mathbb{R}^{3}$ in infinitely many ways,
with the property that for any choice of numbers $\left\{  m_{j}\right\}
_{j=1}^{L},\ m_{j}>0,$ the measure $F=\sum_{j=1}^{L}m_{j}\delta_{K_{j}}$ is a
stationary solution of (\ref{M1E2}) in the sense of Definition \ref{noniso}.
\end{theorem}

\begin{proof}
Given three arbitrary, different points $K_{1},K_{2},K_{3}\in\mathbb{R}^{3}$
we choose a point $K\,_{4}\in\mathbb{R}^{3}$ with the property that the
functions
\[
\Delta=\left(  k_{3}+k_{4}-k_{2}\right)  ^{2}+k_{2}^{2}-k_{3}^{2}-k_{4}^{2}%
\]
are different from zero for any choice of values $\left(  k_{1},k_{2}%
,k_{3},k_{4}\right)  \in\left\{  K_{1},K_{2},K_{3},K\,_{4}\right\}  .$ This
choice of $K\,_{4}$ can be made in infinitely many different ways. It then
follows that for any choice of numbers $\left\{  m_{j}\right\}  _{j=1}%
^{4},\ m_{j}>0$ we have:%
\[
\delta\left(  \omega_{1}+\omega_{2}-\omega_{3}-\omega_{4}\right)  F_{k}%
F_{\ell}F_{m}=0
\]
with $F=\sum_{j=1}^{4}m_{j}\delta_{K_{j}}.$ Therefore $J_{k,\ell,m}$ for any
choice of values of $\left(  k,\ell,m\right)  .$ We can iterate the procedure
in order to add an arbitrary number of particles. Actually it is possible to
form countable sets of particles with the same property This proves the result.
\end{proof}

\subsection{Weak solutions with non interacting condensate.}
\label{fluxes}

We now discuss the Kolmogorov-Zakharov \index{Kolmogorov-Zakharov}solutions, in the framework 
used in this paper. These solutions have the
form $f_{s}(\omega ) =K\omega ^{-7/6}$, whence $
g_{s}\left( \omega \right) =K\omega ^{-2/3}$ and  have been extensively studied in the physical literature, where it
has been seen that they yield a non-zero flux of particles from $\left\{
\omega >0\right\} $ to $\left\{ \omega =0\right\} .$ Different, but
equivalent, expressions for the fluxes have been obtained for instance in 
\cite{DNPZ}, \cite{EMV2}, \cite{Spohn}. We will use the following formulas
for the fluxes:

\begin{equation}
J_{n}\left[ g\right] \left( \omega \right) =J_{n,1}\left[ g\right] \left(
\omega \right) +J_{n,2}\left[ g\right] \left( \omega \right) +J_{n,3}\left[ g%
\right] \left( \omega \right) +J_{n,4}\left[ g\right] \left( \omega \right)
\   \label{KZP2}
\end{equation}%
with:%
\begin{eqnarray}
J_{n,1}\left[ g\right] \left( \omega \right) &=&\int_{0}^{\infty }d\omega
_{1}\int_{0}^{\omega }d\omega _{2}\int_{\omega }^{\infty }d\omega _{3}Q\left[
g\right] \left( \omega _{1},\omega _{2},\omega _{3}\right)  \label{KZP3} \\
J_{n,2}\left[ g\right] \left( \omega \right) &=&\int_{0}^{\omega }d\omega
_{1}\int_{\omega -\omega _{1}}^{\infty }d\omega _{2}\int_{0}^{\omega
_{1}+\omega _{2}-\omega }d\omega _{3}Q\left[ g\right] \left( \omega
_{1},\omega _{2},\omega _{3}\right)  \label{KZP4} \\
J_{n,3}\left[ g\right] \left( \omega \right) &=&-\int_{0}^{\infty }d\omega
_{1}\int_{\omega }^{\infty }d\omega _{2}\int_{0}^{\omega }d\omega _{3}Q\left[
g\right] \left( \omega _{1},\omega _{2},\omega _{3}\right) \   \label{KZP5}
\\
J_{n,4}\left[ g\right] \left( \omega \right) &=&-\int_{\omega }^{\infty
}d\omega _{1}\int_{0}^{\infty }d\omega _{2}\int_{\omega _{1}+\omega
_{2}-\omega }^{\omega _{1}+\omega _{2}}d\omega _{3}Q\left[ g\right] \left(
\omega _{1},\omega _{2},\omega _{3}\right)  \label{KZP6}
\end{eqnarray}%
where $\int_{a}^{b}=\int_{\left( a,b\right) }$ and $Q\left[ g\right] \left(
\omega _{1},\omega _{2},\omega _{3}\right) =\frac{\Phi g_{1}g_{2}g_{3}}{%
\sqrt{\omega _{1}\omega _{2}\omega _{3}}}.$ We define also:%
\begin{equation}
G_{0}\left( \omega \right) =\omega ^{-\frac{2}{3}}  \label{Gaux}
\end{equation}

The methods developed in the papers \cite{EMV1}, \cite{EMV2}
for the Nordheim equation allow to obtain a large class of solutions of (\ref%
{Z2E2a}) which behave asymptotically as the Kolmogorov-Zakharov \index{Kolmogorov-Zakharov}solutions for small values of 
$\omega .$ More precisely, the solutions described in the following Theorem
have the asymptotics $f\left( t,\omega \right) \sim a\left( t\right) \omega
^{-\frac{7}{6}}$ as $\omega \rightarrow 0$ for a suitable function $a\left(
t\right) .$

\begin{theorem}
\label{fluxSol}Given a function $f_{0}\in C^{1}\left( 0,\infty \right) $
satisfying $\left\vert \omega ^{\frac{7}{6}}f_{0}\left( \omega \right)
-A\right\vert +\left\vert \omega ^{\frac{13}{6}}f_{0}\left( \omega \right) +%
\frac{7A}{6}\right\vert \leq C\omega ^{\delta }$ for $0<\omega \leq 1$ and $%
\left\vert \omega ^{\frac{1}{2}+\rho }f_{0}\left( \omega \right) \right\vert
\leq C$ for $\omega \geq 1$ and $\rho >\frac{1}{2},$ there exists\ $T>0$ and
functions $f\in C^{1,0}\left( \left[ 0,T\right] \times \left[ 0,\infty
\right) \right) ,$ $a\in C\left( \left[ 0,T\right] \right) ,\ \left\vert
a\left( t\right) \right\vert \leq 2A$ such that $f$ solves (\ref{E1}), (\ref%
{E2}) and $\left\vert \omega ^{\frac{7}{6}}f\left( t,\omega \right) -a\left(
t\right) \right\vert \leq 2C\omega ^{\frac{\delta }{2}}$ for $0<\omega \leq
1,\ t\in \left[ 0,T\right] $ and $\left\vert \omega ^{\frac{1}{2}+\rho
}f\left( t,\omega \right) \right\vert \leq 2C$ for $\omega \geq 1,\ t\in %
\left[ 0,T\right] .$ Moreover, if $\rho >1$ we have:%
\begin{equation}
\partial _{t}\left( \int f\left( t,\omega \right) \sqrt{\omega }d\omega
\right) =J_{n}\left[ G_{0}\right] \left( 1\right) \left( a\left( t\right)
\right) ^{3}  \label{KZ3}
\end{equation}%
where $J_{n}\left[ G_{0}\right] \left( 1\right) $ is obtained using (\ref%
{KZP2}), (\ref{Gaux}).{}
\end{theorem}

\begin{proof}
The Proof of Theorem \ref{fluxSol} is similar to the Proof of Theorem 2.1 in 
\cite{EMV2}. Its main idea is to linearize (\ref{E1}), (\ref{E2}%
) around the power law $\bar{f}\left( \omega \right) =\omega ^{-\frac{7}{6}%
}. $ The fundamental solution associated to this linearized problem can be
computed explicitly using Wiener-Hopf methods and their properties can be
described with great detail (cf. \cite{EMV1}). The quadratic terms in
Nordheim's equation (\ref{E5}), (\ref{E6}) are lower order terms. Their
contribution must be examined in detail for large values of $\omega ,$ since
their effect is the dominant one in that region. This detailed analysis of
the effect of the quadratic terms in (\ref{E5}) has been made in \cite{EMV2}%
, but in the analysis of (\ref{E1}), (\ref{E2}) we do not need to estimate
the effect of the quadratic terms. This allows to assume initial data $f_{0}$
with less stringent decay conditions, since the contributions of the cubic
terms for $\omega \rightarrow \infty $ can be estimated easily.
\end{proof}

\begin{remark}
Notice that the fluxes given in (\ref{KZ3}) cannot be prescribed for those
solutions, but they arise naturally as a consequence of the evolution of the
equation.
\end{remark}

\begin{remark}
The Kolmogorov-Zakharov \index{Kolmogorov-Zakharov}solutions $f_{s}(\omega )=K\omega ^{-7/6}$ are solutions of (\ref{E1}), (\ref{E2}) in the sense of
Theorem \ref{fluxSol} with initial datum  $f_{s}(\omega )=K\omega ^{-7/6}$. They are defined for arbitrarily large values of 
$T.$
\end{remark}
Our next goal is to make precise how the solutions obtained in
Theorem \ref{fluxSol} can be set in the framework of weak solutions defined
in Sections \ref{weakinteracting} and  \ref{weaknoninteracting}.We first remark that in both concepts of weak solutions
there (cf. Definitions \ref{weakSolution}, \ref{weakSolutionNI}) the
resulting weak solutions must satisfy $\partial _{t}\left( \int_{\left[
0,\infty \right) }g\left( t,d\omega \right) \right) =0$ if the initial mass
of the solutions is finite, as it can be readily seen using the test
function $\varphi =1.$ It is mathematically simpler to work with this type
of mass conserving weak solutions. However, due to (\ref{KZ3}) it would be
impossible to have weak solutions of (\ref{Z2E2a}) in the sense of
Definitions \ref{weakSolution} or \ref{weakSolutionNI} unless an additional
measure is added at the origin. Therefore, in our setting it is natural to
state that the Kolmogorov-Zakharov \index{Kolmogorov-Zakharov}solutions are:%
\begin{equation}
g_{KZ}\left( t,d\omega \right) =-J_{n}\left[ G_{0}\right] \left( 1\right)
K^{3}t\delta _{\omega =0}+\frac{Kd\omega }{\omega ^{\frac{2}{3}}}
\label{KZ1}
\end{equation}
It will be proved later that $J_{n}\left[ G_{0}\right] \left( 1\right) <0$, 
therefore, the mass at $\omega =0$ increases.

In general, given any $f\left( t,\omega \right) $ which solves (\ref{E1}), (%
\ref{E2})~and has the properties in Theorem \ref{fluxSol} we define:%
\begin{eqnarray}
g\left( t,d\omega \right) =m\left( t\right) \delta _{\omega =0}+\sqrt{\omega 
}f\left( t,\omega \right) d\omega \   \label{KZ4}\\
m\left( t\right) =-J_{n}\left[ G_{0}\right] \left( 1\right)
\int_{0}^{t}\left( a\left( s\right) \right) ^{3}ds  \label{KZ5}
\end{eqnarray}

We then have the following result:

\begin{theorem}
\label{WeakNI}Let $\sigma =0.$ The measure $g\left( t,\cdot \right) $
defined by means of (\ref{KZ4}), (\ref{KZ5}) solves (\ref{Z2E2a}) in the
sense of Definition \ref{weakSolutionNI}.
\end{theorem}

The following result has some independent interest, because it states in
which precise sense the Kolmogorov-Zakharov \index{Kolmogorov-Zakharov}solutions  
solve (\ref{Z2E2a}).

\begin{corollary}
Let $\sigma =0.$ The measure $g_{KZ}\left( t,d\omega \right) $ given by (\ref
{KZ1}) solves (\ref{Z2E2a}) in the sense of Definition \ref%
{weakSolutionNI}.
\end{corollary}

We also have the following results which suggests that the correct
definition of weak solutions for solutions with fluxes towards the origin is
Definition \ref{weakSolutionNI}.

\begin{theorem}
\label{WeakI}Let $\sigma =0.$ The measure $g\left( t,\cdot \right) $ defined
by means of (\ref{KZ4}), (\ref{KZ5}) does not satisfy Definition \ref%
{weakSolutionNI}.
\end{theorem}

In order to prove Theorems \ref{WeakNI}, \ref{WeakI} it is convenient to
prove the following Lemma.
\begin{lemma}
\label{grad}Suppose $\rho <-\frac{1}{2},$ $g\in \mathcal{M}_{+}\left( %
\left[ 0,\infty \right) :\left( 1+\omega \right) ^{\rho }\right) ,\ \varphi
\in C_{0}^{2}\left( \left[ 0,\infty \right) \right) $. Then:%
\begin{eqnarray}
&&\iiint_{\left( 0,\infty \right) ^{3}}\frac{\Phi g_{1}g_{2}g_{3}}{%
\sqrt{\omega _{1}\omega _{2}\omega _{3}}}\left[ \varphi \left( \omega
_{1}+\omega _{2}-\omega _{3}\right) +\varphi \left( \omega _{3}\right)
- \label{KZP1}\right.\\
&&\hskip 2cm  \left.-\varphi \left( \omega _{1}\right) -\varphi \left( \omega _{2}\right) \right]
d\omega _{1}d\omega _{2}d\omega _{3} =\int_{\left( 0,\infty \right) }J_{n}%
\left[ g\right] \left( \omega \right) \varphi ^{\prime }\left( \omega
\right) d\omega\nonumber
\end{eqnarray}%
where $J_{n}\left[ g\right] $ is as in (\ref{KZP2}).
\end{lemma}

\begin{proof}
We rewrite the left hand side of (\ref{KZP1}) as:%
\begin{eqnarray*}
&&\iiint_{\left( 0,\infty \right) ^{3}}Q\left[ g\right] \left[
\varphi \left( \omega _{1}+\omega _{2}-\omega _{3}\right) +\varphi \left(
\omega _{3}\right) -\varphi \left( \omega _{1}\right) -\varphi \left( \omega
_{2}\right) \right] d\omega _{1}d\omega _{2}d\omega _{3}= \\
&&\hskip 1.5cm =K_{1}+K_{2}\\
&&\hskip 1cmK_{1} =\iiint_{\left( 0,\infty \right) ^{3}}Q\left[ g\right] \left[
\varphi \left( \omega _{3}\right) -\varphi \left( \omega _{2}\right) \right]
d\omega _{1}d\omega _{2}d\omega _{3} \\
&&\hskip 1cmK_{2} =\iiint_{\left( 0,\infty \right) ^{3}}Q\left[ g\right] \left[
\varphi \left( \omega _{1}+\omega _{2}-\omega _{3}\right) -\varphi \left(
\omega _{1}\right) \right] d\omega _{1}d\omega _{2}d\omega _{3}
\end{eqnarray*}

We now use the fact that $\iiint_{\left\{ \omega _{2}=\omega
_{3}\right\} }Q\left[ g\right] \left[ \varphi \left( \omega _{3}\right)
-\varphi \left( \omega _{2}\right) \right] =0$ to obtain:%
\[
K_{1}=\iiint_{\left\{ \omega _{2}<\omega _{3}\right\} }\left[ \cdot
\cdot \cdot \right] +\iiint_{\left\{ \omega _{2}>\omega _{3}\right\}
}\left[ \cdot \cdot \cdot \right] 
\]

Using then that $\varphi \left( \omega _{3}\right) -\varphi \left( \omega
_{2}\right) =\int_{\omega _{2}}^{\omega _{3}}\varphi ^{\prime }\left( \omega
\right) d\omega $ for $\omega _{2}<\omega _{3}$ and $\varphi \left( \omega
_{3}\right) -\varphi \left( \omega _{2}\right) =-\int_{\omega _{3}}^{\omega
_{2}}\varphi ^{\prime }\left( \omega \right) d\omega $ for $\omega
_{3}<\omega _{2}$ we obtain, applying Fubini's Theorem:%
\[
K_{1}=\int_{\left( 0,\infty \right) }J_{n,1}\left[ g\right] \left( \omega
\right) \varphi ^{\prime }\left( \omega \right) d\omega -\int_{\left(
0,\infty \right) }J_{n,3}\left[ g\right] \left( \omega \right) \varphi
^{\prime }\left( \omega \right) d\omega 
\]

On the other hand:%
\[
K_{2}=\iiint_{\left\{ \omega _{2}<\omega _{3}\right\} }\left[ \cdot
\cdot \cdot \right] +\iiint_{\left\{ \omega _{2}>\omega _{3}\right\}
}\left[ \cdot \cdot \cdot \right] 
\]

Using then $\varphi \left( \omega _{1}+\omega _{2}-\omega _{3}\right)
-\varphi \left( \omega _{1}\right) =\int_{\omega _{1}}^{\omega _{1}+\omega
_{2}-\omega _{3}}\varphi ^{\prime }\left( \omega \right) d\omega $ for $%
\omega _{2}>\omega _{3}$ and $\varphi \left( \omega _{1}+\omega _{2}-\omega
_{3}\right) -\varphi \left( \omega _{1}\right) =\int_{\omega _{1}+\omega
_{2}-\omega _{3}}^{\omega _{1}}\varphi ^{\prime }\left( \omega \right)
d\omega $ for $\omega _{3}>\omega _{2}$ we obtain, arguing in a similar
manner:%
\[
K_{2}=\int_{\left( 0,\infty \right) }J_{n,2}\left[ g\right] \left( \omega
\right) \varphi ^{\prime }\left( \omega \right) d\omega -\int_{\left(
0,\infty \right) }J_{n,4}\left[ g\right] \left( \omega \right) \varphi
^{\prime }\left( \omega \right) d\omega 
\]%
whence the result follows.
\end{proof}
\begin{proof}[Proof of Theorem \protect\ref{WeakNI}]
By assumption $f$ solves (\ref{E1}), (\ref{E2}). Multiplying (\ref{E1}) by $%
\sqrt{\omega _{1}}\tilde{\varphi}\left( \omega _{1}\right) $ with $\tilde{%
\varphi}\in C_{0}^{1}\left( \left[ 0,\infty \right) \right) ,\ $ where and $%
\tilde{\varphi}\left( 0\right) =0$ we obtain, after some changes of
variables:%
\begin{eqnarray}
&&\partial _{t}\left( \int g\tilde{\varphi}\right) =\iiint_{\left(
0,\infty \right) ^{3}}Q\left[ g\right] \left[ \tilde{\varphi}\left( \omega
_{1}+\omega _{2}-\omega _{3}\right) +\tilde{\varphi}\left( \omega
_{3}\right)-\right. \label{KZP8} \\
&&\hskip 6cm \left.  -\tilde{\varphi}\left( \omega _{1}\right) -\tilde{\varphi}\left(
\omega _{2}\right) \right] d\omega _{1}d\omega _{2}d\omega _{3} \nonumber
\end{eqnarray}%
where $g$ is as in (\ref{KZ4}), (\ref{KZ5}).

Given $\varphi \in C_{0}^{2}\left( \left[ 0,\infty \right) \right) $ we
split it as $\varphi =\phi _{\varepsilon }+\tilde{\varphi}_{\varepsilon }$
where $\phi _{\varepsilon },\tilde{\varphi}_{\varepsilon }\in
C_{0}^{1}\left( \left[ 0,\infty \right) \right) ,$ $\phi _{\varepsilon
}\left( \omega \right) =\left( 1-\frac{\omega }{\varepsilon }\right)
_{+}\varphi ,\ \tilde{\varphi}_{\varepsilon }\left( 0\right) =0.$ Then:%
\begin{eqnarray*}
&& \hskip -0.5cm \iiint_{\left( 0,\infty \right) ^{3}}Q\left[ g\right] \left[
\varphi \left( \omega _{1}+\omega _{2}-\omega _{3}\right) +\varphi \left(
\omega _{3}\right) -\varphi \left( \omega _{1}\right) -\varphi \left( \omega
_{2}\right) \right] d\omega _{1}d\omega _{2}d\omega _{3} \\
&&\hskip 1cm  =K_{3}+K_{4} \\
&&\hskip -0.5cm K_{3} =\iiint_{\left( 0,\infty \right) ^{3}}Q\left[ g\right] \left[
\tilde{\varphi}_{\varepsilon }\left( \omega _{1}+\omega _{2}-\omega
_{3}\right) +\tilde{\varphi}_{\varepsilon }\left( \omega _{3}\right) -\tilde{%
\varphi}_{\varepsilon }\left( \omega _{1}\right) -\tilde{\varphi}%
_{\varepsilon }\left( \omega _{2}\right) \right] d\omega _{1}d\omega
_{2}d\omega _{3} \\
&&\hskip -0.5cm K_{4} =\iiint_{\left( 0,\infty \right) ^{3}}Q\left[ g\right] \left[
\phi _{\varepsilon }\left( \omega _{1}+\omega _{2}-\omega _{3}\right) +\phi
_{\varepsilon }\left( \omega _{3}\right) -\phi _{\varepsilon }\left( \omega
_{1}\right) -\phi _{\varepsilon }\left( \omega _{2}\right) \right] d\omega
_{1}d\omega _{2}d\omega _{3}
\end{eqnarray*}

We now remark that, the asymptotics of $f$ as $\omega \rightarrow 0,$ stated
in Theorem \ref{fluxSol} implies:%
\begin{equation}
\lim_{\omega \rightarrow 0}\left[ J_{n}\left[ g\right] \left( \omega \right) %
\right] =\left( a\left( t\right) \right) ^{3}J_{n}\left[ G_{0}\right] \left(
1\right) ,\ \text{with }G_{0}\left( \omega \right) =\left( \omega \right) ^{-%
\frac{2}{3}}  \label{KZP7}
\end{equation}

The proof of (\ref{KZP7}) just requires to see that the asymptotics of $J_{n}%
\left[ g\right] \left( \omega \right) $ depends only on the local behaviour
of $g$ as $\omega \rightarrow 0.$ The arguments requires for the proof are
rather similar to the ones in the computations of the fluxes in \cite{Spohn}.

Applying Lemma \ref{grad} to compute $K_{4}$ we obtain:%
\[
K_{4}=\int_{\left( 0,\infty \right) }J_{n}\left[ g\right] \left( \omega
\right) \phi _{\varepsilon }^{\prime }\left( \omega \right) d\omega 
\]%
and taking the limit $\varepsilon \rightarrow 0$ we obtain:%
\begin{equation}
K_{4}=\varphi \left( 0\right) \lim_{\omega \rightarrow 0}\left[ J_{n}\left[ g%
\right] \left( \omega \right) \right] =-\varphi \left( 0\right) \left(
a\left( t\right) \right) ^{3}J_{n}\left[ G_{0}\right] \left( 1\right)
\label{Z1}
\end{equation}

On the other hand, we compute $\partial _{t}\left( \int_{\left[ 0,\infty
\right) }\varphi g\left( t,d\omega \right) \right) $ as:%
\begin{equation}
\partial _{t}\left( \int_{\left[ 0,\infty \right) }\varphi g\left( t,d\omega
\right) \right) =\partial _{t}\left( \int_{\left[ 0,\infty \right) }\phi
_{\varepsilon }\left( \omega \right) g\left( t,d\omega \right) \right)
+\partial _{t}\left( \int_{\left[ 0,\infty \right) }\tilde{\varphi}\left(
\omega \right) g\left( t,d\omega \right) \right)  \label{KZP9}
\end{equation}

We compute the difference $\partial _{t}\left( \int_{\left[ 0,\infty \right)
}\varphi g\left( t,d\omega \right) \right) -\left( K_{3}+K_{4}\right) .$
Using (\ref{KZP8}), (\ref{KZP9}) we obtain that this difference is:%
\[
\partial _{t}\left( \int_{\left[ 0,\infty \right) }\phi _{\varepsilon
}\left( \omega \right) g\left( t,d\omega \right) \right) -K_{4} 
\]

The integrated version of this equation is:%
\[
\int_{\left[ 0,\infty \right) }\phi _{\varepsilon }\left( \omega \right)
g\left( t,d\omega \right) -\int_{\left[ 0,\infty \right) }\phi _{\varepsilon
}\left( \omega \right) g_{0}\left( d\omega \right) -\int_{0}^{t}K_{4}ds 
\]

Taking the limit $\varepsilon \rightarrow 0$ and using (\ref{Z1}) we obtain:%
\[
\left[ m\left( t\right) +\int_{0}^{t}\left( a\left( s\right) \right)
^{3}J_{n}\left[ G_{0}\right] \left( 1\right) ds\right] \varphi \left(
0\right) 
\]

Using (\ref{KZ5}) we obtain the cancellation of this quantity, and the
result follows.
\end{proof}

The Proof of Theorem \ref{WeakI} is now elementary.
\begin{proof}[Proof of Theorem \protect\ref{WeakI}]
In order to prove that $g$ is a solution of (\ref{Z2E2a}) in the sense of
Definition \ref{weakSolutionNI} we need to compute $\partial _{t}\left(
\int_{\left[ 0,\infty \right) }\varphi g\left( t,d\omega \right) \right)
-\iiint_{\left[ 0,\infty \right) ^{3}}\left[ \cdot \cdot \cdot %
\right] .$ Since, as we have seen in the Proof of Theorem \ref{WeakNI}, $\partial
_{t}\left( \int_{\left[ 0,\infty \right) }\varphi g\left( t,d\omega \right)
\right) -\iiint_{\left( 0,\infty \right) ^{3}}\left[ \cdot \cdot
\cdot \right] =0$,  we have:
\begin{equation}
\partial _{t}\left( \int_{\left[ 0,\infty \right) }\varphi g\left( t,d\omega
\right) \right) -\iiint_{\left[ 0,\infty \right) ^{3}}\left[ \cdot
\cdot \cdot \right] =-m(t) \iiint_{\mathcal{Z}}\left[ \cdot \cdot \cdot %
\right] \ \   \label{SolI}
\end{equation}%
where $\mathcal{Z}=\left( \left\{ 0\right\} \times \left[ 0,\infty \right)
^{2}\right) \cup \left( \left[ 0,\infty \right) \times \left\{ 0\right\}
\times \left[ 0,\infty \right) \right) \cup \left( \left[ 0,\infty \right)
^{2}\times \left\{ 0\right\} \right) .$ We now use the fact that Lemma \ref%
{Cont} implies that the lines $\Gamma _{i,j}=\left\{ \omega _{i}=\omega
_{j}=0\right\} $ do not contribute to the integral. We can then replace the
set  $\mathcal{Z}$ by $\mathcal{Y=}\left( \left\{ 0\right\}
\times \left( 0,\infty \right) ^{2}\right) \cup \left( \left( 0,\infty
\right) \times \left\{ 0\right\} \times \left( 0,\infty \right) \right) \cup
\left( \left( 0,\infty \right) ^{2}\times \left\{ 0\right\} \right) .$ Then
the integral in right-hand side of (\ref{SolI}) becomes: 
\begin{eqnarray*}
\iint_{\left( 0,\infty \right) ^{2}\cap \left\{
\omega _{2}>\omega _{3}\right\} }\!\!\frac{g_{2}g_{3}}{\sqrt{\omega _{2}\omega
_{3}}}\left[ \varphi \left( \omega _{2}-\omega _{3}\right) +\varphi \left(
\omega _{3}\right) -\varphi \left( 0\right) -\varphi \left( \omega
_{2}\right) \right] d\omega _{2}d\omega _{3}+ \\
+ \iint_{\left( 0,\infty \right) ^{2}\cap \left\{
\omega _{1}>\omega _{3}\right\} }\!\!\frac{g_{1}g_{3}}{\sqrt{\omega _{1}\omega
_{3}}}\left[ \varphi \left( \omega _{1}-\omega _{3}\right) +\varphi \left(
\omega _{3}\right) -\varphi \left( \omega _{1}\right) -\varphi \left(
0\right) \right] d\omega _{1}d\omega _{3}+ \\
+\iint_{\left( 0,\infty \right) ^{2}}\!\!\frac{g_{1}g_{2}%
}{\sqrt{\omega _{1}\omega _{2}}}\left[ \varphi \left( \omega _{1}+\omega
_{2}\right) +\varphi \left( 0\right) -\varphi \left( \omega _{1}\right)
-\varphi \left( \omega _{2}\right) \right] d\omega _{1}d\omega _{2}
\end{eqnarray*}

We will assume that $\varphi $ is compactly supported in $\left( 0,\infty
\right) $ in order to have convergence of the integrals. Therefore $\varphi
\left( 0\right) =0.$ Relabelling $\omega _{3}$ to $\omega _{2}$ in the
second integral we obtain: 
\begin{eqnarray*}
&&2 \iint_{\left( 0,\infty \right) ^{2}\cap \left\{
\omega _{1}>\omega _{2}\right\} }\frac{g_{1}g_{2}}{\sqrt{\omega _{1}\omega
_{2}}}\left[ \varphi \left( \omega _{1}-\omega _{2}\right) +\varphi \left(
\omega _{2}\right) -\varphi \left( \omega _{1}\right) \right] d\omega
_{1}d\omega _{2}+ \\
&&+ \iint_{\left( 0,\infty \right) ^{2}}\frac{g_{1}g_{2}%
}{\sqrt{\omega _{1}\omega _{2}}}\left[ \varphi \left( \omega _{1}+\omega
_{2}\right) -\varphi \left( \omega _{1}\right) -\varphi \left( \omega
_{2}\right) \right] d\omega _{1}d\omega _{2}
\end{eqnarray*}

A symmetrization argument and rearrangement of the different terms yields:%
\begin{eqnarray*}
 \iint_{\left( 0,\infty \right) ^{2}}\frac{g_{1}g_{2}}{%
\sqrt{\omega _{1}\omega _{2}}}\left[ \varphi \left( \omega _{1}+\omega
_{2}\right) +\varphi \left( \left\vert \omega _{1}-\omega _{2}\right\vert
\right) +\varphi \left( \min \left\{ \omega _{1},\omega _{2}\right\} \right)-\right.\\
\left.-\varphi \left( \max \left\{ \omega _{1},\omega _{2}\right\} \right)
-\varphi \left( \omega _{1}\right) -\varphi \left( \omega _{2}\right) \right]
d\omega _{1}d\omega _{2}
\end{eqnarray*}

Since the integrand is symmetric under the transformation $\omega
_{1}\longleftrightarrow \omega _{2}$ we can transform this integral in:%
\begin{eqnarray}
&&\iint _{\substack{( 0,\infty ) ^{2}\\ \left\{
\omega _{1}>\omega _{2}\right\}} }\frac{g_{1}g_{2}}{\sqrt{\omega _{1}\omega
_{2}}}\left[ \varphi \left( \omega _{1}+\omega _{2}\right) +\varphi \left(
\omega _{1}-\omega _{2}\right) -2\varphi \left( \omega _{1}\right) \right]
d\omega _{1}d\omega _{2}  \label{OpGs}
\end{eqnarray}

Suppose first that $g\left( \omega \right) =\omega ^{-\frac{2}{3}}.$ In this
case this integral reduces to:%
\[
2\int_{0}^{\infty }\frac{d\omega _{1}}{\left( \omega
_{1}\right) ^{\frac{7}{6}}}\int_{0}^{\omega _{1}}\frac{d\omega _{2}}{\left(
\omega _{2}\right) ^{\frac{7}{6}}}\left[ \varphi \left( \omega _{1}+\omega
_{2}\right) +\varphi \left( \omega _{1}-\omega _{2}\right) -2\varphi \left(
\omega _{1}\right) \right] 
\]

Integrating by parts in the integral in $\omega _{2}$ we obtain:%
\begin{eqnarray*}
&&-12\int_{0}^{\infty }\frac{d\omega _{1}}{\left( \omega
_{1}\right) ^{\frac{4}{3}}}\left[ \varphi \left( 2\omega _{1}\right)
+\varphi \left( 0\right) -2\varphi \left( \omega _{1}\right) \right]+\\
&&\hskip 1cm +12 \int_{0}^{\infty }\frac{d\omega _{1}}{\left( \omega
_{1}\right) ^{\frac{7}{6}}}\int_{0}^{\omega _{1}}d\omega _{2}\left( \omega
_{2}\right) ^{-\frac{1}{6}}\left[ \varphi ^{\prime }\left( \omega
_{1}+\omega _{2}\right) -\varphi ^{\prime }\left( \omega _{1}-\omega
_{2}\right) \right] 
\end{eqnarray*}

Applying Fubini in the second integral this integral becomes:%
\begin{eqnarray*}
&&-12 \int_{0}^{\infty }\frac{d\omega _{1}}{\left( \omega
_{1}\right) ^{\frac{4}{3}}}\left[ \varphi \left( 2\omega _{1}\right)
+\varphi \left( 0\right) -2\varphi \left( \omega _{1}\right) \right]+\\
&&\hskip 1cm 
+12\int_{0}^{\infty }\left( \omega _{2}\right) ^{-\frac{1}{6%
}}d\omega _{2}\int_{\omega _{2}}^{\infty }\left[ \varphi ^{\prime }\left(
\omega _{1}+\omega _{2}\right) -\varphi ^{\prime }\left( \omega _{1}-\omega
_{2}\right) \right] \frac{d\omega _{1}}{\left( \omega _{1}\right) ^{\frac{7}{%
6}}}
\end{eqnarray*}
and integrating by parts in the integral with respect to $\omega _{1}$
the integral becomes:
\begin{eqnarray*}
&&-24 \int_{0}^{\infty }\frac{d\omega _{1}}{\left( \omega
_{1}\right) ^{\frac{4}{3}}}\left[ \varphi \left( 2\omega _{1}\right)
-\varphi \left( \omega _{1}\right) \right] +\\
&&+14\int_{0}^{\infty }\left( \omega _{2}\right) ^{-\frac{1%
}{6}}d\omega _{2}\int_{\omega _{2}}^{\infty }\left[ \varphi \left( \omega
_{1}+\omega _{2}\right) -\varphi \left( \omega _{1}-\omega _{2}\right) %
\right] \frac{d\omega _{1}}{\left( \omega _{1}\right) ^{\frac{13}{6}}}
\end{eqnarray*}

Using the changes of variables $\xi =2\omega _{1}$ in the integral
containing $\varphi \left( 2\omega _{1}\right) ,$ $\xi =\omega _{1}+\omega
_{2}$ in the integral containing $\varphi \left( \omega _{1}+\omega
_{2}\right) $ and $\xi =\omega _{1}-\omega _{2}$ in the integral containing $%
\varphi \left( \omega _{1}-\omega _{2}\right) $ we obtain, applying Fubini
in the second term:%
\begin{align*}
&24 \left( 2^{\frac{1}{3}}-1\right) \int_{0}^{\infty }\frac{%
\varphi \left( \xi \right) d\xi }{\left( \xi \right) ^{\frac{4}{3}}}%
-14\int_{0}^{\infty }\varphi \left( \xi \right) d\xi
\int_{0}^{\frac{\xi }{2}}\left( \omega _{2}\right) ^{-\frac{1}{6}}\left( \xi
-\omega _{2}\right) ^{-\frac{13}{6}}d\omega _{2}+ \\
& +14\int_{0}^{\infty }\varphi \left( \xi \right) d\xi
\int_{0}^{\infty }\left( \omega _{2}\right) ^{-\frac{1}{6}}\left( \xi
+\omega _{2}\right) ^{-\frac{13}{6}}d\omega _{2} =\int_{0}^{\infty }\varphi \left( \xi \right) F\left( \xi
\right) d\xi 
\end{align*}%
\begin{eqnarray*}
&&F\left( \xi \right)  =  \frac{24\left( 2^{\frac{1}{3}}-1\right) }{\left( \xi
\right) ^{\frac{4}{3}}}-14\int_{0}^{\frac{\xi }{2}}\left( \omega _{2}\right)
^{-\frac{1}{6}}\left( \xi -\omega _{2}\right) ^{-\frac{13}{6}}d\omega
_{2}+\\
&&\hskip 6cm +14\int_{0}^{\infty }\left( \omega _{2}\right) ^{-\frac{1}{6}}\left( \xi
+\omega _{2}\right) ^{-\frac{13}{6}}d\omega _{2}
\end{eqnarray*}
A rescaling argument yields:%
\begin{equation}
F\left( \xi \right) =\frac{c_{\ast }}{\left( \xi \right) ^{\frac{4}{3}}}\ 
\label{Fsize}
\end{equation}%
where:%
\[c_{\ast }=24\left( 2^{\frac{1}{3}}-1\right) -14\int_{0}^{\frac{1}{2}}
\left(x\right) ^{-\frac{1}{6}}\left( 1-x\right) ^{-\frac{13}{6}}dx+14\int_{0}^{\infty }
\left( x\right) ^{-\frac{1}{6}}\left( 1+x\right) ^{-\frac{13}{6}}dx\]

Using the change of variables $y=(1+x)^{-1}$ in the last integral we can
transform it in a Beta function. The second integral can be written in terms
of the incomplete Beta function $B_{z}\left( \frac{5}{6},-\frac{7}{6}\right)$
(cf. \cite{AS}). Then: 
\[
c_{\ast }=24\left( 2^{\frac{1}{3}}-1\right) -14B_{z=\frac{1}{2}}\left( \frac{%
5}{6},-\frac{7}{6}\right) +14B\left( \frac{4}{3},\frac{5}{6}\right) 
\]
and its numerical value is $c_{\ast }=0.32964...\neq 0$.
Therefore, if $g=\omega ^{-\frac{2}{3}}$ we obtain that the right-hand side
of (\ref{SolI}) defines a nonzero functional. In particular there exist
functions $\varphi $ for which the integral $\iiint_{\mathcal{Z}}%
\left[ \cdot \cdot \cdot \right] \ $\ is different from zero. 

On the other hand, for arbitrary functions $g$ obtained as $g=\sqrt{\omega }f
$ with $f$ as in Theorem \ref{fluxSol} we obtain a similar result due to the
fact that $g\left( \omega \right) $ is asymptotically close to $K\omega ^{-%
\frac{2}{3}}$ as $\omega \rightarrow 0.$ Indeed, we can consider test
functions $\varphi \left( \omega \right) =\psi \left( \frac{\omega }{%
\varepsilon }\right) ,$ with $\varepsilon >0$ small and $\psi $ compactly
supported in $\left( 0,\infty \right) $. We can approximate $g$ by means of $%
K\omega ^{-\frac{2}{3}}$ with an error of order $\delta _{1}\omega ^{-\frac{2%
}{3}}$ with $\delta _{1}$ small, if $0<\omega <\delta _{2}.$ Due to (\ref%
{Fsize}) we obtain that the contribution to (\ref{OpGs}) of the leading term 
$K\omega ^{-\frac{2}{3}}$ is of order $m\left( t\right) K^{2}\varepsilon ^{-%
\frac{1}{3}}.$ The error term due to the remainder $\delta _{1}\omega ^{-%
\frac{2}{3}}$ in the region $0<\omega <\delta _{2}$ can be estimated using (%
\ref{OpGs}) by:%
\[
m\left( t\right) \delta _{1}\int_{0}^{\delta _{2}}\frac{d\omega _{1}}{\left(
\omega _{1}\right) ^{\frac{7}{6}}}\int_{0}^{\omega _{1}}\frac{d\omega _{2}}{%
\left( \omega _{2}\right) ^{\frac{7}{6}}}\left\vert \varphi \left( \omega
_{1}+\omega _{2}\right) +\varphi \left( \omega _{1}-\omega _{2}\right)
-2\varphi \left( \omega _{1}\right) \right\vert 
\]

Using the form of the function $\varphi $ and a rescaling argument we estimate this term
by $m(t)\delta _{1}\varepsilon ^{-\frac{1}{3}}.$ This
contribution is small compared with the leading term. On the other hand, the
regions where some of the variables $\omega _{1},\ \omega _{2}$ are larger
than $\delta _{2}$ can be estimated, due to the fact that the support of $%
\varphi $ has size $\varepsilon ,$ as well as $\omega _{1}>\omega _{2},$\ by:%
\[
m\left( t\right) \iint_{\left( 0,\infty \right) ^{2}\cap \left\{ \omega
_{1}>\omega _{2}\geq \frac{\delta }{2}\right\} }\frac{g_{1}g_{2}}{\sqrt{%
\omega _{1}\omega _{2}}}\left\vert \varphi \left( \omega _{1}-\omega
_{2}\right) \right\vert d\omega _{1}d\omega _{2}
\]

Since by assumption $g\left( \omega \right) \leq \frac{C}{\left( 1+\omega
\right) ^{\rho }}$ for $\omega \geq 1,$ with $\rho >\frac{1}{2}$ and the
support of $\varphi $ has size $\varepsilon ,$ we can estimate this term as:%
\[
C_{\delta _{2}}m\left( t\right) \varepsilon \iint_{\left( 0,\infty
\right) ^{2}\cap \left\{ \omega _{1}>\omega _{2}\geq \frac{\delta _{2}}{2}%
\right\} }\frac{d\omega }{\left( 1+\omega \right) ^{2\rho +1}}\leq C_{\delta
_{2}}m\left( t\right) \varepsilon 
\]

Therefore, the contribution of this term is negligible if $\varepsilon $ is
small enough. This concludes the proof of the result.
\end{proof}

\begin{remark}\index{condensate}
If we had not added the mass at $\omega =0$ in (\ref{KZ1}) or (%
\ref{KZ4}) the resulting measures $g$ would not be a weak solution of (\ref%
{Z2E2a}) in the sense of neither Definition \ref{weakSolution} or \ref%
{weakSolutionNI}. Due to the absence of condensate in $g$ both definitions
would be equivalent and then it is enough to check that Definition \ref%
{weakSolutionNI} fails. This follows from the fact that the term\ $\partial
_{t}\left( \int_{\left[ 0,\infty \right) }\varphi g\left( t,d\omega \right)
\right) -\left( K_{3}+K_{4}\right) $ computed in the Proof of Theorem \ref%
{WeakNI} yields $\varphi \left( 0\right) \int_{0}^{t}\left( a\left( s\right)
\right) ^{3}J_{n}\left[ G_{0}\right] \left( 1\right) ds\neq 0.$
\end{remark}

\begin{remark}
We have proved that the choice of mass at $\omega =0$ in (\ref{KZ4}) gives a
measure $g$ which does not solve (\ref{Z2E2a}) in the sense of Definition %
\ref{weakSolution}. It is natural to ask if other choice of the mass at $%
\omega =0$ could give such a solution. In the mass conserving case this
cannot happen because the choice of the mass made in (\ref{KZ4}) at $\omega
=0$ is the only one compatible with mass conservation for the measure $g.$
\end{remark}

The existence of weak solutions with interacting condensate has been considered in \cite{Lu3} for the Nordheim equation. \index{Nordheim}

\subsubsection{Negativity of the fluxes.}

We can now prove that $J_{n}\left[ G_{0}\right] \left( 1\right) 
$, the constant  that characterizes the fluxes from $\left\{ \omega >0\right\} $ to $%
\left\{ \omega =0\right\} $, is strictly negative. This has been shown in 
\cite{DNPZ} using a representation formula for the fluxes inspired in the
computations of Zakharov which give the exponents characterizing the
stationary power law solutions of the weak turbulence \index{weak turbulence}equations. In \cite%
{Spohn} the negativity of this constant has been obtained computing it
numerically. We prove here that the negativity of this constant is a
consequence of the Monotonicity Formula. We remark that the choice of signs
in \cite{DNPZ} is the reverse of the one used in \cite{Spohn} as well as in
this paper. It is assumed in \cite{DNPZ} that fluxes tranporting particles
from larger values of $\omega $ to smaller values of $\omega $ are positive.

\begin{theorem}
The constant $J_{n}\left[ G_{0}\right] \left( 1\right) $ has the following
representation formula:%
\begin{eqnarray}
&&J_{n}\left[ G_{0}\right] \left( 1\right) = -\frac{1}{3}\iiint_{\left(
0,\infty \right) ^{3}}\,\frac{d\omega _{1}d\omega _{2}d\omega _{3}}{\left(
\omega _{1}\omega _{2}\omega _{3}\right) ^{\frac{7}{6}}}\times \label{signo}\\
&&\hskip 2cm \times \left[ \sqrt{\omega
_{-}}H_{\varphi }^{1}\left( \omega _{1},\omega _{2},\omega _{3}\right) +%
\sqrt{\left( \omega _{0}+\omega _{-}-\omega _{+}\right) _{+}}H_{\varphi
}^{2}\left( \omega _{1},\omega _{2},\omega _{3}\right) \right] \ 
\nonumber
\end{eqnarray}%
where the functions $H_{\varphi }^{1}$,\ $H_{\varphi }^{2}$ are as in Lemma %
\ref{strictConvex} with $\varphi \left( \omega \right) =\left( 1-\omega
\right) _{+}.$ Moreover $J_{n}\left[ G_{0}\right] \left( 1\right) <0.$
\end{theorem}

\begin{proof}
We take as starting point (\ref{KZP1}) with $g=G_{0}$. In this case $J_{n}%
\left[ g\right] \left( \omega \right) $ is constant. Using the test function 
$\varphi \left( \omega \right) =\left( 1-\omega \right) _{+}$ the right-hand
side of (\ref{KZP1}) reduces to $-J_{n}\left[ G_{0}\right] \left( 1\right) .$
On the other hand, using Proposition \ref{atractiveness} and Lemma \ref%
{strictConvex} we can rewrite the left-hand side of (\ref{KZP1}) as the
right-hand side of (\ref{signo}) with the reverse sign. Due to Lemma \ref%
{strictConvex} we have $H_{\varphi }^{1}\geq 0$,\ $H_{\varphi }^{2}\geq 0.$
Moreover, these functions are strictly positive at least in some compact
subsets of $\left( 0,\infty \right) ^{3},$ whence the result follows.
\end{proof}

\subsubsection{Energy fluxes.}

We can obtain formulas for the energy fluxes analogous to (\ref{KZP1}). This
formula will allow us to prove, using elementary dimensional analysis
arguments, that for the Kolmogorov-Zakharov \index{Kolmogorov-Zakharov}solutions the fluxes of energy vanish.

\begin{lemma}
\label{EnLemma}Let Suppose that $\rho <-\frac{1}{2},$ $g\in \mathcal{M}%
_{+}\left( \left[ 0,\infty \right) :\left( 1+\omega \right) ^{\rho }\right)
,\ \varphi \in C_{0}^{2}\left( \left[ 0,\infty \right) \right) $. Let $%
\varphi \left( \omega \right) =\omega \psi \left( \omega \right) .$ Then:%
\begin{eqnarray}
&&\iiint_{\left( 0,\infty \right) ^{3}}\frac{\Phi g_{1}g_{2}g_{3}}{%
\sqrt{\omega _{1}\omega _{2}\omega _{3}}}\left[ \varphi \left( \omega
_{1}+\omega _{2}-\omega _{3}\right) +\varphi \left( \omega _{3}\right)
-\varphi \left( \omega _{1}\right) -\varphi \left( \omega _{2}\right) \right]
d\omega _{1}d\omega _{2}d\omega _{3}\label{EnFlux}\\
&&\hskip 2cm =\int_{\left( 0,\infty \right) }J_{e}%
\left[ g\right] \left( \omega \right) \varphi ^{\prime }\left( \omega
\right) d\omega  \nonumber
\end{eqnarray}%
where: 
\begin{eqnarray}
&&J_{e}\left[ g\right] \left( \omega \right) =\iiint_{\left( 0,\infty
\right) ^{3}}\frac{\Phi g_{1}g_{2}g_{3}}{\sqrt{\omega _{1}\omega _{2}\omega
_{3}}}F\left( \omega _{1},\omega _{2},\omega _{3};\omega \right) d\omega
_{1}d\omega _{2}d\omega _{3}\   \label{JE}\\
&&F\left( \omega _{1},\omega _{2},\omega _{3};\omega \right) =\left( \omega
_{1}+\omega _{2}-\omega _{3}\right) \Omega \left( \omega _{1},\omega ,\omega
_{1}+\omega _{2}-\omega _{3}\right) +\nonumber\\
&&\hskip 5cm +\omega _{3}\Omega \left( \omega
_{1},\omega ,\omega _{3}\right) -\omega _{2}\Omega \left( \omega _{1},\omega
,\omega _{3}\right) \nonumber
\end{eqnarray}
where the function $\Omega \left( \xi ,\zeta ,\eta \right) $ is defined as
follows:%
\begin{eqnarray*}
\Omega \left( \xi ,\zeta ,\eta \right) &=&1\text{ if }\xi <\zeta <\eta \\
\Omega \left( \xi ,\zeta ,\eta \right) &=&-1\text{ if }\eta <\zeta <\xi \\
\Omega \left( \xi ,\zeta ,\eta \right) &=&0\text{ \ otherwise }
\end{eqnarray*}
\end{lemma}

\begin{proof}
Using the definition of $\varphi $ we can write:%
\begin{eqnarray*}
&&\varphi \left( \omega _{1}+\omega _{2}-\omega _{3}\right) +\varphi \left(
\omega _{3}\right) -\varphi \left( \omega _{1}\right) -\varphi \left( \omega
_{2}\right) = \\
&&=\left( \omega _{1}+\omega _{2}-\omega _{3}\right) \int_{\omega
_{1}}^{\omega _{1}+\omega _{2}-\omega _{3}}\psi ^{\prime }\left( \omega
\right) d\omega +\omega _{3}\int_{\omega _{1}}^{\omega _{3}}\psi ^{\prime
}\left( \omega \right) d\omega -\omega _{2}\int_{\omega _{1}}^{\omega
_{3}}\psi ^{\prime }\left( \omega \right) d\omega
\end{eqnarray*}%
where some of the integrals must be understood with a negative sign if the
limits of integration are not ordered. The Lemma then follows by
Fubini's Theorem.
\end{proof}

\subsubsection{Kolmogorov-Zakharov solutions and dimensional considerations.}\index{Kolmogorov-Zakharov}

The weak formulation    in Definition \ref{weakSolutionNI}
combined with (\ref{KZP1}), (\ref{EnFlux}) imply the following equations in
the sense of distributions:%
\begin{eqnarray}
\partial _{t}\left( g\right) +\partial _{\omega }\left( J_{n}\right) &=&0\ \
,\ \ \omega >0  \label{T1} \\
\partial _{t}\left( \omega g\right) +\partial _{\omega }\left( J_{e}\right)
&=&0\ \ ,\ \ \omega >0  \label{T2}
\end{eqnarray}

These equations describe the transport of mass and energy in the set $%
\left\{ \omega >0\right\}$. The Kolmogorov-Zakharov \index{Kolmogorov-Zakharov}solutions arise
naturally from (\ref{T1}), (\ref{T2}). Indeed,  
$g_{s}(\omega ) = K\omega^{-2/3}$ is a solution of (\ref{T1})
in which $J_{n}=-c_{0}=J_{n}\left[ G_{0}\right] \left( 1\right) K^{3}.$ On
the other hand, it would be natural to define a second Kolmogorov-Zakharov solution using (%
\ref{T2}), and more precisely, assuming that $g$ is a power law yielding $%
J_{e}=constant.$ Using the rescaling properties of $J_{e}$ (cf. (\ref{JE}))
this would suggest $g_{s}\left( \omega \right) =c_1\,\omega^{-1}$  for some constant $c_1>0$. This
solution is usually assumed to be an admissible Kolmogorov-Zakharov \index{Kolmogorov-Zakharov}solution in the physical
literature. It would be associated to a transport of energy from low values
of $\omega $ to higher values. However, since many of the integrals defining
the fluxes are non-convergent we decided not to discuss it. It is not clear
in which sense the solution $g_{s}(\omega )=c_1\, \omega ^{-1}$
is a solution of (\ref{Z2E2a}) and for that reason we prefer not to pursue
its analysis for the moment.

A remarkable fact of the Kolmogorov-Zakharov \index{Kolmogorov-Zakharov}solution 
$g_{s}$ is that its energy
fluxes vanish for any value of $\omega .$ This can be seen in \cite{DNPZ}
using a suitable representation formula for the energy fluxes. We provide a
different proof of this fact here that only requires dimensional analysis
considerations. Due to the homogeneity of the functional $J_{e}$ with
respect to $g$ it is enough to prove that the energy fluxes vanish for $%
g=G_{0},$ with $G_{0}\left( \omega \right) =\omega ^{-\frac{2}{3}}.$

\begin{proposition}
Let $J_{e}\left[ g\right] $ be as in (\ref{JE}). Then $J_{e}\left[ G_{0}%
\right] \left( \omega \right) =0\ $for any $\omega >0.$
\end{proposition}

\begin{proof}
Definition \ref{weakSolutionNI} combined with Lemma \ref{EnLemma} imply that 
$G_{0}$ is a distributional solution of the equation:%
\begin{equation}
\partial _{t}\left( \omega G_{0}\right) +\partial _{\omega }\left( J_{e} 
\left[ G_{0}\right] \right) =0\ \ ,\ \ \omega >0  \label{T3}
\end{equation}

Since $G_{0}$ is a power law we have, due to the homogeneity properties of $%
J_{e}$ that $J_{e}\left[ G_{0}\right] \left( \omega \right) =c_{0}\omega $
for a suitable constant $c_{0}>0.$ Since $G_{0}$ is stationary we obtain
from (\ref{T3}) that $\partial _{\omega }\left( J_{e}\left[ G_{0}\right]
\right) =0.$ Integrating in the interval $\left[ 1,2\right] $ we obtain $%
J_{e}\left[ G_{0}\right] \left( 2\right) =J_{e}\left[ G_{0}\right] \left(
1\right)$, whence $c_{0}=0$ and the result follows. To make this argument
fully rigorous, the weak formulation of (\ref{T3})  with
suitable test functions must be used,  followed by a passage to the limit.
Since this argument is classical and elementary we skip it. \end{proof}
\subsubsection{On more general concepts of weak solutions of (\protect\ref%
{Z2E2a}).}
We have defined two concepts of weak solutions of (\ref{Z2E2a}), namely
Definitions \ref{weakSolution} and \ref{weakSolutionNI}. The main difference
between the two Definitions is the form in which the particles in the
condensate interact with the remaining particles of the system. \index{condensate}A
consequence of this is that for measures $g$ without condensate both
definitions are identical. In the case of Definition \ref{weakSolution} the
particles in the condensate interact with the remaining ones with an
interaction that is obtained taking the limit of the interactions with
particles with small size $\omega <<1.$ On the contrary, in the case of
Definition \ref{weakSolutionNI} it is assumed that the particles in the
condensate do not interact at all with the remaining particles of the
system. 

It is natural to ask if it would be possible to introduce more general types
of interactions between the condensate and the rest of the system in order
to define more general concepts of weak solution of (\ref{Z2E2a}). \index{condensate}The
answer is affirmative. The following definition shows how to introduce in
the system a rather large class of interactions.

\begin{definition}
\label{weakSolutionGen}Suppose  $\alpha ,\beta \in C\left( \left[
0,\infty \right) ^{2}\right) $ are nonnegative functions such that $\alpha
\left( \omega _{1},\omega _{2}\right) =\alpha \left( \omega _{2},\omega
_{1}\right) $ for $\left( \omega _{1},\omega _{2}\right) \in \left[ 0,\infty
\right) ^{2}.$ Let $\rho <-\frac{1}{2}.$ We say that $g\in C\left( %
\left[ 0,T\right) :\mathcal{M}_{+}\left( \left[ 0,\infty \right) :\left(
1+\omega \right) ^{\rho }\right) \right) $ is a weak solution of (\ref{Z2E2a}%
) with initial datum $g_{0}\in \mathcal{M}_{+}\left( \left[ 0,\infty \right)
:\left( 1+\omega \right) ^{\rho }\right) $ and condensate interaction $%
\left( \alpha ,\beta \right) $ if the following identity holds for any test
function $\varphi \in C_{0}^{2}\left( \left[ 0,T\right) \times \left[
0,\infty \right) \right) :$

\begin{eqnarray*}
&& \int_{\left[ 0,\infty \right) }g\left( t_{\ast },\omega \right) \varphi
\left( t_{\ast },\omega \right) d\omega -\int_{\left[ 0,\infty \right)
}g_{0}\varphi \left( 0,\omega \right) d\omega =\int_{0}^{t_{\ast }}\int_{\left[ 0,\infty \right) }g\partial _{t}\varphi
d\omega dt+ \\
&& +\int_{0}^{t_{\ast }}\iiint_{\left( 0,\infty \right) ^{3}}\frac{%
g_{1}g_{2}g_{3}\Phi }{\sqrt{\omega _{1}\omega _{2}\omega _{3}}}\left[
\varphi \left( \omega _{1}+\omega _{2}-\omega _{3}\right) +\varphi \left(
\omega _{3}\right) -\varphi \left( \omega _{1}\right) -\right.\\
&&\left.\hskip 7cm-\varphi \left( \omega
_{2}\right) \right] d\omega _{1}d\omega _{2}d\omega _{3}dt+ \\
&& +\int_{0}^{t_{\ast }}\iiint_{\substack{\left\{ 0\right\} \times \left(
0,\infty \right) ^{2}\\ \left\{ \omega _{2}>\omega _{3}\right\} }}\frac{%
\beta \left( \omega _{2},\omega _{3}\right) g_{1}g_{2}g_{3}}{\sqrt{\omega
_{2}\omega _{3}}}\left[ \varphi \left( \omega _{2}-\omega _{3}\right)
+\varphi \left( \omega _{3}\right) -\varphi \left( 0\right) -\right.\\
&&\left.\hskip 7cm-\varphi \left(
\omega _{2}\right) \right] d\omega _{1}d\omega _{2}d\omega _{3}dt+ \\
&& +\int_{0}^{t_{\ast }}\iiint_{\substack{\left\{ 0\right\} \times \left(
0,\infty \right) ^{2}\\ \left\{ \omega _{1}>\omega _{3}\right\} }}\frac{%
\beta \left( \omega _{1},\omega _{3}\right) g_{1}g_{2}g_{3}}{\sqrt{\omega
_{1}\omega _{3}}}\left[ \varphi \left( \omega _{1}-\omega _{3}\right)
+\varphi \left( \omega _{3}\right) -\varphi \left( \omega _{1}\right)-\right.\\
&&\left.\hskip 7cm
-\varphi \left( 0\right) \right] d\omega _{1}d\omega _{2}d\omega _{3}dt+ \\
&& +\int_{0}^{t_{\ast }}\iiint_{\left\{ 0\right\} \times \left(
0,\infty \right) ^{2}}\frac{\alpha \left( \omega _{1},\omega _{2}\right)
g_{1}g_{2}g_{3}}{\sqrt{\omega _{1}\omega _{2}}}\left[ \varphi \left( \omega
_{1}+\omega _{2}\right) +\varphi \left( 0\right) -\varphi \left( \omega
_{1}\right) -\right.\\
&&\left.\hskip 7cm-\varphi \left( \omega _{2}\right) \right] d\omega _{1}d\omega
_{2}d\omega _{3}dt
\end{eqnarray*}%
for any $t_{\ast }\in \left[ 0,T\right) .$
\end{definition}

\begin{remark}
The function $\alpha $ describes the probability rate for the collision of
two particles with energy $\left( \omega _{1},\omega _{2}\right) $ to
produce a particle with energy $\omega _{3}=0$ and other with $\omega
_{4}=\left( \omega _{1}+\omega _{2}\right) .$ The function $\beta $
describes the probability rate for the collision of one particle with energy 
$\omega _{1}>0$ and a particle in the condensate with energy $\omega _{2}=0$
to produce a particle with energy $\omega _{3}<\omega _{1}$ and other with
energy $\omega _{4}=\left( \omega _{1}-\omega _{3}\right) .$ Notice that
Definition \ref{weakSolutionGen} reduces to Definition \ref{weakSolution} if 
$\alpha =\beta =1$ and to Definition \ref{weakSolutionNI} if $\alpha =\beta
=0.$
\end{remark}

\section{Qualitative behaviors of the solutions.}

In this chapter we study several properties of the solutions whose existence has been proved in Chapter 2. We are interested first in the behavior as $t\to +\infty$ of the global weak solutions with interacting condensate of (\ref{S2E1}), and that is the content of Theorem  \ref{Asympt} below.  We consider next, in the Corollary \ref{AsEnergy} and Proposition \ref{coarsening}, the evolution  of the energy density of the particles.  Then, in the case where the Dirac mass towards which the weak solution  converges is located at the origin, we consider whether it is formed in finite time or it is only asymptotically  achieved. That is the object of Theorem \ref{AsympOsc} and Theorem \ref{weakCond}. We also prove  blow up \index{blow up} of a family of initially bounded solutions.

\subsection{Weak solutions with interacting condensate as $t\to +\infty$.} \index{interacting condensate}

In order to describe the long time asymptotics of the weak solutions of
(\ref{S2E1}) we need some notation that allows us to classify the initial data
$g_{in}.$

We recall that the support of a (nonnegative) Radon measure $\mu$ is defined
as follows:%
\begin{equation}
\operatorname*{supp}\left(  \mu\right)  =\left[  0,\infty\right)
\setminus\bigcup\left\{  \mathcal{U}\text{ open in }\mathbb{R}\text{ :\ }%
\mu\left(  \mathcal{U}\right)  =0\right\}  \ \label{S1}%
\end{equation}
where we assume that $\mu\left(  -\infty,0\right)  =0.$ Notice that
$x\in\operatorname*{supp}\left(  \mu\right)  $ iff for any $\rho>0$ we have
$\mu\left(  B_{\rho}\left(  x\right)  \right)  >0.$

\begin{definition}
\label{Aext}Given a set $A\subset\left[  0,\infty\right)  $, we define an
extended set $A^{\ast}\subset\left(  0,\infty\right)  $ as:%
\[
A^{\ast}=\bigcup_{n=1}^{\infty}A_{n}%
\]
where we define the sets $A_{n}$ inductively by means of:%
\[
A_{1}=A,\ \ \ A_{n+1}=\left(  A_{n}+A_{n}-A_{n}\right)  \cap\left(
0,\infty\right)  \ \ ,\ \ n=1,2,3,...
\]
\end{definition}

It is important to notice that by definition $0\notin A^{\ast}.$ We then have
the following result.

\begin{theorem}
\label{Asympt}Let $\rho<-1$. Suppose $g\in C\left(  \left[
0,\infty\right)\!:\mathcal{M}_{+}\left(  \left[  0,\infty\right)\!:\left(
1+\omega\right)  ^{\rho}\right)  \right)  $ is a weak solution of (\ref{S2E1})
in the sense of Definition \ref{weakSolution} with initial datum $g_{in}%
\in\mathcal{M}_{+}\left(  \left[  0,\infty\right)  :\left(  1+\omega\right)
^{\rho}\right)  .$ Let $m=\int g_{in}d\omega$. Suppose that $m>0.$ Let
$A=\operatorname*{supp}\left(  g_{in}\right)  $ and $A^{\ast}$ as in
Definition \ref{Aext}. We define $R_{\ast}=\inf\left(  A^{\ast}\right)  .$
Then:%
\begin{equation}
\lim_{t\rightarrow\infty}\int g\left(  t,d\omega\right)  \varphi\left(
\omega\right)  =m\varphi\left(  R_{\ast}\right)  \ \label{K3E6}%
\end{equation}
for any test function $\varphi\in C_{0}\left(  \left[  0,\infty\right)
\right)  $.
\end{theorem}

\begin{remark}
Theorem \ref{Asympt} just states that $g\left(  t,\cdot\right)
\rightharpoonup m\delta_{R_{\ast}}$ as $t\rightarrow\infty$ in the weak
topology for measures in $\mathcal{M}_{+}\left(  \left[  0,\infty\right)
:\left(  1+\omega\right)  ^{\rho}\right)  .$ Notice that since $m>0$ we have
that $A\neq\emptyset$ and $R_{\ast}$ is well defined and it satisfies $0\leq
R_{\ast}<\infty.$
\end{remark}

\begin{remark}
The convergence towards a single Dirac mass at the origin  containing the total number of particles has been suggested in several papers of the physical literature. In particular in \cite{DNPZ}, \cite{LY1}, \cite{P}, \cite{Zbook}.
\end{remark}
We first prove the following auxiliary result which will be used to prove some
type of diffusive properties for (\ref{S2E1}).

\begin{lemma}
\label{Diff}Suppose that the assumptions of Theorem \ref{Asympt} hold. Then,
for any $x\in A^{\ast},$ any $t^{\ast}>0$ and any $r>0,$ we have $\int
_{B_{r}\left(  x\right)  }g\left(  t^{\ast},d\omega\right)  >0.$
\end{lemma}

\begin{proof}
Let us consider any ball $B\subset\left(  0,\infty\right)  .$ Taking a
sequence of test functions $\varphi_{n}$ in (\ref{Z2E1}) converging pointwise
to the characteristic function of $\bar{B}$ we obtain the following:%
\begin{align}
&  \int_{\bar{B}}g\left(  t,d\omega\right)  -\int_{\bar{B}}g_{0}\left(
d\omega\right) \nonumber  =\int_{0}^{t}\iiint_{(\left[  0,\infty\right) )^3 }\frac{g_{1}g_{2}g_{3}\Phi}{\sqrt{\omega
_{1}\omega_{2}\omega_{3}}}\times \\
&\hskip 1.5cm \times\left[  \chi_{\bar{B}}\left(  \omega_{1}+\omega
_{2}-\omega_{3}\right)  +\chi_{\bar{B}}\left(  \omega_{3}\right)  -\chi
_{\bar{B}}\left(  \omega_{1}\right)  -\chi_{\bar{B}}\left(  \omega_{2}\right)
\right]  d\omega_{1}d\omega_{2}d\omega_{3}ds\nonumber
\end{align}
for any $t>0.$ Notice that this implies that the function $t\rightarrow
\int_{\bar{B}}g\left(  t,d\omega\right)  $ is Lipschitz continuous and:%
\begin{eqnarray}
&&\partial_{t}\left(  \int_{\bar{B}}g\left(  t,d\omega\right)  \right)
=\iiint_{(\left[  0,\infty\right) )^3 }\frac{g_{1}g_{2}g_{3}\Phi}{\sqrt{\omega_{1}\omega_{2}%
\omega_{3}}}\times \label{K3E1}\\
&&\hskip 1.5cm \times\left[  \chi_{\bar{B}}\left(  \omega_{1}+\omega_{2}-\omega
_{3}\right)  +\chi_{\bar{B}}\left(  \omega_{3}\right)  -\chi_{\bar{B}}\left(
\omega_{1}\right)  -\chi_{\bar{B}}\left(  \omega_{2}\right)  \right]
d\omega_{1}d\omega_{2}d\omega_{3}\nonumber 
\end{eqnarray}
$a.e.$ $t\geq0.$ Since $x\in A^{\ast}$ we have $x\in A_{N}$ for some $N\geq1.$
The definition of the sets $A_{n}$ implies the existence of a family of
points
\[
\mathcal{F}_{N}\left(  x\right)  =\left\{  x_{\mathbf{k}_{N-n}}^{(
n)}\!\!\in A_{n}; \mathbf{k}_{N-n}=\left(  k_{1},k_{2}, \cdots, k_{N-n}\right), k_{j}\in\left\{  1,2,3\right\}, n=1, \cdots, N\right\}
\]
where $x^{\left(  N\right)  }=x.$ Notice that this particular element has an
empty family of indexes $k_{j}.$ The family of points in $\mathcal{F}%
_{N}\left(  x\right)  $ satisfies:%
\[
x_{\mathbf{k}_{N-n}}^{\left(  n\right)  }=x_{\left(  \mathbf{k}_{N-n}%
,1\right)  }^{\left(  n-1\right)  }+x_{\left(  \mathbf{k}_{N-n},2\right)
}^{\left(  n-1\right)  }-x_{\left(  \mathbf{k}_{N-n},3\right)  }^{\left(
n-1\right)  }%
\]
for any $x_{\mathbf{k}_{N-n}}^{\left(  n\right)  }\in\mathcal{F}_{N}\left(
x\right)  $. Notice that the family $\mathcal{F}_{N}\left(  x\right)  $ is not
necessarily unique, but its existence is guaranteed by the definition of the
sets $A_{n}$ and any such a family could be used in the proof. By continuity
we can find a set or radii $r_{n}>0$ such that $r_{N}=r$ and:%
\begin{equation}
B_{r_{n-1}}\left(  x_{\left(  \mathbf{k}_{N-n},1\right)  }^{\left(
n-1\right)  }\right)  +B_{r_{n-1}}\left(  x_{\left(  \mathbf{k}_{N-n}%
,2\right)  }^{\left(  n-1\right)  }\right)  -B_{r_{n-1}}\left(  x_{\left(
\mathbf{k}_{N-n},3\right)  }^{\left(  n-1\right)  }\right)  \subset B_{r_{n}%
}\left(  x_{\mathbf{k}_{N-n}}^{\left(  n\right)  }\right)  \ \label{K3E2}%
\end{equation}
for any $x_{\mathbf{k}_{N-n}}^{\left(  n\right)  }\in\mathcal{F}_{N}\left(
x\right)  .$ Moreover, we choose the numbers $\left\{  r_{n}\right\}  $ in
order to have:%
\begin{equation}
B_{r_{n}}\left(  x_{\mathbf{k}_{N-n}}^{\left(  n\right)  }\right)
\subset\left(  0,\infty\right)  \text{ for any }x_{\mathbf{k}_{N-n}}^{\left(
n\right)  }\in\mathcal{F}_{N}\left(  x\right)  \label{K3E3}%
\end{equation}

Notice that (\ref{K3E2}) implies:%
\begin{equation}
\chi_{B_{r_{n}}\left(  x_{\mathbf{k}_{N-n}}^{\left(  n\right)  }\right)
}\left(  \omega_{1}+\omega_{2}-\omega_{3}\right)  \geq\prod_{\ell=1}^{3}%
\chi_{B_{r_{n-1}}\left(  x_{\left(  \mathbf{k}_{N-n},\ell\right)  }^{\left(
n-1\right)  }\right)  }\left(  \omega_{\ell}\right)  \ \label{K3E3a}%
\end{equation}
for $n=1,2,....,N.$ We understand that the right-hand side of (\ref{K3E3a}) is
zero if $n=1.$

Therefore, applying (\ref{K3E1}) with $\bar{B}=B_{r_{n}}\left(  x_{\mathbf{k}%
_{N-n}}^{\left(  n\right)  }\right)  ,$ using that $\chi_{\bar{B}}\geq0,$ as
well as (\ref{K3E3}) we obtain:%
\begin{eqnarray}
&&\partial_{t}\left(  \int_{B_{r_{n}}\left(  x_{\mathbf{k}_{N-n}}^{\left(
n\right)  }\right)  }g\left(  t,d\omega\right)  \right)\geq \label{K3E4}\\
&& \geq C_{1}\prod
_{\ell=1}^{3}\int_{\chi_{B_{r_{n-1}}\left(  x_{\left(  \mathbf{k}_{N-n}%
,\ell\right)  }^{\left(  n-1\right)  }\right)  }}g\left(  t,d\omega\right)-C_{2}\int_{B_{r_{n}}\left(  x_{\mathbf{k}_{N-n}}^{\left(  n\right)  }\right)
}g\left(  t,d\omega\right)  \nonumber
\end{eqnarray}
where $C_{1}>0,\ C_{2}>0$ depend on the family $\mathcal{F}_{N}\left(
x\right)  $ and $n=1,...N.$ Notice that $C_{1}$ could become very small if
some of the points in $\mathcal{F}_{N}\left(  x\right)  $ becomes large and
$C_{2}$ could become very large if some of the points in $\mathcal{F}%
_{N}\left(  x\right)  $ approaches zero. However, since the family
$\mathcal{F}_{N}\left(  x\right)  $ is finite, both constants are
finite.\ Notice also that the constant $C_{2}$ depends also in $\int g\left(
t,d\omega\right)  =\int g_{0}\left(  d\omega\right)  .$

We now apply (\ref{K3E4}) iteratively, starting with $n=1.$ By assumption:
$$\min_{x_{\mathbf{k}_{N-1}}^{\left(  1\right)  }}\left(  \int_{B_{r_{n}%
}\left(  x_{\mathbf{k}_{N-1}}^{\left(  1\right)  }\right)  }g_{0}\left(
d\omega\right)  \right)  >0.$$ Then:%
\begin{equation}
\min_{0\leq t\leq t^{\ast}}\min_{x_{\mathbf{k}_{N-1}}^{\left(  1\right)  }%
}\left(  \int_{B_{r_{n}}\left(  x_{\mathbf{k}_{N-1}}^{\left(  1\right)
}\right)  }g\left(  t,d\omega\right)  \right)  \geq c_{1}>0\ \label{K3E5}%
\end{equation}
where $c_{0}$ depends in $\mathcal{F}_{N}\left(  x\right)  $ and $t^{\ast}.$
Using (\ref{K3E5}) in (\ref{K3E4}) with $n=2$ we obtain:%
\[
\min_{\frac{t^{\ast}}{N}\leq t\leq t^{\ast}}\min_{x_{\mathbf{k}_{N-1}%
}^{\left(  2\right)  }}\left(  \int_{B_{r_{n}}\left(  x_{\mathbf{k}_{N-1}%
}^{\left(  1\right)  }\right)  }g\left(  t,d\omega\right)  \right)  \geq
c_{2}>0
\]

Iterating, and using the nonnegativity of $g,$ we obtain:%
\[
\min_{\frac{\left(  n-1\right)  t^{\ast}}{N}\leq t\leq t^{\ast}}%
\min_{x_{\mathbf{k}_{N-1}}^{\left(  n\right)  }}\left(  \int_{B_{r_{n}}\left(
x_{\mathbf{k}_{N-1}}^{\left(  1\right)  }\right)  }g\left(  t,d\omega\right)
\right)  \geq c_{n}>0
\]
for $n=1,...N,$ whence the result follows.
\end{proof}

The following Lemma will be used to prove that the dynamics of $g$ can be
reduced to a discrete problem if $R_{\ast}>0.$

\begin{lemma}
\label{SetAstar}Let $\rho<-1,$ $g_{in}\in\mathcal{M}_{+}\left(  \left[
0,\infty\right)  :\left(  1+\omega\right)  ^{\rho}\right)  .$ Define $m=\int
g_{in}d\omega.$ Suppose that $m>0.$ Let $A=\operatorname*{supp}\left(
g_{in}\right)  $ and $A^{\ast}$ as in Definition \ref{Aext} and let be
$R_{\ast}=\inf\left(  A^{\ast}\right)  .$ If $R_{\ast}>0$, there exists a
finite set of positive real numbers $\left\{  D_{k}\right\}  _{k=1}^{L}$ such
that%
\begin{equation}
A^{\ast}=\left\{  R_{\ast}+\sum_{k=1}^{L}n_{k}D_{k}:n_{k}\in\mathbb{N}_{\ast
}\right\}  \label{K5E2}%
\end{equation}

Moreover, for any $j,k\in\left\{  1,2,...,L\right\}  $ the quotients
$\frac{D_{j}}{D_{k}}$ are rational numbers and we have:%
\begin{equation}
\min\left\{  D_{k}\right\}  _{k=1}^{L}\geq R_{\ast} \label{K5E6}%
\end{equation}

\end{lemma}

\begin{proof}
Suppose that $x,y\in A,\ y>x.$ Let $D_{1}\equiv\left(  y-x\right)  .$ Then
$Q_{1}\equiv\left[  x+D_{1}\mathbb{Z}\right]  \cap\mathbb{R}_{+}\subset
A^{\ast}.$ This is just a consequence of the definition of $A^{\ast}.$ Notice
that $R_{\ast}\leq\min Q_{1}\leq D_{1}.$ Suppose that $\left[  A^{\ast
}\setminus Q_{1}\right]  \neq\emptyset.$ We choose $z_{2}\in\left[  A^{\ast
}\setminus Q_{1}\right]  .$ Since $R_{\ast}>0$ we have $z_{2}>0.$ Moreover,
$D_{2}\equiv dist\left(  z_{2},Q_{1}\right)  <D_{1}.$ Then
$Q_{2}\equiv\left(  Q_{1}\cup\left[  \min Q_{1}+D_{2}\mathbb{Z}\right]
\cup\left[  z_{2}+D_{2}\mathbb{Z}\right]  \right)  \cap\mathbb{R}_{+}\subset
A^{\ast}.$ Notice that if $D_{2}>\frac{D_{1}}{2}$ we have $z_{2}\leq
\frac{D_{1}}{2}.$ Since $\min Q_{2}\leq D_{2}$ it follows that $\min Q_{2}%
\leq\frac{D_{1}}{2}.$ Indeed, if $D_{2}\leq\frac{D_{1}}{2}$ this follows
immediately. Otherwise we have that $R_{\ast}\leq\min Q_{2}\leq z_{2}\leq
\frac{D_{1}}{2},$ since $z_{2}\in Q_{2}.$ We then define sets $Q_{k}$ in an
iterative manner. More precisely, as long as $\left[  A^{\ast}\setminus
Q_{k-1}\right]  \neq\emptyset,$ $k=2,3,...$ we can choose $z_{k}\in\left[
A^{\ast}\setminus Q_{k-1}\right]  .$ We have $D_{k}\equiv dist
\left(  z_{k},Q_{k-1}\right)  <D_{k-1}$ and we then define:%
\begin{equation}
Q_{k}\equiv\left(  Q_{k-1}\cup\left[  \min Q_{k-1}+D_{k}\mathbb{Z}\right]
\cup\left[  z_{k}+D_{k}\mathbb{Z}\right]  \right)  \cap\mathbb{R}_{+}\subset
A^{\ast} \label{K5E4}%
\end{equation}

Arguing as in the case $k=2$ we obtain $R_{\ast}\leq D_{k}\leq\frac{D_{k-1}%
}{2}.$ Then $D_{k}$ decreases exponentially in $k$ and since $R_{\ast}>0$ the
process must stop after a finite number of steps, say $M\geq1.$ More
precisely, there exists $1\leq M<\infty$ such that $A^{\ast}\subset Q_{L}.$
Otherwise, if the iteration procedure could be iterated for arbitrarily large
values of $k$ we would arrive at a contradiction. On the other hand, since
(\ref{K5E4}) holds for $k=M$ we have $Q_{M}\subset A^{\ast}$ whence
$Q_{M}=A^{\ast}.$ Then $R_{\ast}=\min\left(  Q_{M}\right)  \in A^{\ast}.$
Moreover, we have proved the existence of points $x_{k},y_{k}\in A^{\ast}$
such that $\left(  y_{k}-x_{k}\right)  =D_{k},\ k=1,...,M.$ Then, using the
definition of $A^{\ast}$ we obtain:%
\begin{equation}
\mathcal{U}_{M}=\left\{  R_{\ast}+\sum_{k=1}^{M}n_{k}D_{k}:n_{k}\in
\mathbb{N}_{\ast}\right\}  \subset A^{\ast}\ \label{K5E5}%
\end{equation}
where:%
\[
\mathbb{N}_{\ast}=\left\{  0,1,2,3,...\right\}
\]

If the inclusion in (\ref{K5E5}) is strict we can find $z\in A\setminus
\mathcal{U}_{L}.$ We must have $z>R_{\ast}$ since, otherwise there would exist
two points in $A^{\ast}$ at a distance smaller than $R_{\ast}$ and this is
impossible as seen above. Otherwise we introduce additional distances
$D_{j},\ j=M+1,....$ and extend iteratively the sets $\mathcal{U}_{M}$ to
$\mathcal{U}_{M+1},\ \mathcal{U}_{M+2},...$ including in the set also the
points $\left\{  nD_{M+1}:n\in\mathbb{N}_{\ast}\right\}  ,\ \left\{
nD_{M+2}:n\in\mathbb{N}_{\ast}\right\}  ,\ ...$ respectively. Since
$D_{k+1}\leq\frac{D_{k}}{2}$ the process must stop in a finite number of
steps. Therefore we obtain $\mathcal{U}_{L}=A^{\ast}$ for some $L.$ This gives
(\ref{K5E2}).

In order to prove that $\frac{D_{j}}{D_{k}}$ are rational numbers for any
$j,k$ we just notice that, if this quotient is irrational for any pair $j,k$
we would have
\begin{equation}
dist\left(  \left\{  R_{\ast}+nD_{j}:n\in\mathbb{N}_{\ast
}\right\}  ,\left\{  R_{\ast}+nD_{k}:n\in\mathbb{N}_{\ast}\right\}  \right)
=0 \label{Z5E1a}%
\end{equation}
whence $R_{\ast}=0$ and the resulting contradiction yields the result. We just
remark that the property  (\ref{Z5E1a}) follows from the well known fact that the set of
points $\left\{  m\alpha\ (\operatorname{mod}\ 1):m\in\mathbb{N}_{\ast
}\right\}  $ is dense in the interval $\left[  0,1\right]  $ for any irrational $\alpha$
(cf.  \cite{Arnold}).
\end{proof}
The following result about the set $A^*$ follows easily.

\begin{corollary}
Let $\rho<-1,$ $g_{in}\in\mathcal{M}_{+}\left(  \left[  0,\infty\right)
:\left(  1+\omega\right)  ^{\rho}\right)  .$ Define $m=\int g_{in}d\omega.$
Suppose that $m>0.$ Let $A=\operatorname*{supp}\left(  g_{in}\right)  $ and
$A^{\ast}$ as in Definition \ref{Aext}. If $R_{\ast}>0,$ the set $A^{\ast}$
has the form (\ref{K5E2}) in Lemma \ref{SetAstar}. If $R_{\ast}=0$ we have
$\overline{A^{\ast}}=\left[  0,\infty\right)  .$
\end{corollary}

\begin{proof}
If $R_{\ast}>0$ we can apply Lemma \ref{SetAstar}. If $R_{\ast}=0$ we have the
following due to the definition of $A^{\ast}.$ Given any $\varepsilon>0,$
there exist $x_{1},x_{2}\in A^{\ast}$ with $0<x_{1}<x_{2}<\varepsilon.$ This
implies that $x_{1}+\ell\left(  x_{2}-x_{1}\right)  \in A^{\ast}$ for any
$\ell=0,1,2,3,...$ and since $\left(  x_{2}-x_{1}\right)  <\varepsilon$ we
then obtain that $A^{\ast}$ is dense in $\left[  0,\infty\right)  .$
\end{proof}

The following result combined with Lemma \ref{Diff} will be used to
characterize the support of the measures $g\left(  t,\cdot\right)  .$ The main
difficulty in the proof of Lemma \ref{gSupport} is due to the fact that
equation (\ref{Z2E2a}) is singular at $\omega=0.$ Therefore we need to obtain
detailed estimates for the measure of $g$ supported in regions with $\omega$
small, because tiny amounts of the measure of $g$ arriving to that region
could have a huge effect in the dynamics of $g.$

\begin{lemma}
\label{gSupport}Suppose that the assumptions of Theorem \ref{Asympt} hold and
$R_{\ast}>0$. Then $\operatorname*{supp}\left(  g\left(  t,.\right)  \right)
\subset\overline{A^{\ast}}$ for any $t\geq0.$
\end{lemma}

\begin{proof}
We will assume, without loss of generality, that $\int g\left(  t,d\omega
\right)  =1.$ We fix $\delta>0$ small, in particular $\delta<\frac{R_{\ast}%
}{8}$. We define $N=N\left(  \delta\right)  $ as the smallest positive integer
such that $3^{N}\delta>\frac{3R_{\ast}}{4}.$ Notice that $3^{N-1}\delta
\in\left[  \frac{R_{\ast}}{4},\frac{3R_{\ast}}{4}\right)  .$ We define the
following sets:%
\begin{align*}
Z_{k}  &  =A_{\ast}+B_{3^{k}\delta}\left(  0\right)
\ \ ,\ \ k=-1,0,1,2,3,....\left(  N-1\right) \\
Z_{N}  &  =Z_{N+1}=\left[  0,\infty\right)
\end{align*}%
\[
\mathcal{U}_{0}=Z_{0}\ \ ,\ \ \mathcal{U}_{k}=Z_{k}\setminus Z_{k-1}%
\ \ ,\ \ k=1,...,N
\]

It is relevant to notice that, using the definition of $N$ as well as
(\ref{K5E6}) and the invariance of $A_{\ast}$ under the transformations
$\left(  \omega_{1},\omega_{2},\omega_{3}\right)  \in A_{\ast}^{3}%
\rightarrow\left(  \omega_{1}+\omega_{2}-\omega_{3}\right)  ,$ we have:%
\begin{equation}
Z_{k-1}+Z_{k-1}-Z_{k-1}=Z_{k}\ \ ,\ \ k=1,2,3,...\left(  N-1\right)
\ \label{Z6E1}%
\end{equation}

We now define a set of nonnegative test functions $\varphi_{k}\in C_{0}%
^{1}\left(  \left[  0,\infty\right)  \right)  ,\ k=0,1,...,N$ as follows. We
will assume that $0\leq\varphi_{k}\leq1,\ k=0,1,...,N$ and we assume that:%
$$
\varphi_{0}\left(  \omega\right) =1\ \ \text{if}\ \ \omega\in
Z_{0} ;\ \ \varphi_{0}\left(  \omega\right)  =0\ \ \text{if\ \ }%
\omega\notin Z_{1}\ 
$$
and, for  $k=1,2,3,...N$:
$$
\varphi_{k}\left(  \omega\right)     =1\ \ \text{if}\ \ \omega\in
Z_{k}\setminus Z_{k-1}=\mathcal{U}_{k}\ \ ;\ \ \varphi_{k}\left(
\omega\right)  =0\ \ \text{if\ \ }\omega\in\left(  \left[  0,\infty\right)
\setminus Z_{k+1}\right)  \cup Z_{k-2}
$$

Moreover, we choose the functions $\varphi_{k}$ satisfying the inequalities:%
\begin{equation}
\left\vert \varphi_{k}^{\prime}\right\vert \leq\frac{C}{3^{k}\delta
}\ \ ,\ \ \left\vert \varphi_{k}^{\prime\prime}\right\vert \leq\frac
{C}{\left(  3^{k}\delta\right)  ^{2}}\ \ ,\ k=0,1,2,3,...N\ \ \label{Z6E4}%
\end{equation}

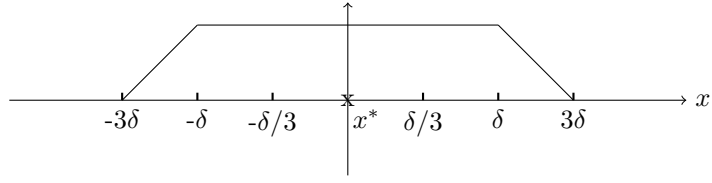
\begin{figure}
\begin{tikzpicture}
\draw[->] (-4.5, 0) -- (4.5, 0) node[right]{$x$};
\draw[->] (-0,-1) -- (-0,1.3) node[above]{};
\draw(0, 0) node {x};
\draw(0.25, 0) node [below]{$x^*$};
\draw[thick] (-3, 0) node[below]{-$3\delta $} -- (-3,0.1);
\draw[thick] (-1, 0) node[below]{-$\delta /3$} -- (-1,0.1);
\draw[thick] (-2, 0) node[below]{-$\delta $} -- (-2, 0.1);
\draw[thick] (1, 0) node[below]{$\delta /3$} -- (1,0.1);
\draw[thick] (2, 0) node[below]{$\delta $} -- (2, 0.1);
\draw[thick] (3, 0) node[below]{$3\delta  $} -- (3,0.1);
\draw(-2, 1)--(2,1);
\draw(2,1)--(3, 0);
\draw(-3,0)--(-2, 1);
\end{tikzpicture}
\caption{The function $\varphi_0$.}
\end{figure}

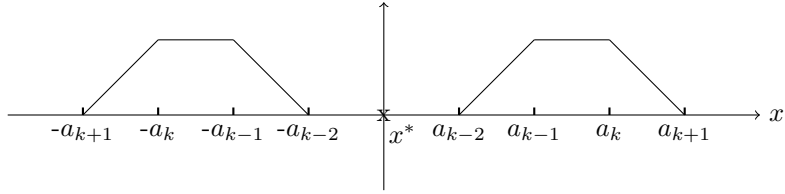
\begin{figure}
\begin{tikzpicture}
\draw[->] (-5,0) -- (5,0) node[right]{$x$};
\draw[->] (-0,-1) -- (-0,1.5) node[above]{$$};
\draw(0, 0) node {x};
\draw(0.25, 0) node [below]{$x^*$};
\draw[thick] (-4, 0) node[below]{ -$a _{ k+1 }$} -- (-4,0.1);
\draw[thick] (-3, 0) node[below]{-$a_k$} -- (-3,0.1);
\draw[thick] (-1, 0) node[below]{-$a _{ k-2 }$} -- (-1,0.1);
\draw[thick] (-2, 0) node[below]{-$a _{ k-1 } $} -- (-2, 0.1);
\draw[thick] (1, 0) node[below]{$a _{ k-2 } $} -- (1,0.1);
\draw[thick] (2, 0) node[below]{$a _{ k-1 }$} -- (2, 0.1);
\draw[thick] (3, 0) node[below]{$a _{ k }$} -- (3,0.1);
\draw[thick] (4, 0) node[below]{$a _{ k+1 }$} -- (4,0.1);
\draw(1, 0)--(2, 1);
\draw(2, 1)--(3, 1);
\draw(3, 1)--(4, 0);
\draw(-4, 0)--(-3, 1);
\draw(-3, 1)--(-2, 1);
\draw(-2, 1)--(-1, 0);
\end{tikzpicture}
\caption{The function $\varphi_k$, $k\ge 1$ where $a_\ell=x^*+3^\ell \delta $.}
\end{figure}

\begin{figure}
\begin{tikzpicture}
\draw[->] (-5.3,0) -- (5.3,0) node[right]{$x$};
\draw[->] (-0,-1) -- (-0,1.3) node[above]{$$};
\draw(0, 0) node {x};
\draw(0.25, 0) node [below]{$x^*$};
\draw[thick] (-5, 0) node[below]{-$a_3$} -- (-5,0.1);
\draw[thick] (-4, 0) node[below]{-$a_2$} -- (-4,0.1);
\draw[thick] (-3, 0) node[below]{-$a_1$} -- (-3,0.1);
\draw[thick] (-1, 0) node[below]{-$a _{ -1 }$} -- (-1,0.1);
\draw[thick] (-2, 0) node[below]{-$a_0$} -- (-2, 0.1);
\draw[thick] (1, 0) node[below]{$a _{ -1 } $} -- (1,0.1);
\draw[thick] (2, 0) node[below]{$a_0 $} -- (2, 0.1);
\draw[thick] (3, 0) node[below]{$a_1 $} -- (3,0.1);
\draw[thick] (4, 0) node[below]{$a_2 $} -- (4,0.1);
\draw[thick] (5, 0) node[below]{$a_3$} -- (5,0.1);
\draw(1, 0)--(2, 1);
\draw(2, 1)--(3, 1);
\draw(3, 1)--(4, 0);
\draw(-4, 0)--(-3, 1);
\draw(-3, 1)--(-2, 1);
\draw(-2, 1)--(-1, 0);
\draw[dotted](-2, 1)--(2,1);
\draw[dotted](2,1)--(3, 0);
\draw[dotted](-3,0)--(-2, 1);
\draw[dashed, thick](-2, 0)--(-3, 1);
\draw[dashed, thick](-3, 1)--(-4, 1);
\draw[dashed, thick](-4, 1)--(-5, 0);
\draw[dashed, thick](2, 0)--(3, 1);
\draw[dashed, thick](3, 1)--(4, 1);
\draw[dashed, thick](4, 1)--(5, 0);
\end{tikzpicture}
\caption{The functions $\varphi_0$ (dotted), $\varphi_1$ and $\varphi_2$ (dashed).}
\end{figure}
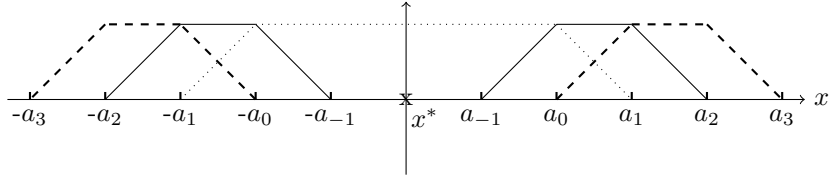
(See figures 1, 2 and 3.) Using (\ref{Z2E1})\ we obtain:%
\begin{eqnarray}
&&\partial_{t}\left(  \int g\left(  t,d\omega\right)  \varphi_{k}\right)
=\sum_{j_{1}=0}^{N}\sum_{j_{2}=0}^{N}\sum_{j_{3}=0}^{N}\int_{\mathcal{U}%
_{j_{1}}}\int_{\mathcal{U}_{j_{2}}}\int_{\mathcal{U}_{j_{3}}}\frac{g_{1}%
g_{2}g_{3}\Phi}{\sqrt{\omega_{1}\omega_{2}\omega_{3}}}\times  \label{Z6E1a}\\
&&\hskip 1.5cm \times \left[  \varphi
_{k}\left(  \omega_{1}+\omega_{2}-\omega_{3}\right)  +\varphi_{k}\left(
\omega_{3}\right)  -\varphi_{k}\left(  \omega_{1}\right)  -\varphi_{k}\left(
\omega_{2}\right)  \right]  d\omega_{1}d\omega_{2}d\omega_{3}\nonumber 
\end{eqnarray}
$a.e.\ t\geq0.$ We now proceed to estimate the different terms on the right of
(\ref{Z6E1a}). We will estimate in a different manner the terms containing at
least two among the indexes $j_{1},j_{2},j_{3}$ equal to $N$ and all the
others choices of indexes. Let us denote the set of indexes $\left(
j_{1},j_{2},j_{3}\right)  $ with at least two values equal to $N$ as
$\mathcal{S}.$ If $\left(  j_{1},j_{2},j_{3}\right)  \notin\mathcal{S}$ we
have that at least two among the values $\omega_{1},\ \omega_{2},\ \omega_{3}$
remain at a distance of $\omega=0$ larger than $\frac{R_{\ast}}{4}.$ Therefore
$\frac{\Phi}{\sqrt{\omega_{1}\omega_{2}\omega_{3}}}\leq C$, with $C>0$
independent on $\delta.$ Then:%
\begin{eqnarray*}
&&\sum_{\left(  j_{1},j_{2},j_{3}\right)  \notin\mathcal{S}}\int_{\mathcal{U}%
_{j_{1}}}\int_{\mathcal{U}_{j_{2}}}\int_{\mathcal{U}_{j_{3}}}\left[
\cdot\cdot\cdot\right]  \leq \\
&&\hskip 1.5cm \leq C\sum_{\left(  j_{1},j_{2},j_{3}\right)
\notin\mathcal{S}}\int_{\mathcal{U}_{j_{1}}}\int_{\mathcal{U}_{j_{2}}}%
\int_{\mathcal{U}_{j_{3}}}g_{1}g_{2}g_{3}\left[  \varphi_{k}\left(  \omega
_{1}+\omega_{2}-\omega_{3}\right)  +\varphi_{k}\left(  \omega_{3}\right)
\right]
\end{eqnarray*}

Using the definition of the functions $\varphi_{k}$ as well as the finiteness
of the measure $g$ we obtain:%
\begin{eqnarray}
&&\sum_{\left(  j_{1},j_{2},j_{3}\right)  \notin\mathcal{S}}\int_{\mathcal{U}%
_{j_{1}}}\int_{\mathcal{U}_{j_{2}}}\int_{\mathcal{U}_{j_{3}}}\left[
\cdot\cdot\cdot\right]  \leq C\int g\left(  t,d\omega\right)  \varphi
_{k}+\label{Z6E2}\\
&&\hskip 1.5 cm +C\sum_{\left(  j_{1},j_{2},j_{3}\right)  \notin\mathcal{S}}%
\int_{\mathcal{U}_{j_{1}}}\int_{\mathcal{U}_{j_{2}}}\int_{\mathcal{U}_{j_{3}}%
}g_{1}g_{2}g_{3}\varphi_{k}\left(  \omega_{1}+\omega_{2}-\omega_{3}\right)\nonumber
\end{eqnarray}

Due to the definition of the test functions $\varphi_{k}$ as well as the
property (\ref{Z6E1}) and the fact that $\ \mathcal{U}_{k}\subset Z_{k}$ it
follows that, if $\max\left\{  j_{1},j_{2},j_{3}\right\}  \leq\left(
k-2\right)  ,$ we have $\varphi_{k}\left(  \omega_{1}+\omega_{2}-\omega
_{3}\right)  =0$ if $\omega_{\ell}\in\mathcal{U}_{\ell}\ ,\ \ell\in\left\{
j_{1},j_{2},j_{3}\right\}  .$ Using that $\varphi_{k}\left(  \omega
_{1}+\omega_{2}-\omega_{3}\right)  \leq1,$ as well as the finiteness of the
mass of $g$ as well as the fact that $\varphi_{\ell}\left(  \omega\right)  =1$
for $\omega\in\mathcal{U}_{\ell}$ it then follows that the last term in
(\ref{Z6E2}) can be estimated as $C\sum_{\ell=\left(  k-2\right)  _{+}}%
^{N}\int g\left(  t,d\omega\right)  \varphi_{\ell},$ with $C$ depends on the
total mass of $g$. Therefore it follows from (\ref{Z6E2}) that:%
\begin{equation}
\sum_{\left(  j_{1},j_{2},j_{3}\right)  \notin\mathcal{S}}\int_{\mathcal{U}%
_{j_{1}}}\int_{\mathcal{U}_{j_{2}}}\int_{\mathcal{U}_{j_{3}}}\left[
\cdot\cdot\cdot\right]  \leq C\sum_{\ell=\left(  k-2\right)  _{+}}^{N}\int
g\left(  t,d\omega\right)  \varphi_{\ell} \label{Z6E3}%
\end{equation}

We now estimate the contribution to the sum in (\ref{Z6E1a}) of the indexes
satisfying  $\left(  j_{1},j_{2},j_{3}\right)  \in\mathcal{S}.$ Suppose first that
$\left(  j_{1},j_{2},j_{3}\right)  \neq\left(  N,N,N\right)  .$ We can then
estimate one of the quotients $\frac{1}{\sqrt{\omega_{k}}},\ k=1,2,3$ by a
constant independent of $\delta$ since $\max\left\{  \omega_{1},\omega
_{2},\omega_{3}\right\}  \geq\frac{R_{\ast}}{4}.$ If $j_{3}<N$ we obtain a
zero contribution to the integrals if $\max\left\{  \omega_{1},\omega
_{2}\right\}  <\frac{R_{\ast}}{8},$ since then $\left(  \omega_{1}+\omega
_{2}-\omega_{3}\right)  <0.~$Therefore, it is enough to consider the case in
which $\max\left\{  \omega_{1},\omega_{2}\right\}  >\frac{R_{\ast}}{8}.$
However, in this case, using the definition of $\Phi$ we obtain:%
\[
\frac{\Phi}{\sqrt{\omega_{1}\omega_{2}\omega_{3}}}\leq C
\]
with $C$ independent of $\delta$ (but depending on $R_{\ast}$). We then
obtain, arguing as in the estimate for the terms with the indexes in the
complement of $\mathcal{S}:$%
\begin{eqnarray}
&&\underset{j_{3}<N}{\sum_{\left(  j_{1},j_{2},j_{3}\right)  \in\mathcal{S}\setminus\left\{
\left(  N,N,N\right)  \right\}}}\int_{\mathcal{U}_{j_{1}}}%
\int_{\mathcal{U}_{j_{2}}}\int_{\mathcal{U}_{j_{3}}}\left[  \cdot\cdot
\cdot\right]  \leq  \label{Z6E6}\\
&&\hskip 1.5cm \leq C\underset{j_{3}<N}{\sum_{\left(  j_{1},j_{2},j_{3}\right)  \in\mathcal{S}}}
\int_{\mathcal{U}_{j_{1}}}\int_{\mathcal{U}_{j_{2}}}%
\int_{\mathcal{U}_{j_{3}}}g_{1}g_{2}g_{3}
d\omega_{1}d\omega_{2}d\omega_{3} C\left(  \int g\left(  t,d\omega\right)  \varphi_{N}\right)  ^{2} \nonumber
\end{eqnarray}

Suppose now that $j_{3}=N.$ Since $\left(  j_{1},j_{2},j_{3}\right)
\in\mathcal{S}\setminus\left\{  \left(  N,N,N\right)  \right\}  $ we have that
exactly one of the indexes $j_{1}$ or $j_{2}$ is equal to $N.$ We can assume,
without loss of generality that $j_{2}=N.$ Then:%
\begin{eqnarray}
&&\underset{j_{3}=N}{\sum_{\left(  j_{1},j_{2},j_{3}\right)  \in\mathcal{S}\setminus\left\{
\left(  N,N,N\right)  \right\}}}\int_{\mathcal{U}_{j_{1}}}%
\int_{\mathcal{U}_{j_{2}}}\int_{\mathcal{U}_{j_{3}}}\left[  \cdot\cdot
\cdot\right]  \leq C\sum_{j_{1}\neq N}\int_{\mathcal{U}_{j_{1}}}%
\int_{\mathcal{U}_{N}}\int_{\mathcal{U}_{N}}\frac{g_{1}g_{2}g_{3}\Phi}%
{\sqrt{\omega_{2}\omega_{3}}}\times \label{Z6E5} \\
&&\hskip 1.5cm \times\left[  \left\vert \varphi_{k}\left(  \omega
_{1}+\omega_{2}-\omega_{3}\right)  -\varphi_{k}\left(  \omega_{1}\right)
\right\vert +\left\vert \varphi_{k}\left(  \omega_{3}\right)  -\varphi
_{k}\left(  \omega_{2}\right)  \right\vert \right]  d\omega_{1}d\omega
_{2}d\omega_{3} \nonumber
\end{eqnarray}

Using (\ref{Z6E4}) we obtain the inequality:%
\[
\left\vert \varphi_{k}\left(  \omega_{1}+\omega_{2}-\omega_{3}\right)
-\varphi_{k}\left(  \omega_{1}\right)  \right\vert +\left\vert \varphi
_{k}\left(  \omega_{3}\right)  -\varphi_{k}\left(  \omega_{2}\right)
\right\vert \leq\frac{C}{3^{k}\delta}\min\left\{  \omega_{2}+\omega_{3}%
,3^{k}\delta\right\}
\]

Then, the right-hand side of (\ref{Z6E5}) can be estimated by:%
\begin{equation}
\frac{C}{\sqrt{3^{k}\delta}}\sum_{j_{1}\neq N}\int_{\mathcal{U}_{j_{1}}}%
\int_{\mathcal{U}_{N}}\int_{\mathcal{U}_{N}}g_{1}g_{2}g_{3}d\omega_{1}%
d\omega_{2}d\omega_{3}\leq\frac{C}{\sqrt{3^{k}\delta}}\left(  \int g\left(
t,d\omega\right)  \varphi_{N}\right)  ^{2} \label{Z6E7}%
\end{equation}

We finally estimate the case $\left(  j_{1},j_{2},j_{3}\right)  =\left(
N,N,N\right)  .$ We derive from (\ref{Z6E4}) the estimate:
\begin{eqnarray}
&&\left\vert \varphi_{k}\left(  \omega_{1}+\omega_{2}-\omega_{3}\right)
\!+\!\varphi_{k}\left(  \omega_{3}\right) \! - \!\varphi_{k}\left(  \omega_{1}\right)
\!-\!\varphi_{k}\left(  \omega_{2}\right)  \right\vert \leq\frac{C\min\left\{
\omega_{+}\omega_{0},\left(  3^{k}\delta\right)  ^{2}\right\}  }{\left(
3^{k}\delta\right)  ^{2}}\ \label{Z6E8}%
\end{eqnarray}
where the functions $\omega_{0},\ \omega_{+}$ are as in Definition \ref{aux}.
Notice that deriving (\ref{Z6E8}) we first substract an affine function from
$\varphi_{k}$ whose contribution to the left-hand side of (\ref{Z6E8})
vanishes. The resulting function to be estimated takes the value zero at
$\omega=0$ as well as its first derivative. Notice that for $k\leq\left(
N-2\right)  $ we should have $\max\left\{  \omega_{1},\omega_{2}\right\}
\geq\frac{R_{\ast}}{8}$ in order to avoid a vanishing integral. In that case
we can argue as in the derivation of (\ref{Z6E7}) to obtain:%
\begin{align}
 \left\vert \iiint_{(\mathcal{U}_{N})^3}%
g_{1}g_{2}g_{3}\Delta_{\varphi_k,0}\left(  \omega_{1},\omega_{2},\omega_{3}\right)
d\omega_{1}d\omega_{2}%
d\omega_{3}\right\vert  \leq\frac{C}{\sqrt{3^{k}\delta}}\left(  \int g\left(  t,d\omega\right)
\varphi_{N}\right)  ^{2}\ \ \label{Z6E9}%
\end{align}
if $k=0,1,...,\left(  N-2\right)$, with  $\Delta_{\varphi_k,0}\left(  \omega_{1},\omega_{2},\omega_{3}\right)$ as defined in
(\ref{Z2E2}), (\ref{Z2E6}). If $k>\left(  N-2\right)  $ we use
(\ref{Z6E8}) which combined with the definition of $N$ implies that $\frac
{1}{3^{k}\delta}$ is bounded and of order one. Then $\sqrt{3^{k}\delta}$ and
$\left(  3^{k}\delta\right)  ^{2}$ are comparable, whence we obtain again the
estimate (\ref{Z6E9}).

Combining (\ref{Z6E1a}), (\ref{Z6E3}), (\ref{Z6E6}), (\ref{Z6E7}),
(\ref{Z6E9}) we arrive at:%
\begin{equation}
\partial_{t}m_{k}\left(  t\right)  \leq C_{0}\left[  \sum_{\ell=\left(
k-2\right)  _{+}}^{N}m_{\ell}\left(  t\right)  +\frac{1}{\sqrt{3^{k}\delta}%
}\left(  m_{N}\left(  t\right)  \right)  ^{2}\right]  \ \ ,\ \ a.e.\ t\geq
0\ \label{K6E1}%
\end{equation}
for $k=0,1,2,..,N,$ where:%
\begin{equation}
m_{k}\left(  t\right)  =\int g\left(  t,d\omega\right)  \varphi_{k}%
\ \label{K6E3}%
\end{equation}
and where $C_{0}$ is a positive constant. We notice also that the definition
of the test functions $\varphi_{k}$ implies:%
\begin{equation}
m_{k}\left(  0\right)  =\delta_{k,0} \label{K6E2}%
\end{equation}

Since the support of the functions $\varphi_{k}$ overlap at most a finite
number of times, and the total mass of $g_{in}$ has been normalized to one, we
have:%
\begin{equation}
\sum_{k=0}^{N}m_{k}\left(  t\right)  \leq3\ \ ,\ \ t\geq0 \label{K6E4}%
\end{equation}
We have also, since each $m_{k}\left(  t\right)  $ is smaller than the total
mass of $g\left(  t,\cdot\right)  :$%
\begin{equation}
m_{k}\left(  t\right)  \leq1\ \ ,\ \ t\geq0\ \ ,\ \ k=0,1,2,..,N \label{K6E6}%
\end{equation}

We first derive some upper estimates for the asymptotics of $m_{k}\left(
t\right)  $ as $t\rightarrow0.$ Using (\ref{K6E1}) and (\ref{K6E4}) we obtain:%
\begin{equation}
m_{k}\left(  t\right)  \leq C_{0}\left(  3+\frac{1}{\sqrt{3^{k}\delta}%
}\right)  t\ \ ,\ \ k=1,2,..,N \label{K6E5}%
\end{equation}

Combining (\ref{K6E6}) and (\ref{K6E5}) we obtain the following estimate:%
\begin{equation}
\max\left\{  m_{1}\left(  t\right)  ,m_{2}\left(  t\right)  \right\}  \leq
C_{0}t+O_{\delta}\left(  t^{2}\right)  \label{K6E7}%
\end{equation}
where from now on we denote as $O_{\delta}\left(  f\left(  t\right)  \right)
$ a function $g\left(  t\right)  $, perhaps depending on $\delta,$ such that
$\lim_{t\rightarrow0}\frac{g\left(  t\right)  }{f\left(  t\right)  }=0.$ We
emphasize the fact that this convergence is not uniform in $\delta$ in
general. Using again (\ref{K6E5}), (\ref{K6E6}) for $k\geq3$ we obtain:%
\begin{equation}
m_{k}\left(  t\right)  =O_{\delta}\left(  t^{2}\right)  \ \ ,\ \ k=3,...,N
\label{K6E8}%
\end{equation}

Combining now (\ref{K6E5}), (\ref{K6E7}), (\ref{K6E8}) we arrive at:%
\[
\max\left\{  m_{3}\left(  t\right)  ,m_{4}\left(  t\right)  \right\}
\leq\frac{2C_{0}^{2}t^{2}}{2!}+O_{\delta}\left(  t^{3}\right)
\]%
\[
m_{k}\left(  t\right)  =O_{\delta}\left(  t^{3}\right)  \ \ ,\ \ k=5,...,N
\]

Iterating the argument we obtain:%
\begin{equation}
\max\left\{  m_{2\ell-1}\left(  t\right)  ,m_{2\ell}\left(  t\right)
\right\}  \leq\frac{1}{2}\frac{\left(  2C_{0}t\right)  ^{\ell}}{\ell
!}+O_{\delta}\left(  t^{\ell+1}\right)  \ \ ,\ \ \ell=1,2,...\left[
\frac{N+1}{2}\right]  \ \label{K6E9}%
\end{equation}
where we will assume that $m_{N+1}\left(  t\right)  \equiv0.$ This term
formally appears in (\ref{K6E9}) if $N$ is an odd number.

Our next goal is to prove the following estimate:%
\begin{equation}
\max\left\{  m_{2\ell-1}\left(  t\right)  ,m_{2\ell}\left(  t\right)
\right\}  \leq\frac{\left(  4C_{0}t\right)  ^{\ell}}{\ell!}\ \ ,\ \ \ell
=1,2,...\left[  \frac{N+1}{2}\right]  \label{K7E1}%
\end{equation}
for $0\leq t\leq T\ $, with $T<\min\left\{  \frac{e\sqrt{R_{\ast}}}%
{12C_{0}},\frac{1}{16eC_{0}}\right\}  $. Notice that due to (\ref{K6E9}) we
have that (\ref{K7E1}) holds for $0\leq t\leq t_{\delta},$ with $t_{\delta
}>0.$ We define
\[
t^{\ast}=\sup\left\{  0\leq\bar{t}\leq T:\text{(\ref{K7E1}) holds in }0\leq
t\leq\bar{t}\right\}
\]

We already know that $t^{\ast}\geq t_{\delta}.$ Our goal is to prove that
$t^{\ast}=T.$ Suppose that $t^{\ast}<T.$ Using (\ref{K6E1}) we obtain:%
\[
\partial_{t}m_{k}\left(  t\right)  \leq C_{0}\left[  \sum_{\ell=\left(
k-2\right)  _{+}}^{k-1}m_{\ell}\left(  t\right)  +\sum_{\ell=k}^{N}m_{\ell
}\left(  t\right)  +\frac{1}{\sqrt{3^{k}\delta}}\left(  m_{N}\left(  t\right)
\right)  ^{2}\right]  \ ,\ \ k=1,2,\cdots, N.
\]
Using (\ref{K7E1}), which holds for $0\leq t\leq t^{\ast}$, we arrive at:%
\[
\partial_{t}m_{k}\left(  t\right)  \leq C_{0}\left[  \sum_{\ell=\left(
k-2\right)  _{+}}^{k-1}m_{\ell}\left(  t\right)  +2\sum_{\ell=\left[
\frac{k+1}{2}\right]  }^{\left[  \frac{N+1}{2}\right]  }\frac{\left(
4C_{0}t\right)  ^{\ell}}{\ell!}+\frac{1}{\sqrt{3^{k}\delta}}\left(
\frac{\left(  4C_{0}t\right)  ^{\left[  \frac{N}{2}\right]  }}{\left(  \left[
\frac{N}{2}\right]  \right)  !}\right)  ^{2}\right].
\]
for $k=1, 2, \cdots, N$. Notice that:%
\begin{eqnarray*}
\sum_{\ell=\left[  \frac{k+1}{2}\right]  }^{\left[  \frac{N+1}{2}\right]
}\frac{\left(  4C_{0}t\right)  ^{\ell}}{\ell!}&=&\frac{\left(  4C_{0}t\right)
^{\left[  \frac{k+1}{2}\right]  }}{\left(  \left[  \frac{k+1}{2}\right]
\right)  !}\sum_{m=0}^{N}\frac{\left(  \left[  \frac{k+1}{2}\right]  \right)
!\left(  4C_{0}t\right)  ^{m}}{\left(  \left[  \frac{k+1}{2}\right]
+m\right)  !}\\
& <&
\frac{\left(
4C_{0}t\right)  ^{\left[  \frac{k+1}{2}\right]  }}{\left(  \left[  \frac
{k+1}{2}\right]  \right)  !}\exp\left(  4C_{0}t\right)
\end{eqnarray*}
Assuming that $t\leq T<\frac{1}{16eC_{0}}<\frac{1}{4C_{0}}$ we obtain that
this sum can be estimated by $\frac{\left(  4C_{0}t\right)  ^{\left[
\frac{k+1}{2}\right]  }}{\left(  \left[  \frac{k+1}{2}\right]  \right)  !}e$.

On the other hand, using that $3^{N-1}\delta\in\left[  \frac{R_{\ast}}%
{4},\frac{3R_{\ast}}{4}\right)  $ we obtain: 
\begin{align*}
&\frac{1}{\sqrt{3^{k}\delta}}\left(  \frac{\left(  4C_{0}t\right)  ^{\left[
\frac{N+1}{2}\right]  }}{\left(  \left[  \frac{N+1}{2}\right]  \right)
!}\right)  ^{2} \leq\frac{2}{\sqrt{R_{\ast}}}\frac{\sqrt{3^{N-k-1}}\left(  4C_{0}t\right)
^{2\left[  \frac{N+1}{2}\right]  }}{\left(  \left(  \left[  \frac{N+1}%
{2}\right]  \right)  !\right)  ^{2}}\\
&  =2\frac{\left(  4C_{0}t\right)  ^{\left[  \frac{k+1}{2}\right]  }}{\left(
\left[  \frac{k+1}{2}\right]  \right)  !}e \times \\
&  \hskip 1cm \times \left[  \frac{1}{e\sqrt{R_{\ast}}}\left(  \sqrt{3}\right)
^{N-2\left[  \frac{N+1}{2}\right]  }\left(  \sqrt{3}\right)  ^{2\left[
\frac{N+1}{2}\right]  }\left(  3\right)  ^{-\frac{k+1}{2}+\left[  \frac
{k+1}{2}\right]  }\left(  3\right)  ^{-\left[  \frac{k+1}{2}\right]  }\times\right. \\
&\hskip 6cm \left. \times \left(
4C_{0}t\right)  ^{2\left[  \frac{N+1}{2}\right]  -\left[  \frac{k+1}%
{2}\right]  }\frac{\left(  \left[  \frac{k+1}{2}\right]  \right)  !}{\left(
\left(  \left[  \frac{N+1}{2}\right]  \right)  !\right)  ^{2}}\right]
\end{align*}
Using that $\sqrt{3}<3,$ $\frac{\left(  \left[  \frac{k+1}{2}\right]  \right)
!}{\left(  \left(  \left[  \frac{N+1}{2}\right]  \right)  !\right)  ^{2}}<1,$
$-\frac{k+1}{2}+\left[  \frac{k+1}{2}\right]  \leq0,\ N-2\left[  \frac{N+1}%
{2}\right]  \leq0,$ we can estimate the term between brackets as as:%
\[
\frac{1}{e\sqrt{R_{\ast}}}\left(  12C_{0}t\right)  ^{2\left[  \frac{N+1}%
{2}\right]  -\left[  \frac{k+1}{2}\right]  }%
\]
Since $2\left[  \frac{N+1}{2}\right]  -\left[  \frac{k+1}{2}\right]  \geq1,$
it then follows that, since $t\leq T\leq\min\left\{  \frac{1}{12C_{0}}%
,\frac{e\sqrt{R_{\ast}}}{12C_{0}}\right\}  $ that this term is smaller than
one. Then:%
\[
\frac{1}{\sqrt{3^{k}\delta}}\left(  \frac{\left(  4C_{0}t\right)  ^{\left[
\frac{N+1}{2}\right]  }}{\left(  \left[  \frac{N+1}{2}\right]  \right)
!}\right)  ^{2}\leq2\frac{\left(  4C_{0}t\right)  ^{\left[  \frac{k+1}%
{2}\right]  }}{\left(  \left[  \frac{k+1}{2}\right]  \right)  !}e
\]
and%
\begin{equation}
\partial_{t}m_{k}\left(  t\right)  \leq C_{0}\left[  \sum_{\ell=\left(
k-2\right)  _{+}}^{k-1}m_{\ell}\left(  t\right)  +4e\frac{\left(
4C_{0}t\right)  ^{\left[  \frac{k+1}{2}\right]  }}{\left(  \left[  \frac
{k+1}{2}\right]  \right)  !}\right]  \ ,\ \ k=1,2,..,N \label{K7E2}%
\end{equation}

We now derive estimates for the terms $m_{k}\left(  t\right)  $ iteratively,
taking as starting point the fact that $m_{0}\left(  t\right)  \leq1$ as well
as (\ref{K6E2}). Arguing by induction in $\ell$ it follows that, for $0\leq
t\leq t^{\ast}$ we have:
\begin{equation}
\max\left\{  m_{2\ell-1}\left(  t\right)  ,m_{2\ell}\left(  t\right)
\right\}  \leq\frac{3}{4}\frac{\left(  4C_{0}t\right)  ^{\ell}}{\ell
!}\ \ ,\ \ \ell=1,2,...\left[  \frac{N+1}{2}\right]  \label{K7E3}%
\end{equation}

Indeed, if $\ell=1$, using (\ref{K7E2}) as well as the fact that
\begin{equation}
16eC_{0}t\leq1\text{ for }0\leq t\leq T\ \label{K7E3c}%
\end{equation}
 we obtain
\begin{equation}
\partial_{t}m_{1}\left(  t\right)  \leq C_{0}\left[  1+16eC_{0}t\right]
<3C_{0}\ \label{K7E3a}%
\end{equation}
whence $m_{1}\left(  t\right)  \leq3C_{0}t.$ The definition of $T$ then
implies that $3C_{0}t\leq1$ for $0\leq t\leq T.$ Then, using again
(\ref{K7E2}) as well as (\ref{K7E3c}): $\ $%
\begin{equation}
\partial_{t}m_{2}\left(  t\right)  \leq C_{0}\left[  2+16eC_{0}t\right]
\leq3C_{0} \label{K7E3b}%
\end{equation}

Integrating (\ref{K7E3a}), (\ref{K7E3b}) we obtain (\ref{K7E3}) for $\ell=1.$
Suppose now that $1<\ell\leq\left[  \frac{N+1}{2}\right]  .$ Then, using the
induction hypothesis and (\ref{K7E2}) we obtain: 
\[
\partial_{t}m_{2\ell-1}\left(  t\right)  \leq C_{0}\left[
2\frac{\left(  4C_{0}t\right)  ^{\ell-1}}{\left(  \ell-1\right)  !}%
+4e\frac{\left(  4C_{0}t\right)  ^{\ell}}{\left(  \ell\right)  !}\right]
\]

Using again (\ref{K7E3c}) as well as the fact that $\left(  \ell\right)
!>\left(  \ell-1\right)  !$ we obtain $\partial_{t}m_{2\ell-1}\left(
t\right)  \leq3C_{0}\frac{\left(  4C_{0}t\right)  ^{\ell-1}}{\left(
\ell-1\right)  !}$ whence:%
\begin{equation}
m_{2\ell-1}\left(  t\right)  \leq\frac{3}{4}\frac{\left(  4C_{0}t\right)
^{\ell}}{\ell!} \label{K7E4}%
\end{equation}

Using again (\ref{K7E2}), combined with the induction hypothesis (\ref{K7E3})
as well as (\ref{K7E4}) and the fact that $\frac{\left(  4C_{0}t\right)
^{\ell}}{\ell!}\leq\frac{\left(  4C_{0}t\right)  ^{\ell-1}}{\left(
\ell-1\right)  !}$ if $\frac{4C_{0}t}{\ell!}<\frac{1}{\left(  \ell-1\right)
!}$ for $0\leq t\leq T,$ we obtain:%
\[
\partial_{t}m_{2\ell}\left(  t\right)  \leq C_{0}\left[  2\frac{3}{4}%
\frac{\left(  4C_{0}t\right)  ^{\ell-1}}{\left(  \ell-1\right)  !}%
+4e\frac{\left(  4C_{0}t\right)  ^{\ell}}{\left(  \ell\right)  !}\right]
\]
whence, using once more (\ref{K7E3c}) we obtain: $\partial_{t}m_{2\ell}\left(
t\right)  \leq3C_{0}\frac{\left(  4C_{0}t\right)  ^{\ell-1}}{\left(
\ell-1\right)  !}$ thus $m_{2\ell}\left(  t\right)  \leq\frac{3}{4}%
\frac{\left(  4C_{0}t\right)  ^{\ell}}{\ell!}.$ This concludes the proof of
(\ref{K7E3}) for $0\leq t\leq t_{\ast}$. Then the inequality (\ref{K7E1}) can
be obtained, due to the continuity of the functions $m_{k}\left(  t\right)  $
to some interval $0\leq t\leq t_{\ast}+\varepsilon_{0}$ with $\varepsilon
_{0}>0,$ but this contradicts the definition of $t_{\ast}$ and implies that
$t^{\ast}=T.$

In order to conclude the Proof of the Lemma we notice that the definition of
$m_{k}\left(  t\right)  $ and (\ref{K7E1}) yield:%
\[
\int_{\mathcal{U}_{k}}g\left(  t,d\omega\right)  \leq\int g\left(
t,d\omega\right)  \varphi_{k}=m_{k}\left(  t\right)  \leq\frac{\left(
4C_{0}t\right)  ^{\left[  \frac{k+1}{2}\right]  }}{\left(  \left[  \frac
{k+1}{2}\right]  \right)  !}%
\]
for $k=1,...,N.$ We now choose $k_{0}=\frac{N}{2}.$ Then $3^{k_{0}}\delta\leq
C\sqrt{\delta}.$ Moreover $k_{0}\rightarrow\infty$ as $\delta\rightarrow0.$
Then:%
\[
\int_{\bigcup_{k\geq k_{0}}\mathcal{U}_{k}}g\left(  t,d\omega\right)  \leq
C\frac{\left(  4C_{0}t\right)  ^{\left[  \frac{k_{0}+1}{2}\right]  }}{\left(
\left[  \frac{k_{0}+1}{2}\right]  \right)  !}\rightarrow0\text{ as }%
\delta\rightarrow0
\]

Classical continuity results for Radon measures then imply that:%
\[
\int_{\left[  0,\infty\right)  \setminus A_{\ast}}g\left(  t,d\omega\right)
=0\ \ ,\ \ 0\leq t\leq t_{\ast}%
\]
and the Lemma follows iterating the argument in time intervals $t_{\ast
}\leq t\leq2t_{\ast},...$
\end{proof}
\begin{proof}
[Proof of the Theorem \ref{Asympt}]We assume without loss of generality that
$m=\int g_{0}\left(  d\omega\right)  =1.$ Let $\varphi\in C_{0}\left(  \left[
0,\infty\right)  \right)  .$ Suppose first that $R_{\ast}=0.$ We define
$\bar{A}=A\setminus\left\{  0\right\}  .$ Let $\bar{R}_{\ast}=\inf\left(
\bar{A}^{\ast}\right)  ,$ where we will use the notation $\inf\left(
\emptyset\right)  =\infty.$ We will consider separately the cases $\bar
{R}_{\ast}=0$ and $\bar{R}_{\ast}>0.$ Suppose first that $\bar{R}_{\ast}=0.$
Our goal is to show (\ref{K3E6}). Let $\varepsilon>0$ arbitrarily small. Since
$\bar{R}_{\ast}=0$ there exist points $z\in\bar{A}^{\ast}$ with $z$
arbitrarily small and $B_{r}\left(  z\right)  \subset\left(  0,\varepsilon
^{2}\right]  $ for some $r>0$. Then, Lemma \ref{Diff} with $t^{\ast}=1$ yields
$c_{0}=\int_{\left(  0,\varepsilon^{2}\right]  }g\left(  1,d\omega\right)
\geq\int_{B_{r}\left(  z\right)  }g\left(  1,d\omega\right)  >0.$ Let
$\eta=\frac{1}{2}$ and $R=\varepsilon^{2}.$ Then Lemma \ref{cotInf} implies
that:%
\begin{equation}
\int_{\left[  0,2\varepsilon^{2}\right]  }g\left(  t,d\omega\right)  \geq
\frac{c_{0}}{2}\ \ ,\ \ t\geq1 \label{K3E7}%
\end{equation}

Suppose now that $R_{\ast}=0$ and $\bar{R}_{\ast}>0.$ Then $\int_{\left(
0,\bar{R}_{\ast}\right)  }g_{0}\left(  d\omega\right)  =0.$ However, since
$R_{\ast}=0$ we have $\int_{\left[  0,\bar{R}_{\ast}\right)  }g_{0}\left(
d\omega\right)  >0,$ whence $\int_{\left\{  0\right\}  }g_{0}\left(
d\omega\right)  >0.$ Then $\int_{\left[  0,\varepsilon^{2}\right]  }%
g_{0}\left(  d\omega\right)  \geq c_{0}.$ Applying then Lemma \ref{cotInf}
with $R=\varepsilon^{2}$ and $\eta=\frac{1}{2}$ we then obtain (\ref{K3E7}).
Therefore we always have (\ref{K3E7}) if $R_{\ast}=0.$

We now have the following alternative. Either there exists $\bar{t}\geq1$ such
that $\int_{\left[  0,4\varepsilon^{2}\right)  }g\left(  \bar{t}%
,d\omega\right)  \geq1-\frac{\varepsilon}{3},$ or otherwise we have
\begin{equation}
\int_{\left[  0,4\varepsilon^{2}\right]  }g\left(  t,d\omega\right)
\leq1-\frac{\varepsilon}{3}\text{ \ \ for any \ }t\geq1\ \label{K4E2}%
\end{equation}
Our aim is to show that second case implies a contradiction. Suppose that
(\ref{K4E2}) takes place. We choose now $R$ sufficiently large to have
$\int_{\left[  0,R\right]  }g_{0}\left(  d\omega\right)  \geq1-\frac
{\varepsilon}{12}.$ Applying Lemma \ref{cotInf} it then follows that
\begin{equation}
\int_{\left[  0,\frac{12R}{\varepsilon}\right]  }g\left(  t,d\omega\right)
\geq1-\frac{\varepsilon}{6}\ \ ,\ \ t\geq0 \label{K4E3}%
\end{equation}
Combining (\ref{K4E2}), (\ref{K4E3}):%
\begin{equation}
\int_{\left(  4\varepsilon^{2},\frac{12R}{\varepsilon}\right]  }g\left(
t,d\omega\right)  \geq\frac{\varepsilon}{6}\ \ ,\ \ t\geq1 \label{K3E8}%
\end{equation}

Let $\varphi\in C^{2}\left(  \left[  0,\infty\right)  \right)  $ a test
function satisfying the following properties:%
\begin{eqnarray*}
&&\text{for any }\omega\in\left[  0,\infty\right):\,\,\,\varphi\left(  \omega\right)  \geq0,\ \ \varphi^{\prime}\left(
\omega\right)  >0,\ \ \varphi^{\prime\prime}\left(  \omega\right)
<0;\\
&&\lim_{\omega\rightarrow\infty}\varphi\left(  \omega\right)  =1\ \ ,\ \ \varphi
\left(  0\right)  =0
\end{eqnarray*}
Let the function $H_{\varphi}^{1}$ as in Lemma \ref{strictConvex}. Suppose
that $\omega_{-}\in\left[  0,2\varepsilon^{2}\right]  $ and $\omega_{0}%
,\omega_{+}\in\left(  4\varepsilon^{2},\frac{12R}{\varepsilon}\right]  .$ Then
$\left(  \omega_{+}-\omega_{-}\right)  \geq\left(  \omega_{0}-\omega
_{-}\right)  \geq2\varepsilon^{2}.$ Using the strict concavity of $\varphi$ in
bounded regions as well as Taylor Theorem we obtain that for such values of
$\left(  \omega_{1},\omega_{2},\omega_{3}\right)  :$
\[
H_{\varphi}^{1}\left(  \omega_{1},\omega_{2},\omega_{3}\right)  \leq
-\kappa\varepsilon^{2}%
\]
where $\kappa>0$ depends on the values of the second derivative of $\varphi$
in $\left[  0,\frac{12R}{\varepsilon}\right]  .$ (Therefore it depends on
$\varepsilon$). Using Proposition \ref{atractiveness} as well as Lemma
\ref{strictConvex} it then follows that:%
\[
\partial_{t}\left(  \int\varphi\left(  \omega\right)  g\left(  t,d\omega
\right)  \right)  \leq-\bar{\kappa}\varepsilon^{2}\left(  \int_{\left(
4\varepsilon^{2},\frac{12R}{\varepsilon}\right]  }g\left(  t,d\omega\right)
\right)  ^{2}\int_{\left[  0,2\varepsilon^{2}\right]  }g\left(  t,d\omega
\right)  \ \ ,\ \ t\geq1
\]

Using (\ref{K3E7}) and (\ref{K3E8}) we then obtain:%
\[
\partial_{t}\left(  \int\varphi\left(  \omega\right)  g\left(  t,d\omega
\right)  \right)  \leq-\frac{\bar{\kappa}c_{0}}{72}\varepsilon^{4}%
\ \ ,\ \ t\geq1
\]
but this contradicts the boundedness of $\int\varphi\left(  \omega\right)
g\left(  t,d\omega\right)  .$ Therefore (\ref{K4E2}) cannot be satisfied and
we then obtain that there exists $\bar{t}=\bar{t}\left(  \varepsilon\right)
\geq1$ such that $\int_{\left[  0,4\varepsilon^{2}\right)  }g\left(  \bar
{t},d\omega\right)  \geq1-\frac{\varepsilon}{3}.$

We then apply again Lemma \ref{cotInf} with $\eta=\frac{\varepsilon}{3}$ to
prove that 
$$\int_{\left[  0,\frac{12\varepsilon^{2}}{\varepsilon}\right)
}g\left(  t,d\omega\right)  =\int_{\left[  0,12\varepsilon^{2}\right)
}g\left(  t,d\omega\right)  \geq1-\frac{2\varepsilon}{3}$$
 for any $t\geq
\bar{t}.$ Then $\int_{\left[  12\varepsilon^{2},\infty\right)  }g\left(
t,d\omega\right)  <\frac{2\varepsilon}{3}.$ Using the continuity of $\varphi$
we then obtain that:%
\[
\left\vert \int g\left(  t,d\omega\right)  \varphi\left(  \omega\right)
-\varphi\left(  0\right)  \right\vert \leq\sup_{\omega\in\left[
0,12\varepsilon^{2}\right]  }\left\vert \varphi\left(  \omega\right)
-\varphi\left(  0\right)  \right\vert +\frac{4\varepsilon}{3}\sup_{\omega
}\left\vert \varphi\left(  \omega\right)  \right\vert \ \ \text{if\ }t\geq
\bar{t}\left(  \varepsilon\right)
\]

Since $\varepsilon$ can be made arbitrarily small we obtain (\ref{K3E6}) if
$R_{\ast}=0.$

Let us assume now that $R_{\ast}>0.$ In this case, due to Lemma \ref{SetAstar}
the set $A^{\ast}$ has the form (\ref{K5E2}). Lemma \ref{Diff} with $t^{\ast
}=1$ implies that $\int_{\left\{  R_{\ast}\right\}  }g\left(  1,d\omega
\right)  >0.$ Moreover, Lemma \ref{gSupport} implies
that\ $\operatorname*{supp}\left(  g\left(  t,.\right)  \right)  \subset
A^{\ast}$ for any $t\geq0.$ We define the test function $\varphi\left(
\omega\right)  =\frac{\left(  \bar{\omega}-\omega\right)  _{+}}{\left(
\bar{\omega}-R_{\ast}\right)  }$ with $\bar{\omega}=\frac{1}{2}\min
_{k=1,...L}\left\{  R_{\ast}+D_{k}\right\}  .$ Since $\varphi$ is convex,
Lemma \ref{Convex} implies that $\partial_{t}\left(  \int_{\left\{  R_{\ast
}\right\}  }g\left(  t,d\omega\right)  \right)  =\partial_{t}\left(  \int
g\left(  t,d\omega\right)  \varphi\left(  \omega\right)  \right)  \geq0.$
Moreover, using the form of the function $\mathcal{G}_{0,\varphi}$ in Lemma
\ref{strictConvex} as well as (\ref{S2E5}) and the fact that $H_{\varphi}%
^{2}\geq0$ we obtain:%
\begin{equation}
\frac{d}{dt}\left(  \int_{\left\{  R_{\ast}\right\}  }g\left(  t,d\omega
\right)  \right)  \geq\frac{1}{3}\int_{0}^{\infty}\int_{0}^{\infty}\int
_{0}^{\infty}\frac{g_{1}g_{2}g_{3}}{\sqrt{\omega_{0}\omega_{+}}}H_{\varphi
}^{1}\left(  \omega_{1},\omega_{2},\omega_{3}\right)  d\omega_{1}d\omega
_{2}d\omega_{3} \label{K7E5}%
\end{equation}

Notice that we can restrict the integration in the right of (\ref{K7E5}) to
the set $\omega_{+}\in A^{\ast}\setminus\left\{  R_{\ast}\right\}  ,$ due to
the fact that $\omega_{+}\in\left\{  R_{\ast}\right\}  $ implies $\omega
_{+}=\omega_{0}=\omega_{-}=R_{\ast}$ and $H_{\varphi}^{1}=0.$ Our choice of
the function $\varphi$ implies that $\varphi\left(  \omega_{+}+\omega
_{0}-\omega_{-}\right)  =\varphi\left(  \omega_{+}\right)  =0$ if $\omega
_{+}\in A^{\ast}\setminus\left\{  R_{\ast}\right\}  .$ Then, using the form of
$H_{\varphi}^{1}$ we obtain:%
\[
\frac{d}{dt}\left(  \int_{\left\{  R_{\ast}\right\}  }g\left(  t,d\omega
\right)  \right)  \geq\frac{1}{3}\iiint\limits_{\left\{  \omega_{+}\in A^{\ast
}\setminus\left\{  R_{\ast}\right\}  \right\}  }\frac{g_{1}g_{2}g_{3}}%
{\sqrt{\omega_{0}\omega_{+}}}\varphi\left(  \omega_{+}+\omega_{-}-\omega
_{0}\right)  d\omega_{1}d\omega_{2}d\omega_{3}%
\]
and since $\varphi\geq0$ we obtain:%
\begin{align*}
\frac{d}{dt}\left(  \int_{\left\{  R_{\ast}\right\}  }g\left(  t,d\omega
\right)  \right)   &  \geq
\iiint_{\left\{\substack{ \omega_{+},\omega_{0}\in A^{\ast
}\setminus\left\{  R_{\ast}\right\} \\ \omega_{-}=R_{\ast}}\right\}  }
\frac{g_{1}g_{2}g_{3}}{\sqrt{\omega_{0}\omega_{+}}%
}\varphi\left(  \omega_{+}+\omega_{-}-\omega_{0}\right)  d\omega_{1}%
d\omega_{2}d\omega_{3}\\
&  =\frac{1}{3}\iiint_{\left\{\substack{ \omega_{+},\omega_{0}\in A^{\ast
}\setminus\left\{  R_{\ast}\right\} \\ \omega_{-}=R_{\ast}}\right\}  }%
\frac{g_{1}g_{2}g_{3}}{\sqrt{\omega_{0}\omega_{+}}}d\omega_{1}d\omega
_{2}d\omega_{3}\\
&  =K_{1}\int_{\left\{  R_{\ast}\right\}  }g\left(  t,d\omega\right)  \left[
\int_{A^{\ast}\setminus\left\{  R_{\ast}\right\}  }g\left(  t,d\omega\right)
\right]  ^{2}%
\end{align*}
where $K_{1}>0$ contains combinatorial factor whose precise value is not
relevant. Due to the monotonicity of $\int_{\left\{  R_{\ast}\right\}
}g\left(  t,d\omega\right)  $ we have:%
\begin{eqnarray*}
\frac{d}{dt}\left(  \int_{\left\{  R_{\ast}\right\}  }g\left(  t,d\omega
\right)  \right) &\geq & K_{1}\int_{\left\{  R_{\ast}\right\}  }g\left(
1,d\omega\right)  \left[  \int_{A^{\ast}\setminus\left\{  R_{\ast}\right\}
}g\left(  t,d\omega\right)  \right]  ^{2}\\
&\geq &K_{2}\left[  \int_{A^{\ast
}\setminus\left\{  R_{\ast}\right\}  }g\left(  t,d\omega\right)  \right]
^{2}=K_{2}\left[  1-\int_{\left\{  R_{\ast}\right\}  }g\left(  t,d\omega
\right)  \right]  ^{2}%
\end{eqnarray*}
for $t\geq1.$ Then $\int_{\left\{  R_{\ast}\right\}  }g\left(  t,d\omega
\right)  \rightarrow1$ as $t\rightarrow\infty$ and the Theorem follows.
\end{proof}

The results in this Section provide detailed information on the asymptotics of
the weak solutions for arbitrary initial data $g_{in}.$ It is important to
remark that, that for any $R_{\ast}>0$ there exist initial data $g_{in}$ such
that the corresponding solution $g\left(  t,\cdot\right)  $ converges to
$m\delta_{R_{\ast}}$ as $t\rightarrow\infty.$ Indeed, any initial distribution
$g_{in}$ supported in a set $A$ such that $A^{\ast}$ is one of the sets
(\ref{K5E2}) yields this asymptotics for $g\left(  t,\cdot\right)  .$

It is interesting to remark that the aggregation of the particles towards
$\omega=0$ does not imply that the energy of the system becomes concentrated
in the region where $\omega$ is close to zero. Indeed, the particles with
$\omega=0$ have zero energy. Since $g\chi_{\left(  0,\infty\right)  }\left(
\omega\right)  \rightharpoonup0$ as $t\rightarrow\infty$ the only alternative
left, due to the conservation of the energy, is the flux of the remaining
energy towards large values of $\omega.$ The fact that the energy tends to
move towards large values of $k$ has been noticed repeatedly in the physical
literature (cf. for instance \cite{DNPZ}, \cite{P}). The precise result that
we prove is the following:

\begin{corollary}
\label{AsEnergy}Let $\rho<-2$ and $g\in C\left(  \left[
0,\infty\right)  :\mathcal{M}_{+}\left(  \left[  0,\infty\right)  :\left(
1+\omega\right)  ^{\rho}\right)  \right)  $ be a weak solution of (\ref{S2E1})
in the sense of Definition \ref{weakSolution} with initial datum $g_{in}%
\in\mathcal{M}_{+}\left(  \left[  0,\infty\right)  :\left(  1+\omega\right)
^{\rho}\right)  .$ Let $m=\int g_{in}d\omega$ and $e=\int\omega g_{in}%
d\omega.$ Suppose that $m>0$ and let $R_{\ast}$ as in Theorem \ref{Asympt}.
Suppose that $e>mR_{\ast}.$ Then, there exists an increasing function
$R\left(  t\right)  $ such that $\lim_{t\rightarrow\infty}R\left(  t\right)
=\infty$ and:%
\[
\lim_{t\rightarrow\infty}\int_{R\left(  t\right)  }^{\infty}\omega g\left(
t,d\omega\right)  =\left(  e-mR_{\ast}\right)
\]

\end{corollary}

\begin{remark}
This Corollary just means that all the excess of energy which does not
accumulate at the point $\omega=R_{\ast}$ escapes to infinity. Notice that the
inequality $e>mR_{\ast}$ is satisfied for any initial distribution $g_{in}$
except the ones given by $g_{in}=m\delta_{R_{\ast}}.$
\end{remark}

\begin{proof}
It is just a consequence of Theorem \ref{Asympt} as well as the conservation
of energy for these distributions (cf. (\ref{K2})).
\end{proof}

\subsection{Energy transfer towards large values of $k$.}
\label{Transfer}
Notice that Theorem \ref{Asympt} implies that the mass of the distribution $g$
tends to concentrate towards the smallest value of $\omega$ compatible with
the collision mechanism. Suppose that $R_{\ast}=\inf\left(  A^{\ast}\right)  $
On the other hand, if we assume that the total energy of the initial
distribution is bounded, i.e. $\int g_{0}\left(  \omega\right)  \omega
d\omega<\infty,$ we would have, due to the conservation of the energy that
$MR_{\ast}<\int g_{0}\left(  \omega\right)  \omega d\omega,$ with $M=\int
g_{0}\left(  \omega\right)  d\omega.$ Then:%
\[
\int\left(  g\left(  t,\omega\right)  -M\delta_{R_{\ast}}\right)  \omega
d\omega\rightarrow\left[  \int g_{0}\left(  \omega\right)  \omega
d\omega-MR_{\ast}\right]
\]
as $t\rightarrow\infty.$ Therefore, since $\int_{\left(  R_{\ast}%
,\infty\right)  }g\left(  t,\omega\right)  \rightharpoonup0$ as $t\rightarrow
\infty,$ it follows that a fraction of the energy of the system should move
towards large values of $\omega.$ Actually, if $R_{\ast}=0,$ all the energy of
the system moves towards large values of $\omega$ as $t\rightarrow\infty.$

It turns out that the rate of transfer of the energy towards larger values of
$k$ is given by the interaction between particles of a given size, say $R,$
with smaller particles. This is made more precise in the following result
which provides a characteristic time scale for the transfer of particles of
size $R$ towards larger values. It is worth to remark that for specific
initial data such a transfer of particles could be nonexistent. For instance,
if $g_{0}\left(  \omega\right)  =M\delta\left(  \omega-R_{0}\right)  ,$ the
transfer of energy towards larger modes does not take place. 

\begin{proposition}
\label{coarsening}Suppose that $g$ is a weak solution of (\ref{S2E1}) in the
sense of Definition \ref{weakSolution}. We define the test function
$\varphi_{R}\left(  \omega\right)  =RQ\left(  \frac{\omega}{R}\right)  ,$ with
$Q\left(  s\right)  =s$ for $0\leq s\leq\frac{1}{2},\ Q^{\prime}%
\geq0,\ Q\left(  s\right)  =1$ for $s\geq\frac{3}{2},\ Q\in C^{2}\left(
\left[  0,\infty\right)  \right)  .$ There exists a constant $c_{0}>0$
independent of $R$ such that, if $\int g_{0}d\omega=1,\ \int g_{0}\omega
d\omega=1,$ $\int g_{0}\varphi_{R}d\omega\geq\frac{1}{2}$ and $\int g\left(
T,\omega\right)  \varphi_{R}d\omega \geq\frac{1}{4}$ we have $T\geq c_{0}%
R^{2}.$
\end{proposition}

\begin{remark}
Given any initial configuration $g_{0}$ such that $\int g_{0}d\omega=M$, 
$\int g_{0}\omega d\omega=E,$ a simple rescaling argument implies that, if $\int
g_{0}\varphi_{R}d\omega\geq\frac{E}{2}$ and $\int g\left(  T,\omega\right)
\varphi_{R}d\omega\geq\frac{E}{4}$ we have $T\geq c_{0}\frac{R^{2}}{ME}$.
\end{remark}

\begin{proof}
We apply (\ref{Z2E1}) with test function $\varphi=\varphi_{R},$ $t_{\ast}%
\in\left[  0,T\right]  $ and $\sigma=0.$ We then obtain:
\begin{equation}
\partial_{t}\left(  \int_{\left[  0,\infty\right)  }g\left(  t,\omega\right)
\varphi\left(  \omega\right)  d\omega\right)  =\iiint_{(\left[  0,\infty\right)^3}
g_{1}g_{2}g_{3}\, \Delta _{ \varphi _k, 0 }(\omega _1, \omega _2, \omega _3)  d\omega_{1}d\omega_{2}d\omega_{3}\ \label{Z2E1a}%
\end{equation}
$a.e.$ $t\in\left[  0,T\right]$, with  $\Delta _{ \varphi _k, 0 }(\omega _1, \omega _2, \omega _3)$  as defined in
(\ref{Z2E2}), (\ref{Z2E6}).

Let us write $I_{R}=\left[  0,R\right]  $
and $I_{R}^{c}=\mathbb{R}\setminus\left[  0,R\right]  .$ We then split the
domains of integration in the last integral of (\ref{Z2E1a}) as follows:%
\begin{align*}
\iiint_{(\left[  0,\infty\right)^3}  &  =\int_{I_{R}}\int_{I_{R}}\int_{I_{R}}+\int_{I_{R}}%
\int_{I_{R}}\int_{I_{R}^{c}}+\int_{I_{R}}\int_{I_{R}^{c}}\int_{I_{R}}%
+\int_{I_{R}^{c}}\int_{I_{R}}\int_{I_{R}}+\\
&  +\int_{I_{R}}\int_{I_{R}^{c}}\int_{I_{R}^{c}}+\int_{I_{R}^{c}}\int_{I_{R}%
}\int_{I_{R}^{c}}+\int_{I_{R}^{c}}\int_{I_{R}^{c}}\int_{I_{R}}+\int_{I_{R}%
^{c}}\int_{I_{R}^{c}}\int_{I_{R}^{c}}\\
&  \equiv\sum_{k=1}^{8}J_{k}%
\end{align*}

We now claim that $J_{k}\geq-\frac{K}{R^{2}}\left(  \int g\varphi_{R}%
d\omega\right)  ^{2}$ for $k=1,...,8.$ Indeed, in the case of $J_{1}$ we
notice that the integrand vanishes unless $\max\left\{  \omega_{1},\omega
_{2}\right\}  \geq\frac{R}{2}.$ A symmetrization argument allows to reduce the
domain of integration to the set $\left\{  \omega_{1}\geq\omega_{2}\right\}
.$ Then $\omega_{1}\geq\frac{R}{2}.$ We then split the domain of integration
in the variables $\left(  \omega_{2},\omega_{3}\right)  $ in the sets
$I_{\frac{R}{8}}\times I_{\frac{R}{8}}$ and its complement. Then:%

\begin{eqnarray*}
J_{1}&=&2\int_{I_{R}\setminus I_{\frac{R}{2}}}d\omega_{1}
\iint_{
I_{\frac{R}{8}}\times I_{\frac{R}{8}} }d\omega_{2}d\omega_{3}%
+2\int_{I_{R}\setminus I_{\frac{R}{2}}}\!\!\!d\omega_{1}\hskip -0.7cm \iint\limits_{ \substack{\null \hskip 0.8 cm
\left\{ \omega_{1}\geq \omega_{2}\right\}\setminus \left(
I_{\frac{R}{8}}\times I_{\frac{R}{8}}\right)  ^{c}   } }\hskip -0.4cm d\omega_{2}d\omega_{3}\\
&=&J_{1,1}+J_{1,2}
\end{eqnarray*}

Using the definition of $\varphi_{R}$ and using also that in the domain of
integration of $J_{1,1}$ we have $\left(  \omega_{1}+\omega_{2}-\omega
_{3}\right)  \geq\max\left\{  \omega_{2},\omega_{3}\right\}  $ we obtain:%
\begin{eqnarray*}
&&J_{1,1}\geq\int_{I_{R}\setminus I_{\frac{R}{2}}}d\omega_{1}\iint
_{I_{\frac{R}{8}}\times I_{\frac{R}{8}}}d\omega_{2}d\omega_{3}\frac{g_{1}%
g_{2}g_{3}\min\left\{  \sqrt{\omega_{2}},\sqrt{\omega_{3}}\right\}  }%
{\sqrt{\omega_{1}\omega_{2}\omega_{3}}}\times \\
&& \hskip 4cm \times \left[  \omega_{3}+\varphi_{R}\left(
\omega_{1}+\omega_{2}-\omega_{3}\right)  -\omega_{2}-\varphi_{R}\left(
\omega_{1}\right)  \right]
\end{eqnarray*}

Due to the symmetry of the integral:%
\begin{equation}
\iint_{I_{\frac{R}{8}}\times I_{\frac{R}{8}}}d\omega_{2}d\omega_{3}%
\frac{g_{2}g_{3}\min\left\{  \sqrt{\omega_{2}},\sqrt{\omega_{3}}\right\}
}{\sqrt{\omega_{2}\omega_{3}}}\left[  \omega_{3}-\omega_{2}\right]  =0
\label{K1E1}%
\end{equation}

Therefore, using Taylor's expansion we obtain:%
\begin{eqnarray}
&&J_{1,1}\geq\int_{I_{R}\setminus I_{\frac{R}{2}}}d\omega_{1}\iint
_{I_{\frac{R}{8}}\times I_{\frac{R}{8}}}d\omega_{2}d\omega_{3}\frac{g_{1}%
g_{2}g_{3}\min\left\{  \sqrt{\omega_{2}},\sqrt{\omega_{3}}\right\}  }%
{\sqrt{\omega_{1}\omega_{2}\omega_{3}}}\times  \label{K1E2}\\
&&\hskip 4cm \times\left[  \frac{\partial\varphi_{R}%
}{\partial\omega}\left(  \omega_{1}\right)   \left(  \omega_{2}-\omega
_{3}\right)  -\frac{C}{R}\left(  \omega_{2}-\omega_{3}\right)  ^{2}\right]\nonumber
\end{eqnarray}
for some $C>0$ independent of $R$ perhaps changing from line to line$.$ Using
(\ref{K1E1}) we obtain the cancellation of the integral of the first term
between brackets in (\ref{K1E2}). Then:%
\[
J_{1,1}\geq-\frac{C}{R^{\frac{3}{2}}}\int_{I_{R}\setminus I_{\frac{R}{2}}%
}d\omega_{1}\iint_{\left(  I_{\frac{R}{8}}\times I_{\frac{R}{8}}\right)
\cap\left\{  \omega_{2}\geq\omega_{3}\right\}  }d\omega_{2}d\omega_{3}%
g_{1}g_{2}g_{3}\left(  \omega_{2}\right)  ^{\frac{3}{2}}%
\]
whence, using the definition of $\varphi_{R}:$%
\begin{equation}
J_{1,1}\geq-\frac{C}{R^{2}}\left(  \int g\varphi_{R}d\omega\right)
^{2}\left(  \int gd\omega\right)  \geq-\frac{C}{R^{2}}\left(  \int
g\varphi_{R}d\omega\right)  ^{2} \label{K1E3}%
\end{equation}

On the other hand, in the integrand of $J_{1,2},$ either $\omega_{2}$ or
$\omega_{3}$ are larger than $\frac{R}{8}.$ Then:%
\begin{equation}
J_{1,2}\geq-\frac{C}{R^{2}}\left(  \int g\varphi_{R}d\omega\right)  ^{2}
\label{K1E4}%
\end{equation}

Combining (\ref{K1E3}), (\ref{K1E4}) we obtain:%
\begin{equation}
J_{1}\geq-\frac{C}{R^{2}}\left(  \int g\varphi_{R}d\omega\right)  ^{2}
\label{K1E5}%
\end{equation}

In order to estimate $J_{2}$ we use the fact that the for the integrand to be
different from zero we need $\max\left\{  \omega_{1},\omega_{2}\right\}
\geq\frac{R}{2}.$ We can assume that, say $\omega_{1}\geq\omega_{2}.$ We then
use that $\Phi\leq\sqrt{\omega_{2}}$ as well as the properties of $\varphi
_{R}$ to obtain:%
\begin{equation}
J_{2}\geq-\frac{C}{R^{2}}\left(  \int g\varphi_{R}d\omega\right)  ^{2}
\label{K1E6}%
\end{equation}

The terms $J_{3},\ J_{4}$ can be estimated in a similar manner. In order to
estimate $J_{4}$ we split the integral as:%
\[
J_{4}=\int_{I_{R}^{c}}d\omega_{1}\iint_{I_{\frac{R}{8}}\times
I_{\frac{R}{8}}}d\omega_{2}d\omega_{3}+\int_{I_{R}\setminus
I_{\frac{R}{2}}}d\omega_{1}\iint_{\left(  I_{\frac{R}{8}}\times I_{\frac
{R}{8}}\right)  ^{c}}d\omega_{2}d\omega_{3}=J_{4,1}+J_{4,2}%
\]

Notice that, using again the symmetry $\omega_{2}\longleftrightarrow\omega
_{3}$ as well as Taylor, we obtain:
\begin{align*}
&  \iint_{ I_{\frac{R}{8}}\times I_{\frac{R}{8}}}
d\omega_{2}d\omega_{3}\frac{g_{2}g_{3}\min\left\{  \sqrt{\omega_{2}}%
,\sqrt{\omega_{3}}\right\}  }{\sqrt{\omega_{2}\omega_{3}}}\left[  \omega
_{3}+\varphi_{R}\left(  \omega_{1}+\omega_{2}-\omega_{3}\right)  -\omega
_{2}-\varphi_{R}\left(  \omega_{1}\right)  \right] \\
&  \geq\iint_{ I_{\frac{R}{8}}\times I_{\frac{R}{8}}
}d\omega_{2}d\omega_{3}\frac{g_{2}g_{3}\min\left\{  \sqrt{\omega_{2}}%
,\sqrt{\omega_{3}}\right\}  }{\sqrt{\omega_{2}\omega_{3}}}\left[
\frac{\partial\varphi_{R}}{\partial\omega}\left(  \omega_{1}\right)  \left(
\omega_{2}-\omega_{3}\right)  -\frac{C}{R}\left(  \omega_{2}-\omega
_{3}\right)  ^{2}\right] \\
&  =-\frac{C}{R}\iint_{ I_{\frac{R}{8}}\times I_{\frac{R}{8}%
}}d\omega_{2}d\omega_{3}\frac{g_{2}g_{3}\min\left\{  \sqrt{\omega
_{2}},\sqrt{\omega_{3}}\right\}  }{\sqrt{\omega_{2}\omega_{3}}}\left(
\omega_{2}-\omega_{3}\right)  ^{2}\geq-\frac{C}{R^{\frac{1}{2}}}\int
g\varphi_{R}d\omega
\end{align*}

Using now that $\omega_{1}\geq\frac{R}{2}$ in the integral $J_{4,1}$ we obtain
$J_{4,1}\geq-\frac{C}{R^{2}}\left(  \int g\varphi_{R}d\omega\right)  ^{2}.$ On
the other hand, $J_{4,2}$ can be estimated as $J_{1,2}.$ Then:%
\begin{equation}
J_{3}+J_{4}\geq-\frac{C}{R^{2}}\left(  \int g\varphi_{R}d\omega\right)  ^{2}
\label{K1E7}%
\end{equation}

Finally, we notice that in the integrals $J_{5},\ J_{6},\ J_{7},\ J_{8}$ there
are at least two integration variables larger than $R.$ Then:%
\begin{equation}
J_{5}+J_{6}+J_{7}+J_{8}\geq-\frac{C}{R^{2}}\left(  \int g\varphi_{R}%
d\omega\right)  ^{2} \label{K1E8}%
\end{equation}

Combining (\ref{K1E5}), (\ref{K1E6}), (\ref{K1E7}), (\ref{K1E8}) and applying
(\ref{Z2E1a}) we obtain:%
\[
\partial_{t}\left(  \int_{\left[  0,\infty\right)  }\varphi\left(
\omega\right)  g\left(  t,d\omega\right)  \right)  \geq-\frac{C}{R^{2}}\left(
\int\varphi_{R}\left(  \omega\right)  g\left(  s,d\omega\right)  \right)
^{2}ds\ \ ,\ \ a.e.\ \ t\in\left[  0,T\right]
\]
whence the result follows.
\end{proof}

Notice that Proposition \ref{coarsening} can be understood as an upper
estimate for the rate of transfer of energy towards higher values of $\omega.$
Indeed, the assumption $\int g_{0}\varphi_{R}d\omega\geq\frac{1}{2}$ in
Proposition \ref{coarsening} means that a significant fraction of the energy
of the system is in the region $\omega\leq R.$ In order to reduce this amount
of energy by a significant amount we need at least times of order $R^{2}$.

\subsection{Detailed  asymptotic behaviour of weak solutions.}

The following result gives a more detailed information on the form in which
$g\left(  t,\cdot\right)  $ approaches to $m\delta_{0}$ as $t\rightarrow
\infty$ if $R_{\ast}=0$ in Theorem \ref{Asympt}.

\begin{theorem}
\label{AsympOsc}Let $\rho,\ g,\ m,\ R^{\ast}$ as in Theorem \ref{Asympt}.
Suppose that $R^{\ast}=0.$ Then, one of the following alternatives hold:

(i) There exists $t^{\ast}\geq0,\ t^{\ast}<\infty$ such that $\int_{\left\{
0\right\}  }g\left(  t,d\omega\right)  >0$ for $\ a.e.\ t\geq t^{\ast}.$
Moreover, for $a.e.\ t_{1},t_{2}$ such that $t_{1}\leq t_{2}$ we have
$\int_{\left\{  0\right\}  }g\left(  t_{1},d\omega\right)  \leq\int_{\left\{
0\right\}  }g\left(  t_{2},d\omega\right)  $.

(ii) There exist an unbounded set of times $A\subset\left[  0,\infty\right)  $
such that for any $t\in A$ there exist $\Omega\left(  t\right)  ,\ \eta\left(
t\right)  \ $such that:%
\begin{eqnarray}
&&\int_{\left(  \Omega\left(  t\right)  \left(  1-\eta\left(  t\right)  \right)
,\Omega\left(  t\right)  \left(  1+\eta\left(  t\right)  \right)  \right)
}g\left(  t,d\omega\right)  \geq m\left(  1-\eta\left(  t\right)  \right)
\ \ \text{for}\ \ t\in A\label{L1E2}\\
&&\int_{\left\{  0\right\}  }g\left(
t, d\omega\right)  =0\ \ \text{for}\ t\geq0  \label{L1E2M}
\end{eqnarray}
with
\begin{equation}
\lim_{t\rightarrow\infty}\ _{A}\eta\left(  t\right)  =0,\ \ \lim
_{t\rightarrow\infty}\ _{A}\Omega\left(  t\right)  =0. \label{L1E3}%
\end{equation}
where, for a function $f:A\rightarrow\mathbb{R}$ we will say that
$\lim_{t\rightarrow\infty}\ _{A}f\left(  t\right)  =0$ iff for any $\delta>0$
there exists $L$ sufficiently large such that, for any $t\in A\cap\left\{
t>L\right\}  $ we have $\left\vert f\left(  t\right)  \right\vert \leq\delta.$

The density of the set $A$ is one, in the sense that:%
\begin{equation}
\lim_{T\rightarrow\infty}\frac{\left\vert A\cap\left[  0,T\right]  \right\vert
}{T}=1 \label{L1E1}%
\end{equation}

In the case (ii) the function $\Omega\left(  t\right)  $ defined in $A$ has an
additional property, namely that, for any $\varepsilon>0,$ there exists
$T_{0}$ large such that, if $t\geq T_{0}$ we have $\Omega\left(  s\right)
\leq\Omega\left(  t\right)  \left(  1+\varepsilon\right)  $ for any $s\geq
t,\ $with $s,t\in A.$
\end{theorem}

\begin{remark}
Theorem \ref{AsympOsc} states that, either a condensate \index{condensate} appears in finite
time, i.e. the quantity $\int_{\left\{  0\right\}  }g\left(  t,dx\right)  $ becomes
positive at some finite $t,$ or alternatively during most of the times
$g\left(  t,\cdot\right)  $ can be
approximated by means of a of Dirac mass at a distances $R\left(  t\right)  $
of the origin. Notice that we can reformulate (ii) using the rescaled
measures:%
\[
G\left(  t,\cdot\right)  =\frac{1}{\Omega\left(  t\right)  }g\left(
t,\frac{\cdot}{\Omega\left(  t\right)  }\right)
\]
Then, the alternative (ii) in Theorem \ref{AsympOsc} implies that:%
\[
\sup_{t\in A\cap\left\{  t>L\right\}  } dist _{\ast}\left(
G\left(  t,\cdot\right)  ,m\delta_{1}\right)  \rightarrow0\ \ \text{as\ \ }%
L\rightarrow\infty\ .
\]

\end{remark}

\begin{remark}
\label{setB}The function $\Omega\left(  t\right)  $ is ``almost-monotone", in
the sense that for $\Omega\left(  t_{2}\right)  $ is smaller than
$\Omega\left(  t_{1}\right)  ,$ plus some small error which can be made
arbitrarily small, if $t_{2}\geq t_{1}$ with $t_{1}$ large. Nevertheless, it
is important to take into account that alternative (ii) in Theorem
\ref{AsympOsc} does not imply that $G\left(  t,\cdot\right)  $ can be
approximated as $m\delta_{1}$ for any large $t.$ In Chapter 4 we
will construct a class of weak solutions for which the alternative (ii) of
Theorem \ref{AsympOsc} holds. One of the properties of those particular
solutions is the existence, for each of them of an unbounded set
$B\subset\left[  0,\infty\right)  ,$ and a positive constant $c_{0}$ such that
for any $t\in B$ and any $\Omega>0$ we have $dist _{\ast
}\left(  \frac{1}{\Omega}g\left(  t,\frac{\cdot}{\Omega}\right)  ,m\delta
_{1}\right)  \geq c_{0}>0.$ Seemingly, such a set $B$ exists for any weak
solution of (\ref{S2E1}) in the sense of Definition \ref{weakSolution} for
which the alternative (ii) in Theorem \ref{AsympOsc} is satisfied. However, we
will not prove this result in this paper with that degree of generality. The
main difficulty proving the existence of these sets $B$ for arbitrary
solutions is to control the displacement of the mass of $g$ for distributions
that are close to a Dirac mass.
\end{remark}

\begin{remark}
We can choose initial values $g_{in}$ such that alternative (i) in Theorem
\ref{AsympOsc} holds (cf. Theorem \ref{weakCond}), and also initial data
$g_{in}$ for which alternative (ii) is satisfied (cf. Theorem \ref{globOsc}).
\end{remark}

\begin{proof}
[Proof of Theorem \texttt{\ref{AsympOsc}} ] 
Some of the methods required to prove Theorem \ref{AsympOsc} have been
introduced in \cite{EV1} in order to study singularity formation for the
Nordheim   \index{Nordheim} equation. The auxiliary results needed to prove
Theorem \texttt{\ref{AsympOsc}} have been included in Chapter 6.

Let $g\in C\left(
\left[  0,\infty\right)  :\mathcal{M}_{+}\left(  \left[  0,\infty\right)
:\left(  1+\omega\right)  ^{\rho}\right)  \right)  $ be a weak solution of
(\ref{S2E1}) in the sense of Definition \ref{weakSolution}. Suppose that for
any $t\geq0$:
\begin{equation}
\int_{\left\{  0\right\}  }g\left(  t,d\omega\right)  =0. \label{T2E2}%
\end{equation}

Given $a>0,$ $\delta>0,R>0$ arbitrarily small, let us denote as $\mathcal{A}%
_{a,R,\delta}\subset\left[  0,\infty\right)  $ the set of times $t\geq0$ with
the property that there exists an interval of the form $\mathcal{I}_{k}\left(
b,R\right)  $ with $b=1+a$ such that:%
\begin{equation}
\int_{\mathcal{I}_{k}\left(  b,R\right)  }g\left(  t,d\omega\right)
d\epsilon\geq\left(  1-\delta\right)  \int_{\left(  0,R\right]  }g\left(
t,d\omega\right)  d\epsilon\ \label{T2E5}%
\end{equation}
where $k=k\left(  t\right)  .$ We now claim that the Lebesgue's measure of the
complement of the set $\mathcal{A}_{a,R,\delta}$ satisfies:%
\begin{equation}
\lim_{T\rightarrow\infty}\frac{\left\vert \left[  0,T\right]  \setminus
\mathcal{A}_{a,R,\delta}\right\vert }{T}=0\ \label{T2E3}%
\end{equation}
for some $C\left(  a,\delta\right)  $ depending on $a,\ \delta$ and the total
mass of $g,$ $n_{0},$ but not on $R.$ Indeed, combining Lemmas \ref{estProd}
and \ref{Compl} as well as the fact that Lemma \ref{altresc} implies that only
the alternatives (i) and (ii) can take place we would have:
\begin{equation}
\nu\int_{\left[  0,T\right]  \setminus\mathcal{A}_{a,R,\delta}\left(
T\right)  }dt\left(  \int_{\left(  0,R\right]  }g\left(  t,d\omega\right)
\right)  ^{3}\leq\frac{2Bb^{\frac{7}{2}}MR}{\rho^{2}\left(  \sqrt{b}-1\right)
^{2}} \label{T2E4}%
\end{equation}

By assumption $R^{\ast}=0.$ Theorem \ref{Asympt}, as well as (\ref{T2E2})
imply that, for any fixed $R>0$ we have:%
\[
\int_{\left(  0,R\right]  }g\left(  t,d\omega\right)  \geq\frac{M}%
{2}>0\ \ \text{with\ \ }M=\int g_{in}\left(  d\epsilon\right)
\]
if $t\geq t^{\ast}$, $t^{\ast}$ is sufficiently large. \ Combining
(\ref{T2E4}) with this inequality we obtain:%
\[
\frac{1}{T}\int_{\left[  t^{\ast},T\right]  \setminus\mathcal{A}_{a,R,\delta
}\left(  T\right)  }dt\leq\frac{C\left(  a,\delta,R,M\right)  }{T}%
\]
whence:%
\[
\frac{1}{T}\int_{\left[  0,T\right]  \setminus\mathcal{A}_{a,R,\delta}\left(
T\right)  }dt\leq\frac{C\left(  a,\delta,R,M\right)  +t^{\ast}}{T}%
\]
and taking the limit $T\rightarrow\infty,$ we obtain (\ref{T2E3}). We can then
construct the set $A$ as follows. We choose decreasing sequences $\left\{
a_{n}\right\}  ,\ \left\{  \delta_{n}\right\}  ,\ \left\{  R_{n}\right\}  $
converging to zero. The definition of the sets $\mathcal{A}_{a_{n}%
,R_{n},\delta_{n}}$ then implies that $\mathcal{A}_{a_{n+1},R_{n+1}%
,\delta_{n+1}}\subset\mathcal{A}_{a_{n},R_{n},\delta_{n}}.$ We then choose
$T_{n}$ sufficiently large to have
\[
\frac{\left\vert \left[  0,T\right]  \setminus\mathcal{A}_{a_{n},R_{n}%
,\delta_{n}}\right\vert }{T}\leq\frac{1}{n}\ \ \text{for }T\geq T_{n}%
\]

Actually, we will assume an stronger condition on the sequence $\left\{
T_{n}\right\}  ,$ namely:%
\[
\frac{\left\vert \mathcal{A}_{a_{n},R_{n},\delta_{n}}\cap\left[
T_{n-1},T\right]  \right\vert }{T}\geq1-\frac{1}{n}\ \ \text{for\ \ }T\geq
T_{n+1}%
\]

Notice that, due to (\ref{T2E3}) this is possible assuming that the sequence
$T_{n}$ increases fast enough in order to have $\frac{T_{n}-T_{n-1}}{T_{n+1}}$
sufficiently small, say smaller than $2^{-n}$. We then define sets
$\mathcal{B}_{n},\ A$ as:%
\[
\mathcal{B}_{n}=\left[  0,T_{n}\right]  \cup\mathcal{A}_{a_{n},R_{n}%
,\delta_{n}}\ \ ,\ \ A=\bigcap_{n}\mathcal{B}_{n}%
\]

Then, if $T_{\ell+1}\leq T$ we obtain $\mathcal{B}_{\ell+1}\cap\left[
T_{\ell-1},T\right)  \subset\mathcal{B}_{\ell-1}\cap\left[  T_{\ell
-1},T\right)  .$ On the other hand, $\mathcal{B}_{n}\cap\left[  0,T\right)
=\left[  0,T\right)  $ if $n<\ell-1.$ We then write:%
\[
A\cap\left[  0,T\right]  =\left[  A\cap\left[  0,T_{\ell-1}\right]  \right]
\cup\left[  A\cap\left[  T_{\ell-1},T\right]  \right]
\]
and:%
\begin{equation}
\frac{\left\vert A\cap\left[  0,T_{\ell-1}\right]  \right\vert }{T}\leq
\frac{T_{\ell-1}}{T}\leq2\cdot2^{-\ell+1}=4\cdot2^{-\ell} \label{F1}%
\end{equation}%
\[
A\cap\left[  T_{\ell-1},T\right]  =\mathcal{B}_{\ell-1}\cap\left[  T_{\ell
-1},T\right]
\]
whence:%
\begin{equation}
\frac{\left\vert A\cap\left[  T_{\ell-1},T\right]  \right\vert }{T}%
=\frac{\left\vert \mathcal{B}_{\ell-1}\cap\left[  T_{\ell-1},T\right]
\right\vert }{T}=\frac{\left\vert \mathcal{A}_{a_{\ell},R_{\ell},\delta_{\ell
}}\cap\left[  T_{\ell-1},T\right]  \right\vert }{T}\geq1-\frac{1}{n}
\label{F2}%
\end{equation}
Combining (\ref{F1}), (\ref{F2}) we obtain (\ref{L1E1}).

The definition of $A$ implies the existence of $\Omega\left(  t\right)  $ with
the properties stated in the Theorem. To conclude the description of the
solutions given in the Theorem in the second case, it only remains to prove
that for any $\varepsilon>0,$ there exists $T_{0}$ large such that
$\Omega\left(  s\right)  \leq\Omega\left(  t\right)  \left(  1+\varepsilon
\right)  $ if $s\geq t\geq T_{0}.$ To this end we use the convex test function
$\varphi\left(  \omega\right)  =\left(  1-\frac{\omega}{\Omega\left(
s\right)  \left(  1+\varepsilon\right)  }\right)  _{+}$. Lemma \ref{Convex}
implies that the function $t\rightarrow\int\left(  1-\frac{\omega}%
{\Omega\left(  s\right)  \left(  1+\varepsilon\right)  }\right)  _{+}g\left(
t,d\omega\right)  $ is increasing. On the other hand, (\ref{L1E2}), (\ref{L1E2M}), 
(\ref{L1E3}) imply that, assuming that $s\in A$ is large enough, we have
$\int\left(  1-\frac{\omega}{\Omega\left(  s\right)  \left(  1+\varepsilon
\right)  }\right)  _{+}g\left(  t,d\omega\right)  \geq\frac{m\varepsilon
}{1+\varepsilon}-\theta,$ where $\theta>0$ can be made arbitrarily small if
$s$ is large enough. Indeed, this is due to the fact that $g\left(
t,\cdot\right)  $ can be approximated as $m\delta_{\Omega\left(  s\right)  }$
if $s\in A$ is sufficiently large and both the dispersion of the mass around
$\omega=\Omega\left(  s\right)  $ and the amount of mass which is not close to
this point can be made arbitrarily small for $t\geq T_{0}$ with $T_{0}$
depending on $\varepsilon.$ On the other hand, if $\Omega\left(  t\right)
\geq\Omega\left(  s\right)  \left(  1+\varepsilon\right)  $ we would obtain
$\int\left(  1-\frac{\omega}{\Omega\left(  s\right)  \left(  1+\varepsilon
\right)  }\right)  _{+}g\left(  t,d\omega\right)  =0.$ This would imply a
contradiction.

In order to conclude the Proof of the Theorem it only remains to show that
the function $t\rightarrow$ $\int_{\left\{  0\right\}  }g\left(  \bar
{t},d\omega\right)  $ is increasing. To this end we construct a family of
\ convex test functions depending on two positive parameters $M,a.$ The
functions of the family which will be denotes as $\varphi_{M,a}\left(
\omega\right)  $ satisfy $\varphi_{M,a}\left(  0\right)  =M,$ are increasing
in $M,$ decreasing in $\omega$ and such that the limit $\lim_{M\rightarrow
\infty}\varphi_{M,a}\left(  \omega\right)  =\varphi_{\infty,a}\left(
\omega\right)  $ is finite for any $\omega>0.$ We will assume also that the
family $\left\{  \varphi_{\infty,a}\left(  \omega\right)  :a>0\right\}  $ is
increasing in $a$ and it satisfies $\lim_{a\rightarrow0^{+}}\varphi
_{M,a}\left(  \omega\right)  =0$ for any $\omega>0.$ Using Lemma \ref{Convex}
we obtain $\int\varphi_{M,a}\left(  \omega\right)  g\left(  t,d\omega\right)
\geq\int\varphi_{M,a}\left(  \omega\right)  g\left(  \bar{t},d\omega\right)  $
$a.e.\ t\geq\bar{t}$ whence:%
\begin{equation}
\int\varphi_{M,a}\left(  \omega\right)  g\left(  t,d\omega\right)  \geq
\varphi_{M,a}\left(  0\right)  \int g\left(  \bar{t},d\omega\right)  =M\int
g\left(  \bar{t},d\omega\right)  \ ,\ \ a.e.\ t\geq\bar{t} \label{L1E5}%
\end{equation}

Taking the limit $a\rightarrow0^{+}$ we obtain, for any $M>0$ fixed:%
\[
\lim_{a\rightarrow0^{+}}\int_{\left(  0,\infty\right)  }\varphi_{M,a}\left(
\omega\right)  g\left(  t,d\omega\right)  =0
\]
whence, using the fact that 
$$\int\varphi_{M,a}\left(  \omega\right)  g\left(
t,d\omega\right)  =\int_{\left\{  0\right\}  }\varphi_{M,a}\left(  0\right)
g\left(  t,d\omega\right)  +\int_{\left(  0,\infty\right)  }\varphi
_{M,a}\left(  \omega\right)  g\left(  t,d\omega\right)  $$
 and taking the limit $a\rightarrow0^{+}$ it follows, using (\ref{L1E5}),  as well as that
$\varphi_{M,a}\left(  0\right)  =M$ that:

\[
M\int_{\left\{  0\right\}  }g\left(  t,d\omega\right)  \geq M\int g\left(
\bar{t},d\omega\right)  \ \ \ ,\ \ a.e.\ t\geq\bar{t}%
\]
whence the result follows.
\end{proof}

\subsection{Finite time condensation. }\index{condensation}

\begin{theorem}
\label{weakCond}For  $\rho<-1$, there exist 
$g_{in}\in\mathcal{M}_{+}\left(  \left[  0,\infty\right)  :\left(
1+\omega\right)  ^{\rho}\right)  $ such that for any $g\in C\left(
\left[  0,\infty\right)  :\mathcal{M}_{+}\left(  \left[  0,\infty\right)
:\left(  1+\omega\right)  ^{\rho}\right)  \right)$, weak solution of
(\ref{S2E1}) in the sense of Definition \ref{weakSolution}, with initial data $g _{ in }$,   the alternative (i) in
Theorem \ref{AsympOsc} holds. 
\end{theorem}

\begin{proof}
The existence of initial data $g_{in}\in\mathcal{M}_{+}\left(  \left[
0,\infty\right)  :\left(  1+\omega\right)  ^{\rho}\right)  $ for which there
exist $\bar{t}\geq0$ such that $\int_{\left\{  0\right\}  }g\left(  \bar
{t},d\omega\right)  >0$ can be proved as in \cite{EV1}, Theorem 10.5. It would
be possible to prove this result also using the arguments in \cite{Lu3}. The
main difference among the results in \cite{EV1} and \cite{Lu3} concerning the
formation of condensate \index{condensate} is that the initial data required in \cite{Lu3} must
be assumed to behave like a suitable power law near the origin. On the
contrary, in the case of the initial data considered in \cite{EV1} it is
possible to assume that the initial function $f$ is bounded. Moreover the
methods in \cite{Lu3}yield instantaneous condensation,\index{condensation} in the sense that the
solutions constructed there have $\int_{\left\{  0\right\}  }g\left(  \bar
{t},d\omega\right)  >0$ for values of $\bar{t}$ arbitrarily small, due to the
singular character of the initial data. On the contrary, the methods in
\cite{EV1} allow to obtain solutions satisfying $\int_{\left\{  0\right\}
}g\left(  t,d\omega\right)  =0$ in some interval $0\leq t<t^{\ast}$ and
$\int_{\left\{  0\right\}  }g\left(  \bar{t},d\omega\right)  $ for values
$\bar{t}>t^{\ast}$ arbitrarily close to $t^{\ast}$. \end{proof}

\subsection{Finite time blow up \index{blow up}  of bounded mild solutions.}
It is worth to remark that it is possible to have blow-up in finite time in
the same manner starting from initially bounded solutions, as it happens for the 
 Nordheim  \index{Nordheim} equation.
\begin{theorem}
\label{BU}Let $M>0,\ E>0,$ $\nu>0,\ \rho<-2$ be given. There exist 
$r=r(M, E, \nu)  > 0, K^{\ast}=K^{\ast}(M,E,\nu)  >0, T_{0}
=T_{0}\left(  M,E\right)$ and $\ \theta_{\ast}>0$
independent on $M, E, \nu$,  such that for any $g_{in}\in L^{\infty}\left(
\mathbb{R}^{+};\sqrt{\omega}\left(  1+\omega\right)  ^{\rho-\frac{1}{2}%
}\right)  $ satisfying$\ $%
\begin{align}
\int_{\mathbb{R}^{+}}g_{in}\left(  d\omega\right)   &  =M\ ,\ \ \int
_{\mathbb{R}^{+}}\omega g_{in}\left(  d\omega\right)  =E,\label{C1}\\
\int_{0}^{R}g_{in}\left(  d\omega\right)   &  \geq\nu R^{\frac{3}{2}%
}\ \ \text{for }0<R\leq r\ \ ,\ \ \int_{0}^{r}g_{in}\left(  d\omega\right)
\geq K^{\ast}\left(  r\right)  ^{\theta_{\ast}} \label{C2}%
\end{align}
there exists a unique mild solution $g\in C\left(  \left[  0,T_{\max}\right)
;L^{\infty}\left(  \mathbb{R}^{+};\sqrt{\omega}\left(  1+\omega\right)
^{\rho-\frac{1}{2}}\right)  \right)  $ of (\ref{S2E1}) defined for a maximal
existence time $T_{\max}<T_{0}$ and such that $f\left(  t,\omega\right)
=\frac{g\left(  t,\omega\right)  }{\sqrt{\omega}}$ satisfies:%
\[
\lim\sup_{t\rightarrow T_{\max}^{-}}\left\Vert g\left(  \cdot,t\right)
\right\Vert _{L^{\infty}\left(  \mathbb{R}^{+}\right)  }=\infty.
\]
\end{theorem}

\begin{theorem}
Let $M>0, E>0, \nu>0, \rho<-2$ be given. There exist $r=r(M, E, \nu)
>0, K^{\ast}=K^{\ast} (M, E, \nu)  > 0, T_{0}=T_{0}(M, E) $ and 
$\theta_{\ast} > 0$ independent on
$M,\ E, \nu$, such that for any $g_{in}\in L^{\infty}\left(  \mathbb{R}%
^{+};\sqrt{\omega}\left(  1+\omega\right)  ^{\rho-\frac{1}{2}}\right)  $
satisfying$\ $%
\begin{align}
\int_{\mathbb{R}^{+}}g_{in}\left(  d\omega\right)   &  =M\ ,\ \ \int
_{\mathbb{R}^{+}}\omega g_{in}\left(  d\omega\right)  =E\label{Z1E8}\\
\int_{0}^{R}g_{in}\left(  d\omega\right)   &  \geq\nu R^{\frac{3}{2}%
}\ \ \text{for }0<R\leq r\ \ ,\ \ \int_{0}^{r}g_{in}\left(  d\omega\right)
\geq K^{\ast}\left(  r\right)  ^{\theta_{\ast}} \label{Z1E9}%
\end{align}
there exists a weak solution $g\in C\left(  \left[  0,\infty\right)
:\mathcal{M}_{+}\left(  \left[  0,\infty\right)  :\left(  1+\omega\right)
^{\rho}\right)  \right)  $ of (\ref{S2E1}) such that there exists $T_{\ast}>0$
such that the following holds:%
\begin{equation}
\sup_{0\leq t\leq T_{\ast}}\left\Vert f\left(  t,\cdot\right)  \right\Vert
_{L^{\infty}\left(  \mathbb{R}^{+}\right)  }<\infty\ \ ,\ \ \sup_{T_{\ast
}<t\leq T_{0}}\int_{\left\{  0\right\}  }g\left(  t,d\omega\right)
>0\ \ \label{Z2E1b}%
\end{equation}
where $g=\sqrt{\omega}f.$
\end{theorem}
The Proof of these Theorems is similar to the Proof of Theorems 2.3 and 10.5
in \cite{EV1} respectively.

\section{Solutions without condensation: Pulsating behavior}\index{condensation} \index{pulsating}

\subsection{Statement of  the result}

Our next goal is to prove the existence of a class of solutions of
(\ref{S2E1}) for which the alternative (ii) in Theorem \ref{AsympOsc} holds.
Moreover, the obtained solutions exhibit the behaviour described in Remark
\ref{setB}, namely the existence of an unbounded set of times $B$ in which
$g\left(  t,\cdot\right)  $ cannot be approximated by means of a Dirac mass in
the form indicated in (\ref{L1E2}). This type of pulsating behavior cannot take 
place for the solutions of the Nordheim equation, as it follows from the results in 
\cite{EV1}, \cite{EV2}. \index{pulsating}

\begin{theorem}
\label{globOsc}There exist a class of weak solutions with interacting condensate of (\ref{S2E1}), 
 in the sense of Definition \ref{weakSolution}, with $g\in
C\left(  \left[  0,\infty\right)  :\mathcal{M}_{+}\left(  \left[
0,\infty\right)  :\left(  1+\omega\right)  ^{\rho}\right)  \right)$, such that the alternative (ii) in
Theorem \ref{AsympOsc} holds. Such solutions have the property $\int_{\left\{
0\right\}  }g\left(  t,dx\right)  =0$ for any $t\geq0.$ Moreover, there exists
a constant $c_1>0$ independent of $g$ such that, 
\begin{equation}
\label{S4distance}
\limsup _{ t\to +\infty }\left(\inf _{a >0 }\left(dist _{\ast}\left(  \frac{1}{a}g\left(  t,\frac{\cdot
}{a}\right)  ,m\delta_{1}\right)\right)\right) \geq c_1
\end{equation}
\end{theorem}

\subsection{Proof of the result.}

We now prove Theorem \ref{globOsc}. To this end, we need to introduce some
notation in order to describe the class of measures under consideration.

\subsubsection{Notation.}

Our goal is to construct a weak solution of (\ref{S2E1}) with initial datum
$g_{in}$ which is a finite measure supported in a countable set of points. We
will assume without loss of generality that $\int g_{in}=1.$ It is convenient
to write the support as a disjoint union of countable sets. We define the
following functions defined for $k\in\left\{  1,2,3,...\right\}  :$\
\[
\theta_{\alpha}\left(  k\right)  =k\ \ \text{ if\ }\alpha=0\ \ ,\ \ \text{
}\theta_{\alpha}\left(  k\right)  =2k-1\ \ \text{ if\ }\alpha\neq0
\]

We then set%
\[
x_{\alpha}\left(  k\right)  =2^{-\alpha}\theta_{\alpha}\left(  k\right)
\ \ \ ,\ \ k=1,2,3,...\ \ ,\ \ \alpha=0,1,2,3,....
\]

We introduce also the following disjoint sets of points:%
\begin{equation}
\Omega_{\alpha}=\left\{  x_{\alpha}\left(  k\right)  :k=1,2,...\right\}
\ \ ,\ \ \alpha=0,1,2,3,.... \label{T2E6}%
\end{equation}

We are interested in measures which have the mass $a_{\alpha}\left(
k,t\right)  $ at a given point $x_{\alpha}\left(  k\right)  $ at the time
$t,\ $i.e., $g$ has the form:%
\begin{equation}
g\left(  t,\cdot\right)  =\sum_{\alpha=0}^{\infty}\sum_{k=1}^{\infty}%
a_{\alpha}\left(  k\right)  \delta_{x_{\alpha}\left(  k\right)  }\left(
\cdot\right)  \label{T2E6a}%
\end{equation}

It is convenient to study separately the masses in each of the families. We
will denote as $\mathcal{M}_{+}\left(  \Omega_{\alpha}\right)  $ the family of
finite Radon measures supported in $\Omega_{\alpha}.$ Then:%
\begin{equation}
g=\sum_{\alpha=0}^{\infty}a_{\alpha}\ \ \label{I1a}%
\end{equation}
where:%
\begin{equation}
a_{\alpha}=\sum_{k=1}^{\infty}a_{\alpha}\left(  k\right)  \delta_{x_{\alpha
}\left(  k\right)  }\in\mathcal{M}_{+}\left(  \Omega_{\alpha}\right)
\label{I1b}%
\end{equation}

We will denote as $m_{\alpha}$ the total mass contained in the family
$\Omega_{\alpha}:$%
\begin{equation}
m_{\alpha}=\sum_{k=1}^{\infty}a_{\alpha}\left(  k\right)  =\left\langle
a_{\alpha},1\right\rangle \label{I1c}%
\end{equation}

We define the following auxiliary families of points:%
\begin{equation}
\mathcal{Z}_{\alpha}=\bigcup_{\beta\leq\alpha}\Omega_{\beta} \label{S1E9a}%
\end{equation}

It is relevant to notice that the sets $\mathcal{Z}_{\alpha}$ can be obtained
from $\Omega_{0}$ by means of a rescaling. More precisely:%
\begin{equation}
\mathcal{Z}_{\alpha}=2^{-\alpha}\Omega_{0}\ ,\ \alpha\geq0 \label{S1E9}%
\end{equation}

\subsubsection{Heuristic description of the pulsating behaviour.} \index{pulsating}

The dynamics of the pulsating solutions will be characterized by the existence
of two sequences of time intervals $\left\{  t_{n}\right\}  ,\ \left\{
s_{n}\right\}  $ such that:%
\begin{align}
0  &  =t_{0}<s_{0}<t_{1}<s_{1}<t_{2}<...<t_{n}<s_{n}<t_{n+1}<s_{n+1}%
<...\label{S4E2}\\
\lim_{n\rightarrow\infty}t_{n}  &  =\lim_{n\rightarrow\infty}s_{n}%
=\infty\label{S4E3}%
\end{align}

During the time intervals\ $\left\{  \left[  t_{n},s_{n}\right)
:n=0,1,2,3,...\right\}  $ the solution $g\left(  \cdot,t\right)  $ will be in
the so-called slow dynamics. The main characteristic of the dynamics of the is
that, during these intervals most of the mass of $g$ remains concentrated at
the point $x_{n}\left(  1\right)  ,\ $i.e., the smallest point of the family
$\Omega_{n}.$ Therefore, we would have, approximately:%
\[
g\left(  \cdot,t\right)  \sim\delta\left(  \cdot-x_{n}\left(  1\right)
\right)  \ \ ,\ \ t\in\left[  t_{n},s_{n}\right)
\]

During the time intervals $\left\{  \left[  s_{n},t_{n+1}\right)
:n=0,1,2,3,...\right\}  $ the solution $g\left(  \cdot,t\right)  $ will pass
through the so-called fast dynamics. These dynamics will be characterized by a
fast transfer of most of the mass of $g$ from $\omega=x_{n}\left(  1\right)  $
to $\omega=x_{n+1}\left(  1\right)  .$ This transfer of mass will be
characterized in a first stage by the spreading of the mass of $g$ which at
the beginning of this phase was mostly at $\omega=x_{n}\left(  1\right)  $ to
the particles of the set $\mathcal{Z}_{n+1}$ defined in (\ref{S1E9a}). During
a second stage, this mass is transferred towards the particle $\omega
=x_{n+1}\left(  1\right)  $ in such a way that at the time $t=t_{n+1}$ most of
the mass of $g$ is at such a point.

Something that will play a crucial role is the remainder of the mass which is
not concentrated at the points $x_{n}\left(  1\right)  $ during the slow
phases or at the points of the sets\ $\mathcal{Z}_{n+1}$ during the fast phases.

During the $n-$th fast phase, the main process taking place is the transfer of
the mass from $\omega=x_{n}\left(  1\right)  $ to $\omega=x_{n+1}\left(
1\right)  .$ More precisely, at the time $t=s_{n}$ most of the mass of $g$ is
concentrated at the point $\omega=x_{n}\left(  1\right)  $ and at the time
$t=t_{n+1}$ most of the mass will be concentrated at the point $\omega
=x_{n+1}\left(  1\right)  .$

During the forthcoming slow phase, which takes place in the time interval
$\left[  t_{n+1},s_{n+1}\right)  $ the most relevant feature taking place is
the interaction between the masses placed at the points $\omega=x_{n+1}\left(
1\right)  ,\ \omega=x_{n+2}\left(  1\right)  .$ The mass $a_{n+2}\left(
1\right)  $ whichis much smaller than the mass $a_{n+1}\left(  1\right)  $ can
be described by means of a linear equation, which predicts a slow increase for
$a_{n+2}\left(  1\right)  $. Due to the slow increment of $a_{n+2}\left(
1\right)  $ the mass $a_{n+2}\left(  1\right)  $ becomes comparable to the
mass $a_{n+1}\left(  1\right)  .$ This event marks the beginning of the next
fast dynamics at time $t=s_{n+1}$.

The description of the masses during the fast stages cannot be approximated by
a system of linear equations. On the contrary, it requires to study a
system of nonlinear system of equations which describes the interactions
between particles with many different values for the masses. However, using
Proposition \ref{atractiveness} it is possible to prove that after a
sufficiently large time scale, most of the mass of $g$ is transferred to
$\omega=x_{n+2}\left(  1\right)  .$ The end of this transfer will take place
at the time $t=t_{n+2}$ and this will mark the starting time of a new slow phase.

The original distribution of masses at the points $x_{n}\left(  k\right)  $ at
time $t=0$ will be made in order to guarantee that most of the remainder of
the mass, not contained in $\omega=x_{n+2}\left(  1\right)  ,$ will be at
$\omega=x_{n+3}\left(  1\right)  $ at the time $t=t_{n+2}.$ Therefore, the
process described above would begin again and this iterative procedure would
be repeated for arbitrarily long times.

In order to gain some intuition for the evolution of $g$ we derive a set of
approximated differential equations which describe the most relevant masses of
$g$ during each of the two phases. Suppose that we take an initial
distribution of masses given by:%
\begin{equation}
a_{\alpha}\left(  k\right)  \left(  0\right)  =\varepsilon_{\alpha}%
\delta_{k,1}\ \ ,\ \ \text{with\ \ }\sum_{\alpha=0}^{\infty}\varepsilon
_{\alpha}=1\ \ ,\ \ \varepsilon_{\alpha}>0\ \ \label{S4E4}%
\end{equation}

We will use in this Subsection the notation $<<$ to indicate that the term on
the left is significantly smaller than the one on the right. We will use also
the symbol $\approx$ to indicate that the quantities in both sides of this
symbol are similar.

We first describe the slow phase which takes place in the interval $\left[
t_{n},s_{n}\right)  .$ During this phase most of the mass is concentrated in
the families $\Omega_{n},\ \Omega_{n+1}$ although the mass in the second
family will be much smaller, i. e. $m_{n+1}\left(  t_{n}\right)
<<m_{n}\left(  t_{n}\right)  .$ At the end of the previous fast dynamics phase
we have $m_{n}\left(  t_{n}\right)  \approx a_{n}\left(  1\right)  \left(
t_{n}\right)  \approx1.$ We neglect the mass of $g$ contained in the families
$\Omega_{\alpha}$ with $\alpha\geq\left(  n+2\right)  $ in the derivation of
the approximate equations describing the evolution of $g$ during this slow phase.

It is convenient to rescale the variables $g,\ \omega$ and $t$ in order to
bring the relevant points during this phase to the integers:%
\begin{equation}
g=2^{\left(  n+1\right)  }\bar{g}\ \ ,\ \ \ \ \omega=2^{-\left(  n+1\right)
}\bar{\omega}\ \ ,\ \ t=t_{n}+2^{-\left(  n+1\right)  }\bar{t} \label{S4E5}%
\end{equation}

This transformation brings the set $\Omega_{n}\cup\Omega_{n+1}$ to $\Omega
_{0}$ and it keeps invariant the equation (\ref{S2E1}). The evolution equation
for measures $\bar{g}$ of the form
\[
\bar{g}\left(  \bar{t},\cdot\right)  =\sum_{k=1}^{\infty}\bar{a}_{k}\left(
\bar{t}\right)  \delta_{k}\left(  \cdot\right)
\]
supported in $\Omega_{0}$ is given by the countable set of equations:
\begin{eqnarray}
&&\partial_{\bar{t}}\bar{a}_{n}\left(  \bar{t}\right) =\sum_{k+m-\ell
=n}\frac{\Phi_{k,m;\ell}}{\sqrt{km\ell}}\bar{a}_{k}\bar{a}_{m}\bar{a}_{\ell
}-\left(  \sum_{m,\ell=1}^{\infty}\frac{\left(  2\Phi_{n,m;\ell}-\Phi
_{\ell,m;n}\right)  }{\sqrt{nm\ell}}\bar{a}_{m}\bar{a}_{\ell}\right)  \bar
{a}_{n}, \label{S4E5a} 
\end{eqnarray}
for $n=1,2,\cdots$,  where
\begin{eqnarray*}
&&\Phi_{k,m;\ell}=\min\left\{  \sqrt{k},\sqrt{m},\sqrt{\ell},\sqrt{\left(
k+m-\ell\right)  _{+}}\right\} \\
&&\bar{a}_{2\ell-1}\left(  \bar{t}\right)  =a_{n+1}\left(  \ell\right)  \left(
t\right)  \ \ \ ,\ \ \ \bar{a}_{2\ell}\left(  \bar{t}\right)  =a_{n}\left(
\ell\right)  \left(  t\right)  \ \ \ ,\ \ \ \ell=1,2,...
\end{eqnarray*}

During the slow phase under consideration we can approximate the evolution of
the functions $\left\{  \bar{a}_{k}\right\}  $ using the equations.%
\begin{align}
\partial_{\bar{t}}\bar{a}_{1}  &  =\frac{1}{2}\left(  \bar{a}_{2}\right)
^{2}\bar{a}_{1}+\frac{1}{2\sqrt{3}}\left(  \bar{a}_{2}\right)  ^{2}\bar{a}%
_{3}\nonumber\\
\partial_{\bar{t}}\bar{a}_{2}  &  =-\left(  \bar{a}_{1}+\frac{1}{\sqrt{3}}%
\bar{a}_{3}\right)  \left(  \bar{a}_{2}\right)  ^{2}\nonumber\\
\partial_{\bar{t}}\bar{a}_{3}  &  =\frac{1}{2}\left(  \bar{a}_{2}\right)
^{2}\bar{a}_{1}+\frac{1}{2\sqrt{3}}\left(  \bar{a}_{2}\right)  ^{2}\bar{a}%
_{3}\ \label{S4E7}%
\end{align}
where we have neglected in (\ref{S4E5a}) all the contributions due to the
terms $\bar{a}_{k},\ k\geq4.$ Notice that this approximation is consistent
given the assumed relative size of the different masses.

Suppose that at the beginning of the slow phase (i.e. $\bar{t}=0$) we have
$\bar{a}_{k}=\alpha_{k},\ k=1,2,3.$ By assumption $\left(  \alpha_{1}%
+\alpha_{3}\right)  <<\alpha_{2}.$ As long as $\left(  \bar{a}_{1}+\bar{a}%
_{3}\right)  $ remains small compared with $\bar{a}_{2}$ we can approximate
the second equation in (\ref{S4E7}) as $\partial_{\bar{t}}\bar{a}_{2}=0,$ i.e.
$\bar{a}_{2}=\alpha_{2}.$ Using this approximation we obtain from (\ref{S4E7}):%

\begin{align*}
\bar{a}_{1}\left(  \bar{t}\right)   &  =\left(  \frac{\sqrt{3}\alpha
_{1}+\alpha_{3}}{\sqrt{3}+1}\right)  e^{K\bar{t}}-\frac{\left(  \alpha
_{3}-\alpha_{1}\right)  }{\sqrt{3}+1}\\
\bar{a}_{3}\left(  \bar{t}\right)   &  =\left(  \frac{\sqrt{3}\alpha
_{1}+\alpha_{3}}{\sqrt{3}+1}\right)  e^{K\bar{t}}-\frac{\sqrt{3}\left(
\alpha_{1}-\alpha_{3}\right)  }{\sqrt{3}+1}%
\end{align*}

These equations indicate that as $\bar{t}$ increases, the masses $\bar{a}%
_{1}\left(  \bar{t}\right)  $ and $\bar{a}_{3}\left(  \bar{t}\right)  $ become
similar and simultaneously both of them increases exponentially. If we write
$\varepsilon=$ $\left(  \sqrt{3}\alpha_{1}+\alpha_{3}\right)  $ we obtain that
$\bar{a}_{1}\left(  \bar{t}\right)  ,\ \bar{a}_{3}\left(  \bar{t}\right)  $
become or order one, for times of order $\bar{t}\approx\frac{\log\left(
\frac{1}{\varepsilon}\right)  }{K}.$ This marks the beginning of the next fast
dynamics phase.

The description of the masses $\bar{a}_{k}$ cannot be made using a linear
system of equations. During the fast phase the three masses $\bar{a}_{1},$
$\bar{a}_{2},$ $\bar{a}_{3}$ become comparable. It can then be readily seen
from (\ref{S4E5a}) that the mass of $g$ is distributed by means of an involved
nonlinear dynamics among many of the functions $\left\{  \bar{a}_{k}\right\}
$. It does not seem feasible to describe this dynamics by means of simple
formulas. However, the result in Theorem \ref{Asympt} indicates that after a
time $\bar{t}$ of order one most of the mass of $g$ becomes concentrated at
the value of $\bar{a}_{1}.$

The values of $\varepsilon_{\alpha}$ in (\ref{S4E4}) will be made in order to
ensure that at the end of the fast phase which transfers the mass from
$x_{n}\left(  1\right)  $ to $x_{n+1}\left(  1\right)  $ most of the remainder
mass which is not contained in $x_{n+1}\left(  1\right)  $ is in the family
$\Omega_{n+2}.$ Therefore an new slow phase begins which can be described as
explained above. The process is then repeated infinitely often as
$t\rightarrow\infty.$

It is important to take into account that the previous picture is an
oversimplified description of the evolution of $g.$ The main reason for this
is that during the $n-$th fast dynamics phase, the mass of the families
$\Omega_{\ell}$ with $\ell$ larger than $\left(  n+1\right)  $ is transported
to the points $x_{\ell}\left(  k\right)  $ with $k$ large. Actually the values
of $k$ to which a meaningful fraction of the mass of the family $\Omega_{\ell
}$ is transported is much larger if $\ell>>n.$ The consequence of this is that
it could take very large times to arrive to a mass distribution which allows
to approximate the dynamics of (\ref{S4E5a}) by means of the system of three
equations (\ref{S4E7}).

Another point to take into account is that in the previous heuristic
description was assumed that only two families $\Omega_{n},\ \Omega_{n+1}$ are
relevant at each time. In a strict sense a careful analysis of the evolution
of the mass in all the families $\Omega_{\ell}$ $\ell>\left(  n+1\right)  $
will be needed. A key point in the whole construction is that arguments
similar to the ones in the proof of Theorem \ref{Asympt} in the case $R_{\ast
}>0,$ and more precisely, arguments like the ones in the proof of Lemma
\ref{SetAstar} will imply that the transfer of mass of the families
$\Omega_{n}$ to $\Omega_{\ell}$ with $\ell>n$ can be estimated by the amount
of mass in the family $\Omega_{\ell}.$ From this point of many of the
arguments in the forthcoming pages can be thought as a some kind of continuous
dependence result of Lemma \ref{SetAstar}. Indeed, Lemma \ref{SetAstar} would
imply, for the class of measures considered in this Section, that the transfer
of mass from $\Omega_{n}$ to $\Omega_{\ell}$ with $\ell>n$ vanishes if the mass
at the family $\Omega_{\ell}$ is zero. We will prove now that this transfer is
small if the mass at the family $\Omega_{\ell}$ is small.

\subsubsection{Existence of global solutions in the class of measures supported
in the sets $\bigcup_{\beta}\Omega_{\beta}$.}

Our next goal is to prove the existence of a class of global measured valued
mild solutions of (\ref{Z2E2a}) such that $\mu\left(  \left[  0,\infty\right)
\setminus \bigcup_{\alpha=0}^{\infty}\Omega_{\alpha}\right)  =0$. The
construction follows similar ideas to the proof of existence of mild solutions
in Section \ref{mildConst}. The main difference is that we will work with a
particular class of measures. In order to use these measures, we need some
results of functional analysis.\\

\noindent
\textbf{Functional analysis preliminaries.}\\

We are interested in proving some existence results for equation (\ref{S2E1}) in a
class of measures $\mu\in\mathcal{M}_{+}\left(  \left[  0,\infty\right)
\right)  $ such that $\mu\left(  \left[  0,\infty\right)  \right)  =\mu\left(
\bigcup_{\alpha=0}^{\infty}\Omega_{\alpha}\right)  $ where the sets
$\Omega_{\alpha}$ are as in (\ref{T2E6}). Notice that according to the
definition of support of a measure given in (\ref{S1}), the support of a
measure is a closed set, and since the closure of $\bigcup_{\alpha=0}^{\infty
}\Omega_{\alpha}$ is $\left[  0,\infty\right)  $ the support of $\mu$ is
different from $\bigcup_{\alpha=0}^{\infty}\Omega_{\alpha}$ in general. 
The following Theorem collects several classical results in  \cite{Rudin} and \cite{Br}
 adapted to our particular setting.

\begin{theorem}
\label{measure}
We will denote as the set of Radon measures $\mathcal{M}_{+}\left(  \left[
0,\infty\right)  \right)  $ the set of all the positive, continuous linear
functionals $\Lambda:C_{0}(\left[  0,\infty\right)  )\rightarrow\mathbb{R}$,
where the topology of $C_{0}(\left[  0,\infty\right)  )$ is the topology of
uniform convergence in compact sets. Let us denote the family of Borel sets of
$\left[  0,\infty\right)  $ as $\mathcal{B}$. Then, there exists a unique
Borel measure $\mu$ such that:

a) $\Lambda f=\int_{X}fd\mu$ for every $f\in C_{c}(X)$.

b) $\mu(K)<\infty$ for every compact set $K\subset X$.

c) The relation $\mu(E)=\sup\{\mu(K):K\subset E,\,\,\,K\,\,\text{compact}\}$
holds for every $E\in\mathcal{B}$.

Moreover, the following additional properties hold:

i) For every $E\in\mathcal{B}$ we have:
$
\mu(E)=\inf\{\mu(V):E\subset V,\,\,V\,\text{open}\}.
$

ii) If $E\in\mathcal{B}$, $A\subset E$, and $\mu(E)=0$, then $\mu(A)=0$.

iii) If $E\in\mathcal{B}$, and $\varepsilon>0$, there is a closed set $F$ and
an open set $V$ such that $F\subset E\subset V$ and $\mu(V-F)<\varepsilon$.

iv) Suppose that we endow $\mathcal{M}_{+}\left(  \left[  0,\infty\right)
\right)  $ with the topology generated by the functionals $L_{\varphi
}:\mathcal{M}_{+}\left(  \left[  0,\infty\right)  \right)  \rightarrow
\mathbb{R}$, with $\varphi\in C_{0}(\left[  0,\infty\right)  ).$ Then, for any
$M>0,$ the set $$\left\{  \mu\in\mathcal{M}_{+}\left(  \left[  0,\infty\right)
\right)  :\int\mu\left(  d\omega\right)  \leq M\right\}$$
 is compact.
\end{theorem}

We will use extensively in the following Sections the fact that Radon measures
can be characterized in two equivalent ways, namely as functions which assign
values to the sets of the $\sigma-$algebra of Borel sets of $\left[
0,\infty\right)  $ or, alternatively as continuous functionals in
$C_{0}\left(  \left[  0,\infty\right)  \right)  .$ In order to avoid heavy
notation we will use the same letter to denote the measure as set funtion and
linear functional on $C_{0}\left(  \left[  0,\infty\right)  \right)  .$

We now define some functional spaces. From now on we will use the notation
$\mathbb{R}_{+}$ to denote the set $\left[  0,\infty\right)  .$ We remark that
in the measures used until the rest of the paper we will have $\mu\left(
\left\{  0\right\}  \right)  =0.$

\begin{definition}
Given $\theta>0,\ \rho^{\ast}>1,$ we will denote as $\mathcal{X}_{\theta
,\rho^{\ast}}$ the space of nonnegative Radon measures $\mu$ in $\mathcal{M}%
_{+}\left(  \left[  0,\infty\right)  \right)  $ such that $\mu\left(  \left[
0,\infty\right)  \setminus \bigcup_{\alpha=0}^{\infty}\Omega_{\alpha}\right)
=0$, satisfying
\begin{equation}
\left\Vert \mu\right\Vert _{\theta,\rho^{\ast}}\equiv\sup_{\alpha\geq0}\left(
\left(  2^{\alpha}\right)  ^{\theta}\mu\left(  \Omega_{\alpha}\right)
\right)  +\sup_{R\geq1}\left(  \frac{\mu\left(  \left[  \frac{R}{2},R\right]
\right)  }{R^{1-\rho^{\ast}}}\right)  <\infty\label{S7E1}%
\end{equation}

We will endow $\mathcal{X}_{\theta,\rho^{\ast}}$ with the weak topology of
measures, i.e. the topology induced by the functionals $\mu\rightarrow\int
\mu\varphi$ with $\varphi\in C_{0}\left(  \mathbb{R}_{+}\right)  .$
\end{definition}

\begin{remark}
Notice that, (\ref{S7E1}) implies that $\sum_{\alpha}\mu\left(  \Omega
_{\alpha}\right)  <\infty.$ In particular, this implies the following
representation formula for the measures $\mu\in\mathcal{X}_{\theta,\rho^{\ast
}}:$%
\begin{equation}
\mu=\sum_{\alpha=0}^{\infty}a_{\alpha}\ \ ,\ \ a_{\alpha}\left(  \ell\right)
=\mu\left(  \left\{  x_{\alpha}\left(  \ell\right)  \right\}  \right)
,\ \ a_{\alpha}=\sum_{\ell=1}^{\infty}a_{\alpha}\left(  \ell\right)
\delta_{x_{\alpha}\left(  \ell\right)  }\ \label{S8E5}%
\end{equation}
where the convergences of the series are understood in the sense of the weak
topology. Equivalently we can understand these measures as functions defined
in Borel sets. We recall that given a Borel set $B$ of $\left[  0,\infty
\right)  $ we have $\delta_{x_{0}}\left(  B\right)  =1$ if $x_{0}\in B$ and
$\delta_{x_{0}}\left(  B\right)  =0$ if $x_{0}\notin B.$
\end{remark}

We will prove now the following result that will play a crucial role in the following.

\begin{proposition}
\label{closure}Let $0<L<\infty,\ \theta>0,\ \rho^{\ast}>1.$ The sets
$$\mathcal{K}_{\theta,\rho^{\ast}}\left(  L\right)  =\left\{  \mu\in
\mathcal{X}_{\theta,\rho^{\ast}}:\left\Vert \mu\right\Vert _{\theta,\rho
^{\ast}}\leq L<\infty\right\}  \subset\mathcal{M}_{+}\left(  \left[
0,\infty\right)  \right)  $$
 are closed in the weak topology.
\end{proposition}

\begin{proof}
Suppose that $\left\{  \mu_{m}\right\}  $ is a sequence of measures contained
in $\mathcal{K}_{\theta,\rho^{\ast}}\left(  R\right)  $ such that $\mu
_{m}\rightharpoonup\mu$ in the weak topology. We must  prove that
$\mu\left(  \left[  0,\infty\right)  \setminus \bigcup_{\alpha=0}^{\infty
}\Omega_{\alpha}\right)  =0$ as well as the inequality $\left\Vert
\mu\right\Vert _{\theta,\rho^{\ast}}\leq R.$ We recall that for any Borel set
$A$ we have, using point (i) in Theorem (\ref{measure}), 
\[
\mathcal{\mu}\left(  A\right)  =\inf\left\{  \mathcal{\mu}\left(  U\right)
:A\subset U\ ,\ U\text{ open}\right\}.
\]
Given any $\beta\geq0$ we define:%
\[
\mathcal{V}_{_{\beta}}=\left[  \mathbb{R}_{+}\setminus \bigcup_{\alpha
=0}^{\beta}\Omega_{\alpha}\right]
\]
Notice that (cf. (\ref{S1E9a})):%
\[
\mathcal{V}_{\beta}\cap\left[  \bigcup_{\alpha=0}^{\infty}\Omega_{\alpha
}\right]  \subset\mathcal{Z}_{\beta+1}%
\]
Then, using that $\mu_{m}\in\mathcal{K}_{\theta,\rho^{\ast}}\left(  L\right)
$ we obtain
\begin{eqnarray}
\mu_{m}\left(  \mathcal{V}_{\beta}\right)  &=&\mu_{m}\left(  \mathcal{V}_{\beta
}\cap\left[  \bigcup_{\alpha=0}^{\infty}\Omega_{\alpha}\right]  \right)
\leq\mu_{m}\left(  \mathcal{Z}_{\beta+1}\right) \label{T2E7} \\
& \leq  & L\sum_{k=\beta
+1}^{\infty}\left(  2^{-k}\right)  ^{\theta}\leq CL\left(  2^{-\left(
\beta+1\right)  }\right)  ^{\theta}\nonumber
\end{eqnarray}
for some $C>0$ depending only in $\theta.$ Taking the limit $m\rightarrow
\infty$ and using the weak convergence $\mu_{m}\rightharpoonup\mu$ we obtain:%
\begin{equation}
\mu\left(  \mathcal{V}_{\beta}\right)  \leq CL\left(  2^{-\left(
\beta+1\right)  }\right)  ^{\theta} \label{T2E9}%
\end{equation}
Then, using (\ref{T2E9}) as well as the fact that $\left[  \mathbb{R}%
_{+}\setminus \bigcup_{\alpha=0}^{\infty}\Omega_{\alpha}\right]  =\bigcap
_{\beta=0}^{\infty}\mathcal{V}_{\beta}$ we obtain:%
\[
\mu\left(  \mathbb{R}_{+}\setminus \bigcup_{\alpha=0}^{\infty}\Omega_{\alpha
}\right)  \leq\sum_{\beta\geq N}\mu\left(  \mathcal{V}_{\beta,N}\right)  \leq
CL\sum_{\beta\geq N}\left(  2^{-\left(  1+N\right)  }\right)  ^{\theta}\leq
CL2^{-\theta\left(  N+1\right)  }%
\]
for $N\geq1$ arbitrary, with $C>0$ independent on $N.$ Taking the limit
$N\rightarrow\infty$ we obtain:%
\begin{equation}
\mu\left(  \mathbb{R}_{+}\setminus \bigcup_{\alpha=0}^{\infty}\Omega_{\alpha
}\right)  =0 \label{T3E1}%
\end{equation}
On the other hand, the measures $\mu_{m}$ satisfy the inequalities:%
\[
\sup_{\alpha\geq0}\left(  \left(  2^{\alpha}\right)  ^{\theta}\mu_{m}\left(
\Omega_{\alpha}\right)  \right)  +\sup_{R\geq1}\left(  \frac{\mu_{m}\left(
\left[  \frac{R}{2},R\right]  \right)  }{R^{1-\rho^{\ast}}}\right)  \leq L
\]
whence:
\[
\left(  2^{\alpha}\right)  ^{\theta}\mu_{m}\left(  \Omega_{\alpha}\right)
+\frac{\mu_{m}\left(  \left[  \frac{R}{2},R\right]  \right)  }{R^{1-\rho
^{\ast}}}\leq L
\]
for any $\alpha\geq0,\ R\geq1.$ Taking the limit $m\rightarrow\infty,$ and
then $\sup_{\alpha\geq0}$ and $\sup_{R\geq1}$ we obtain that:%
\[
\sup_{\alpha\geq0}\left(  \left(  2^{\alpha}\right)  ^{\theta}\mu\left(
\Omega_{\alpha}\right)  \right)  +\sup_{R\geq1}\left(  \frac{\mu\left(
\left[  \frac{R}{2},R\right]  \right)  }{R^{1-\rho^{\ast}}}\right)  \leq L
\]

This inequality, combined with (\ref{T3E1}) yields $\mu\in\mathcal{K}%
_{\theta,\rho^{\ast}}\left(  L\right)  $ and the Proposition follows.
\end{proof}

\begin{remark}
The previous result encodes in a single functional analysis result one of the
main ideas of the sought-for construction of measured valued solutions of
(\ref{S2E1}). Notice that the class of measures satisfying (\ref{T3E1}) is
dense in $\mathcal{M}_{+}\left(  \left[  0,\infty\right)  \right)  $ in the
weak topology. The reason because the set of measures $\mathcal{K}%
_{\theta,\rho^{\ast}}\left(  L\right)  $ is closed in $\mathcal{M}_{+}\left(
\left[  0,\infty\right)  \right)  ,$ in spite of the fact that the closure of
$\bigcup_{\alpha=0}^{\infty}\Omega_{\alpha}$ is $\left[  0,\infty\right)  ,$
is by the condition $\left\Vert \mu\right\Vert _{\theta,\rho^{\ast}}\leq L$
that yields a fast decay of the amount of mass concentrated in the sets
$\Omega_{\alpha}$ with large values of $\alpha.$
\end{remark}

\noindent
\textbf{Study of some auxiliary operators.}\\

We need to study the properties of the following auxiliary function
$A_{g}\left(  \omega_{1}\right)  .$ It is worth to compare this result with
Lemma \ref{LAg}. Notice that, differently from Lemma \ref{LAg}, we assume that
$\sigma=0.$ The function $A_{g}$ can still be defined in spite of this due to
the decay properties of the measures $g\in\mathcal{X}_{\theta,\rho^{\ast}}$
for small $\omega.$

\begin{lemma}
\label{LA}Suppose that $g\in\mathcal{X}_{\theta,\rho^{\ast}}$ for some
$\theta>\frac{1}{2},\ \rho^{\ast}>1.$ Then $A_{g}\left(  \omega_{1}\right)  $
defined by means of%
\begin{equation}
A_{g}\left(  \omega_{1}\right)  =-\iint\Phi\left[  \frac{2g_{2}g_{3}}%
{\sqrt{\omega_{1}\omega_{2}\omega_{3}}}-\frac{g_{3}g_{4}}{\sqrt{\omega
_{1}\omega_{3}\omega_{4}}}\right]  d\omega_{3}d\omega_{4} \label{S7E3}%
\end{equation}
where $\omega_{2}=\omega_{3}+\omega_{4}-\omega_{1}$ defines a continuous
function in $\left[  0,\infty\right)  .$ Moreover, we have:%
\begin{equation}
A_{g}\left(  \omega_{1}\right)  \geq0\ \ ,\ \ \omega_{1}\in\left[
0,\infty\right)  \label{S7E5}%
\end{equation}

\end{lemma}

\begin{proof}
The function $\frac{\Phi}{\sqrt{\omega_{1}}}$ is continuous for $\omega_{1}>0$
and $\left(  \omega_{1},\omega_{2}\right)  \in\mathbb{R}_{+}^{2}.$ Therefore,
each of the terms $\frac{2\Phi g_{2}g_{3}}{\sqrt{\omega_{1}\omega_{2}%
\omega_{3}}},$ $\frac{\Phi g_{3}g_{4}}{\sqrt{\omega_{1}\omega_{3}\omega_{4}}}$
are Radon measures in $\mathbb{R}_{+}^{2}$. In order to prove the convergence
of each of the integrals we just notice that the condition (\ref{S7E1})
implies the estimate:%
\begin{equation}
\int_{\left[  \frac{R}{2},R\right]  }g\left(  d\omega\right)  \leq C\left\Vert
g\right\Vert _{\theta,\rho^{\ast}}\min\left\{  R^{\theta},R^{1-\rho^{\ast}%
}\right\}  \ \ \text{with }\theta>\frac{1}{2},\ \rho^{\ast}>1 \label{S8E1}%
\end{equation}

Therefore, (\ref{S7E3}) defines a continuous function in $\left\{  \omega
_{1}>0\right\}  .$ We can define 
$$A_{g}\left(  0\right)=\lim
_{\omega_{1}}A_{g}\left(  \omega_{1}\right).$$ In order to prove the
existence of this limit we consider separately the two additive terms. In the
case of $J_{2}=\iint\frac{\Phi g_{3}g_{4}}{\sqrt{\omega_{1}\omega_{3}%
\omega_{4}}}$ we decompose the integration region in the sets $Q_{\delta
}=\left\{  \omega_{3}\geq\delta,\ \omega_{4}\geq\delta\right\}  $ with
$\delta>0$ small as well as its complementary $\mathbb{R}_{+}^{2}\setminus
Q_{\delta}$. The integrals $\iint_{Q_{\delta}}\left[  \cdot\cdot
\cdot\right]  $ are independent of $\omega_{1}$ if $\omega_{1}$ is small, due
to the definition of $\Phi.$ On the other hand, the term $\iint
_{\mathbb{R}_{+}^{2}\setminus Q_{\delta}}\left[  \cdot\cdot\cdot\right]  $
converges to zero as $\delta\rightarrow0$ due to (\ref{S8E1}). This implies
the existence of the limit $\lim_{\omega_{1}\rightarrow0}J_{2}.$ In order to
prove the existence of a similar limit for $J_{1}=\iint\frac{2\Phi
g_{2}g_{3}}{\sqrt{\omega_{1}\omega_{2}\omega_{3}}}d\omega_{3}d\omega_{4} $ we
first replace the variable of integration $\omega_{4}$ by $\omega_{2}$ by
means of a change of variables. We now repeat a similar splitting argument of
the integral in the sets $Q_{\delta}$ and $\mathbb{R}_{+}^{2}\setminus
Q_{\delta}$ and use the same argument to prove the existence of $\lim
_{\omega_{1}\rightarrow0}J_{1}.$ This concludes the proof of the existence of
the continuity of the function $A_{g}\left(  \cdot\right)  $ in $\left[
0,\infty\right)  .$

In order to prove (\ref{S7E5}) we rewrite $A_{g}\left(  \omega_{1}\right)  .$
Notice that:%
\begin{equation}
\iint\Phi\frac{2g_{2}g_{3}}{\sqrt{\omega_{1}\omega_{2}\omega_{3}}}%
d\omega_{3}d\omega_{4}=\iint\Phi\frac{g_{2}g_{3}}{\sqrt{\omega_{1}%
\omega_{2}\omega_{3}}}d\omega_{3}d\omega_{4}+\iint\Phi\frac{g_{2}g_{4}%
}{\sqrt{\omega_{1}\omega_{2}\omega_{4}}}d\omega_{3}d\omega_{4} \label{S7E4}%
\end{equation}

We now use the change of variables $\omega_{2}=\omega_{3}+\omega_{4}%
-\omega_{1},$ $d\omega_{2}=d\omega_{4}$ in the first integral and $\omega
_{2}=\omega_{3}+\omega_{4}-\omega_{1},$ $d\omega_{2}=d\omega_{3}$ in the
second one. Then, replacing the variable $\omega_{2}$ by $\omega_{4} $ in the
first resulting integral and $\omega_{2}$ by $\omega_{3}$ in the second, we
obtain that the integral in (\ref{S7E4}) becomes
\[
\iint\frac{g_{3}g_{4}}{\sqrt{\omega_{1}\omega_{3}\omega_{4}}}\Psi
d\omega_{3}d\omega_{4}%
\]
where:%
\[
\Psi=\Psi_{1}+\Psi_{2}%
\]%
\begin{align*}
\Psi_{1}  &  =\chi_{\left\{  \omega_{3}\geq\omega_{4}\right\}  }\chi_{\left\{
\omega_{3}\geq\omega_{1}\right\}  }\sqrt{\left(  \omega_{1}+\omega_{4}%
-\omega_{3}\right)  _{+}}+\chi_{\left\{  \omega_{3}\leq\omega_{4}\right\}
}\chi_{\left\{  \omega_{3}\geq\omega_{1}\right\}  }\sqrt{\omega_{1}}+\\
&  +\chi_{\left\{  \omega_{3}\geq\omega_{4}\right\}  }\chi_{\left\{
\omega_{3}\leq\omega_{1}\right\}  }\sqrt{\omega_{4}}+\chi_{\left\{  \omega
_{3}\leq\omega_{4}\right\}  }\chi_{\left\{  \omega_{3}\leq\omega_{1}\right\}
}\sqrt{\omega_{3}}%
\end{align*}%
\begin{align*}
\Psi_{2}  &  =\chi_{\left\{  \omega_{3}\leq\omega_{4}\right\}  }\chi_{\left\{
\omega_{4}\geq\omega_{1}\right\}  }\sqrt{\left(  \omega_{1}+\omega_{3}%
-\omega_{4}\right)  _{+}}+\chi_{\left\{  \omega_{3}\geq\omega_{4}\right\}
}\chi_{\left\{  \omega_{4}\geq\omega_{1}\right\}  }\sqrt{\omega_{1}}+\\
&  +\chi_{\left\{  \omega_{3}\leq\omega_{4}\right\}  }\chi_{\left\{
\omega_{4}\leq\omega_{1}\right\}  }\sqrt{\omega_{3}}+\chi_{\left\{  \omega
_{3}\geq\omega_{4}\right\}  }\chi_{\left\{  \omega_{4}\leq\omega_{1}\right\}
}\sqrt{\omega_{4}}%
\end{align*}

Notice that $\Psi\geq\Phi$ whence (\ref{S7E5}) follows.
\end{proof}

We now define a nonlinear operator in terms of any given measure
$g\in\mathcal{X}_{\theta,\rho^{\ast}}.$ Notice that it is possible to
characterize Radon measures either by means of the measure of Borel sets or by
means of the action of the measure as an element of the dual of the space of
compactly supported continuous functions. We have decided to follow the second
approach in the definition of $\mathcal{O}\left[  g\right]  $ in the following
Lemma, even if it would be simpler to define the measure of the subsets of
$\bigcup_{\alpha=0}^{\infty}\Omega_{\alpha}$ in order to obtain a definition
consistent with the one given in Lemma \ref{LQg}. 

\begin{lemma}
\label{LOp}Suppose that $g\in\mathcal{X}_{\theta,\rho^{\ast}}$ for some
$\theta>1$ and $\rho^{\ast}>1.$ Then, the following formula defines a mapping
$\mathcal{O}:\mathcal{X}_{\theta,\rho^{\ast}}\rightarrow\mathcal{X}%
_{\theta,\rho^{\ast}}:$%
\begin{equation}
\mathcal{O}\left[  g\right]  =\iint\Phi\frac{g_{2}g_{3}g_{4}}{\sqrt
{\omega_{2}\omega_{3}\omega_{4}}}d\omega_{3}d\omega_{4}\ ,\ \omega_{2}%
=\omega_{3}+\omega_{4}-\omega_{1}\ \label{S8E2}%
\end{equation}
where the action of the measure $\mathcal{O}\left[  g\right]  $ acting over a
test function $\varphi\in C_{0}\left(  \mathbb{R}_{+}\right)  $ is given by:%
\begin{equation}
\left\langle \mathcal{O}\left[  g\right]  ,\varphi\right\rangle =\iiint\Phi\frac{g_{2}g_{3}g_{4}}{\sqrt{\omega_{2}\omega_{3}\omega_{4}}}%
\varphi\left(  \omega_{1}\right)  d\omega_{3}d\omega_{4}d\omega_{1}
\label{S8E3}%
\end{equation}

Moreover, we have the estimate:%
\begin{equation}
\left\Vert \left(  \mathcal{O}\left[  g\right]  \right)  \right\Vert
_{\theta,\rho^{\ast}}\leq C\left\Vert g\right\Vert _{\theta,\rho^{\ast}}^{3}
\label{S8E3a}%
\end{equation}

\end{lemma}

\begin{proof}
Using the definition of the measure $g_{2}$ (i.e. the change of variables) we
would have:%
\begin{eqnarray}
\left\langle \mathcal{O}\left[  g\right]  ,\varphi\right\rangle =
\iiint\Phi\left(  \omega_{3}+\omega_{4}-\omega_{2},\omega_{2},\omega_{3}%
,\omega_{4}\right)  \frac{g_{2}g_{3}g_{4}}{\sqrt{\omega_{2}\omega_{3}%
\omega_{4}}}\times\nonumber \\ 
\hskip 0.1cm \times \varphi\left(  \omega_{3}+\omega_{4}-\omega_{2}\right)
d\omega_{2}d\omega_{3}d\omega_{4} \label{S8E4}%
\end{eqnarray}

Using the definition of $\Phi$ as well as (\ref{S8E1}) we immediately obtain
that (\ref{S8E4}) converges for any $\varphi\in C_{0}\left(  \mathbb{R}%
_{+}\right)  .$ Moreover $\left\vert \left\langle \mathcal{O}\left[  g\right]
,\varphi\right\rangle \right\vert \leq C\left\Vert g\right\Vert _{\theta
,\rho^{\ast}}^{3}\left\Vert \varphi\right\Vert _{\infty}$ and therefore
$\mathcal{O}\left[  g\right]  \in\mathcal{M}_{+}\left(  \mathbb{R}_{+}\right)
.$ Notice that the constant $C$ is independent of $\varphi,$ due to the decay
assumptions made for $g$ for large and small values$.$ Therefore, the operator
$\left\langle \mathcal{O}\left[  g\right]  ,\varphi\right\rangle $ is well
defined for any $\varphi\in C_{b}\left(  \mathbb{R}_{+}\right)  .$

In order to prove that $\mathcal{O}\left[  g\right]  \in\mathcal{X}%
_{\theta,\rho^{\ast}}$ let us  show that $\mathcal{O}\left[  g\right]
\left(  \mathbb{R}_{+}\setminus\bigcup_{\alpha=0}^{\infty}\Omega_{\alpha
}\right)  =0$ as well as $\left\Vert \mathcal{O}\left[  g\right]  \right\Vert
_{\theta,\rho^{\ast}}<\infty.$ To this end we first approximate $\mathcal{O}%
\left[  g\right]  $ in the weak topology by a sequence $\left(  \mathcal{O}%
\left[  g\right]  \right)  _{N}\in\mathcal{X}_{\theta,\rho^{\ast}}$ as
follows. Suppose that $g=\sum_{\alpha=0}^{\infty}a_{\alpha}$ , with
$a_{\alpha}\in\mathcal{X}_{\theta,\rho^{\ast}}$ satisfying $a_{\alpha}\left(
\mathbb{R}_{+}\setminus\Omega_{\alpha}\right)  $ (cf. (\ref{S8E5})). We then
define:%
\begin{align*}
\left\langle \left(  \mathcal{O}\left[  g\right]  \right)  _{N},\varphi
\right\rangle  &  =\sum_{\alpha\leq N}\sum_{\beta\leq N}\sum_{\gamma\leq
N}\iiint\Phi\left(  \omega_{3}+\omega_{4}-\omega_{2},\omega_{2}%
,\omega_{3},\omega_{4}\right)\times \nonumber \\
& \hskip 3.5cm \times  \frac{a_{\alpha,2}a_{\beta,3}a_{\gamma,4}%
}{\sqrt{\omega_{2}\omega_{3}\omega_{4}}}\varphi\left(  \omega_{3}+\omega
_{4}-\omega_{2}\right)  d\omega_{2}d\omega_{3}d\omega_{4}\\
&  =\iiint\Phi\left(  \omega_{3}+\omega_{4}-\omega_{2},\omega_{2}%
,\omega_{3},\omega_{4}\right)  \frac{G_{N,2}G_{N,3}G_{N,4}}{\sqrt{\omega
_{2}\omega_{3}\omega_{4}}}\times \nonumber \\
& \hskip 3.5cm \times\varphi\left(  \omega_{3}+\omega_{4}-\omega
_{2}\right)  d\omega_{2}d\omega_{3}d\omega_{4}%
\end{align*}
where:%
\[
G_{N}=\sum_{\alpha\leq N}a_{\alpha}%
\]

We claim that $\lim_{N\rightarrow\infty}$ $\left(  \mathcal{O}\left[
g\right]  \right)  _{N}=\mathcal{O}\left[  g\right]  $ in the weak topology.
To prove this, we use that  $\Phi\leq\min\left\{
\sqrt{\omega_{2}},\sqrt{\omega_{3}},\sqrt{\omega_{4}}\right\}$,  and obtain
the estimate:
\begin{equation}
\left\vert \left\langle \left(  \mathcal{O}\left[  g\right]  \right)
_{N}-\mathcal{O}\left[  g\right]  ,\varphi\right\rangle \right\vert
\leq7\left\Vert \varphi\right\Vert _{\infty}\left(  \int_{\bigcup_{\eta\geq
N}\Omega_{\eta}}g\right)\left(  \int\frac{g\left(  \omega\right)
}{\sqrt{\omega}}d\omega\right)  ^{2},\ \label{S8E6}%
\end{equation}
where, in order to compute the difference $\left(  \mathcal{O}\left[  g\right]  \right)
_{N}-\mathcal{O}\left[  g\right]  $ we have  written  $g=G_{N}+H_{N}$ with
$H_{N}\left(  \mathbb{R}_{+}\setminus\bigcup_{\alpha=N}^{\infty}\Omega
_{\alpha}\right)  =0.$ The difference $g_{2}g_{3}g_{4}-G_{N,2}G_{N,3}G_{N,4}$
can be written in terms of sums of products of functions$\ G_{N}$ and $H_{N}$
containing at least one measure $H_{N}.$ Estimating $\Phi$ by one of the
square roots, and using the fact that $H_{N}\left(  \mathbb{R}_{+}%
\setminus\bigcup_{\alpha=N}^{\infty}\Omega_{\alpha}\right)  =0$ and
$H_{N}\left(  \bigcup_{\alpha=N}^{\infty}\Omega_{\alpha}\right)  =G_{N}\left(
\bigcup_{\alpha=N}^{\infty}\Omega_{\alpha}\right)  $ we obtain (\ref{S8E6}).
Using the definition of $\mathcal{X}_{\theta,\rho^{\ast}}$ we then obtain:%
\[
\left\vert \left\langle \left(  \mathcal{O}\left[  g\right]  \right)
_{N}-\mathcal{O}\left[  g\right]  ,\varphi\right\rangle \right\vert \leq
C\left\Vert \varphi\right\Vert _{\infty}\left\Vert g\right\Vert _{\theta
,\rho^{\ast}}^{3}\sum_{\eta\geq N}\left(  2^{-\eta}\right)  ^{\theta
}\rightarrow0\ \ \text{as\ \ }N\rightarrow\infty
\]

This gives the desired convergence $\left(  \mathcal{O}\left[  g\right]
\right)  _{N}\rightharpoonup\mathcal{O}\left[  g\right]  $ as $N\rightarrow
\infty.$ We then obtain the representation formula:%
\begin{align}
&  \left\langle \mathcal{O}\left[  g\right]  ,\varphi\right\rangle
  =\sum_{\alpha}\sum_{\beta}\sum_{\gamma}\iiint\Phi\left(  \omega
_{3}+\omega_{4}-\omega_{2},\omega_{2},\omega_{3},\omega_{4}\right) \times \label{S8E7}\\
& \hskip 4cm \times 
\frac{a_{\alpha,2}a_{\beta,3}a_{\gamma,4}}{\sqrt{\omega_{2}\omega_{3}%
\omega_{4}}}\varphi\left(  \omega_{3}+\omega_{4}-\omega_{2}\right)
d\omega_{2}d\omega_{3}d\omega_{4}\nonumber
\end{align}
for any $\varphi\in C_{0}\left(  \mathbb{R}_{+}\right)  .$ Our next goal is to prove:
\begin{equation}
\mathcal{O}\left[  g\right]  \left(  \mathbb{R}_{+}\setminus\bigcup_{\alpha
=0}^{\infty}\Omega_{\alpha}\right)  =0 \label{Q1}%
\end{equation}
Using again property (i) in Theorem (\ref{measure}) we have:
\begin{align}
\mathcal{O}\left[  g\right]  \left(  \mathbb{R}_{+}\setminus\bigcup_{\alpha
=0}^{\infty}\Omega_{\alpha}\right)    \leq & \,\,\mathcal{O}\left[  g\right]
\left(  \mathbb{R}_{+}\setminus\mathcal{Z}_{N}\right) \label{T1E1}\\
= \,\, &\mathcal{O}\left[  g\right]  \left(  \mathbb{R}_{+}\setminus\bigcup
_{x\in\mathcal{Z}_{N}}\left(  x-\frac{\varepsilon}{2^{N}},x+\frac{\varepsilon
}{2^{N}}\right)  \right) +\nonumber \\
&+\mathcal{O}\left[  g\right]  \left(  \bigcup
_{x\in\mathcal{Z}_{N}}\left[  \left(  x-\frac{\varepsilon}{2^{N}}%
,x+\frac{\varepsilon}{2^{N}}\right)  \setminus\left\{  x\right\}  \right]
\right) \nonumber
\end{align}

We can estimate the first term on the right hand side, using a nonnegative and
continuous test function $\varphi_{\varepsilon,N}$ which takes the value $1$
in the set $\mathbb{R}_{+}\setminus\left[x-\frac{\varepsilon}{2^{N}}%
,x+\frac{\varepsilon}{2^{N}}\right]  $ and vanishes in a neighbourhood of the
points $\left\{  x\in\mathcal{Z}_{N}\right\}  .$ Then:%
\begin{equation}
\mathcal{O}\left[  g\right]  \left(  \mathbb{R}_{+}\setminus\bigcup
_{x\in\mathcal{Z}_{N}}\left(  x-\frac{\varepsilon}{2^{N}},x+\frac{\varepsilon
}{2^{N}}\right)  \right)  \leq\left\langle \mathcal{O}\left[  g\right]
,\varphi_{\varepsilon,N}\right\rangle \label{T1E2}%
\end{equation}

We then use the representation formula (\ref{S8E7}) to compute the right-hand
side of (\ref{T1E2}). We split the triple sum as follows:%
\begin{equation}
\left\langle \mathcal{O}\left[  g\right]  ,\varphi_{\varepsilon,N}%
\right\rangle =\sum_{\substack{\alpha,\beta,\gamma \\ \max\left\{  \alpha,\beta
,\gamma\right\}  \leq N}}\iiint\left[  \cdot\cdot\cdot\right]
+\sum_{\substack{\alpha,\beta,\gamma \\ \max\left\{  \alpha,\beta,\gamma\right\}
>N}}\iiint\left[  \cdot\cdot\cdot\right]  \label{T1e}%
\end{equation}

The first term on the right vanishes, because, due to our choice of the
function $\varphi_{\varepsilon,N},$ this term contains only contributions of
points such that $\omega_{3}+\omega_{4}-\omega_{2}\in\bigcup_{\sigma
=N+1}^{\infty}\Omega_{\sigma},$ with $\omega_{2}\in\Omega_{\alpha}%
,\ \omega_{3}\in\Omega_{\beta},$ $\omega_{4}\in\Omega_{\gamma}.$ However, this
set is empty because $\max\left\{  \alpha,\beta,\gamma\right\}  \leq N.$
Indeed, if such a set of values of $\left(  \omega_{2},\omega_{3},\omega
_{4}\right)  $ exists we would have:%
\begin{equation}
2^{-\beta}\theta_{\beta}+2^{-\gamma}\theta_{\gamma}-2^{-\alpha}\theta_{\alpha
}=2^{-\sigma}\theta_{\sigma}\ \ ,\ \ \sigma\geq N+1\ \ ,\ \ \max\left\{
\alpha,\beta,\gamma\right\}  \leq N\ \label{P1}%
\end{equation}
where $\theta_{\alpha},\ \theta_{\beta},\ \theta_{\gamma},\ \theta_{\sigma}$
are positive integers and in addition $\theta_{\sigma}$ is an odd number.
However, (\ref{P1}) implies:%
\[
\theta_{\sigma}=2^{\sigma-\beta}\theta_{\beta}+2^{\sigma-\gamma}\theta
_{\gamma}-2^{\sigma-\alpha}\theta_{\alpha}%
\]
and since $\sigma\geq\max\left\{  \alpha,\beta,\gamma\right\}  +1$ this
implies that $\theta_{\sigma}$ is an odd number, that would be a
contradiction. Therefore:%
\begin{equation}
\sum_{\substack{\alpha,\beta,\gamma \\ \max\left\{  \alpha,\beta,\gamma\right\}  \leq
N}}\iiint\left[  \cdot\cdot\cdot\right]  =0 \label{I1}%
\end{equation}

On the other hand, the last term in (\ref{T1e}) can be estimated, using the
fact that each of the integrals contains at least one index $\alpha
,\beta,\gamma$ larger than $N.$ Then:%
\begin{align}
\sum_{\substack{\alpha,\beta,\gamma \\ \max\left\{  \alpha,\beta,\gamma\right\}
>N}}\iiint\left[  \cdot\cdot\cdot\right]   &  \leq3\left\Vert
\varphi_{\eta,\varepsilon}\right\Vert _{\infty}\left(  \int_{\bigcup
_{\alpha\geq\left(  N+1\right)  }\Omega_{\alpha}}g\right)  \left(  \int
\frac{g\left(  \omega\right)  }{\sqrt{\omega}}d\omega\right)  ^{2}\label{I2}\\
&  \leq C\left\Vert g\right\Vert _{\theta,\rho^{\ast}}^{3}\left(  2^{-\left(
N+1\right)  }\right)  ^{\theta}\nonumber
\end{align}
and this approaches to zero as $N\rightarrow\infty.$

We now estimate the last term in (\ref{T1E1}). To this end we use the
regularity properties of the measure $\mathcal{O}\left[  g\right]  .$ We first
compute $\mathcal{O}\left[  g\right]  \left(  \left\{  x\right\}  \right)
,\ x\in\mathcal{Z}_{N}$ by means of:%
\begin{equation}
\mathcal{O}\left[  g\right]  \left(  \left\{  x\right\}  \right)
=\lim_{\delta\rightarrow0}\mathcal{O}\left[  g\right]  \left(  \left(
x-\delta,x+\delta\right)  \right)  \label{T1E3}%
\end{equation}

Notice that $\mathcal{O}\left[  g\right]  \left(  \left(  x-\delta
,x+\delta\right)  \right)  $ can be estimated from below and above, using
nonnegative continuous test functions $\varphi_{1},\ \varphi_{2}$ satisfying
$\varphi_{1}\leq\varphi_{2},\ \varphi_{1}=1\ $in $\left[  x-\frac{\delta}%
{2},x+\frac{\delta}{2}\right]  ,\ \varphi_{1}=0\ $in $\mathbb{R}_{+}%
\setminus\left(  x-\delta,x+\delta\right)  ,\ \varphi_{2}=1\ $in\ $\left[
x-\delta,x+\delta\right]  ,\ \ \varphi_{2}=0\ $in $\mathbb{R}_{+}%
\setminus\left(  x-2\delta,x+2\delta\right)  .$ We then compute $\left\langle
\mathcal{O}\left[  g\right]  ,\varphi_{1}\right\rangle ,\ \left\langle
\mathcal{O}\left[  g\right]  ,\varphi_{2}\right\rangle $ using (\ref{S8E7}).
Splitting the contribution of the terms of the sum yielding $x$ and the rest
in the interval of integration we obtain, for any of these functions:%
\[
\left\vert \left\langle \mathcal{O}\left[  g\right]  ,\varphi_{k}\right\rangle
-M_{x}\right\vert \leq R\left(  \delta\right)  \ \ ,\ \ k=1,2
\]
where, using that $\varphi_{k}\left(  x\right)  =1\ ,\ k=1,2$
\[
M_{x}=\hskip -0.65cm \mathop{\sum\sum\sum}_{_{\substack {x_{\beta}+x_{\gamma}=x_{\alpha}+x \\ x_{\alpha}\in
\Omega_{\alpha},\ x_{\beta}\in\Omega_{\beta} ,\ x_{\gamma}\in\Omega_{\gamma} } }}\hskip -0.5 cm\Phi\left(  x_{\beta}+x_{\gamma}-x_{\alpha},x_{\alpha
},x_{\beta},x_{\gamma}\right)  \frac{a_{\alpha}\left(  \left\{  x_{\alpha
}\right\}  \right)  a_{\beta}\left(  \left\{  x_{\beta}\right\}  \right)
a_{\gamma}\left(  \left\{  x_{\gamma}\right\}  \right)  }{\sqrt{x_{\alpha
}x_{\beta}x_{\gamma}}}%
\]
and, using $\left\Vert \varphi_{k}\right\Vert _{\infty}\leq1:$%
\[
R\left(  \delta\right) =\hskip -0.65cm \mathop{\sum\sum\sum}_{ \substack{ 0<\left\vert x_{\beta}+x_{\gamma
}-x_{\alpha}-x\right\vert <\delta \\ x_{\alpha}\in\Omega_{\alpha},\ x_{\beta}%
\in\Omega_{\beta} ,\ x_{\gamma}\in\Omega_{\gamma}}}\hskip -0.5 cm
\Phi\left(  x_{\beta}+x_{\gamma}-x_{\alpha},x_{\alpha},x_{\beta},x_{\gamma
}\right)  \frac{a_{\alpha}\left(  \left\{  x_{\alpha}\right\}  \right)
a_{\beta}\left(  \left\{  x_{\beta}\right\}  \right)  a_{\gamma}\left(
\left\{  x_{\gamma}\right\}  \right)  }{\sqrt{x_{\alpha}x_{\beta}x_{\gamma}}}%
\]

Then, using arguments analogous to those yielding (\ref{I1}), (\ref{I2}) we
obtain:
\[
M_{x}\leq C\left\Vert g\right\Vert _{\theta,\rho^{\ast}}^{3}\left(  2^{-\eta
}\right)  ^{\theta}\ \ \text{if\ \ }x\in\Omega_{\eta}%
\]

We must estimate now  the remainder $R\left(  \delta\right)  .$ We claim
that $\lim_{\delta\rightarrow0}R\left(  \delta\right)  =0$ for each $x$
fixed$.$ Indeed, let us denote  as $\sigma=\max\left\{\alpha,\beta
,\gamma, N\right\}$. We have $0<\left\vert x_{\beta}+x_{\gamma}-x_{\alpha
}-x\right\vert <\delta$.  Then:%
\[
0<\left\vert 2^{\sigma}x_{\beta}+2^{\sigma}x_{\gamma}-2^{\sigma}x_{\alpha
}-2^{\sigma}x\right\vert <\delta2^{\sigma}%
\]
Notice that $2^{\sigma}x_{\beta}+2^{\sigma}x_{\gamma}-2^{\sigma}x_{\alpha
}-2^{\sigma}x$ is an integer. Then its absolute value is larger than one,
whence $1\leq\delta \, 2^{\sigma}.$ Since $x$ and therefore $N$ is fixed this
implies:%
\[
\max\left\{  \alpha,\beta,\gamma,N\right\}  \geq\frac{\log\left(  \frac
{1}{\delta}\right)  }{\log\left(  2\right)  }\rightarrow\infty
\ \ \text{as\ \ }\delta\rightarrow0
\]

We can then estimate $R\left(  \delta\right)  $ as:%
\[
R\left(  \delta\right)  \leq C\left\Vert g\right\Vert _{\theta,\rho^{\ast}%
}^{3}\left(  2^{-\sigma}\right)  ^{\theta}=C\left\Vert g\right\Vert
_{\theta,\rho^{\ast}}^{3}\left(  \delta\right)  ^{\theta}\rightarrow
0\ \ \text{as\ \ }\delta\rightarrow0
\]
whence, using (\ref{T1E3}):\
\begin{equation}
\mathcal{O}\left[  g\right]  \left(  \left\{  x\right\}  \right)  =M_{x}
\label{T1E4}%
\end{equation}

We can now estimate the last term in (\ref{T1E1}). A similar argument shows
that, for $\varepsilon$ small:%
\begin{equation}
\mathcal{O}\left[  g\right]  \left(  \left(  x-\frac{\varepsilon}{2^{N}%
},x+\frac{\varepsilon}{2^{N}}\right)  \right)  \leq M_{x}+C\left\Vert
g\right\Vert _{\theta,\rho^{\ast}}^{3}\left(  \frac{\varepsilon}{2^{N}%
}\right)  ^{\theta} \label{T1E5}%
\end{equation}
Using (\ref{T1E4}) as well as (\ref{T1E5}) we obtain:%
\[
\mathcal{O}\left[  g\right]  \left(  \left(  x-\frac{\varepsilon}{2^{N}%
},x+\frac{\varepsilon}{2^{N}}\right)  \setminus\left\{  x\right\}  \right)
\leq C\left\Vert g\right\Vert _{\theta,\rho^{\ast}}^{3}\left(  \frac
{\varepsilon}{2^{N}}\right)  ^{\theta}%
\]
whence:%
\[
\mathcal{O}\left[  g\right]  \left(  \bigcup_{x\in\mathcal{Z}_{N}}\left(
x-\frac{\varepsilon}{2^{N}},x+\frac{\varepsilon}{2^{N}}\right)  \cap\left\{
\omega\leq R_{0}\right\}  \setminus\left\{  x\right\}  \right)  \leq
CR_{0}\left\Vert g\right\Vert _{\theta,\rho^{\ast}}^{3}\left(  \frac
{\varepsilon}{2^{N}}\right)  ^{\theta}\cdot2^{N}%
\]
where we use the fact that the number of points of $\mathcal{Z}_{N}%
\cap\left\{  \omega\leq R_{0}\right\}  $ can be estimated as $CR_{0}2^{N}.$
Since $\theta>1$ it then follows that this measure converges to zero as
$N\rightarrow\infty.$

On the other hand%

\[
\mathcal{O}\left[  g\right]  \left(  \bigcup_{x\in\mathcal{Z}_{N}}\left(
x-\frac{\varepsilon}{2^{N}},x+\frac{\varepsilon}{2^{N}}\right)  \cap\left\{
\omega>R_{0}\right\}  \setminus\left\{  x\right\}  \right)
\]
can be estimated as $C\left\Vert g\right\Vert _{\theta,\rho^{\ast}}^{3}%
R_{0}^{1-\rho^{\ast}}$. This term can be made small choosing $R_{0}$ large.
Therefore, all the terms on the right-hand side of (\ref{T1E1}) can be made
arbitrarily small, whence (\ref{Q1}) follows.

To conclude the proof of the Lemma it only remains to obtain (\ref{S8E3a}). We
first estimate in the formula for $\left\Vert \left(  \mathcal{O}\left[
g\right]  \right)  \right\Vert _{\theta,\rho^{\ast}}$ the contributions from
the regions\ where $\omega_{1}\geq\frac{1}{2}.$ To this end, let $R\geq1,$ and
define a continuous test function $\varphi=\varphi\left(  \omega_{1}\right)  $
such that $\varphi\left(  \omega_{1}\right)  =1$ if $\omega_{1}\geq\frac{R}%
{2},$ $\varphi\left(  \omega_{1}\right)  =0$ if $\omega_{1}\leq\frac{R}{4},$
$0\leq\varphi\left(  \omega_{1}\right)  \leq1$ if $\omega_{1}\in\mathbb{R}%
_{+}.$ Then:%
\begin{eqnarray*}
&&\mathcal{O}\left[  g\right]  \left(  \left[  \frac{R}{2},R\right]  \right)
\leq\iiint\Phi\left(  \omega_{3}+\omega_{4}-\omega_{2},\omega_{2}%
,\omega_{3},\omega_{4}\right) \times\\
&& \hskip 4cm \times\frac{g_{2}g_{3}g_{4}}{\sqrt{\omega_{2}%
\omega_{3}\omega_{4}}}\varphi\left(  \omega_{3}+\omega_{4}-\omega_{2}\right)
d\omega_{2}d\omega_{3}d\omega_{4}%
\end{eqnarray*}

Using the symmetry $\omega_{3}\leftrightarrow\omega_{4},$ as well as the fact
that $\Phi\leq\sqrt{\omega_{2}}$ we obtain:%
\[
\mathcal{O}\left[  g\right]  \left(  \left[  \frac{R}{2},R\right]  \right)
\leq2\iiint_{\left\{  \omega_{4}\geq\omega_{3}\right\}  }\frac
{g_{2}g_{3}g_{4}}{\sqrt{\omega_{3}\omega_{4}}}\varphi\left(  \omega_{3}%
+\omega_{4}-\omega_{2}\right)  d\omega_{2}d\omega_{3}d\omega_{4}%
\]

Since the function $\varphi\left(  \omega_{3}+\omega_{4}-\omega_{2}\right)  $
vanishes for $\omega_{3}+\omega_{4}-\omega_{2}\leq\frac{R}{4}$, it follows
that the set where the integrand does not vanishes is included in the set
where $\omega_{3}+\omega_{4}\geq\frac{R}{4},$ and since $\omega_{3}\leq
\omega_{4},$ we can then obtain an upper bound for the integral restricting
the integration to the set $\left\{  \omega_{4}\geq\frac{R}{8}\right\}  .$
Since $\varphi\leq1$ it then follows that:%
\[
\mathcal{O}\left[  g\right]  \left(  \left[  \frac{R}{2},R\right]  \right)
\leq2\iiint_{\left\{  \omega_{4}\geq\frac{R}{8}\right\}  }\frac
{g_{2}g_{3}g_{4}}{\sqrt{\omega_{3}\omega_{4}}}d\omega_{2}d\omega_{3}%
d\omega_{4}\leq C\left\Vert g\right\Vert _{\theta,\rho^{\ast}}^{2}%
\int_{\left[  \frac{R}{8},\infty\right)  }\frac{g\left(  \omega\right)
}{\sqrt{\omega}}d\omega
\]

Using then the definition of $\left\Vert g\right\Vert _{\theta,\rho^{\ast}}$
we obtain:\
\begin{equation}
\mathcal{O}\left[  g\right]  \left(  \left[  \frac{R}{2},R\right]  \right)
\leq C\left\Vert g\right\Vert _{\theta,\rho^{\ast}}^{3}R^{1-\rho^{\ast}%
}\ \ ,\ \ R\geq1 \label{S8E9}%
\end{equation}

We now derive estimates for the measures $\mathcal{O}\left[  g\right]  \left(
\Omega_{\alpha}\right)  .$ To this end we use the representation formula
(\ref{S8E7}). We consider a family of functions $\psi_{\varepsilon}\in
C_{0}\left(  \mathbb{R}\right)  $ satisfying $\psi_{\varepsilon}\left(
0\right)  =1,\ \psi_{\varepsilon}\left(  s\right)  =0$ if $\left\vert
s\right\vert \geq\varepsilon,\ 0\leq\psi_{\varepsilon}\leq1.$ We then consider
a sequence of test functions $\varphi_{\eta,\varepsilon}\left(  \omega\right)
=\sum_{\ell=1}^{\infty}\psi_{\varepsilon}\left(  \omega-x_{\eta}\left(
\ell\right)  \right)  .$ Notice that these test functions are not compactly
supported, but they are $C_{b}\left(  \mathbb{R}_{+}\right)  .$ Therefore, it
is possible to define $\left\langle \mathcal{O}\left[  g\right]
,\varphi_{\eta,\varepsilon}\right\rangle .$ Our assumptions on $\psi
_{\varepsilon}$ as well as (\ref{S8E7}) imply:%
\begin{equation}
\left\langle \mathcal{O}\left[  g\right]  ,\Omega_{\eta}\right\rangle
\leq\left\langle \mathcal{O}\left[  g\right]  ,\varphi_{\eta,\varepsilon
}\right\rangle \label{S9E3}%
\end{equation}

We compute $\left\langle \mathcal{O}\left[  g\right]  ,\varphi_{\eta
,\varepsilon}\right\rangle $ using (\ref{S8E7}). We split the sum as:%
\begin{equation}
\left\langle \mathcal{O}\left[  g\right]  ,\varphi_{\eta,\varepsilon
}\right\rangle =\hskip -0.65 cm\sum_{\substack{\alpha,\beta,\gamma \\ \max\left\{  \alpha,\beta
,\gamma\right\}  <\eta}}\hskip -0.4 cm\iiint\left[  \cdot\cdot\cdot\right]
+\hskip -0.65 cm\sum_{\substack{\alpha,\beta,\gamma \\ \max\left\{  \alpha,\beta
,\gamma\right\} \ge \eta}}\hskip -0.4 cm\iiint\left[  \cdot\cdot\cdot\right]  \label{S9E1}%
\end{equation}

We now claim that the first term on the right of (\ref{S9E1}) is identically
zero if $\varepsilon$ is sufficiently small. Indeed, the integrations in that
term are restricted to those in:%
\[
\bigcup_{\ell=1}^{\infty}\left\{  \left\vert \omega_{3}+\omega_{4}-\omega
_{2}-x_{\eta}\left(  \ell\right)  \right\vert \leq\varepsilon:\left(
\omega_{2},\omega_{3},\omega_{4}\right)  \in\Omega_{\alpha}\times\Omega
_{\beta}\times\Omega_{\gamma}\right\}
\]

The elements of this set satisfy:%
\[
\left\vert 2^{-\beta}\theta_{\beta}\left(  j\right)  +2^{-\gamma}%
\theta_{\gamma}\left(  k\right)  -2^{-\alpha}\theta_{\alpha}\left(  m\right)
-2^{-\eta}\theta_{\eta}\left(  \ell\right)  \right\vert \leq\varepsilon
\]
or equivalently:%
\[
\left\vert 2^{\eta-\beta}\theta_{\beta}\left(  j\right)  +2^{\eta-\gamma
}\theta_{\gamma}\left(  k\right)  -2^{\eta-\alpha}\theta_{\alpha}\left(
m\right)  -\theta_{\eta}\left(  \ell\right)  \right\vert \leq\varepsilon
2^{\eta}%
\]
where $\theta_{\eta}\left(  \ell\right)  $ is a positive integer. However,
since $\max\left\{  \alpha,\beta,\gamma\right\}  <\eta$ this set is empty, if
$\varepsilon$ is sufficiently small, whence:%
\begin{equation}
\sum_{\substack{\alpha,\beta,\gamma \\ \max\left\{  \alpha,\beta,\gamma\right\}  <\eta}
}\hskip -0.3cm \iiint\left[  \cdot\cdot\cdot\right]  =0 \label{S9E2}%
\end{equation}

In order to estimate the last term in (\ref{S9E1}) we use the fact that at
least one of the indexes $\alpha,\beta,\gamma$ is larger than $\eta.$ Suppose
without loss of generality that such index is $\alpha.$ We then estimate
$\Phi$ by $\sqrt{\omega_{2}}$ to arrive at the estimate:%
\[
\sum_{\substack{\alpha,\beta,\gamma \\ \max\left\{\alpha,\beta,\gamma\right\}
\geq\eta}}\hskip -0.3cm \iiint\left[  \cdot\cdot\cdot\right]  \leq3\left\Vert
\varphi_{\eta,\varepsilon}\right\Vert _{\infty}\left(  \int_{\bigcup
_{\alpha\geq\eta}\Omega_{\alpha}}g\right)  \left(  \int\frac{g\left(
\omega\right)  }{\sqrt{\omega}}d\omega\right)  ^{2}%
\]
whence, since $\left\Vert \varphi_{\eta,\varepsilon}\right\Vert _{\infty}=1:$
\[
\left\langle \mathcal{O}\left[  g\right]  ,\varphi_{\eta,\varepsilon
}\right\rangle \leq C\left\Vert g\right\Vert _{\theta,\rho^{\ast}}^{3}%
\sum_{\alpha\geq\eta}\left(  2^{-\alpha}\right)  ^{\theta}\leq C\left\Vert
g\right\Vert _{\theta,\rho^{\ast}}^{3}\left(  2^{-\eta}\right)  ^{\theta}%
\]
and using (\ref{S9E3}) we obtain:%
\begin{equation}
\left\langle \mathcal{O}\left[  g\right]  ,\Omega_{\eta}\right\rangle \leq
C\left\Vert g\right\Vert _{\theta,\rho^{\ast}}^{3}\left(  2^{-\eta}\right)
^{\theta} \label{S9E4}%
\end{equation}

Combining (\ref{S8E9}) and (\ref{S9E4}) we obtain (\ref{S8E3a}).
\end{proof}

\begin{remark}
Notice that Lemma \ref{LOp} implies that the operator $\mathcal{O}\left[
\cdot\right]  $ transforms measures $g\in\mathcal{X}_{\theta,\rho^{\ast}}$
into measures in $\mathcal{X}_{\theta,\rho^{\ast}}.$ This will allow to obtain
mild solutions of (\ref{S2E1}) with values $g\left(  t,\cdot\right)
\in\mathcal{X}_{\theta,\rho^{\ast}}$ for $t\geq0.$
\end{remark}

\begin{remark}
It is interesting to remark that Lemma \ref{LOp} implies also that, assuming
that $g$ is given by (\ref{T2E6a}), it is possible to give the following
representation formula for $\mathcal{O}\left[  g\right]  :$
\begin{eqnarray}
&&\mathcal{O}\left[  g\right]  =\sum_{\gamma}\sum_{k=1}^{\infty}\left[
\sum_{\alpha,\beta,\eta}\sum_{\ell,j,m=1}^{\infty}\frac{a_{\eta}\left(
\ell\right)  a_{\alpha}\left(  j\right)  a_{\beta}\left(  m\right)  }%
{\sqrt{x_{\eta}\left(  \ell\right)  x_{\alpha}\left(  j\right)  x_{\beta
}\left(  m\right)  }} \times \right. \label{Z4E3}\\
&& \left.\times \Phi\left(  x_{\gamma}\left(  k\right)  ,x_{\eta}\left(
\ell\right)  ,x_{\alpha}\left(  j\right)  ,x_{\beta}\left(  m\right)  \right)
\delta_{\left(  x_{\gamma}\left(  k\right)  +x_{\eta}\left(  \ell\right)
-x_{\alpha}\left(  j\right)  -x_{\beta}\left(  m\right)  \right)  ,0}\right]
\delta_{x_{\gamma}\left(  k\right)  } \nonumber
\end{eqnarray}
\end{remark}

The following definition is similar to Definition \ref{MeasMildSol} with the
only difference that we restrict the measures $g\left(  t,\cdot\right)  $ to
be in $\mathcal{X}_{\theta,\rho^{\ast}}.$ In addition to the discrete
character of the measures, we assume also more stringent decay conditions at
infinity because we are interested in solutions with finite mass.

\begin{definition}
\label{MiSol}Given $\theta>1,\ \rho^{\ast}>1,\ T\in\left(  0,\infty\right]  $
and $g_{in}\in\mathcal{X}_{\theta,\rho^{\ast}}$ we will say that $g\in
C\left(  \left[  0,T\right]  :\mathcal{X}_{\theta,\rho^{\ast}}\right)  $ is a
mild solution of (\ref{S2E1}) with values in $\mathcal{X}_{\theta,\rho^{\ast}%
}$ and with initial value $g\left(  \cdot,0\right)  =g_{in}$ if the following
identity holds in the sense of measures:%
\begin{eqnarray}
g\left(  \omega_{1},t\right)  &=& g_{in}\left(  \omega_{1}\right)  \exp\left(
-\int_{0}^{t}A_{g}\left(  \omega_{1},s\right)  ds\right)+ \label{S7E2}\\ 
&& +\int_{0}^{t}%
\exp\left(  -\int_{s}^{t}A_{g}\left(  \omega_{1},\xi\right)  d\xi\right)
\mathcal{O}\left[  g\right]  \left(  \cdot,s\right)  ds\nonumber%
\end{eqnarray}
for $0\leq t<T,$ where $A_{g}\left(  \cdot,s\right)  $ is defined as in Lemma
\ref{LA} for each $g\left(  \cdot,s\right)  $ and $\mathcal{O}\left[
g\right]  \left(  \cdot,s\right)  $ is defined as in Lemma \ref{LOp} for each
$g\left(  \cdot,s\right)  .$
\end{definition}

\subsubsection{Proof of a local existence Theorem of measured mild solutions with
values in the space $\mathcal{X}_{\theta,\rho^{\ast}}.$}

As a first step we need to construct local measured valued solutions in
$\mathcal{X}_{\theta,\rho^{\ast}}$ in the sense of Definition \ref{MiSol}.

\begin{theorem}
\label{mildCount}Let $\theta>1,\ \rho^{\ast}>1$ and $g_{0}\in\mathcal{X}%
_{\theta,\rho^{\ast}}$ there exists $T>0,$ and at least one mild solution of
(\ref{S2E1}) in $C\left(  \left[  0,T\right]  :\mathcal{X}_{\theta,\rho^{\ast
}}\right)  $ with initial value $g\left(  \cdot,0\right)  =g_{in}$ in the
sense of the Definition \ref{MiSol}. Moreover, the following identities hold:%
\begin{equation}
\int g\left(  t,d\omega\right)  =\int g_{in}\left(  d\omega\right)
\ \ \text{for any }t\in\left[  0,T\right]  \label{Z3E6}%
\end{equation}%
\begin{equation}
\int_{\left\{  0\right\}  }g\left(  t,d\omega\right)  =0\ \ \text{for any
}t\in\left[  0,T\right]  \label{Z3E7}%
\end{equation}

\end{theorem}

\begin{proof}
We define a space of measures as:%
\[
Y\left(  g_{in}\right)  =\left\{  g\in C\left(  \left[  0,T\right]
:\mathcal{X}_{\theta,\rho^{\ast}}\right)  :\sup_{0\leq t\leq T}\left\Vert
g\right\Vert _{\theta,\rho^{\ast}}\leq2\left\Vert g_{in}\right\Vert
_{\theta,\rho^{\ast}}\right\}
\]
and define an operator $\mathcal{T}:Y\left(  g_{in}\right)  \rightarrow
Y\left(  g_{in}\right)  $ as the right-hand side of (\ref{S7E2}), or more
precisely:%
\begin{align*}
\mathcal{T}\left[  g\right]  \left(  t,\omega_{1}\right)   &  =g_{in}\left(
\omega_{1}\right)  \exp\left(  -\int_{0}^{t}A_{g}\left(  s,\omega_{1}\right)
ds\right)  +\\
&  +\int_{0}^{t}\exp\left(  -\int_{s}^{t}A_{g}\left(  \xi,\omega_{1}\right)
d\xi\right)  \left(  \iint\Phi\frac{g_{2}g_{3}g_{4}}{\sqrt{\omega_{2}%
\omega_{3}\omega_{4}}}d\omega_{3}d\omega_{4}\right)  ds\\
&  \equiv\mathcal{T}_{1}\left[  g\right]  \left(  t,\omega_{1}\right)
+\mathcal{T}_{2}\left[  g\right]  \left(  t,\omega_{1}\right)
\end{align*}

Notice that the operator $\mathcal{T}\left[  g\right]  $ is well defined due
to Lemmas \ref{LA}, \ref{LOp}.

We now prove that the operator $\mathcal{T}$ brings $Y\left(  g_{in}\right)  $
to itself if $T$ is sufficiently small. To check this, we integrate
$\mathcal{T}\left[  g\right]  $ in the interval $\left[  \frac{R}{2},R\right]
.$ Then:%
\[
\int_{\left[  \frac{R}{2},R\right]  }\mathcal{T}\left[  g\right]  \left(
t,d\omega\right)  =\int_{\left[  \frac{R}{2},R\right]  }\mathcal{T}_{1}\left[
g\right]  \left(  t,d\omega\right)  +\int_{\left[  \frac{R}{2},R\right]
}\mathcal{T}_{2}\left[  g\right]  \left(  t,d\omega\right)
\]
where, using (\ref{S7E5}) and the definition of $\left\Vert g\right\Vert
_{\theta,\rho^{\ast}}$
\begin{equation}
\int_{\left[  \frac{R}{2},R\right]  }\mathcal{T}_{1}\left[  g\right]  \left(
t,d\omega\right)  \leq\int_{\left[  \frac{R}{2},R\right]  }g_{in}\left(
d\omega\right)  \leq\left\Vert g_{in}\right\Vert _{\theta,\rho^{\ast}}%
R\min\left\{  R^{\theta},R^{-\rho^{\ast}}\right\}  \label{S7E6}%
\end{equation}

We have also, using again\ (\ref{S7E5}), as well as the symmetry of the
integral with respect to the symmetrization $\omega_{3}\leftrightarrow
\omega_{4}:$
\[
\int_{\left[  \frac{R}{2},R\right]  }\mathcal{T}_{2}\left[  g\right]  \left(
t,d\omega\right)  \leq2\int_{0}^{t}\int_{\left[  \frac{R}{2},R\right]
}\left(  \iint_{\left\{  \omega_{3}\leq\omega_{4}\right\}  }\Phi\frac
{g_{2}g_{3}g_{4}}{\sqrt{\omega_{2}\omega_{3}\omega_{4}}}d\omega_{3}d\omega
_{4}\right)  d\omega_{1}ds
\]

We estimate $\Phi$ by $\sqrt{\omega_{1}}.$ Then:%
\begin{equation}
\int_{\left[  \frac{R}{2},R\right]  }\mathcal{T}_{2}\left[  g\right]  \left(
t,d\omega\right)  \leq2\sqrt{R}\int_{0}^{t}\int_{\left[  \frac{R}{2},R\right]
}\left(  \iint_{\left\{  \omega_{3}\leq\omega_{4}\right\}  }\frac
{g_{2}g_{3}g_{4}}{\sqrt{\omega_{2}\omega_{3}\omega_{4}}}d\omega_{3}d\omega
_{4}\right)  d\omega_{1}ds \label{S7E7a}%
\end{equation}

We now distinguish two cases. Suppose that $R\geq1.$ We then use that in the
region of integration we have $\omega_{4}\geq\frac{R}{4}$. Replacing the
integration in $\omega_{1}$ by the integration in $\omega_{2}$ by means of a
change of variables, we obtain the estimate:%
\[
\int_{\left[  \frac{R}{2},R\right]  }\mathcal{T}_{2}\left[  g\right]  \left(
t,d\omega\right)  \leq4\int_{0}^{t}\left(  \int_{\mathbb{R}_{+}}\frac{g\left(
s,d\omega\right)  }{\sqrt{\omega}}\right)  ^{2}\int_{\frac{R}{4}}^{\infty
}g\left(  s,d\omega\right)  ds
\]

Notice that, since $\rho^{\ast}>1$ we have $\int_{\mathbb{R}_{+}}%
\frac{g\left(  s,d\omega\right)  }{\sqrt{\omega}}\leq C\left\Vert g\left(
s,\cdot\right)  \right\Vert _{\theta,\rho^{\ast}},$ as it can be seen
decomposing the region of integration in dyadic intervals. On the other hand,
a similar argument yields $\int_{\frac{R}{4}}^{\infty}g\left(  s,d\omega
\right)  ds\leq CR^{1-\rho^{\ast}}\left\Vert g\left(  s,\cdot\right)
\right\Vert _{\theta,\rho^{\ast}}$ if $R\geq1.$ Then:%
\begin{equation}
\int_{\left[  \frac{R}{2},R\right]  }\mathcal{T}_{2}\left[  g\right]  \left(
t,d\omega\right)  \leq CR^{1-\rho^{\ast}}\int_{0}^{t}\left\Vert g\left(
s,\cdot\right)  \right\Vert _{\theta,\rho^{\ast}}^{3}ds\text{ \ if \ }R\geq1
\label{S7E7}%
\end{equation}

Suppose now that $R\leq1.$ Then (\ref{S7E7a}) implies:%
\begin{eqnarray}
\int_{\left[  \frac{R}{2},R\right]  }\mathcal{T}_{2}\left[  g\right]  \left(
t,d\omega\right)  &\leq & 2\sqrt{R}\int_{0}^{t}\left(  \int_{\mathbb{R}_{+}}%
\frac{g\left(  s,d\omega\right)  }{\sqrt{\omega}}\right)  ^{3}ds  \label{S7E8}\\
&\leq &C\sqrt
{R}\int_{0}^{t}\left\Vert g\left(  s,\cdot\right)  \right\Vert _{\theta
,\rho^{\ast}}^{3}ds\ \text{\ if \ }R\leq1 \nonumber
\end{eqnarray}

Combining (\ref{S7E6}), (\ref{S7E7}), (\ref{S7E8}) and using that
$\theta=\frac{1}{2}$ we obtain:%
\[
\sup_{R>0}\frac{1}{\min\left\{  R^{\theta},R^{-\rho^{\ast}}\right\}  }\frac
{1}{R}\int_{\left[  \frac{R}{2},R\right]  }\mathcal{T}\left[  g\right]
\left(  t,d\omega\right)  \leq\left\Vert g_{in}\right\Vert _{\theta,\rho
^{\ast}}+CT\sup_{0\leq t\leq T}\left\Vert g\left(  t,\cdot\right)  \right\Vert
_{\theta,\rho^{\ast}}^{3}%
\]

Then the operator $\mathcal{T}\left[  \cdot\right]  $ maps $Y\left(
g_{in}\right)  $ into itself. Moreover, arguing as in the Proof of Lemma
\ref{globRegul} we obtain that the operator $\mathcal{T}\left[  \cdot\right]
$ defines a continuous mapping from $Y\left(  g_{in}\right)  $ to $C\left(
\left[  0,T\right]  :\mathcal{X}_{\theta,\rho^{\ast}}\right)  $ in the weak
topology. Notice that in this case $\sigma=0,$ and therefore some of the
functions appearing in the integrals defining $A_{g}$ and $\mathcal{O}\left[
g\right]  $ are singular near $\omega=0.$ However, the contribution to those
integrals of the regions close to the origin can be made estimated if $g\in
C\left(  \left[  0,T\right]  :\mathcal{X}_{\theta,\rho^{\ast}}\right)  $ using
the fact that $\left\Vert g\left(  t,\cdot\right)  \right\Vert _{\theta
,\rho^{\ast}}$ is bounded. Therefore, it is possible to adapt the argument in
the Proof of Lemma \ref{globRegul} to prove the desired continuity of the
operator $\mathcal{T}\left[  \cdot\right]  .$ Moreover, since the set
$\mathcal{X}_{\theta,\rho^{\ast}}$ is closed in $\mathcal{M}_{+}\left(
\left[  0,\infty\right)  \right)  $ it follows that it is compact in the weak
topology. Therefore, applying also Arzela-Ascoli as in the Proof of Lemma
\ref{globRegul}. The existence of solutions then follows using Schauder's
Theorem. To prove the identity (\ref{Z3E6}), we can argue as in the Proof of
Proposition \ref{relSolutions}, in order to show that $g$ is also a weak
solution of (\ref{S2E1}) in the sense of Definition \ref{weakSolution}. This
follows from Proposition \ref{relSolutions} due to the fact that mild measured
values solutions with values in $\mathcal{X}_{\theta,\rho^{\ast}}$ in the
sense of Definition \ref{MiSol} are also mild measured valued solutions in the
sense of Definition \ref{MeasMildSol}. Taking a sequence of test functions converging to $1$ in $\omega\geq0$ we
obtain (\ref{Z3E6}). Finally, we notice that (\ref{Z3E7}) follows by
construction, since $0\in\left[  \mathbb{R}_{+}\setminus\bigcup_{\alpha
=0}^{\infty}\Omega_{\alpha}\right]  ,$ whence the result follows.
\end{proof}

Actually, it turns out that the solutions can be extended as long as
$\left\Vert g\right\Vert _{\theta,\rho^{\ast}}$ remains bounded.
\begin{theorem}
\label{Prolong}Let $g\in C\left(  \left[  0,T\right]  :\mathcal{X}%
_{\theta,\rho^{\ast}}\right)  $ the mild solution of (\ref{S2E1}) obtained in
Theorem \ref{mildCount}. Suppose that $\sup_{0\leq t\leq T}\left\Vert g\left(
t,\cdot\right)  \right\Vert _{\theta,\rho^{\ast}}<\infty.$ Then, there exists
$\delta>0$ and $\tilde{g}\in C\left(  \left[  0,T+\delta\right]
:\mathcal{X}_{\theta,\rho^{\ast}}\right)  $ such that $g\left(  t,\cdot
\right)  =\tilde{g}\left(  t,\cdot\right)  $ for $t\in\left[  0,T\right]  $
and $\tilde{g}$ is a mild solution of (\ref{S2E1}) in the interval
$t\in\left[  0,T+\delta\right]  .$
\end{theorem}

\begin{proof}
We just construct a mild solution in the time interval $\left[  T,T+\delta
\right]  $ with initial datum $\tilde{g}\left(  T,\cdot\right)  \in
\mathcal{X}_{\theta,\rho^{\ast}}.$ Such solution is well defined for
$\delta>0$ as it can be seen using the argument in the Proof of Theorem
\ref{mildCount}. The function $\tilde{g}$ obtained combining the values of $g$
in $t\in\left[  0,T\right]  $ and $\tilde{g}$ for $t\in\left[  T,T+\delta
\right]  $ gives the provides mild solution of (\ref{S2E1}) as it can be seen
using Definition \ref{MiSol}.
\end{proof}

\subsubsection{Global existence of measure mild solutions with values in
$\mathcal{X}_{\theta,\rho^{\ast}}.$}

We will now prove that, if the masses $m_{\alpha}=\sum_{k=1}^{\infty}%
a_{\alpha}\left(  k\right)  $ contained in each of the families $\Omega
_{\alpha}$ decrease fast enough as $\alpha\rightarrow\infty,$ the mild
solutions obtained in Theorem \ref{mildCount} are globally defined in time. To
this end, we first need to prove the following result which will has a
consequence that the mass cannot propagate from the families $\left\{
\Omega_{\beta}\right\}  _{\beta>\alpha}$ to the family $\Omega_{\alpha}$
unless some meaningful amount of mass is already present in this last family.

In order to prove a global well posedness results in the space
$C\left(  \left[  0,\infty\right)  :\mathcal{X}_{\theta,\rho^{\ast}}\right)  $
we need the following auxiliary Lemmas.

\begin{lemma}
\label{ODEUpper}Let $g\in C\left(  \left[  0,T\right]  :\mathcal{X}%
_{\theta,\rho^{\ast}}\right)  $ be the mild solution of (\ref{S2E1}) obtained in
Theorem \ref{mildCount}. Let us write:%
\begin{equation}
m_{\gamma}=\sum_{k=1}^{\infty}a_{\gamma}\left(  k\right)
\ \ \ ,\ \ \ M_{\gamma+1}=\sum_{\eta\geq\gamma+1}m_{\eta}\ \ ,\ \ S_{\gamma
+1}=\sum_{\alpha>\gamma}\frac{m_{\alpha}}{\sqrt{x_{\alpha}\left(  1\right)  }}
\label{Z4E1}%
\end{equation}

Then, the following inequality holds $a.e.$ $t\in\left[  0,T\right]  :$%
\begin{equation}
\partial_{t}m_{\gamma}\leq\frac{6m_{\gamma}}{x_{\gamma}\left(  1\right)
}+6M_{\gamma+1}\left(  S_{\gamma+1}+\frac{1}{\sqrt{x_{\gamma}\left(  1\right)
}}\right)  ^{2} \label{Z4E2}%
\end{equation}

\end{lemma}

\begin{proof}
Notice that the Definition of mild solution (cf. Definition \ref{MiSol})
implies the following identity, in the sense of measures:%
\[
\partial_{t}g\left(  t,\cdot\right)  =\mathcal{O}\left[  g\right]  \left(
\cdot,t\right)  -A_{g}\left(  t,\cdot\right)  g\left(  t,\cdot\right)
\ \ ,\ \ a.e.\ t\in\left[  0,T\right]
\]

Due to Lemma \ref{LA} we have $A_{g}\left(  \omega_{1},\cdot\right)  \geq0.$
Then:%
\begin{equation}
\partial_{t}g\left(  t,\cdot\right)  \leq\mathcal{O}\left[  g\right]  \left(
t,\cdot\right)  \label{Z4E4}%
\end{equation}

Due to the definition of $\mathcal{X}_{\theta,\rho^{\ast}}$ implies that for
each $t\geq0$ the measure $g$ has the form (\ref{T2E6a}), (\ref{I1c}). Then
$\mathcal{O}\left[  g\right]  \left(  \cdot,t\right)  $ is given by
(\ref{Z4E3}). We then use that:%
\begin{eqnarray*}
\Phi\left(  x_{\gamma}\left(  k\right)  ,x_{\eta}\left(  \ell\right)
,x_{\alpha}\left(  j\right)  ,x_{\beta}\left(  m\right)  \right) & \leq &
\min\left\{  x_{\eta}\left(  \ell\right)  ,x_{\alpha}\left(  j\right)
,x_{\beta}\left(  m\right)  \right\}\\
& = &\Phi\left(  x_{\eta}\left(
\ell\right)  ,x_{\alpha}\left(  j\right)  ,x_{\beta}\left(  m\right)  \right)
\end{eqnarray*}

Combining then (\ref{Z4E4}), (\ref{Z4E3}) we obtain:%
\begin{eqnarray*}
&&\partial_{t}g\left(  \cdot,t\right)  \leq  \sum_{\gamma}\sum_{k=1}^{\infty}
\left[  \sum_{\alpha,\beta,\eta}\sum_{\ell,j,m=1}^{\infty}\frac{a_{\eta
}\left(  \ell\right)  a_{\alpha}\left(  j\right)  a_{\beta}\left(  m\right)
}{\sqrt{x_{\eta}\left(  \ell\right)  x_{\alpha}\left(  j\right)  x_{\beta
}\left(  m\right)  }}\times \right.\\
&& \left. \hskip 2cm \times \Phi\left(  x_{\eta}\left(  \ell\right)  ,x_{\alpha
}\left(  j\right)  ,x_{\beta}\left(  m\right)  \right)  \delta_{\left(
x_{\gamma}\left(  k\right)  +x_{\eta}\left(  \ell\right)  -x_{\alpha}\left(
j\right)  -x_{\beta}\left(  m\right)  \right)  ,0}\right]  \delta_{x_{\gamma
}\left(  k\right)  }%
\end{eqnarray*}

Adding the contributions associated to the family $\gamma$ and using the
definition of $m_{\gamma}$ in (\ref{Z4E1}) we obtain:%
\begin{eqnarray*}
&&\partial_{t}m_{\gamma}\leq\sum_{\alpha,\beta,\eta}\sum_{\ell,j,s=1}^{\infty
}\frac{a_{\eta}\left(  \ell\right)  a_{\alpha}\left(  j\right)  a_{\beta
}\left(  s\right)  }{\sqrt{x_{\eta}\left(  \ell\right)  x_{\alpha}\left(
j\right)  x_{\beta}\left(  s\right)  }}\Phi\left(  x_{\eta}\left(
\ell\right)  ,x_{\alpha}\left(  j\right)  ,x_{\beta}\left(  s\right)  \right)\times \\
&&\hskip 6cm \times \sum_{k=1}^{\infty}\delta_{\left(  x_{\gamma}\left(  k\right)  +x_{\eta
}\left(  \ell\right)  -x_{\alpha}\left(  j\right)  -x_{\beta}\left(  s\right)
\right)  ,0}%
\end{eqnarray*}

We now claim that:%
\[
\sum_{k=1}^{\infty}\delta_{\left(  x_{\gamma}\left(  k\right)  +x_{\eta
}\left(  \ell\right)  -x_{\alpha}\left(  j\right)  -x_{\beta}\left(  s\right)
\right)  ,0}\leq F\left(  \alpha,\beta,\eta;\gamma\right)
\]
where $F\left(  \alpha,\beta,\eta;\gamma\right)  =1$ if $\gamma\leq
\max\left\{  \alpha,\beta,\eta\right\}  $ and $F\left(  \alpha,\beta
,\eta;\gamma\right)  =0$ otherwise. This can be proved with the same argument
yielding (\ref{I1}) in the Proof of Lemma \ref{LOp}. Then:%
\[
\partial_{t}m_{\gamma}\leq\sum_{\alpha,\beta,\eta}\sum_{\ell,j,m=1}^{\infty
}\frac{a_{\eta}\left(  \ell\right)  a_{\alpha}\left(  j\right)  a_{\beta
}\left(  m\right)  }{\sqrt{x_{\eta}\left(  \ell\right)  x_{\alpha}\left(
j\right)  x_{\beta}\left(  m\right)  }}\Phi\left(  x_{\eta}\left(
\ell\right)  ,x_{\alpha}\left(  j\right)  ,x_{\beta}\left(  m\right)  \right)
F\left(  \alpha,\beta,\eta;\gamma\right)
\]

Using the symmetry in the indexes we obtain:%
\[
\partial_{t}m_{\gamma}\leq6\sum_{\alpha\leq\beta\leq\eta}\sum_{\ell
,j,m=1}^{\infty}\frac{a_{\eta}\left(  \ell\right)  a_{\alpha}\left(  j\right)
a_{\beta}\left(  m\right)  }{\sqrt{x_{\eta}\left(  \ell\right)  x_{\alpha
}\left(  j\right)  x_{\beta}\left(  m\right)  }}\Phi\left(  x_{\eta}\left(
\ell\right)  ,x_{\alpha}\left(  j\right)  ,x_{\beta}\left(  m\right)  \right)
F\left(  \alpha,\beta,\eta;\gamma\right)
\]
Then, using $\Phi\left(  x_{\eta}\left(  \ell\right)  ,x_{\alpha}\left(
j\right)  ,x_{\beta}\left(  m\right)  \right)  \leq\sqrt{x_{\eta}\left(
\ell\right)  }:$
\begin{align*}
\partial_{t}m_{\gamma}  &  \leq6\sum_{\alpha\leq\beta\leq\eta}\sum
_{\ell,j,m=1}^{\infty}\frac{a_{\eta}\left(  \ell\right)  a_{\alpha}\left(
j\right)  a_{\beta}\left(  m\right)  F\left(  \alpha,\beta,\eta;\gamma\right)
}{\sqrt{x_{\eta}\left(  \ell\right)  x_{\alpha}\left(  j\right)  x_{\beta
}\left(  m\right)  }}\sqrt{x_{\eta}\left(  \ell\right)  }\leq\\
&  \leq6\sum_{\alpha\leq\beta\leq\eta}\sum_{\ell,j,m=1}^{\infty}\frac{a_{\eta
}\left(  \ell\right)  a_{\alpha}\left(  j\right)  a_{\beta}\left(  m\right)
F\left(  \alpha,\beta,\eta;\gamma\right)  }{\sqrt{x_{\alpha}\left(  j\right)
x_{\beta}\left(  m\right)  }}%
\end{align*}
whence:%
\[
\partial_{t}m_{\gamma}\leq6\sum_{\alpha\leq\beta\leq\eta}m_{\eta}\sum
_{j,m=1}^{\infty}\frac{a_{\alpha}\left(  j\right)  a_{\beta}\left(  m\right)
F\left(  \alpha,\beta,\eta;\gamma\right)  }{\sqrt{x_{\alpha}\left(  j\right)
x_{\beta}\left(  m\right)  }}%
\]
Then, using the estimates $x_{\alpha}\left(  j\right)  \geq x_{\alpha}\left(
1\right)  ,\ x_{\beta}\left(  m\right)  \geq x_{\beta}\left(  1\right)  $ and
adding in $j,m:$%
\[
\partial_{t}m_{\gamma}\leq6\sum_{\alpha\leq\beta\leq\eta\ ,\ \eta\geq\gamma
}\frac{m_{\eta}m_{\alpha}m_{\beta}}{\sqrt{x_{\alpha}\left(  1\right)
x_{\beta}\left(  1\right)  }}%
\]
We now split the sum in two cases:%
\begin{align}
\partial_{t}m_{\gamma}  &  \leq6\sum_{\alpha\leq\beta\leq\eta=\gamma}%
\frac{m_{\eta}m_{\alpha}m_{\beta}}{\sqrt{x_{\alpha}\left(  1\right)  x_{\beta
}\left(  1\right)  }}+6\sum_{\alpha\leq\beta\leq\eta\ ,\ \eta>\gamma}%
\frac{m_{\eta}m_{\alpha}m_{\beta}}{\sqrt{x_{\alpha}\left(  1\right)  x_{\beta
}\left(  1\right)  }}\ \label{Z4E6}\\
&  \leq6m_{\gamma}\sum_{\alpha\leq\beta\leq\gamma}\frac{m_{\alpha}m_{\beta}%
}{\sqrt{x_{\alpha}\left(  1\right)  x_{\beta}\left(  1\right)  }}%
+6M_{\gamma+1}\sum_{\alpha\leq\beta}\frac{m_{\alpha}m_{\beta}}{\sqrt
{x_{\alpha}\left(  1\right)  x_{\beta}\left(  1\right)  }}\ \nonumber
\end{align}
where $M_{\gamma+1}$ is as in (\ref{Z4E1}). Then:%
\begin{equation}
6m_{\gamma}\sum_{\alpha\leq\beta\leq\gamma}\frac{m_{\alpha}m_{\beta}}%
{\sqrt{x_{\alpha}\left(  1\right)  x_{\beta}\left(  1\right)  }}\leq
\frac{6m_{\gamma}}{x_{\gamma}\left(  1\right)  }\sum_{\alpha\leq\beta
\leq\gamma}m_{\alpha}m_{\beta}\leq\frac{6m_{\gamma}}{x_{\gamma}\left(
1\right)  }\ \label{Z4E7}%
\end{equation}
where we use the fact that the total mass of $g$ is bounded. On the other
hand:%
\begin{align}
\sum_{\alpha\leq\beta}\frac{m_{\alpha}m_{\beta}}{\sqrt{x_{\alpha}\left(
1\right)  x_{\beta}\left(  1\right)  }}  &  \leq\left(  \sum_{\alpha}%
\frac{m_{\alpha}}{\sqrt{x_{\alpha}\left(  1\right)  }}\right)  ^{2}=\left(
\sum_{\alpha>\gamma}\frac{m_{\alpha}}{\sqrt{x_{\alpha}\left(  1\right)  }%
}+\sum_{\alpha\leq\gamma}\frac{m_{\alpha}}{\sqrt{x_{\alpha}\left(  1\right)
}}\right)  ^{2}\label{Z4E8}\\
&  \leq\left(  S_{\gamma+1}+\frac{1}{\sqrt{x_{\gamma}\left(  1\right)  }%
}\right)  ^{2}\ \nonumber
\end{align}
where we define $S_{\gamma+1}$ as in (\ref{Z4E1}). Plugging (\ref{Z4E7}),
(\ref{Z4E8}) into (\ref{Z4E6}) we obtain (\ref{Z4E2}) and the Lemma follows.
\end{proof}

\begin{lemma}
\label{approxODEnew}Suppose that $g\in C\left(  \left[  0,T\right]
:\mathcal{X}_{\theta,\rho^{\ast}}\right)  $ solves (\ref{S2E1}) in the sense
of Definition \ref{MiSol}. Suppose that $g_{0}=g\left(  0,\cdot\right)  $
satisfies
\begin{equation}
\int_{\left[  0,R2^{-\gamma}\right]  }g_{0}d\omega\geq1-\eta\label{S2E3n}%
\end{equation}
for suitable $R$ and $\eta.$ Let us assume that $g$ has the form
(\ref{T2E6a})-(\ref{I1c}). Then, the following inequality holds:%
\begin{equation}
\partial_{t}\left[  a_{\gamma}\left(  1\right)  \left(  t\right)  +R_{\gamma
}\left(  t\right)  \right]  \geq\frac{a_{\gamma}\left(  1\right)  \left(
t\right)  \left[  \left(  1-2\eta\right)  -a_{\gamma}\left(  1\right)  \left(
t\right)  \right]  ^{2}}{L^{2}2^{\gamma}} \label{S2E2n}%
\end{equation}
where:%
\[
R_{\gamma}\left(  t\right)  \leq3M_{\gamma+1}\left(  t\right)  \ .
\]

\end{lemma}

\begin{proof}
We use (\ref{S1E12a}). The function $\mathcal{G}_{\varphi}(\mathcal{\omega
}_{1}\,\mathcal{\omega}_{2}\,\mathcal{\omega}_{3})$ can be written as
(cf. Lemma \ref{strictConvex}):%
$$
\mathcal{G}_{\varphi}\left(  \omega_{1},\omega_{2},\omega_{3}\right)   
=\frac{1}{3}\left[  \sqrt{\omega_{-}}H_{\varphi}^{1}\left(  \omega_{1}%
,\omega_{2},\omega_{3}\right)  +\sqrt{\left(  \omega_{0}+\omega_{-}-\omega
_{+}\right)  _{+}}H_{\varphi}^{2}\left(  \omega_{1},\omega_{2},\omega
_{3}\right)  \right] \nonumber
$$
Using  the
convex test function $\varphi\left(  \omega\right)  =2\left(  \frac{3}%
{2}-\frac{\omega}{x_{\gamma}\left(  1\right)  }\right)  _{+}$ we deduce by 
 Lemma \ref{strictConvex}:

\begin{equation}
\mathcal{G}_{\varphi}\left(  \omega_{1},\omega_{2},\omega_{3}\right)
\geq\frac{\sqrt{\mathcal{\omega}_{-}}}{3}\left[  \varphi\left(  \omega
_{+}+\omega_{-}-\omega_{0}\right)  +\varphi\left(  \omega_{+}+\omega
_{0}-\omega_{-}\right)  -2\varphi\left(  \omega_{+}\right)  \right]
\geq0\ \label{S1E8a}%
\end{equation}

We split the measure $g$ in two pieces:%
\[
g=g_{\gamma}+\tilde{g}_{\gamma}%
\]
where:%
\[
g_{\gamma}=\sum_{\alpha\leq\gamma}a_{\alpha}\ \ ,\ \ \tilde{g}_{\gamma}%
=\sum_{\alpha>\gamma}a_{\alpha}%
\]

We now use the monotonicity formula (\ref{S2E5}). Then, using also
(\ref{S1E8a}) and ignoring all the terms containing $\tilde{g}_{\gamma}$
(since they are nonnegative):
\begin{eqnarray*}
&&\frac{d}{dt}\left(  \int_{0}^{\infty}g\left(  \omega\right)  \varphi\left(
\omega\right)  d\omega\right)     =6\iiint_{\left\{  \omega_{-}%
\leq\omega_{0}\leq\omega_{+}\right\}  }\frac{g\left(  \omega_{-}\right)
g\left(  \omega_{0}\right)  g\left(  \omega_{+}\right)  }{\sqrt{\omega
_{-}\omega_{0}\omega_{+}}} \times \\
&& \hskip 7cm\times \mathcal{G}_{\varphi}\left(  \omega_{-},\omega
_{0},\omega_{+}\right)  d\omega_{-}d\omega_{0}d\omega_{+}\\
&&\hskip 3.9cm   \geq6\iiint_{\left\{  \omega_{-}\leq\omega_{0}\leq\omega_{+}\leq
L\right\}  }\frac{g_{\gamma}\left(  \omega_{-}\right)  g_{\gamma}\left(
\omega_{0}\right)  g_{\gamma}\left(  \omega_{+}\right)  }{\sqrt{\omega
_{-}\omega_{0}\omega_{+}}}\times \\
&&\hskip 7cm\times \mathcal{G}_{\varphi}\left(  \omega_{-},\omega
_{0},\omega_{+}\right)  d\omega_{-}d\omega_{0}d\omega_{+}%
\end{eqnarray*}

Notice that the smallest particle with mass in $g_{\gamma}$ is $x_{\gamma
}\left(  1\right)  .$ Then we have the following inequalities in the sense of
measures for $\omega_{-}\leq\omega_{0}\leq\omega_{+}\leq L:$
\begin{align*}
&  \frac{g_{\gamma}\left(  \omega_{-}\right)  g_{\gamma}\left(  \omega
_{0}\right)  g_{\gamma}\left(  \omega_{+}\right)  }{\sqrt{\omega_{-}\omega
_{0}\omega_{+}}}\mathcal{G}_{\varphi}\left(  \omega_{-},\omega_{0},\omega
_{+}\right) \\
&  \geq\frac{g_{\gamma}\left(  \omega_{-}\right)  g_{\gamma}\left(  \omega
_{0}\right)  g_{\gamma}\left(  \omega_{+}\right)  }{3\sqrt{\omega_{0}%
\omega_{+}}}\left[  \varphi\left(  \omega_{+}+\omega_{-}-\omega_{0}\right)
+\varphi\left(  \omega_{+}+\omega_{0}-\omega_{-}\right)  -2\varphi\left(
\omega_{+}\right)  \right] \\
&  \geq\frac{g_{\gamma}\left(  \omega_{-}\right)  g_{\gamma}\left(  \omega
_{0}\right)  g_{\gamma}\left(  \omega_{+}\right)  }{3L}\left[  \varphi\left(
\omega_{+}+\omega_{-}-\omega_{0}\right)  +\varphi\left(  \omega_{+}+\omega
_{0}-\omega_{-}\right)  -2\varphi\left(  \omega_{+}\right)  \right]
\end{align*}
This expression is nonnegative at every point. Moreover, if $\omega_{+}%
=\omega_{0}$ and $\omega_{+}>\omega_{-},$ and using the fact that $\varphi$
vanishes for the points $x_{\gamma}\left(  k\right)  ,\ k\geq2,$ $x_{\alpha
}\left(  j\right)  ,\ j\geq1,\ \alpha<\gamma$ it then follows that:
\begin{align*}
&  6\iiint_{\left\{  \omega_{-}\leq\omega_{0}\leq\omega_{+}\leq
L\right\}  }\frac{g_{\gamma}\left(  \omega_{-}\right)  g_{\gamma}\left(
\omega_{0}\right)  g_{\gamma}\left(  \omega_{+}\right)  }{\sqrt{\omega
_{-}\omega_{0}\omega_{+}}}\mathcal{G}_{\varphi}\left(  \omega_{-},\omega
_{0},\omega_{+}\right)  d\omega_{-}d\omega_{0}d\omega_{+}\\
&  \geq\frac{2a_{\gamma}\left(  1\right)  }{L}\sum_{\alpha\leq\gamma}%
\sum_{j=1}^{\infty}\left(  a_{\alpha}\left(  j\right)  \left(  1-\delta
_{j,1}\delta_{\alpha,\gamma}\right)  \right)  \varphi\left(  x_{\gamma}\left(
1\right)  \right)
\end{align*}
Since $\varphi\left(  x_{\gamma}\left(  1\right)  \right)  =1,$  we deduce:
\begin{align*}
&  6\iiint_{\left\{  \omega_{-}\leq\omega_{0}\leq\omega_{+}\leq
L\right\}  }\frac{g_{\gamma}\left(  \omega_{-}\right)  g_{\gamma}\left(
\omega_{0}\right)  g_{\gamma}\left(  \omega_{+}\right)  }{\sqrt{\omega
_{-}\omega_{0}\omega_{+}}}\mathcal{G}_{\varphi}\left(  \omega_{-},\omega
_{0},\omega_{+}\right)  d\omega_{-}d\omega_{0}d\omega_{+}\\
&  \geq\frac{2a_{\gamma}\left(  1\right)  }{L}\sum_{\alpha\leq\gamma}%
\sum_{j=1,\left\{  x_{\alpha}\left(  j\right)  \leq L\right\}  }^{\infty
}\left(  a_{\alpha}\left(  j\right)  \left(  1-\delta_{j,1}\delta
_{\alpha,\gamma}\right)  \right)  ^{2}\\
&  \geq\frac{2a_{\gamma}\left(  1\right)  }{L}\sum_{\alpha\leq\gamma}\left[
\int_{\left\{  x\leq L\right\}  \setminus\left\{  x_{\gamma}\left(  1\right)
\right\}  }g_{\gamma}\right]  ^{2}  \geq\frac{2a_{\gamma}\left(  1\right)  }{LN_{\gamma}}\left(  \sum
_{\alpha\leq\gamma}\int_{\left\{  x\leq L\right\}  \setminus\left\{
x_{\gamma}\left(  1\right)  \right\}  }g_{\gamma}\right)  ^{2}\\
&  \geq\frac{2a_{\gamma}\left(  1\right)  }{LN_{\gamma}}\left(  \sum
_{\alpha\leq\gamma}\left[  \left(  1-2\eta\right)  -a_{\gamma}\left(
1\right)  \right]  \right)  ^{2}%
\end{align*}
where $N_{\gamma}$ is the number of elements of $\mathcal{Z}_{\gamma}$ which
are smaller than $L.$ We have used Jensen's inequality in the sum. Notice that
$N_{\gamma}\leq L2^{\gamma}.$ Therefore:%
\begin{align*}
&  6\iiint_{\left\{  \omega_{-}\leq\omega_{0}\leq\omega_{+}\leq
L\right\}  }\frac{g_{\gamma}\left(  \omega_{-}\right)  g_{\gamma}\left(
\omega_{0}\right)  g_{\gamma}\left(  \omega_{+}\right)  }{\sqrt{\omega
_{-}\omega_{0}\omega_{+}}}\mathcal{G}_{\varphi}\left(  \omega_{-},\omega
_{0},\omega_{+}\right)  d\omega_{-}d\omega_{0}d\omega_{+}\\
&  \geq\frac{2a_{\gamma}\left(  1\right)  }{L^{2}2^{\gamma}}\left(
\sum_{\alpha\leq\gamma}\left[  \left(  1-2\eta\right)  -a_{\gamma}\left(
1\right)  \right]  \right)  ^{2}%
\end{align*}

We then have:%
\[
\frac{d}{dt}\left(  \int_{0}^{\infty}g\left(  \omega\right)  \varphi\left(
\omega\right)  d\omega\right)  \geq\frac{2a_{\gamma}\left(  1\right)  }%
{L^{2}2^{\gamma}}\left(  \sum_{\alpha\leq\gamma}\left[  \left(  1-2\eta
\right)  -a_{\gamma}\left(  1\right)  \right]  \right)  ^{2}%
\]

We then write $\int_{0}^{\infty}g\varphi d\omega=\int_{0}^{\infty}g_{\gamma
}\varphi d\omega+\int_{0}^{\infty}\tilde{g}_{\gamma}\varphi d\omega.$ Let us
denote as $R_{\gamma}\left(  t\right)  $ the quantity $\int_{0}^{\infty}%
\tilde{g}_{\gamma}\varphi d\omega.$ Then, since $\varphi\leq3$ we obtain:%
\[
R_{\gamma}\left(  t\right)  =\int_{0}^{\infty}\tilde{g}_{\gamma}\varphi
d\omega\leq3\int_{0}^{\infty}\tilde{g}_{\gamma}d\omega=3M_{\gamma+1}%
\]
This concludes the Proof of Lemma \ref{approxODEnew}.
\end{proof}

\begin{lemma}
\label{ODELemma}Suppose that the set of functions $\left\{  a_{\gamma}\left(
1\right)  \right\}  ,\ \left\{  m_{\gamma}\right\}  $ defined above for
$\gamma=0,1,2,...$ satisfy the following set of inequalities:%
\[
\frac{db_{\gamma}}{dt}\geq C_{1}\left(  \gamma\right)  \left[  b_{\gamma
}\left[  \left(  1-2\eta_{\gamma}\right)  -b_{\gamma}\right]  ^{2}%
-BM_{\gamma+1}\right]
\]
where $b_{\gamma}=a_{\gamma}\left(  1\right)  +R_{\gamma},$ $R_{\gamma}%
\leq3M_{\gamma+1}$, $\lim_{\gamma\rightarrow\infty}\eta_{\gamma}=0,$ as well
as:%
\[
\frac{dm_{\gamma}}{dt}\leq C_{2}\left(  \gamma\right)  m_{\gamma}%
+6M_{\gamma+1}\left(  S_{\gamma+1}+\frac{1}{\sqrt{x_{\gamma}\left(  1\right)
}}\right)  ^{2}%
\]
where $S_{\gamma+1}$ is as in (\ref{Z4E1}). Then, there exists a sequence
$\left\{  \varepsilon_{\gamma}\right\}  $ of positive numbers satisfying
$\sum_{\gamma=0}^{\infty}\varepsilon_{\gamma}=1,$ such that, if we assume
that:\
\[
a_{\gamma}\left(  \cdot\right)  \left(  0\right)  =\varepsilon_{\gamma}%
\delta_{x_{\gamma}\left(  1\right)  }%
\]
there exists a solution $g$ of (\ref{S2E1}) in the sense of Definition
\ref{MiSol} globally defined in time and there exists an increasing sequence
of times $\left\{  t_{n}\right\}  $ such that $\lim_{n\rightarrow\infty}%
t_{n}=\infty$ and:%
\begin{equation}
a_{n}\left(  1\right)  \left(  t\right)  \geq1-4\eta_{n}\ \ \text{for\ \ }%
t_{n}\leq t\leq2t_{n} \label{S6E1}%
\end{equation}

\end{lemma}

\begin{remark}
Notice that the Theorem implies
\[
\sup_{t_{n}\leq t\leq2t_{n}}\left(  \sum_{\alpha\neq\gamma}a_{\alpha}\left(
\cdot\right)  \left(  t\right)  +\sum_{\ell=2}^{\infty}a_{\gamma}\left(
\ell\right)  \left(  t\right)  \right)  \rightarrow0\ \ \text{as\ \ }%
n\rightarrow\infty
\]

\end{remark}

\begin{remark}
The result can be reformulated also in terms of the weak topology of measures
for \thinspace$g.$ We recall $dist_{\ast}$ denotes the
distance associated to the weak topology of measures (cf. Notation
\ref{distWeak}). Then:%
\[
\sup_{t_{n}\leq t\leq2t_{n}}\left[dist_{\ast}\left(
2^{n}g\left(  2^{n}\left(  \cdot\right)  ,t\right)  ,\delta_{1}\right)
\right]  \rightarrow0\ \ \text{as\ \ }n\rightarrow\infty
\]

\end{remark}

\begin{proof}
We will assume, without loss of generality that $C_{1}\left(  \gamma\right)
\geq0,\ C_{2}\left(  \gamma\right)  \geq0.$ We need to guarantee the existence
of several inequalities for a suitable range of times. These inequalities are:%
\begin{align}
BM_{\gamma+1}  &  \leq\frac{1}{8}\min\left\{  b_{\gamma},\eta_{\gamma}%
,\frac{Bb_{\gamma}}{3}\right\} \nonumber\\
S_{\gamma+1}  &  \leq\frac{1}{\sqrt{x_{\gamma}\left(  1\right)  }}\nonumber\\
6M_{\gamma+1}\left(  \frac{2}{\sqrt{x_{\gamma}\left(  1\right)  }}\right)
^{2}  &  \leq C_{2}\left(  \gamma\right)  m_{\gamma} \label{S6E2}%
\end{align}

As long as these inequalities are satisfied we have:%
\begin{equation}
\frac{db_{\gamma}}{dt}\geq\frac{C_{1}\left(  \gamma\right)  b_{\gamma}}%
{2}\left[  \left(  1-2\eta_{\gamma}\right)  -b_{\gamma}\right]  ^{2}
\label{S6E3}%
\end{equation}%
\begin{equation}
\frac{dm_{\gamma}}{dt}\leq2C_{2}\left(  \gamma\right)  m_{\gamma} \label{S6E4}%
\end{equation}

We need to precise sufficient conditions to have (\ref{S6E2}). Notice that if
(\ref{S6E2}) holds we have $\frac{db_{\gamma}}{dt}\geq0.$ Moreover, the first
inequality in (\ref{S6E2}) guarantees also that $a_{\gamma}\left(  1\right)
\left(  t\right)  $ is comparable to $b_{\gamma}\left(  t\right)  .$ Then
$b_{\gamma}\left(  0\right)  \geq\frac{\varepsilon_{\gamma}}{2}.$ We can
assume that $\eta_{\gamma}\geq\frac{\varepsilon_{\gamma}}{2}$ and also that
$B\geq3.$ Then, since $m_{\gamma}\geq a_{\gamma}\left(  1\right)  $ and
$a_{\gamma}\left(  1\right)  \left(  t\right)  $ is comparable to $b_{\gamma
}\left(  t\right)  ,$ we would have the first and the third inequalities in
(\ref{S6E2}) if we have:%
\begin{equation}
BM_{\gamma+1}\leq\varepsilon_{\gamma}\ \ \text{and\ \ }6M_{\gamma+1}\left(
\frac{2}{\sqrt{x_{\gamma}\left(  1\right)  }}\right)  ^{2}\leq C_{2}\left(
\gamma\right)  \varepsilon_{\gamma} \label{S6E5}%
\end{equation}

Notice that (\ref{S6E4}) implies:%
\[
m_{\gamma}\left(  t\right)  \leq m_{\gamma}\left(  0\right)  \exp\left(
2C_{2}\left(  \gamma\right)  t\right)
\]
and since $m_{\gamma}\left(  0\right)  =\varepsilon_{\gamma}$ we obtain:%
\[
m_{\gamma}\left(  t\right)  \leq\varepsilon_{\gamma}\exp\left(  2C_{2}\left(
\gamma\right)  t\right)
\]

Therefore:%
\begin{equation}
M_{\gamma+1}\leq\sum_{\alpha\geq\gamma+1}\varepsilon_{\alpha}\exp\left(
2C_{2}\left(  \alpha\right)  t\right)  \ \label{S6E6}%
\end{equation}
and:%
\begin{equation}
S_{\gamma+1}\leq\sum_{\alpha\geq\gamma+1}\frac{m_{\alpha}}{\sqrt{x_{\alpha
}\left(  1\right)  }}\leq\sum_{\alpha\geq\gamma+1}\frac{\varepsilon_{\alpha
}\exp\left(  2C_{2}\left(  \alpha\right)  t\right)  }{\sqrt{x_{\alpha}\left(
1\right)  }}=\sum_{\alpha\geq\gamma+1}\varepsilon_{\alpha}2^{\frac{\alpha}{2}%
}\exp\left(  2C_{2}\left(  \alpha\right)  t\right)  \label{S6E7}%
\end{equation}

Using (\ref{S6E3}) as well as the fact that \thinspace$b_{\gamma}\left(
0\right)  \geq\frac{\varepsilon_{\gamma}}{2}$ we obtain that, if (\ref{S6E3})
holds during the time interval $0\leq t\leq2t_{\gamma}$ for we would have
$b_{\gamma}\left(  t\right)  \geq1-3\eta_{\gamma}$ for $t_{\gamma}\leq
t\leq2t_{\gamma},$ where:%
\begin{eqnarray*}
&&t_{\gamma}\left(  \varepsilon_{\gamma},\eta_{\gamma}\right)  =\frac
{C_{1}\left(  \gamma\right)  }{2\left(  1-2\eta_{\gamma}\right)  ^{2}}\left[
\frac{\left(  1-2\eta_{\gamma}\right)  }{\eta_{\gamma}}-\frac{2\left(
1-2\eta_{\gamma}\right)  }{2+4\eta_{\gamma}-\varepsilon_{\gamma}}+\right.\\
&&\hskip 5cm +\left.\log\left(
\frac{\left(  1-3\eta_{\gamma}\right)  \left(  2+4\eta_{\gamma}-\varepsilon
_{\gamma}\right)  }{\eta_{\gamma}\varepsilon_{\gamma}}\right)  \right]
\end{eqnarray*}
Notice that $t_{\gamma}$ tends to infinity if $\eta_{\gamma}$ or
$\varepsilon_{\gamma}$ approach zero.

We would have (\ref{S6E3}), (\ref{S6E4}) for $0\leq t\leq2t_{\gamma}$ if the
inequalities (\ref{S6E2}) hold in the same time interval. Sufficient condition
for this are the inequalities (\ref{S6E5}) as well as the second inequality in
(\ref{S6E2}). Using (\ref{S6E6}), (\ref{S6E7}) we would have those
inequalities for $0\leq t\leq2t_{\gamma}$ if:%
\begin{align*}
B\sum_{\alpha\geq\gamma+1}\varepsilon_{\alpha}\exp\left(  4C_{2}\left(
\alpha\right)  t_{\gamma}\left(  \varepsilon_{\gamma},\eta_{\gamma}\right)
\right)   &  \leq\varepsilon_{\gamma}\ \ \\
\text{\ }6\left(  \frac{2}{\sqrt{x_{\gamma}\left(  1\right)  }}\right)
^{2}\sum_{\alpha\geq\gamma+1}\varepsilon_{\alpha}\exp\left(  4C_{2}\left(
\alpha\right)  t_{\gamma}\left(  \varepsilon_{\gamma},\eta_{\gamma}\right)
\right)   &  \leq C_{2}\left(  \gamma\right)  \varepsilon_{\gamma}\\
\sum_{\alpha\geq\gamma+1}\varepsilon_{\alpha}2^{\frac{\alpha}{2}}\exp\left(
4C_{2}\left(  \alpha\right)  t_{\gamma}\left(  \varepsilon_{\gamma}%
,\eta_{\gamma}\right)  \right)   &  \leq\frac{1}{\sqrt{x_{\gamma}\left(
1\right)  }}%
\end{align*}

The three inequalities hold if we have:%
\begin{equation}
\sum_{\alpha\geq\gamma+1}\varepsilon_{\alpha}2^{\frac{\alpha}{2}}\exp\left(
4C_{2}\left(  \alpha\right)  t_{\gamma}\left(  \varepsilon_{\gamma}%
,\eta_{\gamma}\right)  \right)  \leq\min\left\{  \frac{\varepsilon_{\gamma}%
}{B},\frac{C_{2}\left(  \gamma\right)  2^{-\gamma}\varepsilon_{\gamma}}%
{24},2^{\frac{\gamma}{2}}\right\}  \equiv Q\left(  \gamma;\varepsilon_{\gamma
}\right)  \label{S6E8}%
\end{equation}

We can now construct the sequence $\varepsilon_{\gamma}$ inductively. We will
assume that $\varepsilon_{0}\geq\frac{1}{2}.$ We will then select
$\varepsilon_{\alpha}$ inductively for $\alpha\geq1,$ as any positive number
satisfying the inequalities:%
\begin{equation}
\varepsilon_{\alpha}\leq\min\left\{  \frac{1}{2^{\alpha}}\frac{\min\left\{
Q\left(  \gamma;\varepsilon_{\gamma}\right)  ,\frac{1}{2}\right\}  }%
{2^{\frac{\alpha}{2}}\exp\left(  4C_{2}\left(  \alpha\right)  t_{\gamma
}\left(  \varepsilon_{\gamma},\eta_{\gamma}\right)  \right)  }:0\leq\gamma
\leq\alpha-1\right\}  \ \label{S6E9}%
\end{equation}
where we replace in these inequalities that $\varepsilon_{0}=\frac{1}{2}.$
Since $\exp\left(  4C_{2}\left(  \alpha\right)  t_{\gamma}\left(
\varepsilon_{\gamma},\eta_{\gamma}\right)  \right)  \geq1,$ these inequalities
imply that:%
\[
\sum_{\alpha\geq1}\varepsilon_{\alpha}\leq\frac{1}{2}\sum_{\alpha\geq1}%
\frac{1}{2^{\alpha}}=\frac{1}{4}<\frac{1}{2}%
\]

We then choose $\varepsilon_{0}$ as $\left(  1-\sum_{\alpha\geq1}%
\varepsilon_{\alpha}\right)  .$ Then $\sum_{\alpha\geq0}\varepsilon_{\alpha
}=1.$ Notice that $\varepsilon_{0}\geq\frac{1}{2}$ and since the right-hand
side of (\ref{S6E9}) is increasing in $\varepsilon_{0},$ it follows that this
inequalities hold with this new choice of $\varepsilon_{0},$ since they were
valid with $\varepsilon_{0}=\frac{1}{2}.$ Moreover, since $a_{\gamma}\left(
1\right)  \geq b_{\gamma}-3M_{\gamma+1}$ it then follows that $a_{\gamma
}\left(  1\right)  \left(  t\right)  \geq1-4\eta_{\gamma}$ if $t_{\gamma}\leq
t\leq2t_{\gamma}.$
\end{proof}

\subsubsection{Proof of Theorem \ref{globOsc}.}


\begin{proof}
[Proof of Theorem \ref{globOsc}] Lemma \ref{ODELemma} combined with Lemmas
\ref{ODEUpper} and \ref{approxODEnew} imply the existence of $g\in
C\left(  \left[  0,\infty\right)  :\mathcal{X}_{\theta,\rho^{\ast}}\right)  $
which solves (\ref{S2E1}) in the sense of Definition \ref{MiSol}. Notice that,
by construction, this solution satisfies $\int_{\left\{  0\right\}  }g\left(
t,d\omega\right)  =0$ for all $t>0.$ Using then Theorem \ref{AsympOsc} we
obtain that the alternative (ii) holds. 

It only remains to prove  (\ref{S4distance}). The construction of the family implies that for $t$ sufficiently large, most of the mass of the measure $g$ is contained in $\Omega _J\cup \Omega  _{ J+1 }$, with $J$ depending of $t$ and the mass contained out of this set tends to zero as $t\to \infty$. Rescaling the unit of length we may assume that $x_J(1)=1$. We denote as $\tilde g$ the measure $g$ using this new length scale. The construction of $g$ implies, the existence for any $\delta >0$, of $t_1<t_2$ both sufficiently large, such that: $ \int  _{ (1-\delta, 1+\delta  ) }\tilde g(t_1)\ge m(1-\delta )$ and $ \int  _{ (1/2-\delta, 1/2+\delta  ) }\tilde g(t_2)\ge m(1-\delta )$. Then, by continuity, there exists $t^*\in (t_1, t_2)$ such that $\int  _{ (1/2-\delta , 1/2+\delta ) } \tilde g =1/2$. Since the integral of $\tilde g$ over $\cup _{ k\not =J, J+1 }\Omega _k$ tends to zero as $t\to +\infty$, we obtain, using the definition of $\tilde g$:
\begin{equation}
\label{S4MP3}
\inf _{a >0 }\left(dist _{\ast}\left(  \frac{1}{ a}g\left(  t^*,\frac{\cdot
}{a}\right)  ,m\delta_{1}\right)\right) \geq c_1>0.
\end{equation} 
whence the result follows.
\end{proof}

\section{Heuristic arguments and open problems.}

We present in this Chapter several heuristic arguments and formal calculations concerning long time asymptotic properties of the solutions of equation (\ref{E1}), (\ref{E2}). To deal with this problem scaling arguments  have been repeatedly used  in the physical literature, cf. in particular \cite{P}, \cite{DNPZ}, \cite{Zbook}. From this point of view, the main goal of this Chapter is to formulate some precise PDE's problems covering several different cases.
\subsection{Transport of the energy towards large values of $\omega.$}
\label{Transfer2}
\subsubsection{Weak solutions with interacting condensates. The case of finite particle mass} \index{weak solution} \index{interacting condensate}
\label{Heurist}
Notice that Corollary \ref{AsEnergy} implies that for general initial data,
the energy of the solution is transported towards large values of $\omega
.$\ We remark that it is possible to derive heuristically one equation that
describes the transfer of energy towards larger scales.\ Let us assume by
definiteness that $g$ has the following two scale form:%
\begin{equation}
g\left(  t,\cdot\right)  =g_{comp}\left(  t,\cdot\right)  +\frac{1}{R^{2}%
}G\left(  t,\frac{\cdot}{R}\right)  \ \ ,\ \ g_{comp}\left(  t,\cdot\right)
=M\delta_{0}\left(  \cdot\right)  \label{K2E1}%
\end{equation}
where $R>>1.$ Notice that the form (\ref{K2E1}) implicitly assumes that most
of the mass of $g$ concentrates in $\omega=0.$ On the other hand, the energy
of the solution $\int\omega gd\omega$ is at distances of order $R$ from the
origin. Notice that Theorem \ref{Asympt} indicates that after a transient
state most of the mass of the solutions should concentrate at $\omega=0.$ We
are assuming in (\ref{K2E1}) that the value of $R_{\ast}$ in Theorem
\ref{Asympt} is $R_{\ast}=0.$ If $R_{\ast}>0$ and $R>>R_{\ast}$ it would be
possible to argue in a similar manner, although in such a case a fraction of
the energy would remain trapped at distances of order $R_{\ast}$ from the origin.

Assuming that $g$ has the form (\ref{K2E1}) we can derive an evolution
equation for $G$ as follows. We use in (\ref{Z2E1a}) a test function $\varphi$
with the form $\varphi\left(  \omega\right)  =\psi\left(  \frac{\omega}%
{R}\right)  ,$ with $\psi$ compactly supported in $\left(  0,\infty\right)  .$
Then:%
\begin{equation}
\partial_{t}\left(  \int_{\left[  0,\infty\right)  }g\left(  t,\omega\right)
\varphi\left(  \omega\right)  d\omega\right)  =\frac{1}{R}\partial_{t}\left(
\int_{\left[  0,\infty\right)  }G\left(  t,\tilde{\omega}\right)  \psi\left(
\tilde{\omega}\right)  d\tilde{\omega}\right)  \label{K2E2}%
\end{equation}
On the other hand we have:

\begin{align*}
g_{1}g_{2}g_{3}  &  =M^{3}\delta_{0}\left(  \omega_{1}\right)  \delta
_{0}\left(  \omega_{2}\right)  \delta_{0}\left(  \omega_{2}\right)  +\frac{M^{2}}{R^{2}}\delta_{0}\left(  \omega_{1}\right)  \delta_{0}\left(
\omega_{2}\right)  G\left(  t,\frac{\omega_{3}}{R}\right)  +\\
&  +\frac{M^{2}}%
{R^{2}}\delta_{0}\left(  \omega_{1}\right)  \delta_{0}\left(  \omega
_{3}\right)  G\left(  t,\frac{\omega_{2}}{R}\right)  +\frac{M^{2}}{R^{2}%
}\delta_{0}\left(  \omega_{2}\right)  \delta_{0}\left(  \omega_{3}\right)
G\left(  t,\frac{\omega_{1}}{R}\right)  +\\
&  +\frac{M}{R^{4}}\delta_{0}\left(  \omega_{1}\right)  G\left(
t,\frac{\omega_{2}}{R}\right)  G\left(  t,\frac{\omega_{3}}{R}\right)
+\frac{M}{R^{4}}\delta_{0}\left(  \omega_{2}\right)  G\left(  t,\frac
{\omega_{1}}{R}\right)  G\left(  t,\frac{\omega_{3}}{R}\right)  +\\
&  +\frac
{M}{R^{4}}\delta_{0}\left(  \omega_{3}\right)  G\left(  t,\frac{\omega_{1}}%
{R}\right)  G\left(  t,\frac{\omega_{2}}{R}\right)  +\\
&  +\frac{1}{R^{6}}G\left(  t,\frac{\omega_{1}}{R}\right)  G\left(
t,\frac{\omega_{2}}{R}\right)  G\left(  t,\frac{\omega_{3}}{R}\right) =\sum_{k=1}^{8}I_{k}%
\end{align*}

The contribution of the term $I_{1}$ in the integral of the right-hand side of
(\ref{Z2E1a}) vanishes. Notice that for the test functions under consideration
this would hold also if $M\delta_{0}$ is replaced by a distribution supported
in regions $\omega$ of order one. The contribution of the term $I_{2}$ is
$o\left(  \frac{1}{R^{2}}\right)  $ if $G\left(  t,\bar{\omega}\right)  $
contains a small amount of mass in regions with $\bar{\omega}$ small. The
contribution of the terms $I_{3},\ I_{4}$ also vanishes. Actually, if
$\ g_{comp}$ is replaced by a distribution supported in values with $\omega$
of order one we would obtain terms containing the derivatives of the test
function $\psi.$ However, the contribution of those terms would be $o\left(
\frac{1}{R^{2}}\right)  $\ due to the factor $\frac{1}{R^{2}},$ as well as the
form of the test function $\varphi$ which gives an additional term $\frac
{1}{R}$ upon differentiation. More precisely, if we assume that $\ g_{comp}$
is supported in a bounded range of values of $\omega$ we would obtain the
following contributions due to$\ I_{3}$ in the integral on the right-hand side
of (\ref{Z2E1a}):%
\begin{eqnarray*}
&&\frac{1}{R^{3}}\iiint_{(\left[  0,\infty\right) )^3 }\frac{g_{comp,1}g_{comp,3}G\left(
t,\frac{\omega_{2}}{R}\right)  \Phi}{\sqrt{\omega_{1}\omega_{2}\omega_{3}}
}\times \\
&&\hskip 3cm \times\left[  \psi^{\prime}\left(  \frac{\omega_{1}}{R}\right)  \left(  \omega
_{1}-\omega_{3}\right)  +O\left(  \frac{\left(  \omega_{1}-\omega_{3}\right)
^{2}}{R}\right)  \right]  d\omega_{1}d\omega_{2}d\omega_{3}%
\end{eqnarray*}

The integral of the term containing $\left(  \omega_{1}-\omega_{3}\right)  $
vanishes by symmetry and the last term gives then a contribution of order
$O\left( R^{-4}\right)  .$ The one of $I_{4}$ is similar.

The term $I_{8}$ in  the right-hand side of (\ref{Z2E1a})
can be computed by means of rescaling arguments. It turns out to be of order
$\frac{1}{R^{4}}=o\left(  \frac{1}{R^{3}}\right)  .$

The main contribution to the integral on the right of (\ref{Z2E1a}) is due to
the terms $I_{5},\ I_{6},\ I_{7}.$ The  terms $I_{5}+I_{6}$ yield:
\begin{equation}
\frac{2M}{R^{3}}\iint_{(\left[  0,\infty\right) )^2 }\frac{G\left(  t,\tilde{\omega}_{2}\right)  G\left(  t,\tilde{\omega}%
_{3}\right)  }{\sqrt{\tilde{\omega}_{2}\tilde{\omega}_{3}}}\left[  \psi\left(
\tilde{\omega}_{2}-\tilde{\omega}_{3}\right)  +\psi\left(  \tilde{\omega}%
_{3}\right)  -\psi\left(  \tilde{\omega}_{2}\right)  \right]  d\tilde{\omega
}_{2}d\tilde{\omega}_{3} \label{K2E3}%
\end{equation}

On the other hand, the contribution of the term $I_{7}$ is:%
\begin{equation}
\frac{M}{R^{3}}\iint_{(\left[  0,\infty\right)  )^2}\frac{G\left(  t,\tilde{\omega}_{1}\right)  G\left(  t,\tilde{\omega}%
_{2}\right)  }{\sqrt{\tilde{\omega}_{1}\tilde{\omega}_{2}}}\left[  \psi\left(
\tilde{\omega}_{1}+\tilde{\omega}_{2}\right)  -\psi\left(  \tilde{\omega}%
_{1}\right)  -\psi\left(  \tilde{\omega}_{2}\right)  \right]  d\tilde{\omega
}_{1}d\tilde{\omega}_{2}\ \label{K2E4}%
\end{equation}

It then follows, combining (\ref{K2E2})-(\ref{K2E4}) that to the leading
order, the evolution equation for $G$ is given by:%
\begin{align}
\partial_{t}\left(  \int_{\left[  0,\infty\right)  }G\left(  t,\tilde{\omega
}\right)  \psi\left(  \tilde{\omega}\right)  d\tilde{\omega}\right)   &
=\frac{2M}{R^{2}}\iint_{(\left[  0,\infty\right) )^2 }\frac{G\left(  t,\tilde{\omega}_{2}\right)  G\left(  t,\tilde
{\omega}_{3}\right)  }{\sqrt{\tilde{\omega}_{2}\tilde{\omega}_{3}}}\times \label{K2E5}\\
& \hskip 0.3 cm \times \left[
\psi\left(  \tilde{\omega}_{2}-\tilde{\omega}_{3}\right)  +\psi\left(
\tilde{\omega}_{3}\right)  -\psi\left(  \tilde{\omega}_{2}\right)  \right]
d\tilde{\omega}_{2}d\tilde{\omega}_{3}+\nonumber\\
&  +\frac{M}{R^{2}}\iint_{(\left[  0,\infty\right) )^2 }\frac{G\left(  t,\tilde{\omega}_{1}\right)  G\left(
t,\tilde{\omega}_{2}\right)  }{\sqrt{\tilde{\omega}_{1}\tilde{\omega}_{2}}%
}\times \nonumber\\
& \hskip 0.3 cm \times\left[  \psi\left(  \tilde{\omega}_{1}+\tilde{\omega}_{2}\right)
-\psi\left(  \tilde{\omega}_{1}\right)  -\psi\left(  \tilde{\omega}%
_{2}\right)  \right]  d\tilde{\omega}_{1}d\tilde{\omega}_{2} \nonumber
\end{align}

Equation (\ref{K2E5}) is the weak formulation of the following
coagulation - fragmentation equation:%
\begin{align}
\frac{R^{2}}{2M}\partial_{t}G\left(  t,\tilde{\omega}\right)   
&=\frac{G\left(  t,\tilde{\omega}\right)  }{\sqrt{\tilde{\omega}}}\int
_{\tilde{\omega}}^{\infty}\frac{G\left(  t,\xi\right)  d\xi}{\sqrt{\xi}}%
-\frac{G\left(  t,\tilde{\omega}\right)  }{\sqrt{\tilde{\omega}}}\int
_{0}^{\tilde{\omega}}\frac{G\left(  t,\xi\right)  d\xi}{\sqrt{\xi}}+ \label{K2E6}\\
&+\int
_{0}^{\infty}\frac{G\left(  t,\tilde{\omega}+\xi\right)  G\left(
t,\xi\right)  d\xi}{\sqrt{\left(  \tilde{\omega}+\xi\right)  \xi}} +\frac{1}{2}\int_{0}^{\tilde{\omega}}\frac{G\left(  t,\tilde{\omega}%
-\xi\right)  G\left(  t,\xi\right)  d\xi}{\sqrt{\left(  \tilde{\omega}%
-\xi\right)  \xi}}-\nonumber\\
&-\frac{G\left(  t,\tilde{\omega}\right)  }{\sqrt
{\tilde{\omega}}}\int_{0}^{\infty}\frac{G\left(  t,\xi\right)  d\xi}{\sqrt
{\xi}}\nonumber
\end{align}

Equation (\ref{K2E6}) can be expected to describe the flux of energy of the
solution towards $\tilde{\omega}\rightarrow\infty.$ More precisely,
(\ref{K2E6}) describes the distribution $G$ which describes the part of the
distribution $g$ in which the energy of the initial distribution is
concentrated. Notice that the characteristic time scale for the equation
(\ref{K2E6}) is of order $\frac{R^{2}}{2M}.$ This agrees with the result
obtained in Proposition \ref{coarsening}. Moreover, due to the fact that the
energy associated to the solution of (\ref{S2E1}), (\ref{S2E1a}) escapes to
large values of $\omega$ as $t\rightarrow\infty,$ we can expect the
asymptotics of (\ref{K2E6}) to be given by a self-similar behaviour with the
form:%
\begin{equation}
G\left(  t,\tilde{\omega}\right)  =\frac{R^{2}}{2Mt+R^{2}}\Phi\left(
y\right)  \ \ ,\ \ \ y=\frac{R\tilde{\omega}}{\sqrt{2Mt+R^{2}}} \label{K2E7}%
\end{equation}

Notice that such a rescaling indicates that the energy of the initial
distribution $g,$ which is concentrated at values $\omega\approx R$ for $t=0,$
would be concentrated at distances $\omega\approx\sqrt{2Mt+R^{2}}$ for
arbitrary values of $t\geq0.$

Equation (\ref{K2E6}) is reminiscent of the equation which has been obtained
in several papers considering the linearization of the isotropic
Nordheim  \index{Nordheim} equation near Bose-Einstein condensates (cf. \cite{DNPZ}, \cite{LLPR}, \cite{SK1, SK2}, 
\cite{Spohn}). The equations derived in
those papers contain additional terms which are due to the fact that the
Nordheim  \index{Nordheim} equation contains, besides the cubic terms in (\ref{E1})
additional quadratic terms. Moreover, in some of these papers, it is assumed
that $M$ is a function of $t,$ which is due to the fact that the mass of the
condensate it is assumed to change in time. In the case of the
Nordheim  \index{Nordheim} equation the energy of the distribution is expected to
remain in bounded regions of $\omega,$ and therefore the previous analysis
would be meaningless. A nonisotropic version of (\ref{K2E6}) has been obtained
also in \cite{DNPZ},  linearizing also near a condensate \index{condensate}at $\omega=0.$ However, some of the
analysis in \cite{DNPZ} and \cite{P} can be used in order to describe the
behaviour which can be expected for the solutions of (\ref{K2E6}). In
particular the paper\ \cite{P} has obtained the rescaling laws for the
transfer of energy towards infinity by means of dimensional arguments. In
order to describe asymptotically how this transfer takes place, we notice that
there exists a set of "thermal equilibria" for (\ref{K2E6}) having the form:%
\begin{equation}
G\left(  \tilde{\omega}\right)  =\frac{a}{\sqrt{\tilde{\omega}}}%
\ \ ,\ \ \ a\geq0 \label{K2E8}%
\end{equation}

The first and last integrals in (\ref{K2E6}) are divergent for if $G$ is as in
(\ref{K2E8}) but the integrals can be made meaningful combining the first and
last integrals on the right-hand side of (\ref{K2E6}). Notice that the
stationary solutions correspond to a balance between the aggregation and the
fragmentation terms in (\ref{K2E6}).

The asymptotics of the function $\Phi\left(  y\right)  $ as $y\rightarrow0$
can be expected to be given by the equilibria (\ref{K2E8}). We would then have
$\Phi\left(  y\right)  \sim\frac{a}{\sqrt{y}}$ as $y\rightarrow0$ for some
$a>0.$ Notice that this implies the following behaviour for the function $g$:%
\[
g\left(  t,\omega\right)  \sim\frac{a}{\left(  2Mt+R^{2}\right)  ^{\frac{3}%
{4}}}\frac{1}{\sqrt{\omega}}\ \ ,\ \ \left(  2Mt+R^{2}\right)  \rightarrow
\infty\ \ ,\ \ \ 1<<\omega<<\sqrt{2Mt+R^{2}}%
\]

\subsubsection{Weak solutions with interacting condensates. The case of infinite particle mass}  \index{weak solution} \index{interacting condensate}

We remark that the well posedness Theorem \ref{localExBounded} allows to
obtain solutions of (\ref{Z2E2a}) for initial data $g_{in}\left(
\omega\right)  $ bounded as $\omega^{-\rho}$ as $\omega\rightarrow\infty,$
with $\rho<-\frac{1}{2}.$ In particular this suggests that it is possible to
obtain global measured valued solutions of (\ref{Z2E2a}) for a large class of
nonintegrable initial data, although we have proved existence of global solutions 
only for $\rho <-1$ (cf. Remark \ref{infinitemass}) .  Although many of the results of this paper apply
only to solutions satisfying $\int g\left(  d\omega\right)  <\infty,$ and the
results in Subsection \ref{Heurist} require finite energy (i.e. $\rho>2$), it
is interesting to remark that the arguments leading to the equation
(\ref{K2E6}) can be adapted to\ cover also the case in which $g_{in}\left(
\omega\right)  \sim\frac{K}{\omega^{\rho}}$ as $\omega\rightarrow\infty$ with
$K>0,\ \frac{1}{2}<\rho<2.$ The main difference between this case and the one
studied above, is the fact that the number of particles with small values of
$\omega$ increases without limit as $t\rightarrow\infty.$ We need to modify
the ansatz (\ref{K2E1}) in order to take into account that $M=M\left(
t\right)  $ changes in time. Suppose that the characteristic length scale for
the particles aggregating (or fragmenting) their energy is $R=R\left(
t\right)  .$ If we assume that the function $g\left(  \omega\right)  $ behaves
like the power law $\omega^{-\rho}$ as $\omega\rightarrow\infty$ we should
have the following rescaling for "large" energies:%
\begin{equation}
g\left(  t,\cdot\right)  \simeq\frac{1}{R^{\rho}}G\left(  t,\frac{\cdot}%
{R}\right)  \ \ ,\ \ R=R\left(  t\right)  \label{K8E6}%
\end{equation}

On the other hand, the number of small particles will be denoted as
$M=M\left(  t\right)  .$ Notice that this amounts to approximate $g\left(
t,\cdot\right)  $ in all the regions as:%
\begin{equation}
g\left(  t,\cdot\right)  =M\left(  t\right)  \delta_{0}+\frac{1}{R^{\rho}%
}G\left(  t,\frac{\cdot}{R}\right)  \label{K8E2}%
\end{equation}

We have to distinguish two different cases. If $1<\rho$ the number of
particles of the system is finite and therefore we would have $M\left(
t\right)  \rightarrow M\left(  \infty\right)  =\int g_{in}\left(
d\omega\right)  .$ In this case the dynamics of the particles with large
energies would be given by a solution of (\ref{S2E1}) with the form
(\ref{K8E6}). The rescaling properties of (\ref{K8E6}) give the rescaling
$R=t^{\frac{1}{\rho}}.$ The distribution of particles containing most of the
energy would be given then by a selfsimilar solution of the equation
(\ref{S2E1}) with one of the functions $g$ replaced by $M\left(
\infty\right)  \delta_{0}$ and the two remaining functions $g$ replaced by the
ansatz (\ref{K8E6}). This corresponds to one self-similar solution of the
coagulation-fragmentation model (\ref{K2E6}) with "fat tails".

Suppose now that $\frac{1}{2}<\rho<1.$ We will ignore critical cases in which
logarithmic corrections can be expected. Using the approximation (\ref{K8E2})
into (\ref{S2E1}) we can derive formally one equation for the change of
$M\left(  t\right)  .$ To this end, we need to integrate (\ref{S2E1}) in
regions where $\omega$ is bounded. Using the fact that the increment of
$M\left(  t\right)  $ is due mainly to interactions of two particles described
by the distribution $G$ with one particle with $\omega$ of order one we obtain
the rescaling:%
\begin{equation}
\frac{\left[  M\right]  }{\left[  t\right]  }\simeq\frac{\left[  M\right]
}{\left[  R\right]  ^{2\rho}}\left[  R\right]  \label{K8E4}%
\end{equation}

On the other hand, the equation (\ref{S2E1}) yields, assuming that the change
of $G$ is also due mostly to the interaction between two particles described
by $G$ with one particle described by the distribution $M\left(  t\right)
\delta_{0}:$%
\begin{equation}
\frac{1}{\left[  t\right]  \left[  R\right]  ^{\rho}}\simeq\frac{\left[
M\right]  }{\left[  R\right]  ^{2\rho}} \label{K8E5}%
\end{equation}

Combining (\ref{K8E4}), (\ref{K8E5}) we obtain the following scaling laws for
$\omega$ and the particles of order one:%
\[
R=t^{\frac{1}{2\rho-1}}\ \ ,\ \ M=t^{\frac{1-\rho}{2\rho-1}}%
\]
This gives the characteristic rescaling for the self-similar solutions
describing the distribution of frequencies for large particles.

\subsubsection{Weak solutions with non interacting condensate}  \index{weak solution} \index{non interacting condensate}

The results of this paper refer mainly to the weak solutions of (\ref{S2E1}) in the
sense of Definition \ref{weakSolution}. However, as it has been explained in
Sections \ref{weaknoninteracting} and \ref{fluxes}, other types of solutions are possible.

In this Subsection we speculate about the possible long time asymptotics of
the solutions of (\ref{S2E1}) obtained in Theorem \ref{fluxSol} assuming that
they could be extended for arbitrarily long times. As we have indicated in the
previous Subsections, in the case of the weak solutions which
satisfy\ Definition \ref{weakSolution} the long time asymptotics can be
described using the coagulation-aggregation model (\ref{K2E6}). This model is
a simplification of the original model (\ref{S2E1}), which is possible due to
the fact that the most relevant process in order to determine the behaviour of
particles with large values of $\omega$ is the interaction of two particles
with one particle with $\omega=0,$ or at least $\omega$ small.

For the solutions obtained in Theorem \ref{fluxSol} the particles escaping
towards $\omega=0$ do not interact any longer with the remaining particles of
the system. We examine the long time asymptotics of solutions such that
$g_{0}\left(  \omega\right)  \sim\frac{1}{\omega^{\rho}},\ \rho>\frac{1}{2}.$
We can then look for self-similar solutions of (\ref{S2E1}) with the form:%
\[
g\left(  t,\omega\right)  =\frac{1}{t^{\alpha}}G\left(  \frac{\omega}%
{t^{\beta}}\right)  \ \ ,\ \ \xi=\frac{\omega}{t^{\beta}}%
\]
Due to the differential equation (\ref{S2E1}) we must assume that
$2\alpha-\beta=1.$ If we denote the left-hand side of (\ref{S2E1}) as
$\mathcal{Q}\left[  g\right]  $ we obtain the following equation for $G:$%
\begin{equation}
-\alpha G-\beta\xi\frac{\partial G}{\partial\xi}=\mathcal{Q}\left[  G\right]
\label{sF}%
\end{equation}

We are interested in solutions of (\ref{sF}) with the behaviour $G\left(
\xi\right)  \sim\frac{1}{\xi^{\rho}}.$ This requires $\alpha=\beta\rho.$ This
identity implies also the asymptotics $g\left(  t,\omega\right)  \sim\frac
{1}{\omega^{\rho}}$ as $\omega\rightarrow\infty$ for each fixed $t.$ Notice
that the integral term $\mathcal{Q}\left[  G\right]  $ can be expected to
behave then as $\frac{1}{\xi^{3\rho-1}},$ and since $\rho>\frac{1}{2}$ this
implies that the contribution of the integral term is negligible compared with
the left-hand side of (\ref{sF}). Combining the constraints for $\alpha
,\ \beta$ we obtain $\alpha=\frac{\rho}{2\rho-1},\ \beta=\frac{1}{2\rho-1}.$
These exponents determine the long time asymptotics of the solutions with this
initial data. It is interesting to remark that the type of solution derived
does not depend in the finiteness of the mass. If $\rho>1$ the mass of the
solutions is finite. We can readily see that in that case:%
\[
M\left(  t\right)  =\int g\left(  t,d\omega\right)  =\frac{1}{t^{\alpha-\beta
}}\int G\left(  d\omega\right)  =\frac{M_{0}}{t^{\frac{\rho-1}{2\rho-1}}%
}\rightarrow0\ \ \text{as\ \ }t\rightarrow\infty
\]

A rigorous construction of the self-similar solutions described in this and
previous Sections is not currently available. Their construction would provide
further insight in the long time asymptotics of the solutions of this problem.

\subsection{Open problems.}

In this paper we have obtained several mathematical results for the Weak
Turbulence Equation associated to the cubic nonlinear Schr\"odinger equation \index{Schr\"{o}dinger equation}
in three space dimensions. Nevertheless, there are many questions which
still remain unsolved as well as some questions which arise naturally from
the results in this paper. We list here some open problems suggested by the results of this paper:

\begin{itemize}
\item Uniqueness of weak solutions defined in Definitions \ref{weakSolution}
, \ref{weakSolutionNI} or more generally in Definition \ref{weakSolutionGen}
(cf. Sections \ref{weakinteracting},  \ref{weaknoninteracting} and \ref{fluxes} ).

\item Existence of weak solutions in the sense of Definitions \ref%
{weakSolutionNI} and \ref{weakSolutionGen} (cf. Sections \ref{weaknoninteracting} and \ref{fluxes} ).

\item Our construction of weak solutions in the sense of Definition \ref%
{weakSolution} takes as a starting point the solution of the weak turbulence\index{weak turbulence} 
equation with a regularized kernel $\Phi _{\sigma }$ which has the property
that cuts the singular terms of $\Phi $ for small values of $\omega .$ For
this regularized kernels the solutions are global and it is possible to
consider their limit as $\sigma \rightarrow 0.$ It is a natural question to
determine if using different approximating kernels $\Phi _{\sigma }$ in
order to approximate the kernel $\Phi $ it is possible to derive limit
solutions which are weak solutions of (\ref{S2E1}), (\ref{S2E1a}) in a sense different from
Definition \ref{weakSolution}. In particular this includes to obtain (if
they exist) approximating kernels $\Phi _{\sigma }$ for which the
corresponding solutions could yield as limit weak solutions in the sense of
Definitions \ref{weakSolutionNI} and \ref{weakSolutionGen} (cf. Sections
\ref{weaknoninteracting} and \ref{fluxes} ).

\item Smoothing effects. The results in the physical literature suggest that
weak solutions in the sense of Definition \ref{weakSolution} yield an
asymptotics $g(t, \omega ) \sim a(t)\, \omega ^{-1/2}$ as $\omega \rightarrow 0^{+}$ if there is a condensate. \index{condensate}On the
other hand, we can expect similarly the behaviour $g(t,\omega )
\sim a (t)\, \omega ^{-\frac{2}{3}}$ as $\omega \rightarrow
0^{+}$ for the weak solutions in the sense of Definition \ref{weakSolutionNI}%
. It is unlikely that this asymptotics holds pointwise. Most likely the
asymptotics takes place in some kind of averaged sense or in some suitable
weak topology. These issues are closely related to the stability analysis of
the Kolmogorov-Zakharov \index{Kolmogorov-Zakharov}and Rayleigh-Jeans \index{Rayleigh Jeans}solutions (cf. Sections \ref{fluxes} ).

\item Prove of disprove global existence of weak solutions of the Wave
Turbulence Equation without finite mass (cf. Section \ref{mildConst}).

\item Prove or disprove rigorously the asymptotic results conjectured for the transfer of mass and energy
in Section \ref{Transfer2}.
More generally, to derive information about the asymptotics of the solutions
of weak turbulence \index{weak turbulence}theory.

\item We have obtained a family of solutions yielding "pulsating" behaviour. \index{pulsating}
We have proved also that the solutions of weak turbulence \index{weak turbulence}with finite mass
and that do not develop a condensate \index{condensate}in finite time have the property that $g
$ can be approximated by a Dirac mass at a positive distance of the origin
during most of the times. Our result does not rule out the possibility of
the Dirac mass moving continuously towards $\omega =0,$ although in the case
of the pulsating solutions constructed in Chapter 4 this is not the case.
Is it possible to prove that the pulsating behavior occurs for any solution
of weak turbulence \index{weak turbulence}that does not develop a condensate \index{condensate}in finite time? (Cf. Remark \ref{setB}). \index{pulsating}
\end{itemize}

There exists some long standing open questions, although not completely precisely formulated
 in some cases. 

\begin{itemize}
\item Derive a precise asymptotics of the solutions near blow-up (and condensate \index{condensate}
formation). Numerical simulations in the Physical literature suggest
self-similar behaviour (cf.  ( \cite{JPR}, \cite{LLPR}, \cite{SK1, SK2}).

\item Establish the precise mathematical connection between the cubic
nonlinear Schr\"odinger equation \index{Schr\"{o}dinger equation}and the equation 
for weak turbulence. \index{weak turbulence}This
 could yield a very large class of
problems ranging from the approximation of particular solutions of NLS by
means of solutions of weak turbulence \index{weak turbulence}equations, to the precise statistical
conditions which ensure the validity of weak turbulence theory. Many issues
related to this problem could be complicated due to the presence of the
condensate.\index{condensate}

\item The mathematical theory of the Cauchy problem for non isotropic weak turbulence \index{weak turbulence}is a widely open area. One of the questions considered in the physical literature is the dynamics  of non isotropic perturbations of the isotropic KZ solutions (cf. in particular \cite {BZ}, \cite{Zbook}). 
\end{itemize}

\section{Auxiliary results.}

In this Chapter we give several auxiliary results which can be proved by means of minor adaptions of some of the arguments used in \cite{EV1} in the proof of blow up for the Nordheim equation. Therefore, we will just state here the results used
emphasizing the points where differences with \cite{EV1} arise. The results of this Chapter  are used in the proof of 
 Theorems \ref{StatIsot} and \ref{AsympOsc}.

\begin{proposition}
Suppose that $g\in C\left(  \left[  0,\infty\right)  :\mathcal{M}_{+}\left(
\left[  0,\infty\right)  :\left(  1+\omega\right)  ^{\rho}\right)  \right)
$\ is a weak solution of (\ref{S2E1}) in the sense of Definition
\ref{weakSolution}. There exists a positive constant $B<\infty,$ independent
on $g$ such that, for any $T>0$ and $R\in\left(  0,1\right)  $ we have:
\begin{equation}
\int_{0}^{T}dt\int_{\left[  0,\frac{R}{2}\right]  ^{3}}\left[  \prod_{m=1}%
^{3}\,g_{m}d\omega_{m}\right]  \frac{\left(  \omega_{0}\right)  ^{\frac{3}{2}%
}}{\left(  \omega_{+}\right)  ^{\frac{3}{2}}}\left(  \frac{\omega_{0}%
-\omega_{-}}{\omega_{0}}\right)  ^{2}\leq BR\int gd\epsilon\label{T1E6}%
\end{equation}
where the functions $\omega_{-},\ \omega_{0},\ \omega_{+}$ are as in
Definition \ref{aux}.
\end{proposition}

\begin{proof}
It is similar to the Proof of Proposition 5.1 of \cite{EV1}. The main idea is
to use in (\ref{Z2E1}) the test function $\varphi\left(  \omega\right)
=\psi\left(  \frac{\omega}{R}\right)  \ \ ,\ \ \ R>0,\ \ \omega>0$ with
$\psi\left(  s\right)  =s^{\theta}\ $for$\ 0<s<1,\ \psi\left(  s\right)  =1,$
$s\geq1\ ,\ \ 0<\theta<1.$ The monotonicity property described in Subsection
\ref{Mon} then yields several inequalities, which can be transformed in
(\ref{T1E6}) after some computations. The only difference with the argument in
\cite{EV1} is that in that paper, an estimate for some additional quadratic
terms, analogous to the ones appearing in the classical Boltzmann equation,
must be obtained, and this results in an additional term on the right-hand
side of (\ref{T1E6}). These terms are not present in (\ref{S2E1}) and this
results in the simpler estimate (\ref{T1E6}), which contains only one term on
the right-hand side, due to the contribution of the initial value $g_{in}.$
\end{proof}

It is now possible to reformulate the estimate (\ref{T1E6}) in a form that
makes clearer the fact that this estimate basically allows to control the mass
associated to the product measure $\prod_{m=1}^{3}\,g_{m}d\omega_{m}$
contained outside the diagonal set $\left\{  \left(  \omega_{1},\omega_{2}%
,\omega_{3}\right)  :\omega_{1}=\omega_{2}=\omega_{3}\right\}  .$ We define
the following family of sets:%

\begin{eqnarray}
&&\mathcal{S}_{R,\rho}=\left\{  \left(  \omega_{1},\omega_{2},\omega_{3}\right)
\in\left[  0,R\right]  ^{3}:\left\vert \omega_{0}-\omega_{-}\right\vert
>\rho\omega_{0}\right\},\ 0<R\leq1,\  0<\rho<1\ . \label{B5b}%
\end{eqnarray}

We then have the following result:

\begin{lemma}
\label{estProd}Suppose that $g\in L_{+}^{\infty}\left(  \left[  0,T\right]
;\mathcal{M}_{+}\left(  \left[  0,1\right]  \right)  \right)  ,$ satisfies
(\ref{T1E6}) for any $0\leq R\leq1$ and $T>0.$ Suppose also that
$\int_{\left\{  0\right\}  }g\left(  \omega,t\right)  d\epsilon=0$ for any
$t\in\left[  0,T\right]  .$ Let $0<\rho<1$ and $\mathcal{S}_{R,\rho}$ as in
(\ref{B5b}). Then, for any $T>0$ we have:%
\begin{equation}
\int_{0}^{T}dt\int_{\mathcal{S}_{R,\rho}}\left[  \prod_{m=1}^{3}\,g_{m}%
d\omega_{m}\right]  \leq\frac{2Bb^{\frac{7}{2}}MR}{\rho^{2}\left(  \sqrt
{b}-1\right)  ^{2}},\ \ R\in\left[  0,\frac{1}{2}\right]  \ \label{T2E1}%
\end{equation}
with $b=\frac{1}{1-\rho}$ and $B$ as in (\ref{T1E6}).
\end{lemma}

\begin{proof}
It is an adaptation of the proof of Lemma 5.4 of \cite{EV1}.
\end{proof}

We recall now a Key Measure Theory result that has been used in \cite{EV1}. In
order to formulate it we need some additional notation. Given $b>1,$ we define
a sequence of intervals $\left\{  \mathcal{I}_{k}\right\}  _{k=0}^{\infty}$
contained in the interval $\left[  0,1\right]  $ by means of:%
\begin{equation}
\mathcal{I}_{k}\left(  b\right)  =b^{-k}\left(  \frac{1}{b},1\right]
\ \ ,\ \ k=0,1,2,...\ \ ,\ \ b=1+a>1 \label{B3}%
\end{equation}

Notice that $\bigcup_{k=0}^{\infty}\mathcal{I}_{k}\left(  b\right)  =\left(
0,1\right]  ,\ \mathcal{I}_{k}\left(  b\right)  \cap\mathcal{I}_{j}\left(
b\right)  =\emptyset$ if $k\neq j.$

We need to define also some ``extended" intervals:%
\begin{equation}
\mathcal{I}_{k}^{\left(  E\right)  }\left(  b\right)  =\mathcal{I}%
_{k-1}\left(  b\right)  \cup\mathcal{I}_{k}\left(  b\right)  \cup
\mathcal{I}_{k+1}\left(  b\right)  \ \ ,\ \ k=0,1,2,...\ \ \ \label{B5}%
\end{equation}
where, by convenience, we assume that $\mathcal{I}_{-1}\left(  b\right)
=\emptyset.$

We will write $\mathcal{I}_{k}=\mathcal{I}_{k}\left(  b\right)  ,\ \mathcal{I}%
_{k}^{\left(  E\right)  }=\mathcal{I}_{k}^{\left(  E\right)  }\left(
b\right)  $ if the dependence of the intervals in $b$ is clear in the argument.

We also define for further reference a family $\mathcal{P}_{b}$ of unions of
elements of the family $\left\{  \mathcal{I}_{k}\left(  b\right)  \right\}  .$
We define:%
\begin{equation}
\mathcal{P}_{b}\mathcal{=}\left\{  A\subset\left[  0,1\right]  :A=\bigcup
_{j=1}^{\infty}\mathcal{I}_{k_{j}}\left(  b\right)  \text{ for some sequence
}\left\{  k_{j}\right\}  \subset\left\{  1,2,...\right\}  \right\}
\label{B5c}%
\end{equation}

Given $A\in\mathcal{P}_{b}$ we can define an extended set $A^{\left(
E\right)  }$ as follows. Suppose that $A=\bigcup_{j=1}^{\infty}\mathcal{I}%
_{k_{j}}\left(  b\right)  .$ We then define:%
\begin{equation}
A^{\left(  E\right)  }=\bigcup_{j=1}^{\infty}\mathcal{I}_{k_{j}}^{\left(
E\right)  }\left(  b\right)  \label{B5d}%
\end{equation}

Notice that given a measure $g\in\mathcal{M}_{+}\left(  \left[  0,1\right]
\right)  ,$ such that $\int_{\left\{  0\right\}  }g\left(  \omega\right)
d\epsilon=0,$ we have:%
\begin{equation}
\int_{\left[  0,1\right]  }g\left(  \omega\right)  d\epsilon=\sum
_{k=0}^{\infty}\int_{\mathcal{I}_{k}\left(  b\right)  }g\left(  \omega\right)
d\epsilon\label{B4}%
\end{equation}

We need also a rescaled version of the sets $\left\{  \mathcal{I}_{k}\left(
b\right)  \right\}  ,\ \left\{  \mathcal{I}_{k}^{\left(  E\right)  }\left(
b\right)  \right\}  ,\ \mathcal{P}_{b}.$ Given $R\in\left(  0,1\right]  $ and
$b>1$ we define two families of intervals $\left\{  \mathcal{I}_{k}\left(
b,R\right)  \right\}  ,\ \left\{  \mathcal{I}_{k}^{\left(  E\right)  }\left(
b,R\right)  \right\}  $ by means of:%
\begin{equation}
\mathcal{I}_{k}\left(  b,R\right)  =R\mathcal{I}_{k}\left(  b\right)
\ \ ,\ \ \mathcal{I}_{k}^{\left(  E\right)  }\left(  b,R\right)
=R\mathcal{I}_{k}^{\left(  E\right)  }\left(  b\right)
\ \ ,\ \ k=0,1,2,...\ \label{Z1E1}%
\end{equation}
with $\left\{  \mathcal{I}_{k}\left(  b\right)  \right\}  ,\ \left\{
\mathcal{I}_{k}^{\left(  E\right)  }\left(  b\right)  \right\}  $ as in
(\ref{B3}), (\ref{B5}). We define also a class of sets $\mathcal{P}_{b}\left(
R\right)  $ as follows:%
\begin{equation}
\mathcal{P}_{b}\left(  R\right)  =\left\{  A\subset\left[  0,R\right]
:A=RB,\ \ B\in\mathcal{P}_{b}\right\}  \ \label{Z1E2}%
\end{equation}
where $\mathcal{P}_{b}$ is as in (\ref{B5d}). We can also define the concept
of extended sets. Given $A\in\mathcal{P}_{b}\left(  R\right)  ,$ with the form
$A=RB,\ B\in\mathcal{P}_{b}$ we define:%
\begin{equation}
A^{\left(  E\right)  }=RB^{\left(  E\right)  } \label{Z1E3}%
\end{equation}

\bigskip The following result has been proved in \cite{EV1}.

\begin{lemma}
\label{altresc} (Lemma 6.3 of \cite{EV1}). Suppose that $b>1,\ 0<R\leq1.$ We
define intervals $\left\{  \mathcal{I}_{k}\left(  b,R\right)  \right\}
,\ \left\{  \mathcal{I}_{k}^{\left(  E\right)  }\left(  b,R\right)  \right\}
$ as in (\ref{Z1E1}). Let $\mathcal{P}_{b}\left(  R\right)  $ as in
(\ref{Z1E2}) and $A^{\left(  E\right)  }$ as in (\ref{Z1E3}) for
$A\in\mathcal{P}_{b}\left(  R\right)  $. Given $0<\delta<\frac{2}{3},$ we
define $\eta=\min\left\{  \left(  \frac{1}{3}-\frac{\delta}{2}\right)
,\frac{\delta}{6}\right\}  >0.$ Then, for any $g\in\mathcal{M}^{+}\left[
0,R\right]  $ satisfying $\int_{\left\{  0\right\}  }g\left(  d\epsilon
\right)  =0,$ at least one of the following statements is satisfied:

(i) Either there exist an interval $\mathcal{I}_{k}\left(  b,R\right)  $ such
that:%
\begin{equation}
\int_{\mathcal{I}_{k}^{\left(  E\right)  }\left(  b,R\right)  }g\left(
d\omega\right)  \geq\left(  1-\delta\right)  \int_{\left[  0,R\right]
}g\left(  d\omega\right)  \ \ , \label{Z1E4}%
\end{equation}

(ii) or, either there exist two sets $\mathcal{U}_{1},\mathcal{U}_{2}%
\in\mathcal{P}_{b}\left(  R\right)  $ such that $\mathcal{U}_{2}%
\cap\mathcal{U}_{1}^{\left(  E\right)  }=\emptyset$ and:%
\begin{equation}
\min\left\{  \int_{\mathcal{U}_{1}}g\left(  d\omega\right)  ,\int
_{\mathcal{U}_{2}}g\left(  d\omega\right)  \right\}  \geq\eta\int_{\left[
0,R\right]  }g\left(  d\omega\right)  .\ \label{Z1E4a}%
\end{equation}
Moreover, in the case (ii) the set $\mathcal{U}_{1}$ can be written in the
form:%
\begin{equation}
\mathcal{U}_{1}=\bigcup_{j=1}^{L}\mathcal{I}_{k_{j}}\left(  b,R\right)
\ \label{Z1E4b}%
\end{equation}
for some sequence $\left\{  k_{j}\right\}  $ and some finite $L.$ We have:%
\begin{equation}
\mathcal{I}_{k_{m}}\left(  b,R\right)  \cap\left(  \bigcup_{j=1}%
^{m-1}\mathcal{I}_{k_{j}}^{\left(  E\right)  }\left(  b,R\right)  \right)
=\emptyset\ \ ,\ \ m=2,3,...L,\ \ \label{Z1E4c}%
\end{equation}
and also:%
\begin{equation}
\sum_{j=1}^{L}\left(  \int_{\mathcal{I}_{k_{j}}\left(  b,R\right)  }\!\!\!g\left(
d\omega\right)  \right)  ^{2}\leq\left(  \int_{\mathcal{I}_{k_{1}}\left(
b,R\right)  }\!\!\!g\left(  d\omega\right)  \right)  ^{2}+\sum_{j=2}^{L}%
\int_{\mathcal{I}_{k_{1}}\left(  b,R\right)  }\!\!\!g\left(  d\omega\right)
\int_{\mathcal{I}_{k_{j}}\left(  b,R\right)  }\!\!\!g\left(  d\omega\right),\ \label{Z1E4d}%
\end{equation}%
\begin{equation}
\int_{\mathcal{I}_{k_{1}}\left(  b,R\right)  }\!\!\!g\left(  d\omega\right)
<\left(  1-\delta\right)  \int_{\left[  0,R\right]  }\!\!\!g\left(  d\omega\right)
. \label{Z1E5d}%
\end{equation}
\bigskip
\end{lemma}

This Lemma basically states that either the measure $g$ is concentrated in one
of the intervals $\mathcal{I}_{k}^{\left(  E\right)  }\left(  b,R\right)  ,$
or alternatively its mass is spread among some sets  ``sufficiently separated".
Using this Lemma we can then obtain the following result, which has been
proved also in \cite{EV1}.

\begin{lemma}
\label{Compl}Let $0<\delta<\frac{2}{3},$ $0<\rho<1.$ For any $R\in\left(
0,1\right)  $ we define $\mathcal{S}_{R,\rho}$ as in (\ref{B5b}). Let us
assume also that $b=\frac{1}{\left(  1-\rho\right)  }.$ There exists $\nu
=\nu\left(  \delta\right)  >0$ independent on $R$ and $\rho$ such that, for
any $g\in\mathcal{M}^{+}\left[  0,R\right]  $ satisfying $\int_{\left\{
0\right\}  }g\left(  d\omega\right)  =0$ if the alternative (ii) in Lemma
\ref{altresc} takes place we have:%
\begin{equation}
\int_{\mathcal{S}_{R,\rho}}\left[  \prod_{m=1}^{3}\,g_{m}\left(  d\omega
_{m}\right)  \right]  \geq\nu\left(  \int_{\left(  0,R\right]  }g\left(
d\omega\right)  \right)  ^{3}>0.\ \label{H1bis}%
\end{equation}

\end{lemma}

\bibliographystyle{amsalpha}

\end{document}